\documentclass[12pt, a4paper,oldfontcommands]{memoir}
\usepackage[english, science, hyperref, dropcaps, submissionstatement]{ku-frontpage}

\makepagestyle{myplain}
\makeevenfoot{myplain}{}{}{}
\makeevenhead{myplain}{\thepage}{}{\leftmark}
\makeoddfoot{myplain}{}{}{}
\makeoddhead{myplain}{\leftmark}{}{\thepage}
\nouppercaseheads

\usepackage[T1]{fontenc}    
\usepackage{times}
\usepackage{hyperref}       
\setlrmarginsandblock{4cm}{3cm}{*}
\setulmarginsandblock{3.5cm}{3.5cm}{*}
\checkandfixthelayout
\usepackage{booktabs}       
\usepackage{amsfonts}       
\usepackage{microtype}      
\usepackage{graphicx}
\usepackage{amsmath} 
\usepackage{amsthm}
\theoremstyle{plain}
\newtheorem{theorem}{Theorem}[section]
\newtheorem{corollary}[theorem]{Corollary}
\newtheorem{lemma}[theorem]{Lemma}
\newtheorem{remark}[theorem]{Remark}
\newtheorem{definition}[theorem]{Definition}
\newtheorem{example}[theorem]{Example}
\DeclareMathOperator{\Tr}{Tr}
\DeclareMathOperator{\Prob}{Prob}

\DeclareMathOperator{\id}{id}
\DeclareMathOperator{\Loc}{Loc}
\DeclareMathOperator{\Spec}{Spec}

\DeclareMathOperator{\AVP}{AVP}
\DeclareMathOperator{\Dec}{\Gamma^{\mathcal{D}}}
\DeclareMathOperator{\DecI}{Dec}
\DeclareMathOperator{\Enc}{\Gamma^{\mathcal{E}}}
\DeclareMathOperator{\EncI}{Enc}

\DeclareMathOperator{\EC}{EC}
\DeclareMathOperator{\Dist}{Dist}

\DeclareMathOperator{\tr}{tr}
\DeclareMathOperator{\CNOT}{CNOT}
\DeclareMathOperator{\SWAP}{SWAP}
\DeclareMathOperator{\gpi}{gpi}
\DeclareMathOperator{\gpii}{gpi2}
\DeclareMathOperator{\ms}{ms}
\usepackage[utf8]{inputenc} 
\usepackage{graphicx} 
\usepackage{braket} 
\usepackage{caption}
\usepackage{physics}
\usepackage{mathtools}
\usepackage{appendix}
\usepackage{subcaption} 
\usepackage{mathrsfs}
\usepackage{bbm} 
\usepackage{float} 
\usepackage{wrapfig}
\usepackage{xcolor} 
\usepackage{quantikz}
\usetikzlibrary{shapes}
\usetikzlibrary{positioning}
\usepackage{listings}
\usepackage{colortbl}

 \assignment{PhD thesis}
 \author{Paula Belzig}
 \title{Quantum communication \\ and fault-tolerance}

 \subtitle{  }
 \advisor{Advisor: Matthias Christandl}

\kupdfsetup{Quantum communication and fault-tolerance}{}{Paula Belzig}

\begin{document}
\pagestyle{myplain}
\setlength{\parskip}{7pt}
\setlength{\parindent}{0pt}

\DeclareRobustCommand
  \Compactcdots{\mathinner{\cdotp\mkern-1mu\cdotp\mkern-1mu\cdotp}}

   
\begingroup
  \fontencoding{T1}\fontfamily{cmr}\selectfont
  \maketitle
\endgroup
\null\thispagestyle{empty}
\\
\hspace*{-1cm}\textbf{Paula Belzig}\\
\hspace*{-1cm}QMATH\\
\hspace*{-1cm}Department of Mathematical Sciences\\
\hspace*{-1cm}University of Copenhagen\\
\hspace*{-1cm}\href{mailto:pb@math.ku.dk}{pb@math.ku.dk}\\
\hspace*{-1cm}\href{mailto:paulabelzig@gmail.com}{paulabelzig@gmail.com}
\vspace{1cm}
\\
\hspace*{-1cm}\large{\textbf{Advisor:}}\normalsize\\
\hspace*{-1cm}\textbf{Matthias Christandl}, University of Copenhagen, Copenhagen, Denmark\\
\vspace{1cm}
\\
\hspace*{-1cm}\large{\textbf{Assessment committee:}}\normalsize\\
\hspace*{-1cm}\textbf{Laura Man\v{c}inska} (Chair), University of Copenhagen, Copenhagen, Denmark\\
\hspace*{-1cm}\textbf{Omar Fawzi}, École Normale Supérieure de Lyon, Lyon, France\\
\hspace*{-1cm}\textbf{Mario Berta}, RWTH Aachen University, Aachen, Germany\\
\vspace{5cm}
\\
\hspace*{-1.12cm}\begin{tabular}{ll}
    Date of submission: &\qquad July 31, 2023\\
    Edited for print: &\qquad \today\\
    Date of defense: &\qquad  September 12, 2023\\
\end{tabular}\\
\vspace{2.5cm}
\\
\hspace*{-1cm} \footnotesize\copyright~ \normalsize 2023 by the author \\
\hspace*{-1cm}ISBN 978-87-7125-216-3

\chapter*{Abstract}
\addcontentsline{toc}{chapter}{Abstract}
\thispagestyle{empty}


In this thesis, we are interested in the limits of quantum communication with and without entanglement, and with and without noise assumptions on the communication setup. When a sender and a receiver are connected by a communication line that is governed by noise which is modelled by a quantum channel, they hope to design a coding scheme, i.e. messages and message decoders, in such a way that they are robust to this noise. The amount of message bits per channel use is called the achievable rate of the coding scheme, and the maximal achievable rate for a given quantum channel is called capacity of the quantum channel. Here, we are interested in coding schemes and capacities under various assumptions, in particular in the case where the sender and the receiver share quantum entanglement, which turns out to be the most natural generalization of the classical communication analogue that lies at the basis of many of our modern technologies. 

In order to communicate with a quantum computer, the encoding and decoding of a message for a quantum channel on a quantum computer is implemented by a sequence of gates from a finite gate set. For classical channels, these gates are assumed to be noise-free at the timescales relevant for communication, but this assumption is likely not justified for quantum circuits. In the first part of this thesis, we show that the encoding and decoding for entanglement-assisted communication can be protected against noise by the use of techniques from quantum fault-tolerance and error correction without significantly decreasing the capacity. In particular, we recover the usual, noiseless case as the probability of errors in the circuit approaches zero.

In the second part, we study a proposal for fault-tolerance on trapped ion quantum devices in order to gauge their potential for implementing communication protocols and to assess our assumptions about noise models, with a particular focus on the preparation of maximally entangled states.

In the final part of the thesis, we revisit the comparison between quantum communication setups where the sender and the receiver have access to entanglement and setups where they do not. It is well-known that the communication rates are improved when entanglement is present. Here, we show that this improvement can be bounded for a large class of finite dimensional channels, taking a step towards resolving a conjecture from one of the earliest papers on entanglement-assisted communication.

\chapter*{Resumé}
\addcontentsline{toc}{chapter}{Resumé}
\thispagestyle{empty}

I denne afhandling er vi interesseret i grænserne for kvantekommunikation med og uden sammenfiltring, og med og uden støjantagelser på kommunikationsopsætningen. Når en afsender og en modtager er forbundet med en kommunikationslinje, der er styret af støj, som er modelleret af en kvantekanal, håber de at udvikle et kodningsskema, dvs. beskeder og beskedafkodere, på en sådan måde, at de er robuste over for denne støj. Mængden af meddelelsesbits pr. kanalbrug kaldes kodningsskemaets opnåelige rate, og den maksimalt opnåelige rate kaldes kvantekanalens kapacitet. Her er vi interesseret i kodningsskemaer og kapaciteter under forskellige forudsætninger, især i det tilfælde, hvor afsender og modtager deler sammenfiltring, hvilket viser sig at være den mest naturlige generalisering af den klassiske kommunikations-analog, der danner grundlaget for mange af vores moderne teknologier.

For at kommunikere med en kvantecomputer implementeres kodningen og afkodningen af en besked til en kvantekanal på en kvantecomputer ved hjælp af en sekvens af porte fra et begrænset portmængde. For klassiske kanaler antages disse porte at være støjfrie på de tidsskalaer, der er relevante for kommunikation, men denne antagelse er sandsynligvis ikke berettiget for kvantekredsløb. I afhandlingens første del viser vi, at indkodning og afkodning til sammenfiltringsassisteret kommunikation kan beskyttes mod støj ved brug af teknikker fra kvantefejltolerance og fejlkorrektion uden at reducere kapaciteten væsentligt. Især genopretter vi det støjfri tilfælde, da sandsynligheden for fejl i kredsløbet nærmer sig nul.

I den anden del studerer vi et forslag til fejltolerance på ion-baserede kvanteenheder for at måle deres potentiale for implementering af kommunikationsprotokoller og for at vurdere vores antagelser om støjmodeller, med særligt fokus på maksimalt sammenfiltrede tilstande.

I afhandlingens sidste del ser vi igen på sammenligningen mellem kvantekommunikationsopsætninger, hvor afsender og modtager har adgang til sammenfiltring og opsætninger hvor de ikke har. Det er velkendt, at kommunikationsrater er højere, når der er sammenfiltring. Her viser vi, at denne forbedring kan afgrænses til en stor klasse af endeligt-dimensionelle kanaler, og tager et skridt i retning af at løse en formodning fra en af de tidligste artikler om sammenfiltringsassisteret kommunikation.

\chapter*{Contributions and structure}
\addcontentsline{toc}{chapter}{Contributions and structure}
\thispagestyle{empty}


\begin{itemize}
    \item Chapter~\ref{chapter-fteacap} is based on joint work with Matthias Christandl and Alexander Müller-Hermes. It has appeared on a preprint server under the identification number arXiv:2210.02939 \cite{BCMH22},
and a shortened version features in the proceedings of the 2023 International Symposium on Information Theory \cite{BCMH23SHORT}.
\item 
Chapter~\ref{chapter-aws} is based on an ongoing work with Matthias Christandl.
\item 
Chapter~\ref{chapter-divrad} is based on discussions with Matthias Christandl, Alexander Müller-Hermes, Marco Tomamichel, Lasse H. Wolff, Bergfinnur Durhuus and Andreas Winter, and the results are in preparation for publication.
\item Some objects and concepts in Chapter~\ref{chapter-preliminaries} are also introduced in our work in \cite{BCMH22}, which is reproduced in Chapter~\ref{chapter-fteacap} for self-containedness of the manuscript. Due to this, the introduced notation, theorems and definitions may have overlap.
\end{itemize}



\newpage
\tableofcontents
\thispagestyle{empty}







\chapter{Introduction}
\label{chapter-preliminaries}

The mathematical description of a world with quantum effects gives rise to a theory of information and computation that fundamentally differs from the classical notion which facilitates and simplifies our everyday lives. Developing and improving our understanding of the fundamental limitations as well as the potential of a quantum world has been the goal of theoretical and experimental research for decades, and some important aspects of this theory are the subject of this introductory chapter. In the later chapters of this thesis, we hope to contribute to this effort.


Section~\ref{sec-intro-comm} introduces the fundamental mathematical and information theoretical concepts that will be utilized in the subsequent chapters of this thesis.

Section~\ref{sec-ft-intro} gives an overview over quantum fault-tolerance, highlighting important aspects and state-of-the-art techniques for useful computation with quantum systems.

In Section~\ref{sec-ft-and-comm}, we discuss the possibilites and pitfalls of combining the techniques from quantum Shannon theory and quantum fault-tolerance in order to study communication in the presence of noise.

\section{Quantum information theory}
\label{sec-intro-comm}

 We will present a brief overview of the mathematical representation of quantum states and quantum channels, before delving into how the unique properties of quantum systems can be utilized for manipulating and encoding information. For comprehensive introductions to quantum information theory, we refer to \cite{NC00,Watrous,Wilde13}.


\subsection{Quantum states and quantum channels}

Let $\mathcal{M}_d$ denote the matrix algebra of complex $d\times d$-matrices, and let $\mathcal{V}_d$ denote a complex $d$-dimensional vector space, which we will refer to as a Hilbert space of $d$-level quantum state vectors. For all of this thesis, unless stated otherwise, we will be considering finite dimensional spaces $d<\infty$.

In Dirac notation, we denote an element of the Hilbert space by $\ket{\phi}\in\mathcal{V}_d$. On this vector space, we denote the inner product by $\braket{\phi}{\psi}\in \mathbbm{C}$, and we define the outer product: $\ketbra{\phi}{\psi} \in \mathcal{M}_d$. We assign the space $\mathcal{V}_d$ a standard basis
$\{\ket{x}\}_{x=0,1,...,d-1}$, which is often called the \emph{computational basis} in quantum information theory.

A matrix $A\in \mathcal{M}_d$ is \emph{unitary} if $A^{\dagger}A=\mathbbm{1}_d$, where $\mathbbm{1}_d$ denotes the identity matrix in $\mathcal{M}_d$.
A matrix $B\in \mathcal{M}_d$ is \emph{Hermitian} if $B^{\dagger}=B$.
A matrix $C\in \mathcal{M}_d$ is \emph{positive semidefinite}, which we denote by $C\geq 0$, if it is Hermitian and has only non-negative eigenvalues. A matrix is \emph{positive} if it is Hermitian and has only positive eigenvalues.

A quantum state $\rho\in\mathcal{M}_d$ is a positive semidefinite Hermitian matrix with $\trace(\rho)=1$. These conditions ensure that the $d$ eigenvalues of a quantum state $\rho\in\mathcal{M}_d$ are non-negative and add to $1$, meaning that the set of eigenvalues forms a probability distribution.


The set of quantum states is convex, and therefore, any state in it can be written as a convex combination of at most $d^2$ extremal points of the set. The extremal points of the set of quantum states are states of rank $1$, which we refer to as \emph{pure states}. The eigenvalues of pure states correspond to deterministic probability distributions.
States that are not pure are referred to as \emph{mixed states}.
One state of particular importance is the \emph{maximally mixed state} $\frac{\mathbbm{1}}{d}\in\mathcal{M}_d$ with a uniform distribution of eigenvalues.

In quantum many-body theory, quantum systems are governed by interactions with a larger environment, and often understood as pure states on a larger Hilbert space.
Composite Hilbert spaces are mathematically described by the tensor product, where $\rho_{AB}\in \mathcal{M}_{d_A} \otimes \mathcal{M}_{d_B}$. The subsystem $B$'s individual state is governed by the partial trace 
$\trace_A= \trace \otimes \id_{d_B}$, and analogous for subsystem $A$. We will often denote the subsystem's states by $\tr_A (\rho_{AB}) =\rho_B$ and $\tr_B (\rho_{AB}) =\rho_A$.

For any quantum state $\rho\in \mathcal{M}_{d_A}$, it is possible to construct a pure state $\ketbra{\phi}\in\mathcal{M}_{d_A}\otimes \mathcal{M}_{d_B}$ such that $\rho =\trace_B (\ketbra{\phi})$, and $\ketbra{\phi}\in\mathcal{M}_{d_A}\otimes \mathcal{M}_{d_B}$ is referred to as a \emph{purification} of $\rho$.

A quantum state $\rho_{AB}\in  \mathcal{M}_{d_A} \otimes \mathcal{M}_{d_B}$ is called a \emph{product state} if $\rho_{AB}=\rho_A \otimes \rho_B$. A quantum state is called \emph{separable} if is a probabilistic mixture of product states, $ \rho_{AB}=\sum_x p_X(x) \rho^x_A \otimes \rho^x_B$ where $p_X(x)$ is a probability distribution. The set of product states is thus a subset of the set of separable states.

If a quantum state is not separable, then it is referred to as \emph{entangled}. Entanglement of two quantum systems is a quantum effect that plays a prominent role in quantum information processing, and will lie at the center of many of the questions of this thesis. We will also frequently refer to a particular example of an entangled state, the \emph{maximally entangled state} of two qubits $\phi_+=\ketbra{\phi_+} \in \mathcal{M}_2\otimes \mathcal{M}_2$, where $\ket{\phi_+}=\frac{1}{\sqrt{2}} (\ket{00}+\ket{11})$. 

We refer to a quantum state as a classical state if it is diagonal in the computational basis, i.e. $\rho=\sum_x p_X(x) \ketbra{x}{x}$ with a probability distribution $p_X(x)$. We refer to a quantum state as classical-quantum, or a \emph{cq-state}, if $\rho_{AB}=\sum_x p_X(x) \ketbra{x}{x} \otimes \rho_x$ for some probability distribution $p_X(x)$ and quantum states $\rho_x\in\mathcal{M}_{d_B}$.

The transformations that a quantum system can undergo are mathematically described by quantum operations, which are linear, completely positive and trace non-increasing maps between the spaces of quantum states. A map $T :\mathcal{M}_{d_A}\rightarrow \mathcal{M}_{d_B}$ is \emph{trace non-increasing} if $\trace(T(\rho))\leq\trace(\rho)\ \forall \rho \in \mathcal{M}_{d_A}$ and \emph{trace preserving} if $\trace(T(\rho))=\trace(\rho)\ \forall \rho \in \mathcal{M}_{d_A}$. A map $T$ is \emph{completely positive} if $\id_{d_R} \otimes T$ is a positive map $\forall d_R$, which means that it preserves positive semi-definiteness. 

We refer to the subset of quantum operations that is strictly trace preserving as quantum channels. 
The conditions of complete positivity and trace preservation ensure that quantum channels do indeed transform quantum states into quantum states.

 

\subsection{Measures of distance and information}
\label{sec-measures-of-distance-and-information}

\subsubsection{Measures of distance}

Here, we will go over some of the many interesting and useful matrix norms and quantifiers of similarity between quantum states and channels. We begin by defining Schatten matrix-norms.

\begin{definition}
For any real number $p\geq 1$ and any matrix $A\in \mathcal{M}_d$, the Schatten $p$-norm of $A$ is given by
\[
\|A\|_{p} := \left(\Tr\left((A^{\dagger}A)^{\frac{p}{2}}\right)\right)^{\frac{1}{p}},
\]
\end{definition}

For any matrix $A$ and any $p\in[0,\infty)$, the Schatten p-norm is non-negative ($\|A\|_{p} \geq 0$) and homogeneous ($\|\alpha A\|_{p} = |\alpha| \|A\|_{p}$ for any scalar $\alpha$).
For any matrices $A$ and $B$, the Schatten p-norm fulfills the triangle inequality \[\|A + B\|_{p} \leq \|A\|_{p} + \|B\|_{p},\]
and is sub-multiplicative:\[\|AB\|_{p}\leq \|A\|_{p} \|B\|_{p}.\]

For $p = 1$, this norm corresponds to the nuclear norm or the trace norm, which is the sum of the singular values of the matrix $A$. For $p = 2$, the Schatten $p$-norm is the square root of the sum of the squares of the entries of the matrix $A$, and corresponds to the Frobenius norm or the Hilbert-Schmidt norm. For $p \rightarrow \infty$, this norm corresponds to the spectral norm or the operator norm, which is the maximum singular value of the matrix $A$.

On finite dimensional spaces, these norms are equivalent, meaning that, for any $p,q$, there exists real numbers $0< C_1\leq C_2$ (depending on the dimension) such that $C_1 \|\cdot\|_p \leq \|\cdot\|_q \leq C_2\|\cdot\|_p$ \cite{Johnson12}.

The distance between two matrices can be defined as the Schatten $p$-norm of their difference. For objects from quantum theory, there are some distinguished ways to quantify their distance; in particular, it is often the case that a notion of distance arises as an operationally justified quantifier of success of a quantum computing task. This is commonly referred to as "having an operational interpretation".

The trace distance, which is commonly defined as half of the Schatten 1-norm of the difference between the matrices:
\[\|\rho-\sigma\|_{\trace}=\frac{1}{2}\| \rho - \sigma \|_1\]
is one example of a distance with an operational interpretation, as it naturally quantifies the success probability of distinguishing two quantum states $\rho$ and $\sigma$ \cite[Section~2.2.4]{NC00}.

A quantifier of similarity between two quantum states is the fidelity of two quantum states. This measure is not a distance of quantum states in the mathematical sense; for equal states, it is equal to one and not zero. In fact, it induces a distance which is referred to as the purified distance.
For this reason, and because it has a particularly simple form for pure states, it is the preferred measure of distance in some scenarios.


\begin{definition}
The fidelity of two quantum states $\rho\in\mathcal{M}_d$ and $\sigma\in\mathcal{M}_d$ is given by
    \[F(\rho,\sigma)=\|\sqrt{\rho}\sqrt{\sigma}\|_{1}^2.\]
If one of the states is pure, e.g. $\sigma=\ketbra{\phi}$, it holds that 
    \[F(\rho,\ketbra{\phi})= |\bra{\phi} \rho \ket{\phi}|^2,\]
    and if both states are pure, i.e. $\rho=\ketbra{\psi}$ and $\sigma=\ketbra{\phi}$, we have
    \[F(\ketbra{\psi},\ketbra{\phi})= |\braket{\psi}{\phi} |^2.\]
\end{definition}

It is not uncommon to see an alternative notation where the fidelity is defined without the square, for example in \cite{NC00}, but here, we will use the above definition.

\begin{theorem}[Uhlmann's theorem, \cite{Uhlmann76}]
    Let $\rho\in\mathcal{M}_d$ and $\sigma\in\mathcal{M}_d$ be quantum states. Then,
    \[F(\rho,\sigma)= \max_{\ket{\psi_{\rho}},\ket{\phi_{\sigma}}} |\braket{\psi_{\rho}}{\phi_{\sigma}}|^2\]
    where the maximization is taken over all purifications $\ketbra{\psi_{\rho}}{\psi_{\rho}}\in \mathcal{M}_d\otimes \mathcal{M}_d $ of $\rho$, and $\ketbra{\phi_{\sigma}}{\phi_{\sigma}}\in \mathcal{M}_d\otimes \mathcal{M}_d $ of $\sigma$.
\end{theorem}

It is easily seen from Uhlmann's theorem that the fidelity is symmetric in its arguments, i.e. $F(\rho,\sigma)=F(\sigma,\rho)$, and "faithful" in the sense that $F(\rho,\sigma)=1 \iff \rho=\sigma$. The fidelity of two quantum states is related to the trace distance in the following way:

\begin{theorem}[Fuchs-van de Graaf inequality, \cite{FvdG97}]\label{thm-fidelity-and-trace-dist}  Let $\rho\in\mathcal{M}_d$ and $\sigma\in\mathcal{M}_d$ be quantum states. Then, we have
\begin{equation*}1-\sqrt{F(\rho,\sigma)}\leq \frac{1}{2}\|\rho-\sigma\|_{1} \leq \sqrt{1-F(\rho,\sigma)}. \end{equation*}
\end{theorem}

The  fidelity obeys the following triangle inequality:

\begin{corollary}\label{thm-triangle_ineq_for_fidelity}
Let $\rho\in\mathcal{M}_d$, $\sigma\in\mathcal{M}_d$ and $\tau\in\mathcal{M}_d$ be quantum states. Then, we have
\begin{equation*}
F(\rho,\sigma) \geq 1- \sqrt{1-F(\rho,\tau)}-\sqrt{1-F(\tau,\sigma)} .\end{equation*}
\end{corollary}

\begin{proof} We have
\begin{equation*}\begin{split}
F(\rho,\sigma) &\geq \sqrt{F(\rho,\sigma)}\\ &\geq 1-\frac{1}{2}\|\rho-\sigma\|_1 \\& \geq1-\frac{1}{2}\|\rho-\tau\|_1 -\frac{1}{2}\|\tau-\sigma\|_1 \\&\geq 1- \sqrt{1-F(\rho,\tau)}-\sqrt{1-F(\tau,\sigma)}.\end{split}\end{equation*}
The second and fourth inequality are consequences of Theorem~\ref{thm-fidelity-and-trace-dist}, and the third inequality follows because the Schatten 1-norm obeys the triangle inequality.
\end{proof}

Quantum channels map quantum states to quantum states. Therefore, the norm of a quantum channel can be quantified by looking at the norm of the quantum states which the channels output. 
This induces a norm for quantum channels in the following way:

\begin{definition}
For any $1\leq p,q\leq \infty$, and for any linear map $\Lambda:\mathcal{M}_{d_{A}}\rightarrow \mathcal{M}_{d_B}$, the induced $q\rightarrow p$-norm is given by
    \[\|\Lambda\|_{q\rightarrow p}= \sup_{\substack{x\in\mathcal{M}_{d_{A}} \\ x\neq 0}} \frac{\|\Lambda (x) \|_p}{\|x\|_q}.\]
\end{definition}

The associated measure of distance between two quantum channels $S:\mathcal{M}_{d_{A}}\rightarrow \mathcal{M}_{d_B}$ and $T:\mathcal{M}_{d_{A}}\rightarrow \mathcal{M}_{d_B}$, $\| T-S\|_{1\rightarrow 1}$, quantifies the distance between quantum channels as the distance between channel outputs for "worst-case" input states.

However, the various induced measures of this kind do not capture all aspects of quantum theory. In many quantum information processing tasks, we leverage the quantum effect of entanglement by considering input which is part of an entangled state on a larger system (e.g., the environment, a reference system).  


This is why the appropriate distance will often be a distance that takes an environment into consideration, as was first noted by Kitaev in \cite{Kitaev97}:

\begin{definition}
   For two quantum channels $T:\mathcal{M}_{d_{A'}}\rightarrow \mathcal{M}_{d_B}$ and $S:\mathcal{M}_{d_{A'}}\rightarrow \mathcal{M}_{d_B}$, we define their diamond distance as
    \[\|T-S\|_{\diamond} =\sup_{d_A}\| \id_{d_A} \otimes (T-S)\|_{1\rightarrow 1}. \]
\end{definition}

In line with conventions from \cite{Watrous04,JKP07,CMW16}, we refer to such norms with an extra reference system as stabilized norms. The diamond distance stabilizes at $d_A=d_B$, meaning that a higher dimensional reference system does not change the value \cite{Kitaev97,Watrous04}. 


\subsubsection{Measures of information}

One concept from quantum information theory which is of great interest in the course of this thesis is a quantum state's \emph{von Neumann entropy}. This entropy quantifies the inherent randomness of a given quantum state, which emerges from its eigenvalue probability distribution.

In classical information theory, an information source is represented by a random variable $X$ that takes values $x$ from a discrete set $\{x\}_{x=0,1,...D-1}$ according to a probability distribution $\{p_X(x)\}_{x=0,1,...D-1}$. Then, the Shannon entropy quantifies the inherent randomness of the probability distribution:

\begin{definition}
The (Shannon) entropy of a discrete random variable $X$ with probability
distribution $p_X(x)$ is
\[H(X)_{p_X}=-\sum_x p_X(x) \log(p_X(x)).\]
\end{definition}
Unless stated otherwise, we always take $\log$ to mean the logarithm with base $2$, which corresponds to measuring information in units of bits. We further follow the convention that $0\log(0)=0$.
For the special case of binary probability distributions, i.e. when $X$ is a Bernoulli random variable with parameter $p$, we refer to the associated entropy
\[h(p)=-p\log(p)-(1-p)\log(1-p)\]
as the binary entropy.

For quantum systems, a notion of entropy is defined in the following way:
\begin{definition}
The von Neumann entropy of a quantum state $\rho\in\mathcal{M}_{d_A}$ is
\[H(A)_{\rho}=-\trace(\rho\log(\rho)).\]
\end{definition}

For a quantum state $\rho\in\mathcal{M}_{d_A}$ with eigenvalues $\{p_X(x)\}_{x=0,1,...,d_A-1}$, the von Neumann entropy is equal to the Shannon entropy of the eigenvalue distribution, i.e. $H(A)_{\rho}=\sum_x p_X(x) \log(p_X(x))$.

This entropy and many quantifiers of information based on this entropy turn out to characterize how well certain quantum states or ensembles perform in a given task, and we will elaborate on this in Section~\ref{sec-shannon-theory}.

As quantum entropy is an important property of quantum states, there is also a related measure that quantifies the degree to which two states differ with respect to their inherent randomness, the \emph{quantum relative entropy} of two states, also referred to as their \emph{quantum divergence} or their \emph{Umegaki relative entropy}.

\begin{definition}\label{def-q-divergence} Let $\rho\in\mathcal{M}_d$ and $\sigma\in\mathcal{M}_d$ be quantum states. Then, their quantum relative entropy is given by
\begin{equation*}
D(\rho||\sigma)= \begin{cases}
 \trace\Big(\rho\big(\log(\rho)-\log(\sigma)\big)\Big)&\text{ if } supp(\sigma)\subseteq supp(\rho) \\
\infty &\text{else.}
\end{cases}
\end{equation*}
\end{definition}
The relative entropy of two quantum states is faithful, i.e. $D(\rho||\sigma)=0 \iff \rho=\sigma$. It is not symmetric under exchanging of the arguments, and is therefore not a metric. It is a special case of a quantum $f$-divergences, which is a quantum generalization of the concept of $f$-divergences from classical probability theory.

The relative entropy of two quantum states is related to their trace distance in the following way:

\begin{theorem}[Quantum Pinsker's inequality, \cite{HOT81,Rouze19}]\label{thm-pinsker-ineq}
Let $\rho\in\mathcal{M}_d$ and $\sigma\in\mathcal{M}_d$ be quantum states. Then, we have 
    \[D(\rho||\sigma)\geq \frac{1}{2\ln(2)} \| \rho-\sigma\|^2_1 .\]
\end{theorem}


For quantum channels $T:\mathcal{M}_{d_{A'}}\rightarrow \mathcal{M}_{d_B}$ and $S:\mathcal{M}_{d_{A'}}\rightarrow \mathcal{M}_{d_B}$, we also define a channel divergence
\[D(T||S)=\sup_{\rho\in\mathcal{M}_{d_{A'}}} D(T(\rho)||S(\rho)), \]
and a stabilized channel divergence
\[D_{stab}(T||S)=\sup_{\rho_{AA'} \in\mathcal{M}_{d_{A}} \otimes\mathcal{M}_{d_{A'}}} D( (\id\otimes T)(\rho_{AA'}) ||(\id\otimes S)(\rho_{AA'})) .\]
These notions of distance are intimately connected to the information theoretic objects we study in Chapter~\ref{chapter-divrad}.

There is one important special case of the quantum relative entropy, which quantifies the difference between a two-system quantum state and the closest product state. This is a notable quantity not only because of its fundamental ties to quantum entropy, but also because it turns out to have an operational interpretation in the context of (quantum) Shannon theory in Section~\ref{sec-shannon-theory} and \ref{sec-quantum-shannon-theory}.

It is a quantum generalization of the following quantity, which characterizes how independent two random variables are:

\begin{definition}[Classical mutual information]
Let $X$ and $Y$ be discrete random variables with the joint probability
distribution $p_{XY}(x,y)$. Then, the mutual information of $X$ and $Y$ is given by
\[I(X:Y)_{p_{XY}} = H(X)_{p_{X}}+H(Y)_{p_{Y}}-H(XY)_{p_{XY}}.\]
\end{definition}

For quantum systems, the analogue quantity is defined in terms of the von Neumann entropy rather than the classical Shannon entropy:

\begin{definition}[Quantum mutual information]
Let $\rho_{AB}\in\mathcal{M}_{d_A}\otimes\mathcal{M}_{d_B}$ be a quantum state. Then, the quantum mutual information of the systems $A$ and $B$ is given by
\[I(A:B)_{\rho_{AB}} = H(A)_{\rho_{A}}+H(B)_{\rho_{B}}-H(AB)_{\rho_{AB}}.\]
\end{definition}

This is equivalent to the following formulation in terms of a relative entropy:

\begin{theorem}
    Let $\rho_{AB}\in\mathcal{M}_{d_A}\otimes\mathcal{M}_{d_B}$ be a quantum state. Then, we have
\[I(A:B)_{\rho_{AB}} = D(\rho_{AB}||\rho_A\otimes \rho_B).\]
\end{theorem}

The difference between functions of two states, such as their von Neumann entropy or their mutual information (with respect to the same state), is governed by continuity bounds, many of which have been found and refined in \cite{AF04,Shirokov17,AE05}.



\subsection{Communication via a classical channel}
\label{sec-shannon-theory}

The field of information theory was pioneered by the mathematician Claude Shannon \cite{Shannon48}, whose contributions laid the foundation for many of our modern digital technologies. 
Beyond revolutionizing the fields of classical communication theory, cryptography and circuit design, his ideas would later provide a theoretical foundation that has been adapted and extended to quantum information theory. 

Let $X$ be a random variable that takes values $x$ from a discrete set $\{x\}_{x=0,1,...D-1}$ according to a probability distribution $\{p_X(x)\}_{x=0,1,...D-1}$. Then, if the random variable is sampled according to the probability distribution $p_X(x)$ and a specific realization $x$ is obtained, our surprise upon observing $x$ can be quantified by $-\log(p_X(x))$, which is often referred to as $x$'s surprisal, or its information content, measured in the unit of bits. The average information content of the set is then given by the Shannon entropy $H(X)_{p_X}$ of the probability distribution $\{p_X(x)\}_{x=0,1,...,D-1}$. 

When sampling a sequence of length $n$ from this random variable, we denote by $f(x)$ the observed number of times the symbol $x$ occurs in the sequence. Then, we can define a notion of sample entropy $-\sum_x \frac{f(x)}{n} \log{\frac{f(x)}{n}}$. As $n$ becomes large, the number of times we observe the symbol $x$ in the sequence will approach $f(x)\rightarrow p_X(x)n$ by the law of large numbers, meaning that the sample entropy approaches the entropy of the underlying probability distribution $p_X(x)$.

In other words, for a large enough sequence length $n$, it is highly likely that the sample entropy of the sequence is close to the underlying entropy of the probability distribution. This idea is made mathematically rigorous by defining typical sequences, which have a sample entropy which is close to the underlying entropy (i.e. similar information content), see \cite[Chapter~14]{Wilde13} for more details. As it turns out, if the variable is not a uniform random variable, the size of the set of typical sequences (the typical set) is significantly smaller than the size of the set of all sequences \cite[Property~14.3.2]{Wilde13}.

This idea naturally proves useful for data compression, i.e. for sending or storing shorter messages with the same (or at least highly similar) average information content. The concept of typical sequences suggests that it might be a good idea to place a focus on encoding typical sequences, rather than sequences that are less likely to appear. We would like to choose shorter codewords for more likely sequences; for example, we could choose codeword lengths in relation to a symbol's information content. Then, when we find a strategy to encode the information content of a sequence of length $m$ (as generated by the source) in a (shorter) sequence of length $n\leq m$, we say that we can compress the source's output at a rate of $m/n$. As it turns out, the optimal rate at which we can compress information with vanishing error probability is exactly given by the Shannon entropy \cite[Theorem~9]{Shannon48}, which is referred to as \emph{Shannon's data compression theorem}, \emph{Shannon's source coding theorem} or \emph{Shannon's noiseless channel coding theorem}.

If the symbols additionally become subject to noise during transmission, this leads to another of Shannon's fundamental results referred to as \emph{Shannon's noisy channel coding theorem}. In this theorem, Shannon showed that the reliable communication over noisy classical channels can be achieved using \emph{channel coding}.

Let $T:\mathcal{A}\rightarrow\mathcal{B}$ denote the classical channel that describes the noise affecting individual symbols. The input to the channel is given by a symbol $a\in\{0,1\}$ of random variable $A$, which is distributed according to the probability distribution $p_A(a)$, and the channel outputs a symbol $b\in\{0,1\}$ of the random variable $B$ according to the probability distribution $p_{B|A}(b|a)$. Then, a coding scheme consists of an encoder and decoder that encode and decode messages in a way such that they are transmitted correctly with high probability (see also Figure~\ref{fig-shannon-cap}).

\begin{figure}[htbp]
    \centering
    \includegraphics[width=10cm]{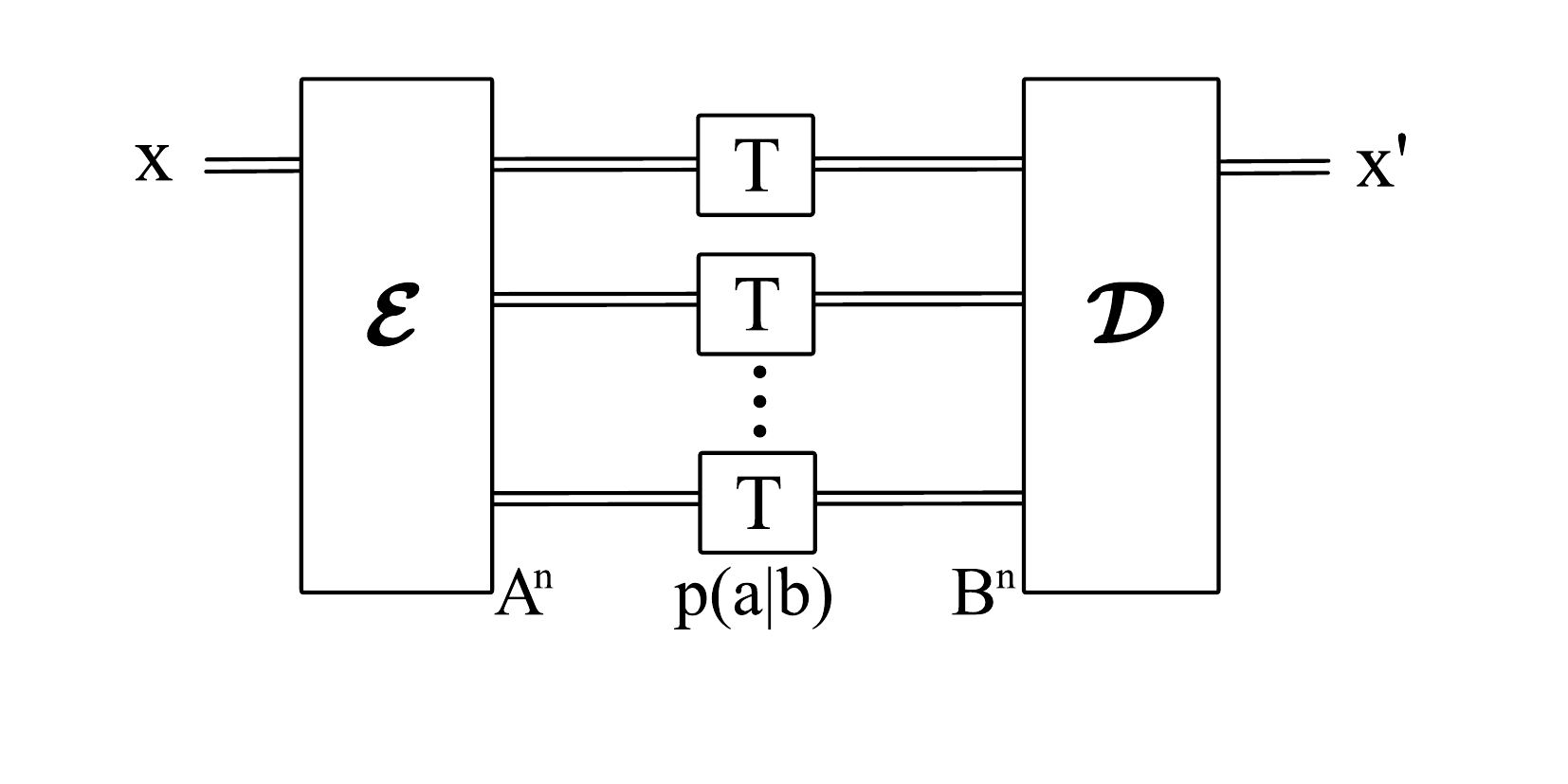}
    \caption{\textbf{Basic setup for classical communication over a classical channel.} An encoder encodes a message of $m$ bits, $x\in \{1,2,...,2^m\}$ into a codeword $a_n\in \mathcal{A}^n$ which is a bitstring of $n$ bits. The codeword is sent from the sender to the receiver by $n$ uses of the channel $T$. After receiving the noisy output $b_n  \in \mathcal{B}^n$ of $n$ applications of $T$, the decoder produces an estimate $x'\in\{1,2,...,2^m\}$, which should ideally be identical to the original message $x$. (All the information in this picture is classical, but we already use our convention where classical information transfer is indicated by double lines, while the transfer of quantum information is indicated by single lines.)}
    \label{fig-shannon-cap}
\end{figure}

\begin{definition}[Shannon channel coding scheme]
Let $T:\mathcal{A}\rightarrow\mathcal{B}$, be a classical channel, and let $n,m \in \mathbbm{N}$ and $\epsilon>0$.

Then, an $(n,m,\epsilon)$-coding scheme consists of an encoding operation $\mathcal{E}:\{1,2,...,2^m\}\rightarrow\mathcal{A}^n$ and a decoding operation $\mathcal{D}:\mathcal{B}^n\rightarrow \{1,2,...,2^m\}$ such that
\begin{align*}
    \Prob[x \neq \mathcal{D} \circ  T^{\otimes n } \circ \mathcal{E}  (x)] \leq \epsilon
\end{align*}
for all bit strings $x\in \{1,2,...,2^m\}$.
\end{definition}

A given coding scheme can then achieve a certain communication rate (although perhaps the rate is zero), and the tightest bound on the achievable rate is the channel's Shannon capacity:

\begin{definition}[Shannon capacity]
Let $T:\mathcal{A}  \rightarrow \mathcal{B} $ be a classical channel.

If, for every $n\in \mathbbm{N}$, there exists an $(n,m(n),\epsilon(n))$ Shannon channel coding scheme, then a rate $R\geq 0$ is called achievable for classical communication via the classical channel $T$ if
\begin{align*}
   R & \leq 
   \liminf_{n\rightarrow \infty} \Big\{ 
    \frac{m(n)}{n} \Big\}
    \end{align*}
and
\begin{align*}
   \lim_{n\rightarrow \infty} \epsilon(n) \rightarrow 0.
    \end{align*}
The Shannon capacity of $T$ is given by
\begin{align*}
  C_S(T) = \sup_{R} \{ R | & R \text{ achievable for classical communication via $T$} \}.
    \end{align*}
\end{definition}

Again, we may consider encoding messages in such a way that typical sequences have shorter codewords - but since the channel corrupts some symbols, we also have to consider which symbol sequences the channel may typically map to, and how the two sets are related. This gives rise to a notion of conditional typicality \cite[Section~14.6,14.9]{Wilde13}. It turns out that the size of the conditional typical set and the typical set of input sequences can be related to the rates at which we can communicate. As a consequence, the following non-asymptotic single-letter formula for the channel capacity can be obtained:

\begin{theorem}[{Shannon's noisy channel coding theorem, \cite[Section~13]{Shannon48}}] \label{thm-shannons-noisy-channel-coding}
Let $X$ denote a discrete random variable with a set of possible values $\{x\}$ that is distributed according to a probability distribution $p_X$. Let $T(X)=Y$ be the output of a classical channel $T$, which transmits the symbols according to a probability distribution $p_{Y|X}$. Then, we have:
    \[C_S(T) = \max_{p_X} I(X:Y)_{p_{XY}}.\]
\end{theorem}

The original proof was provided by Shannon in \cite[Section~13]{Shannon48}, but a detailed version of the arguments can also be found in \cite[Theorem~14.10.1]{Wilde13}.

This theorem establishes the mutual information between the the random variable $X$ and the output $Y=T(X)$ as the fundamental rate at which we can communicate via a channel $T$. 

\subsection{Communication via a quantum channel}\label{sec-quantum-shannon-theory}

When quantum effects occur, it is natural to ask how they change compression and  communication. On the one hand, for example, a transmission line may naturally have quantum properties. On the other hand, some quantum effects between distant systems, like quantum entanglement, suggest a notion of communication which is fundamentally different from classical communication.


Many ideas from classical information theory have quantum generalizations and analogues that appear in the study of quantum communication, such as a notion of quantum typicality \cite[Chapter~15]{Wilde13}. Based on these concepts, Schumacher \cite{Schumacher95} devised a scheme that achieves quantum data compression of a quantum information source at a rate which is equal to the von Neumann entropy of the source, suggesting the interpretation of the qubit as a unit of quantum information, and the von Neumann entropy as a measure of quantum information.

Some classical results, such as Theorem~\ref{thm-shannons-noisy-channel-coding}, do not have a straightforward counterpart in quantum information theory. Suitable methods and achievable rates for quantum channel coding depend on the communication setup and can vary significantly. Two such setups will be the main players in this thesis and will be described in the rest of this section.

\subsubsection{Classical communication via a quantum channel}

The ability of quantum systems to transmit classical information can be analyzed in a setup where the information is encoded into a quantum system's state instead of a classical bitstring, which serves as input into the channel (see also Figure~\ref{fig-clcap}).

\begin{definition}[Classical communication coding scheme]
Let $T:\mathcal{M}_{d_A} \rightarrow \mathcal{M}_{d_B}$ be a quantum channel, and let $n,m \in \mathbbm{N}$ and $\epsilon>0$.

Then, an $(n,m,\epsilon)$-coding scheme for classical communication consists of quantum channels $\mathcal{E}:\mathbbm{C}^{2^m}\rightarrow\mathcal{M}_{d_A}^{\otimes n}$ and $\mathcal{D}:\mathcal{M}_{d_A}^{\otimes n}\rightarrow \mathbbm{C}^{2^m}$ such that
\begin{align*}
    F\Big(X,\mathcal{D} \circ  T^{\otimes n } \circ \mathcal{E}  (X)\Big) \geq 1-\epsilon
\end{align*}
where $X=\ketbra{x}$, for all bit strings  $x\in \{0,1\}^m$.
\end{definition}

When a coding scheme exists for a channel $T$ for some $\epsilon>0$ and some combination of $m$ and $n$, then we can communicate via this channel at a rate of $m/n$ and with error $\epsilon$. The maximum such rate in the asymptotic limit with vanishing communication error is called the classical capacity of a quantum channel $T$.

\begin{figure}[!t]
\centering
       \includegraphics[width=10cm]{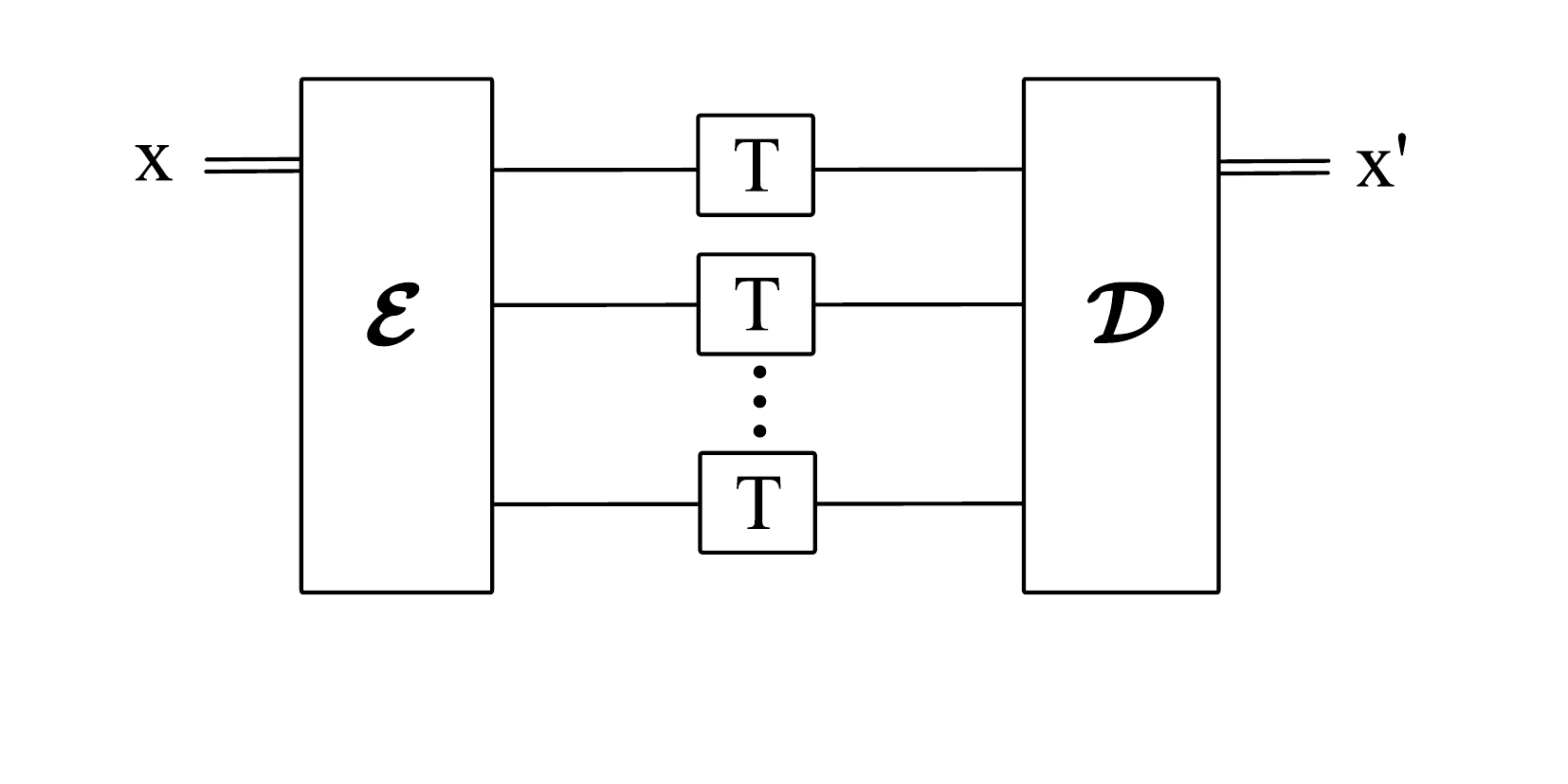}
\caption{\textbf{Basic setup for classical communication over a quantum channel.} The encoding map $\mathcal{E}$ maps a bit string $x$ of length $m$ to a quantum state in $\mathcal{M}_{d_A}^{\otimes n}$. The quantum channel $T$ acts on each of the $n$ subsystems, and the decoder $\mathcal{D}$ decodes the received quantum state to a bit string $x'$, which should be identical to the input bit string $x$. Note that classical information transfer is indicated by double lines (input into encoder, output of decoder), while the transfer of quantum states is indicated by single lines.}
\label{fig-clcap}
\end{figure}

\begin{definition}[Classical capacity of a quantum channel]
Let $T:\mathcal{M}_{d_A} \rightarrow \mathcal{M}_{d_B}$ be a quantum channel.

If, for every $n\in \mathbbm{N}$, there exists an $(n,m(n),\epsilon(n))$-coding scheme for classical communication, then a rate $R\geq 0$ is called achievable for classical communication via the quantum channel $T$ if
\begin{align*}
   R & \leq 
   \liminf_{n\rightarrow \infty} \Big\{ 
    \frac{m(n)}{n} \Big\}
    \end{align*}
and
\begin{align*}
   \lim_{n\rightarrow \infty} \epsilon(n) \rightarrow 0.
    \end{align*}
The classical capacity of $T$ is given by
\begin{align*}
  C(T) = \sup_{R} \{ R | & R \text{ achievable for classical communication via $T$} \}.
    \end{align*}
\end{definition}

To study whether the quantum nature of the channel could aid in information transfer or not, Holevo introduced the following helpful quantity, which is closely connected to the classical capacity because of the Holevo-Schumacher-Westmoreland theorem \cite{Holevo96,SW97} (restated in Theorem~\ref{thm-hsw-thm}).

\begin{definition}[Holevo quantity and Holevo capacity]\label{def-holevo-capacity}
Let $\{p_i,\rho_i\}$ be an ensemble of quantum states $\rho_i\in\mathcal{M}_{d_B}$ drawn according to a probability distribution $\{p_i\}$. The Holevo quantity of this ensemble is given by
\[\chi_H (\{p_i,\rho_i\}) =  H(B)_{\sum_i p_i \rho_i } -\sum_i p_i H(B)_{\rho_i}   .\]
Let $T:\mathcal{M}_{d_A} \rightarrow \mathcal{M}_{d_B}$ be a quantum channel.
The Holevo capacity of a quantum channel is given by
\[C_H (T) =\sup_{\{p_i,\rho_i\},\rho_i \in\mathcal{M}_{d_A}}  \Big\{H(B)_{\sum_i p_i T( \rho_i) } -\sum_i p_i H(B)_{ T(\rho_i)} ) \Big\} .\]
\end{definition}

\begin{theorem}[{\cite[Section~5]{SW01}}]
    The ensemble that maximizes $C_H(T)$ always exists for any quantum channel $T:\mathcal{M}_{d_A} \rightarrow \mathcal{M}_{d_B}$.
\end{theorem}


There have been alternative formulations of the Holevo capacity. Here, we will introduce two important alternative formulations and prove that they are equivalent to the definition. The first is a formulation in terms of a mutual information with a supremum over cq-state, highlighting the similarity and difference to the classical scenario. The second is a formulation in terms of a relative entropy which will play an integral role in Chapter~\ref{chapter-divrad}.

\begin{theorem}[Alternative formulation 1]\label{thm-class-cap-1}
   For any quantum channel $T:\mathcal{M}_{d_{A'}} \rightarrow \mathcal{M}_{d_B}$, we have
\[C_H(T)=\max_{\rho_{AA'}\text{ a cq state}} I(B:A)_{(\id\otimes T)(\rho_{AA'})}.\]
\end{theorem}
\begin{proof}
This is a consequence of the properties of cq-states:
   \begin{align*}
      & \max_{\rho\text{ a cq state}} I(B:A)_{(\id\otimes T)(\rho)} \\&\hspace{1cm} = \max_{\{p_i,\rho_i\}} I(B:A)_{(\id\otimes T)(\sum_i p_i \ketbra{i}\otimes  \rho_i)}
       \\& \hspace{1cm} =\max_{\{p_i,\rho_i\}} \Big\{ H(A)_{\sum_i p_i \ketbra{i}} + H(B)_{\sum_i p_i T( \rho_i) } - H(AB)_{\sum_i p_i \ketbra{i}\otimes  T(\rho_i))}  \Big\}\\&\hspace{1cm} 
       =\max_{\{p_i,\rho_i\}}  \Big\{ H(\{p_i\} ) + H(B)_{\sum_i p_i T( \rho_i) } -  \sum_i p_i H(AB)_{ \ketbra{i}\otimes T(\rho_i)} -H(\{p_i\} )  \Big\}
       \\&\hspace{1cm} 
       =\max_{\{p_i,\rho_i\}}  \Big\{H(\{p_i\} ) + H(B)_{\sum_i p_i T( \rho_i) } -\sum_i p_i H(B)_{ T(\rho_i)} -H(\{p_i\} ) \Big\}
       \\&\hspace{1cm} 
       =\max_{\{p_i,\rho_i\}}  \Big\{  H(B)_{\sum_i p_i T( \rho_i) } -\sum_i p_i H(B)_{ T(\rho_i)}   ) \Big\}\\& \hspace{1cm}=C_H(T)
   \end{align*}
   The first equality is by definition of the set of cq-states; the second equality is by definition of the quantum mutual information. The third equality is true due to the orthogonality of the classical states on system $A$. The fourth equality is true because $\ketbra{i}\otimes T(\rho_i)$ is separable, which implies $H(AB)_{\ketbra{i}\otimes T(\rho_i)}=H(A)_{\ketbra{i}}+ H(B)_{T(\rho_i)}= H(B)_{T(\rho_i)}$, because $\ketbra{i}$ is pure.
\end{proof}

In Theorem~\ref{thm-class-cap-2}, we give an alternative formulation in terms of relative entropy that was first studied in \cite{SW01}, and subsequently in \cite{TT15}. The proof is less straight-forward and will require two extra lemmata (\ref{thm-chi-as-D} and \ref{thm-equal-distance-prop}) that are given below:

\begin{lemma}[Average of relative entropies] \label{thm-chi-as-D} Let $\{p_i,\rho_i\}$ be an ensemble of quantum states $\rho_i\in\mathcal{M}_d$, associated with a probability $p_i$. Then, we have
    \[\chi_H(\{p_i,\rho_i\}) =\sum_i p_i D(\rho_i||\sum_j p_j \rho_j) .\]
\end{lemma}

\begin{proof}
For any ensemble $\{p_i,\rho_i\}$, we have that
    \begin{equation*}\begin{split}
     \chi_H(\{p_i,\rho_i\}) & = H(B)_{\sum_i p_i  \rho_i } -\sum_i p_i H(B)_{ \rho_i}   ) \\&   =  - \trace\Big(\sum_i p_i \rho_i \log(\sum_j p_j\rho_j)\Big) +\sum_k p_k \trace\Big(\rho_k\log(\rho_k) \Big) 
        \\&   = \sum_i p_i ( - \trace\Big(\rho_i \log(\sum_j p_j  \rho_j)\Big) + \trace\Big(\rho_i\log(\rho_i) \Big)
        \\&   =    \sum_i p_i ( \trace\Big( \rho_i \big(\log(\rho_i)- \log(\sum_j p_j \rho_j) \big) \Big) 
        \\&  =  \sum_i p_i D(\rho_i||\sum_j p_j \rho_j)  .
\end{split}\end{equation*}
The first two equalities are obtained by explicitly writing out the von Neumann entropy. The third and fourth equalities are due of the linearity of the trace, and the final equality is true by the definition of the quantum state divergence, Definition~\ref{def-q-divergence}.
\end{proof}

\begin{lemma}[{Equal distance property, \cite[Section~5]{SW01}}] \label{thm-equal-distance-prop}
    Suppose $\rho^*_k\in\mathcal{M}_d$ is one of the states in the ensemble $\{p^*_i,\rho^*_i\}$ that maximizes the Holevo quantity, i.e. we have $\chi_H(\{p^*_i,\rho^*_i\})=\max_{\{p_i,\rho_i\}}\chi_H(\{p_i,\rho_i\})$.
    Then, we have that \[D(\rho^*_k||\sum_j p^*_j \rho^*_j)= \chi_H(\{p^*_i,\rho^*_i\}).\]
This also implies that $D(\rho^*_k||\sum_j p^*_j \rho^*_j)= \chi_H(\{p^*_i,\rho^*_i\})$ $\forall k$, since $\chi_H(\{p^*_i,\rho^*_i\})$ is equal to the average in Lemma~\ref{thm-chi-as-D}.
\end{lemma}

With these two lemmata, we prove that the Holevo capacity is equal to the solution of the following optimization problem:

\begin{theorem} [Alternative formulation 2, \cite{SW01}]
\label{thm-class-cap-2}
For any quantum channel $T:\mathcal{M}_{d_A} \rightarrow \mathcal{M}_{d_B}$, we have
\[C_H (T)= \min_{\sigma\in\mathcal{M}_{d_{B}}} \max_{\rho\in\mathcal{M}_{d_A}} D( T(\rho) ||\sigma) .\]
\end{theorem}

\begin{proof}
To prove this, we use the fact that we can write $\chi_H(\{p_i,\rho_i\})$ of an ensemble $\{p_i,\rho_i\}$ of quantum states $\rho_i\in\mathcal{M}_{d_B}$ associated with a probability $p_i$ as an average of quantum divergences as in Lemma~\ref{thm-chi-as-D}. For any quantum state $\sigma\in\mathcal{M}_{d_B}$, we can rewrite this in the following way:
\begin{align*}
 \chi_H(\{p_i,\rho_i\}) & = \sum_i p_i D(\rho_i||\sum_j p_j \rho_j)  \\ &  =    \sum_i p_i ( \trace\Big( \rho_i \big(\log(\rho_i)- \log(\sum_j p_j \rho_j) \big) + \log(\sigma)-\log(\sigma) \Big) 
        \\& =  \sum_i p_i D(\rho_i||\sigma) - D(\sum_i p_i \rho_i || \sigma) 
    \end{align*}
As a consequence, we have that 
$\chi_H (\{p_i,\rho_i\})\leq  \sum_i p_i D(\rho_i||\sigma) \leq \max_i D(\rho_i||\sigma)$ for all quantum states $\sigma$, and therefore also for the quantum state that minimizes the expression. Therefore, we have
\[\chi_H (\{p_i,\rho_i\})\leq  \min_{\sigma}\max_{i} D(\rho_i||\sigma).\]
This is also true for the ensemble that maximizes the Holevo quantity in the expression for $C_H(T)$, and therefore, we have
\[C_H(T)= \chi_H (\{p^*_i,T(\rho^*_i)\})\leq \min_{\sigma}\max_{\rho} D(T(\rho)||\sigma) . \]
On the other hand, due to the equal distance property in Lemma~\ref{thm-equal-distance-prop}, we have that, for all $k$,
\begin{align*}
   C_H(T)&= \chi_H (\{p^*_i,T(\rho^*_i)\}) \\&= D(T(\rho^*_k)||\sum_j p^*_j \rho^*_j) \\& \geq \min_{\sigma} D(T(\rho^*_k)||\sigma)  \\& =\min_{\sigma} \max_{\rho} D(T(\rho)||\sigma)  
\end{align*}
In combination, we have equality between the Holevo capacity and the optimization problem above.

\end{proof}

The classical capacity and the Holevo capacity are related in the following way:

\begin{theorem}[{Holevo–Schumacher–Westmoreland theorem, \cite[Chapter~20.3]{Wilde13}}]\label{thm-hsw-thm}
Let $T:\mathcal{M}_{d_A} \rightarrow \mathcal{M}_{d_B}$ be a quantum channel. Let $C(T)$ denote the channel's classical capacity, and $C_H(T)$ denote the channel's Holevo capacity. For any channel $T$, we have
    \[C(T) = \lim_{n\rightarrow \infty} \frac{1}{n} C_H(T^{\otimes n}).\]
    This also implies that $C_H(T)\leq C(T)$ for any channel $T$.
\end{theorem}

For a while, it was an open question whether or not the Holevo quantity is additive (in the tensor-product sense) for all quantum channels. This so called \emph{additivity conjecture} was widely believed to be true by many researchers, because it holds true for many commonly considered classes of quantum channels \cite{King02,Shor02_2}, and because numerical searches did not turn up any counter-example. In 2009, Hastings \cite{Hastings09} provided a non-constructive counter-example to the additivity conjecture; it was further shown that this counter-example appears in rather high dimensions, and for a high number of channel copies \cite[Proposition~3]{FKM10}. It is an open question whether counter-examples to the additivity conjecture exist for small input and output channel dimension and small numbers of tensor products.


\begin{theorem}[\cite{Hastings09}]
There exists a quantum channel $T:\mathcal{M}_{d_A} \rightarrow \mathcal{M}_{d_B}$ and a number $N\in\mathbbm{N}_+$ with
    \[C_H(T)< \frac{1}{N} C_H(T^{\otimes N}).\]
\end{theorem}


\subsubsection{Entanglement-assisted communication via a quantum channel}

One of the most important basic quantum communication protocols is superdense coding \cite{BW92}, which allows a sender and a receiver with access to one maximally entangled state to transmit two classical bits through the use of one perfect quantum channel by the procedure sketched in Figure~\ref{fig-superdense}. In other words, a single identity channel $\id_2:\mathcal{M}_2\rightarrow \mathcal{M}_2$, when assisted by entanglement, can transmit classical information at a rate of $2$ bits per channel use. Without entanglement, the channel can only transmit one classical bit per channel use.

\begin{figure}[htbp]
    \centering
    \includegraphics[width=10cm]{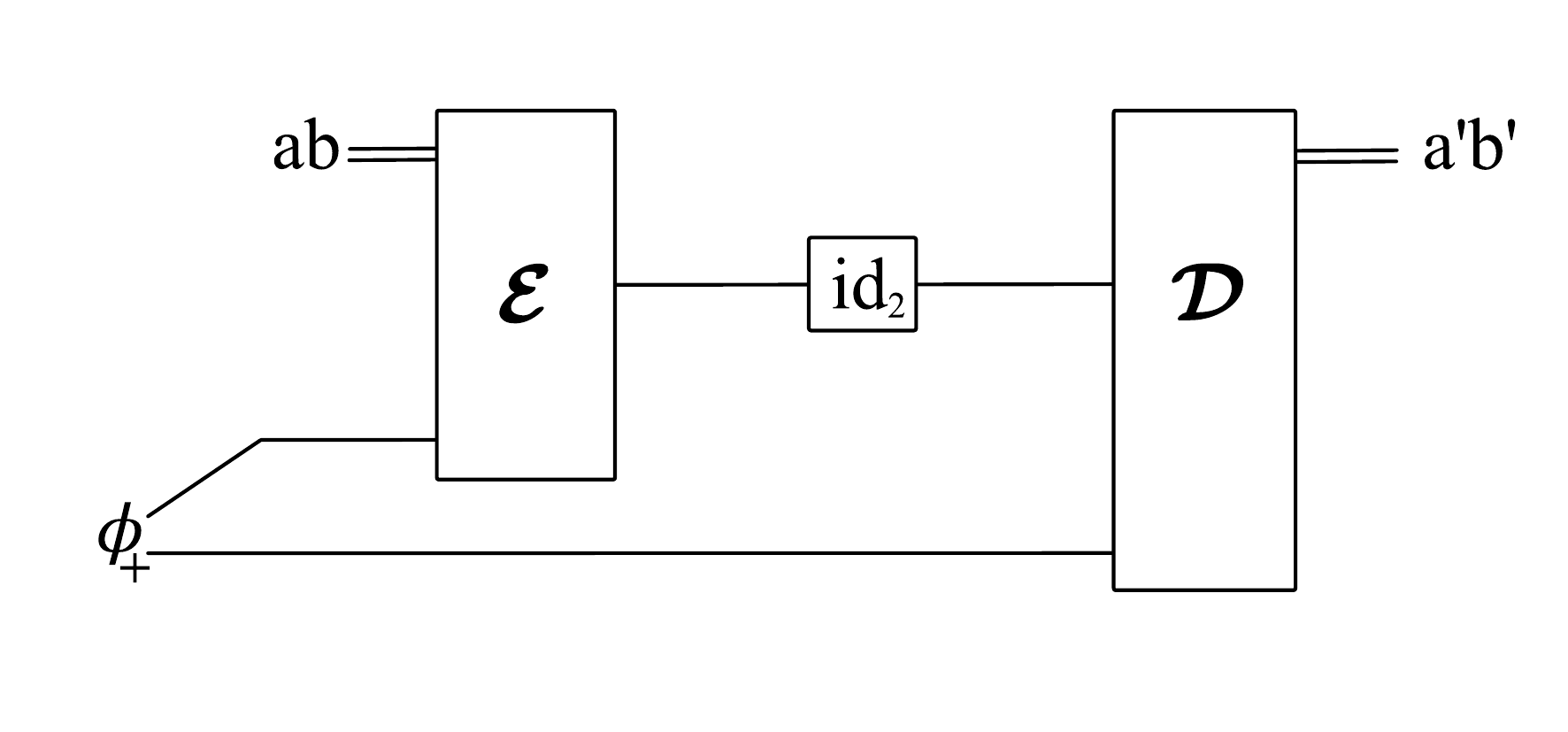}
    \caption{\textbf{Basic setup of the superdense coding protocol.} Let the sender and the receiver each hold one qubit from a maximally entangled state $\ket{\phi_+}=\frac{1}{\sqrt{2}} (\ket{00}+\ket{11})$.
The sender performs an encoding $\mathcal{E}:\mathbbm{C}^2\otimes \mathcal{M}_2\rightarrow \mathcal{M}_2$ of two classical bits $a\in\{0,1\}$ and $b\in\{0,1\}$ by performing one of the 4 quantum gates $G_{ab}=\{\mathbbm{1}_2,\sigma_x,\sigma_y,\sigma_z\}$ on their part of the maximally entangled state. Subsequently, they send their part of the entangled state through the channel $\id_2$, so that the receiver now holds the entire 2-qubit state $(G_{ab}\otimes\mathbbm{1}_2) \ket{\phi_+}$. By performing a measurement in the Bell basis as a decoder $\mathcal{D}:\mathcal{M}_2\otimes\mathcal{M}_2\rightarrow \mathbbm{C}^2$, the receiver can deduce which message the sender intended to send from the measurement output.}
    \label{fig-superdense}
\end{figure}

This protocol highlights how entanglement between the sender and the receiver can enhance communication. Based on this idea, we can generally consider a coding scheme where the sender and receiver have entanglement (see also Figure~\ref{fig-ea-cap}):

\begin{definition}[Entanglement-assisted coding scheme]\label{def-ea-coding-scheme}
Let $T:\mathcal{M}_{d_A} \rightarrow \mathcal{M}_{d_B}$ be a quantum channel, and let $n,m \in \mathbbm{N}$, $R_{ea}\in\mathbbm{R}^+$ and $\epsilon>0$.

Then, an $(n,m,\epsilon,R_{ea})$-coding scheme for entanglement-assisted communication consists of quantum channels $\mathcal{E}:\mathbbm{C}^{2^m}\otimes\mathcal{M}_2^{ \otimes \lfloor nR_{ea} \rfloor}\rightarrow\mathcal{M}_{d_A}^{\otimes n}$ and $\mathcal{D}:\mathcal{M}_{d_A}^{\otimes n}\otimes \mathcal{M}_2^{ \otimes \lfloor nR_{ea} \rfloor} \rightarrow \mathbbm{C}^{2^m}$ such that
\begin{align*}
    F\Big(X,\mathcal{D} \circ  \big( (T^{\otimes n } \circ \mathcal{E} )\otimes \id_2^{\otimes \lfloor nR_{ea} \rfloor} \big) (X \otimes \phi_+^{\otimes \lfloor nR_{ea} \rfloor })\Big) \geq 1-\epsilon
\end{align*}
where $X=\ketbra{x}$, for all bit strings  $x\in \{0,1\}^m$.
\end{definition}

\begin{remark}
Here, $\phi_+=\ketbra{\phi_+}{\phi_+}$, where $\ket{\phi_+}=\frac{1}{\sqrt{2}}( \ket{00}+\ket{11})$ denotes a maximally entangled state of two qubits. Like in \cite{BSST02}, we define entanglement-assistance with respect to copies of the maximally entangled state $\phi_+$. Without loss of generality, we could allow assistance by copies of arbitrary pure entangled states, since they can be prepared efficiently from maximally entangled states by the process of entanglement dilution \cite{LP99,BSST02}  using of a sublinear amount of classical communication from one party to the other \cite[Theorem~1]{HL04}. It turns out that even entirely arbitrary entangled states (not of product form) cannot increase the communication rate \cite{BDHSW14} and it is therefore sufficient to use maximally entangled states as the entanglement resource.

For a rate of entanglement-assistance $R_{ea}=0$ in the above definition, the scenario reduces to the scheme for classical communication with no entanglement-assistance as introduced in \cite{SW97,Holevo96}. For $R_{ea} \geq  \sup_{\varphi} H(A)_{\varphi}$, where the supremum goes over pure bipartite states $\varphi$, the entanglement-assisted capacity does not increase with more entangled states \cite{HL04}, and we will henceforth focus on this scenario.
\end{remark}

\begin{figure}[!t]
\centering
       \includegraphics[width=10cm]{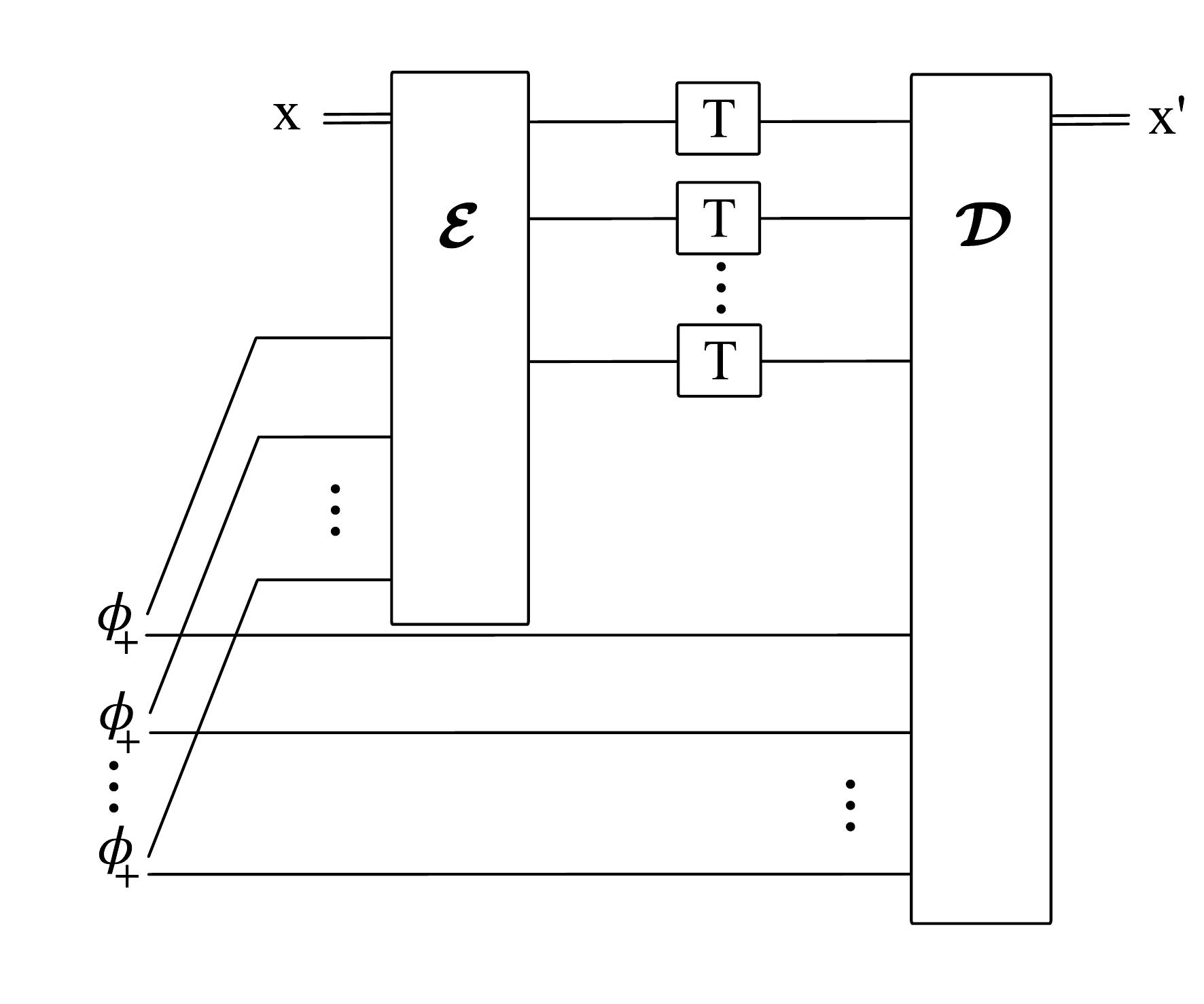}
\caption{\textbf{Basic setup for entanglement-assisted communication.} The encoding map $\mathcal{E}$ maps a bit string $x$ of length $m$ and one part of each entangled state $\phi_+$ to a quantum state in $\mathcal{M}_{d_A}^{\otimes n}$. The quantum channel $T$ acts on each of the $n$ subsystems, and the decoder $\mathcal{D}$ uses the other part of each maximally entangled state to decode the received quantum state to a bit string $x'$, which should be identical to the input bit string $x$. Note that classical information transfer is indicated by double lines (input into encoder, output of decoder), while the transfer of quantum states is indicated by single lines.}
\label{fig-ea-cap}
\end{figure}

\begin{definition}[Entanglement-assisted capacity]
Let $T:\mathcal{M}_{d_A} \rightarrow \mathcal{M}_{d_B}$ be a quantum channel and let $R_{ea}\in \mathbbm{R}_+$ be the rate of entanglement-assistance.

If, for some $R_{ea}$ and for every $n\in \mathbbm{N}$, there exists an $(n,m(n),\epsilon(n),R_{ea})$-coding scheme for entanglement-assisted communication, then a rate $R\geq 0$ is called achievable for entanglement-assisted communication via the quantum channel $T$ if

\begin{align*}
   R & \leq 
   \liminf_{n\rightarrow \infty} \Big\{ 
    \frac{m(n)}{n} \Big\}
    \end{align*}
and
\begin{align*}
   \lim_{n\rightarrow \infty} \epsilon(n) \rightarrow 0.
    \end{align*}
The entanglement-assisted capacity of $T$ is given by
\begin{align*}
  C^{ea}(T) = \sup_{R} \{ R | & R \text{ achievable for entanglement-assisted communication via $T$} \}.
    \end{align*}
\end{definition}

Importantly, the entanglement-assisted capacity is equal to the quantum mutual information of the channel input and output:

\begin{theorem}[{Entanglement-assisted capacity Theorem, \cite{BSST99,BSST02}}]\label{thm-ent-ass-cap-bsst}
Let $T:\mathcal{M}_{d_{A'}} \rightarrow \mathcal{M}_{d_B}$ be a quantum channel. Then, we have: 
    \[C^{ea}(T) = \sup_{\rho_{AA'}} I(B:A)_{(\id \otimes T)(\rho_{AA'})}\]
    where the supremum is over all pure quantum states $\rho_{AA'}\in\mathcal{M}_{d_A}\otimes \mathcal{M}_{d_{A'}}$.
\end{theorem}

For an alternative, detailed proof, see also \cite[Section~21.4-21.5]{Wilde13}.

This theorem highlights that the entanglement-assisted capacity can be understood as the true quantum analogue of the completely classical version from Shannon's noisy channel coding theorem~\ref{thm-shannons-noisy-channel-coding}. In particular, this means that the entanglement-assisted capacity does indeed satisfy additivity.

In Chapter~\ref{chapter-divrad}, we will show that the entanglement-assisted capacity can be rewritten in terms of a quantum divergence, which resembles the divergence radius formulation for the Holevo capacity in Theorem~\ref{thm-class-cap-2}.

These three notions of capacity (the classical capacity, the Holevo capacity and the entanglement-assisted capacity) are the ones who appear most in this thesis; other notions of capacity, such as quantum capacity and private capacity have also been defined and studied in many interesting contexts.

In our work in Chapter~\ref{chapter-fteacap}, we consider communication setups for entanglement-assisted communication which confirm and supplement the practical relevance of Theorem~\ref{thm-ent-ass-cap-bsst}. We study a notion of capacity where we take noise effects into account which arise from the quantum nature of the encoding and decoding operation, see Section~\ref{sec-ft-cap-setup}, and we study a notion of capacity for channels with highly correlated noisy input, which we refer to as capacity under arbitrarily varying perturbations, see Section~\ref{sec-avp-model}. While both of these communication scenarios are not covered by standard techniques, we show that we can still achieve rates that are close to the rates from the standard scenario in Theorem~\ref{thm-ent-ass-cap-bsst}.


Given all of these different notions of capacity in the quantum realm, it is also an interesting question of how they compare to one another; some such relationships have been shown in \cite{BSST02}. For example, it is straightforward to see that the classical capacity of a given channel is smaller or equal to the entanglement-assisted capacity of the same channel as the maximization of the mutual information runs over a smaller set of states. It is not known how much larger the entanglement-assisted capacity can be, but it has been conjectured that it is upper-bounded by the classical capacity multiplied by a dimension-dependent factor in \cite{BSST02}. In Chapter~\ref{chapter-divrad}, we prove this conjecture for a large class of channels.

\subsubsection{Examples of quantum channels and their capacities}

There are some important examples of quantum channels which will frequently be referred to in this thesis. Many of these channel models have closed-form expressions for capacities which will play a role in Chapter~\ref{chapter-divrad}. Here, we will introduce some such quantum channels and review what is known about their capabilities to transmit quantum information.

\begin{example}[Erasure channel]
An erasure channel with erasure probability $p$ is a quantum channel which applies an identity with probability $1-p$, and with probability $p$, it erases the state and returns an erasure flag in the form of an orthogonal state $\ketbra{e}{e}$, i.e. \[\rho\rightarrow (1-p)\rho + p \ketbra{e}{e}.\]
For this channel, the Holevo capacity is \cite[Exercise~8.7]{Holevo13}:
\[C_H(T_p) = (1-p) C_H(\id) = (1-p) \log d\]
and the entanglement-assisted capacity is \cite[Exercise~9.2]{Holevo13}:
\[C^{ea}(T_p)=2(1-p)\log d .\]
\end{example}

\begin{example}[Depolarizing channel]\label{ex-depo-channel}
    
Let $T_p$ be the depolarizing channel, which acts on quantum states in the following way:
\[T_p(\rho)=(1-p) \rho + p  \frac{\mathbbm{1}_d}{d} \trace(\rho) .\]
This channel is completely positive for $0\leq p \leq \frac{d^2}{d^2-1}$. 

This channel's Holevo capacity is given by \[ C_H(T_p) = \log d + (1- p \frac{d-1}{d}) \log (1- p \frac{d-1}{d}) +p \frac{d-1}{d} \log(\frac{p}{d}) \] as can also be found in \cite[Exercise~8.9]{Holevo13}, \cite[Eq.~1]{BSST99} and \cite[Theorem~20.4.3]{Wilde13}.
The entanglement-assisted classical capacity of the depolarizing channel is given by \cite[Exercise~9.4]{Holevo13}, \cite[Eq.~3]{BSST99} as \[ C^{ea}(T_p) = \log d^2 + (1- p \frac{d^2-1}{d^2}) \log (1- p \frac{d^2-1}{d^2}) +p \frac{d^2-1}{d^2} \log(\frac{p}{d^2}) . \]
\end{example}

\begin{example}[Qubit Pauli channels]
    As a special case of highly symmetric quantum channels, Pauli channels \cite{NR07} can be interpreted as depolarizing quantum channels with a probability of depolarization that is axis dependent. In the qubit case, Pauli channels act on states in the following way: \[T_{p_0,p_1,p_2,p_3} (\rho)= p_0 \rho+ p_1 X\rho X +p_2 Y \rho Y +p_3 Z\rho Z\]
    where $p_i\geq 0 \ \forall i\in\{0,1,2,3\}$, and $\sum_i p_i=1$, described by three independent parameters.
    
    When two parameters are equal, we refer to the resulting Pauli channel as a covariant qubit Pauli channel, e.g. 
\[T_{p_0,p_3} (\rho)= p_0 \rho+ \frac{(1-p_0-p_3)}{2} (X\rho X + Y \rho Y) +p_3 Z\rho Z.\] 
This channel is covariant under rotations around the Z-axis; since X and Y have the same coefficient, this channel is described by two independent parameters.

The qubit depolarizing channel with depolarizing probability $p$ is a special case of a qubit Pauli channel with $p_0=(1-\frac{3p}{4})$, $p_1=p_2=p_3=(1-p_0)=p/4$.

For qubit Pauli channels, an expression for the Holevo capacity can be found in \cite{Siudzinska20}; for covariant Pauli channels, this expression can be made more explicit and was computed in \cite{PK22}. The classical capacity of a covariant qubit Pauli channel is given by
\begin{equation*}
C_H (T_{p_0,p_3})= \begin{cases}
C_1(T_{p_0,p_3}) & \text{ if } (p_0-p_3)^2 \geq (2p_0+2p_3-1)^2 \\
C_2(T_{p_0,p_3}) &\text{else}
\end{cases}
\end{equation*}
where
\[C_1(T_{p_0,p_3})=1+\frac{1-p_0+p_3}{2}\log(\frac{1-p_0+p_3}{2}) + \frac{1+p_0-p_3}{2}\log(\frac{1+p_0-p_3}{2}),\]
\[C_2(T_{p_0,p_3})=1+(p_0+p_3)\log(p_0+p_3) + (1-p_0-p_3)\log(1-p_0-p_3) .\]

The entanglement-assisted capacity of a covariant qubit Pauli channel is
\[C^{ea}(T_{p_0,p_3})=2+p_0 \log(p_0)+p_3 \log(p_3) + (1-p_0-p_3)\log( \frac{1-p_0-p_3}{2}) .\]
\end{example}


\begin{example}[Qubit amplitude damping channel]
Let $\rho=\begin{pmatrix}\rho_{00} & \rho_{01}\\  \rho_{10} &  \rho_{11}\end{pmatrix}$ be a qubit quantum state. Then, the amplitude damping channel with damping parameter $\gamma\in[0,1]$
transforms the state in the following way:
\[T_{\gamma}(\rho)=\begin{pmatrix}\rho_{00}+\gamma \rho_{11}& \sqrt{(1-\gamma)}\rho_{01}\\ \sqrt{(1-\gamma)}\rho_{10} & (1-\gamma)\rho_{11}\end{pmatrix}.\]

The Holevo capacity of the amplitude damping channel is given by \cite[Eq.~90]{BSST02} as
\begin{align*} C_{H}(T_{\gamma})& =  \max_{p\in[0,1]}  h(\gamma p) - h( \frac{1}{2} (1+\sqrt{1-4\gamma (1-\gamma)p^2})),
\end{align*}
and the entanglement-assisted capacity of the qubit amplitude damping channel is given by \cite[Eq.~88]{BSST02} as
\begin{align*} C^{ea}(T_{\gamma})& =  \max_{p\in[0,1]} h(p)+h((1-\gamma)p)- h(\gamma p) .
\end{align*}
\end{example}

\section{Quantum fault-tolerance}
\label{sec-ft-intro}

Due to the potential of leveraging quantum effects for information transmission and computation, there have been significant theoretical and experimental steps towards building and manipulating quantum systems in a controlled manner. One major challenge in this is that quantum systems are highly susceptible to fluctuations and noise, and mitigating this is the objective of designing a quantum fault-tolerant architecture. In this section, we give an overview over state-of-the-art techniques for fault-tolerance and quantum error correction. More comprehensive reviews can be found in \cite{Preskill97,Gottesman09,Steane06}.





\subsection{Noise in a quantum computer}\label{sec-intro-noise-and-errors}


In a quantum computing device, such as a superconducting qubit platform, a given quantum computing task is typically realised by a quantum circuit, which is a sequence of elementary quantum gates applied to a quantum system, typically one or multiple qubits. In other words, quantum circuits are quantum channels that can be written as a combination of elementary operations from a gate set.

In the course of this thesis, we will refer to the following list of elementary operations: 

\begin{itemize}
    \item identity gate $\mathbbm{1}=\begin{pmatrix}
    1&0\\0&1\end{pmatrix}$
    \item the three Pauli gates $X=\begin{pmatrix}
    0&1\\1&0\end{pmatrix}$, $Y=\begin{pmatrix}
    0&-i\\ i&0\end{pmatrix}$ and $Z=\begin{pmatrix}
    1&0\\0&-1\end{pmatrix}$
    \item Hadamard gate $H=\begin{pmatrix}
    1&1\\1&-1\end{pmatrix}$
    \item T gate $T=\begin{pmatrix}
    1&0\\0&e^{i\pi/4} \end{pmatrix}$
    \item CNOT gate $\CNOT_{12}=\begin{pmatrix}
    1&0&0&0\\0&1&0&0\\0&0&0&1\\0&0&1&0 \end{pmatrix}$
    \item preparation in the computational basis
    \item discarding of a qubit (i.e. performing a trace of a subsystem)
    \item measurements in the computational basis
\end{itemize}

We assume that a quantum circuit $\Gamma$ is specified by a particular decomposition, which we may call its \emph{circuit diagram}. The set of elementary operations in the circuit diagram is the circuit's \emph{set of locations} $\Loc(\Gamma)$, and the number of elementary operations in the decomposition is then denoted by $| \Loc(\Gamma)|$.

The physical setup of qubits and gates that makes up the circuit is then subject to the noise that governs quantum effects.



One source of noise in a quantum system is decoherence \cite{Zeh70,Schlosshauer19}. Decoherence occurs when a quantum system interacts with its environment and loses its coherence, or its ability to maintain a superposition of states. Experiments observing and studying decoherence effects have been performed for a variety of quantum computing setups, including superconducting qubits \cite{CNHM03,VAC+02} and trapped ion qubits \cite{HRB08,MKT+00,TMK+00}.

Another source of error during the course of a quantum computation is the noise caused by manipulating the quantum system's state by imperfect quantum gates \cite{DiVincenzo2000}. These effects can occur due to imperfect control of the gates, or due to fluctuations in the system's parameters during the time duration of the gate operation, and they have also been observed for many quantum computing platforms, including superconducting qubits \cite{GCS17,WNJRM17} and trapped ion qubits \cite{BWC+11}.

In our work in this thesis, we will largely consider circuits that are subject to noise described by the i.i.d. Pauli noise model. In this model, each gate of the circuit is affected locally by noise in the form of additional Pauli channels being inserted into the circuit diagram according to the following rules (see also Definition II.1 in \cite{CMH20}):
\begin{itemize}
\item For single qubit gates, a Pauli channel is applied after the gate.
\item For two qubit gates, a tensor product of two Pauli channels is applied after the gate.
\item For preparation gates, a Pauli channel is applied after the gate.
    \item For a measurement gate or a trace gate, a Pauli channel is applied before the gate.
\end{itemize}
The locations where a Pauli channel is inserted, and which Pauli channel is inserted, are specified by a \emph{fault pattern} $F$. Any possible fault pattern $F$ is selected with a (classical) probability $P(F)=(1-p)(p/3)^{l_x+l_y+l_z}$ where $l_k$ is the number of locations where a fault of type $k=\{\id,x,y,z\}$ appears. We will often represent a fault-affected quantum circuit as a mixture of circuits affected by different Pauli fault patterns $F$ with probability $P(F)$. In our model, each gate is faulty with probability $p$.

This model for the noise in a quantum circuit is frequently employed in theoretical works such as \cite{AGP05,Gottesman16}, and has been found to be a decent approximation to the noise observed on the experimental data of real quantum devices, for example in \cite{Google21} for superconducting qubits and \cite[Fig.~S2]{ZLZCLZK20} for trapped ion qubits. In Section~\ref{sec-overhead-and-scaling}, will briefly remark on the use of this model for our work on communication theory in Chapter~\ref{chapter-fteacap} .

We denote a quantum circuit $\Gamma$ which is affected by noise where a fault pattern $F$ occurs with a probability according to a distribution $\mathcal{F}$ by $[\Gamma]_{\mathcal{F}}$ ("the circuit $\Gamma$ affected by the noise model $\mathcal{F}$"). 

The prevalence of errors in a quantum system dictates whether or not useful long computations can be performed on a quantum computer. If all gates, preparations and measurements are noisy, it becomes almost inevitable that an error will occur during a long computation, which can have devastating effects on the result. It is therefore integral that we construct our quantum systems and quantum operations so as to protect against this noise - for example, it would seemingly be unwise to perform a gate, immediately followed by its inverse. At least to some degree, our implementation of the protocol can be either well-protected or badly protected from noise. And it would be even better if we could implement our computation in such a way that it is robust against quantum errors, so that we can detect or correct them!

\emph{Quantum fault-tolerance} is a set of techniques and methods used in quantum computing to protect against errors and ensure the reliable operation of quantum computers. The basic idea behind fault-tolerance is to design quantum algorithms and architectures in such a way that they can tolerate a certain amount of noise and errors in the system without compromising the accuracy of the computation.

This is typically achieved by employing techniques from \emph{quantum error correction} in order to make the quantum information highly redundant so that errors which affect a low enough number of qubits become correctable.

There are two crucial aspects to this.
One is the assumption that errors are local, in the sense that errors on single qubits or gates dominate over errors that affect mutual qubits at once. The other is that the redundancy of the quantum information should be achieved with a small number of resources; this is often referred to as overhead. The more qubits are used, the more likely it becomes that errors will affect the system.


This means that simply adding more qubits is not the solution; instead, the methods of quantum error correction rely on clever tricks that simultaneously correct or detect the errors while keeping the overhead low. 
In the well-known stabilizer formalism, this corresponds to finding a subgroup of the Pauli group that is well-suited for protecting the information \cite{Gottesman97}. Then, a computation can be implemented in a protected subspace with the help of fault-tolerant procedures or gadgets, which limit the propagation of error. We will introduce the mathematical formalism related to these concepts more precisely in Section~\ref{sec-error-correction-intro}.

Even as experimental setups and devices keep improving, it is believed that some quantum effects such as decoherence will never make it possible to execute gates quickly enough so as to avoid that any error happens at all \cite{Preskill18,Terhal23}. It is not unreasonable to suppose that fault-tolerance is a concept that will likely have relevance beyond near-term devices and algorithms (at least, if quantum computers will have relevance).








There exist several proposals for quantum error correcting codes with unique advantages and challenges for near-term implementation. Each quantum error correcting code comes with an overhead, connectivity, thresholds, a set of transversal gates, and various other aspects which can make them more or less suited for a computing task. We will introduce some examples of quantum error correcting codes in Section~\ref{sec-error-correction-codes-examples}.

\subsection{Fault-tolerant implementation of quantum circuits}
\label{sec-error-correction-intro}

In order to protect a quantum circuit against noise, it is integral to choose a clever implementation (by which we mean qubit setup and gate setup) that protects against errors as much as possible. Quantum error correction provides one framework for clever circuit implementation.

Given a specific quantum circuit, we will refer to its qubits as logical qubits, and its gates as logical gates. The state of a logical qubit is encoded in a fault-tolerant way into a larger number of physical qubits. For this section, we will consider codes where one single logical qubit is encoded in the quantum state of $K$ physical qubits. This can be referred to as "implementing a quantum circuit in a quantum error correcting code".

With a well-selected construction, the whole computation can then be protected against noise. This is sometimes referred to as a "fault-tolerant implementation of a quantum circuit in a quantum error correcting code" or a "fault-tolerant simulation of the quantum circuit".

The notion of a "fault-tolerant implementation" by itself is not necessarily tied to implementations in a quantum error correcting code; sometimes, a circuit implementation may be fault-tolerant without any redundancy or code-based construction, and/or due to some completely different reason. Fault-tolerance chiefly refers to the ability of an implemented circuit which is subject to some noise model to perform like the ideal, noiseless circuit with high probability, as long as the noise is not too bad. This is formally described by a \emph{threshold theorem}, which states that a circuit is fault-tolerant (logical errors occur with low probability) if the probability of a physical error occuring is below a certain threshold. (As an example, consider Theorem~\ref{thm-threshold} for the threshold theorem of the concatenated 7-qubit Steane code.)

In essence, quantum error correction works because, when assuming a local error model such as the i.i.d. Pauli model at the level of the physical qubits, the probability of our logical qubits being corrupted (decoded incorrectly) is significantly smaller. This means that well-constructed implementations in a quantum error correcting code often fulfill a threshold theorem, and can thus achieve fault-tolerance.

Let $\Pi_K$  be the Pauli group on $K$ qubits. A \emph{stabilizer code} is obtained by selecting a commuting (abelian) subgroup $\mathcal{S}$ of $\Pi_K$ that does not contain $-\mathbbm{1}^{\otimes K}$, and any stabilizer code has an associated simultaneous $+1$-eigenspace $\mathcal{C}\in(\mathbbm{C}^{2})^{\otimes K}$, which is called the \emph{code space}. This commuting subgroup $\mathcal{S}$ is referred to as the \emph{stabilizer group} of this code.
 
 Stabilizer codes with $L=\log(\dim(\mathcal{C}))$ logical qubits encoded into $K$ physical qubits, which have code distance $\delta$ are denoted by $[[K,L,\delta]]$ codes in double brackets. The code distance $\delta$ refers to the minimum weight (i.e. number of physical qubits on which it acts non-trivially) of an operator in the $N(\mathcal{S})/\mathcal{S}$, where $N(\mathcal{S})$ is the normalizer of the stabilizer group. A code with distance $\delta$ can detect up to $\delta-1$ errors, and it can correct up to $\frac{\delta-1}{2}$ errors. Therefore, in order to correct one error, a code must thus have a distance of $\delta=3$ or more.
 
 Unless stated otherwise (e.g. in Chapter~\ref{chapter-aws}), we will assume in this thesis that the code space has dimension $2$, i.e. a single qubit is encoded ($L=1$), and the stabilizer subgroup is generated by $K-1$ independent elements. We will denote these generators by $g_1,\ldots ,g_{K-1}$, and we will denote the basis of the code space by $\left\lbrace\ket{\overline{0}},\ket{\overline{1}}\right\rbrace$.

Any product Pauli operator $E:(\mathbbm{C}^{2})^{\otimes K}\rightarrow (\mathbbm{C}^{2})^{\otimes K}$ either commutes or anti-commutes with the elements of this stabilizer group and can be associated to a vector
\[
s = (s_1,\ldots , s_{K-1})\in \{0,1\}^{K-1} ,
\]
where $s_i=0$ if the Pauli operator $E$ commutes with the generator $g_i$, or $s_i=1$ if they anti-commute. This vector $s$ is referred to as the \emph{syndrome} corresponding to the error $E$.


A general quantum state can be decomposed in terms of eigenspaces associated to the syndromes, as
\[
(\mathbbm{C}^{2})^{\otimes K} = \bigoplus_{s\in\{0,1\}^{K-1}} W_s ,
\]
where $W_s$ is the common eigenspace of the operators $g_1,\ldots ,g_{K-1}$ where we have an eigenvalue $(-1)^{s_i}$ for $g_i$ for each $i$. Each $W_s$ has the same dimension as the codespace (which we take to be equal to $2$) and can be associated (non-uniquely) to some Pauli operator $E_s$ such that 
\begin{equation}  \label{eq-choosing-Paulis}
W_s = \text{span}\left\lbrace E_s\ket{\overline{0}},E_s\ket{\overline{1}}\right\rbrace.
\end{equation}
Then, we can find a basis transformation $D:(\mathbbm{C}^{2})^{\otimes K}\rightarrow \mathbbm{C}^2\otimes (\mathbbm{C}^{2})^{\otimes (K-1)}$ such that 
\begin{equation}\label{eq-choosing-Dec}
    D\left( E_s\ket{\overline{b}}\right) = \ket{b}\otimes \ket{s}
\end{equation}
for any $b\in \left\lbrace 0,1\right\rbrace$ and any $s\in\{0,1\}^{K-1}$. This selection of $E_s$ determines which errors will be correctable by the code, and we illustrate this with an example in Section~\ref{sec-error-correction-codes-examples}. We will frequently denote the zero syndrome corresponding to the error $E_0=\mathbbm{1}_2^{\otimes K}$ (informally, "no error") by $\ketbra{0}{0}\in\mathcal{M}_2^{\otimes (K-1)}$.

To map between the space of the logical qubits (containing ancillas/syndrome space) and the physical qubits, we define the quantum channels $\DecI^*:\mathcal{M}_2^{\otimes K}\rightarrow \mathcal{M}_2^{\otimes K}$ and its inverse $\EncI^*:\mathcal{M}_2^{\otimes K}\rightarrow \mathcal{M}_2^{\otimes K}$ (see also \cite[Section~IIC]{CMH20}), acting as $\DecI^*(X)=D X D^{\dagger}$, and an ideal error correcting channel:
\[\EC^*=\EncI^*\circ (\id_2 \otimes\ketbra{0}{0} \trace )\circ \DecI^*.\]

These channels are mathematical objects which are useful to our analysis of noisy gates rather than real, implemented circuits. They conveniently allow us to relate the ideal, noiseless gate decomposition of a circuit to its implementation in the error correcting code.

Gates from the ideal circuit are implemented as \emph{gadgets} on physical qubits, as defined more precisely below:
\begin{definition}[{Gadgets, \cite[Definition~II.3]{CMH20}}]\label{def-gadget} We define gadgets for 4 distinct cases of elementary gates:
\begin{enumerate}
    \item For each elementary single qubit gate $G:\mathcal{M}_2\rightarrow \mathcal{M}_2$, we define a gadget $\overline{G}:\mathcal{M}_2^{\otimes K}\rightarrow \mathcal{M}_2^{\otimes K}$ such that
\[\DecI^*\circ \ \overline{G} \circ \EncI^*( \cdot \otimes \ketbra{0}{0})= G(\cdot) \otimes \ketbra{0}{0} . \]
\item For each elementary two qubit gate $G:\mathcal{M}_2^{\otimes 2}\rightarrow \mathcal{M}_2^{\otimes 2}$, we define a gadget $\overline{G}:\mathcal{M}_2^{\otimes 2K}\rightarrow \mathcal{M}_2^{\otimes 2K}$ such that
\[(\DecI^*)^{\otimes 2}\circ \overline{G} \circ (\EncI^*)^{\otimes 2}( \cdot \otimes \ketbra{0}{0}^{\otimes 2})= G(\cdot) \otimes \ketbra{0}{0}^{\otimes 2} . \]
\item Let $G_P:\mathbbm{C} \rightarrow  \mathcal{M}_2 $ be a preparation gate. Then, we define a preparation gadget $\overline{G}_P:\mathbbm{C} \rightarrow  \mathcal{M}_2^{\otimes K}$ such that
\[\DecI^*\circ\ \overline{G}_P= G_P \otimes \ketbra{0}{0}  .\]
\item Let $G_M:\mathcal{M}_2 \rightarrow \mathbbm{C} $ be a measurement gate. Then, we define a measurement gadget $\overline{G}_M:\mathcal{M}_2^{\otimes K}\rightarrow \mathbbm{C}$ such that
\[ \overline{G}_M \circ \EncI^*( \cdot \otimes \ketbra{0}{0})= G_M(\cdot) . \]
\end{enumerate}
\end{definition}


To implement a circuit in an error correcting code such that its behaviour matches with a noiseless version of the same circuit with high probability, an error correction gadget can be performed between every operation to stop errors from accumulating. The concatenated object obtained from a gate gadget and an error correction gadget is referred to as a rectangle of an elementary operation in \cite{AGP05,CMH20}.

\begin{example} \label{ex-error-corr-implementation}
    This sketch illustrates how a small example circuit is implemented in an error-correcting code. This circuit has two logical qubits, one logical single-qubit gate $G$ and one logical two-qubit gate $G'$.
    In the implementation in an error-correcting code, each gate is replaced by the concatenation of its corresponding gadget (e.g. $\overline{G}$ or $\overline{G}'$) and an error-correction gadget $\EC$, except the measurement gate, which is just replaced by its gadget. 

\centering
\begin{tikzcd}
 \gate[style={draw, shape=semicircle, minimum size=0.5cm, xscale=0.7,yscale=0.7, inner sep=0pt, outer sep=0pt, shape border rotate=90}]{}& \gate{G} & \gate[wires=2,nwires=2]{G'} &\gate[style={draw, shape=semicircle, minimum size=0.5cm,  xscale=0.7,yscale=0.7, inner sep=0pt, outer sep=0pt, shape border rotate=270}]{}  \\
\gate[style={draw, shape=semicircle, minimum size=0.5cm, xscale=0.7,yscale=0.7, inner sep=0pt, outer sep=0pt, shape border rotate=90}]{}&\qw&\qw&\gate[style={draw, shape=semicircle, minimum size=0.5cm,  xscale=0.7,yscale=0.7, inner sep=0pt, outer sep=0pt, shape border rotate=270}]{}  \\
\end{tikzcd} \textbf{$\mapsto$}
\begin{tikzcd}
 \gate[wires=2,nwires=2,style={draw, shape=semicircle, minimum size=0.5cm, xscale=1,yscale=1, inner sep=0pt, outer sep=0pt, shape border rotate=90}]{}& \gate[wires=2,nwires=2]{\EC} & \gate[wires=2,nwires=2]{\overline{G}} & \gate[wires=2,nwires=2]{\EC} &\gate[wires=4,nwires=4]{\overline{G}'} &\gate[wires=2,nwires=2]{\EC} &\gate[wires=2,nwires=2,style={draw, shape=semicircle, minimum size=0.5cm,  xscale=1,yscale=1, inner sep=0pt, outer sep=0pt, shape border rotate=270}]{}  \\
 \qw& \qw& \qw& \qw& \qw& \qw&\qw\\
\gate[wires=2,nwires=2,style={draw, shape=semicircle, minimum size=0.5cm, xscale=1,yscale=1, inner sep=0pt, outer sep=0pt, shape border rotate=90}]{}&\gate[wires=2,nwires=2]{\EC} &\qw&\qw&\qw&\gate[wires=2,nwires=2]{\EC} &\gate[wires=2,nwires=2,style={draw, shape=semicircle, minimum size=0.5cm,  xscale=1,yscale=1, inner sep=0pt, outer sep=0pt, shape border rotate=270}]{}  \\
 \qw& \qw& \qw& \qw& \qw& \qw&\qw&\\
\end{tikzcd}
\end{example}

For any code, we can think of ways to construct such gadgets. However, this does not automatically mean that a construction is fault-tolerant; the fault-tolerance of a given implementation has to be analyzed carefully. The need for this arises from the fact that a rectangle may receive noisy input; then, limiting the number of faults we allow in each rectangle to 1 may still lead to uncorrectable errors. Aliferis, Gottesman and Preskill \cite{AGP05} therefore introduce a method wherein gadgets are grouped together with the error-correction gadgets preceding and following it, limiting the amount of errors in the rectangle \emph{and} its input. This grouping is referred to as \emph{extended rectangle (exRec)}. Then, the circuit can be regarded as a construction made up of (overlapping) exRecs, and the characteristics of this construction determine whether the implementation of a quantum circuit protects well against noise.



\begin{example}
    Consider the encoded circuit from Example~\ref{ex-error-corr-implementation}. Here, we illustrate how this circuit can be understood as a construction of exRecs, which are illustrated as blue, yellow and red boxes. 
    
\begin{center}
    \begin{quantikz}
\gate[wires=2,nwires=2,style={draw, shape=semicircle, minimum size=0.5cm, xscale=1,yscale=1, inner sep=0pt, outer sep=0pt, shape border rotate=90}]{}\gategroup[wires=2,steps=2,style={dashed,rounded
corners,draw=blue,fill=blue!20,fill opacity=0.5, inner
xsep=2pt,xscale=1.2,yscale=0.9,xshift=-0.1cm},background]{} & \gate[wires=2,nwires=2]{\EC}\gategroup[wires=2,steps=3,style={dashed,rounded
corners,fill=yellow,fill opacity=0.2, inner
xsep=2pt},background]{} & \gate[wires=2,nwires=2]{\overline{G}} & \gate[wires=2,nwires=2]{\EC} \gategroup[wires=4,steps=3,style={dashed,rounded
corners,draw=red,fill=red!20,fill opacity=0.5, inner
xsep=2pt,xscale=1,yscale=1.1},background]{} &\gate[wires=4,nwires=4]{\overline{G}'} &\gate[wires=2,nwires=2]{\EC} \gategroup[wires=2,steps=2,style={dashed,rounded
corners,draw=blue,fill=blue!20,fill opacity=0.5, inner
xsep=2pt,xscale=1.2,yscale=0.9,xshift=0.1cm},background]{} &\gate[wires=2,nwires=2,style={draw, shape=semicircle, minimum size=0.5cm,  xscale=1,yscale=1, inner sep=0pt, outer sep=0pt, shape border rotate=270}]{}  \\
  \qw& \qw& \qw&\qw& \qw& \qw& \qw\\
 \gate[wires=2,nwires=2,style={draw, shape=semicircle, minimum size=0.5cm, xscale=1,yscale=1, inner sep=0pt, outer sep=0pt, shape border rotate=90}]{}\gategroup[wires=2,nwires=2,steps=4,style={dashed,rounded
corners,draw=blue,fill=blue!20,fill opacity=0.5,xscale=1.05,yscale=0.9,xshift=-0.15cm},background]{} &\qw&\qw&\gate[wires=2,nwires=2]{\EC} &\qw&\gate[wires=2,nwires=2]{\EC} \gategroup[wires=2,steps=2,style={dashed,rounded
corners,draw=blue,fill=blue!20,fill opacity=0.5, inner
xsep=2pt,xscale=1.2,yscale=0.9,xshift=0.1cm},background]{} &\gate[wires=2,nwires=2,style={draw, shape=semicircle, minimum size=0.5cm,  xscale=1,yscale=1, inner sep=0pt, outer sep=0pt, shape border rotate=270}]{}  \\
&\qw&\qw&\qw\qw& \qw& \qw&\qw&
\end{quantikz}
\end{center}

It can clearly be seen that the exRecs corresponding to the gadgets overlap, and that their overlap contains the error correction gadgets between them.
\end{example}

Then, in order to analyze a circuit's behaviour under noise, we define a notion of \emph{well-behavedness} for exRecs under a fault pattern $F$. This is the key characteristic that exRecs in a circuit must fulfill in order to ensure that the implementation protects the information. It is defined in terms of the following transformation rules:

\begin{definition}[{Well-behaved exRecs, \cite[Definition~II.5]{CMH20}}]
\label{def-well-behaved}
For each elementary gate type $G$, we denote the corresponding gadget as in Definition~\ref{def-gadget} by $\overline{G}$. Then, we define a notion of well-behavedness for each case:
\begin{enumerate}
\item For each elementary single qubit gate $G:\mathcal{M}_2\rightarrow \mathcal{M}_2$, we define a gadget $\overline{G}:\mathcal{M}_2^{\otimes K}\rightarrow \mathcal{M}_2^{\otimes K}$. Let $E_G=\EC \ \circ \overline{G} \circ \EC $ be the corresponding exRec and let $F$ be a i.i.d. Pauli fault pattern on the quantum circuit $E_G$. Then, we call $E_G $ well-behaved under $F$ if
\[(\DecI^*) \circ [E_G]_F = (G\otimes S_G^F) \circ (\DecI^*) \circ [\EC ]_F\]
with some quantum channel $S_G^F:\mathcal{M}_2^{\otimes (K-1)}\rightarrow \mathcal{M}_2^{\otimes (K-1)}$.
\item Let $G:\mathcal{M}_2^{\otimes 2}\rightarrow \mathcal{M}_2^{\otimes 2}$ be an elementary two qubit gate, and denote the corresponding gadget by $\overline{G}:\mathcal{M}_2^{\otimes 2K}\rightarrow \mathcal{M}_2^{\otimes 2K}$. Let $E_G=\EC^{\otimes 2}\circ \overline{G} \circ \EC^{\otimes 2}$ be the corresponding exRec and let $F$ be a i.i.d. Pauli fault pattern on the quantum circuit $E_G$. Then, we call $E_G $ well-behaved under $F$ if
\[(\DecI^*)^{\otimes 2} \circ [E_G]_F = (G\otimes S_G^F) \circ (\DecI^*)^{\otimes 2} \circ [\EC^{\otimes 2}]_F\]
with some quantum channel $S_G^F:\mathcal{M}_2^{\otimes2(K-1)}\rightarrow \mathcal{M}_2^{\otimes2(K-1)}$.
\item Let $G_P:\mathbbm{C} \rightarrow  \mathcal{M}_2 $ be a preparation gate, and denote the corresponding gadget by $\overline{G}_P:\mathbbm{C} \rightarrow  \mathcal{M}_2^{\otimes K}$. Let $E_{G_P}=\EC\circ \overline{G}_P$ be the corresponding exRec and let $F$ be a i.i.d. Pauli fault pattern on the quantum circuit $E_{G_P}$. Then, we call $E_{G_P}$ well-behaved under $F$ if
\[\DecI^*\circ [E_{G_P}]_F = (G_P\otimes \sigma_{G_P}^F)\]
with some state $\sigma_{G_P}^F\in\mathcal{M}_2^{\otimes (K-1)}$.
\item Let $G_M:\mathcal{M}_2 \rightarrow \mathbbm{C} $ be a measurement gate, and denote the corresponding gadget by $\overline{G}_M:\mathcal{M}_2^{ \otimes K}\rightarrow \mathbbm{C}$. Let $E_{G_M}= \overline{G}_M \circ \EC$ be the corresponding exRec and let $F$ be a i.i.d. Pauli fault pattern on the quantum circuit $E_{G_M}$. Then, we call $E_{G_M}$ well-behaved under $F$ if
\[[E_{G_M}]_F = ({G_M} \otimes\trace) \circ \DecI^*\circ [\EC]_F.\]
\end{enumerate}
Note that the above transformation rules are sketched in Figure~\ref{fig-well-behaved}.
\end{definition}

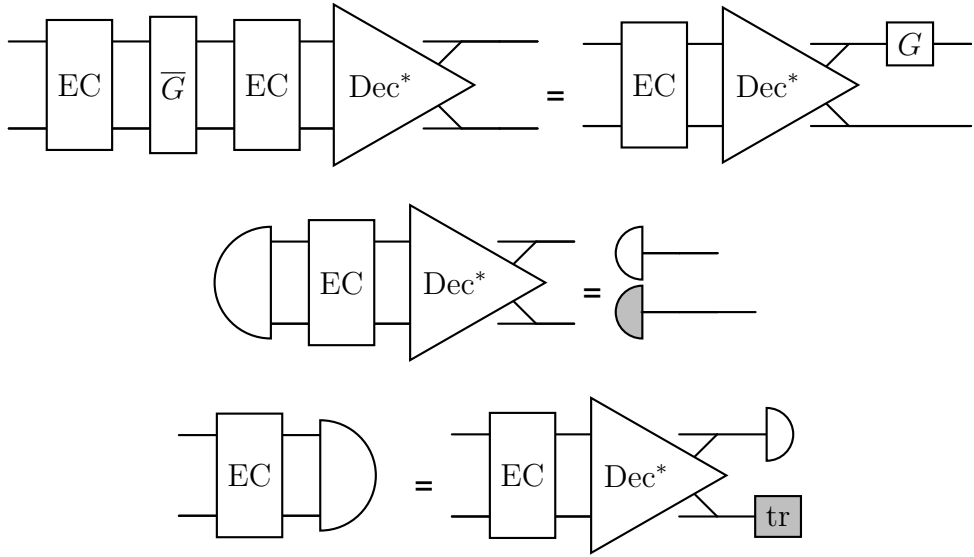
\begin{figure}[htbp]
   \centering
\begin{tikzcd}
\qw & \gate[wires=2,nwires=2]{\EC} &\gate[wires=2,nwires=2]{\overline{G}}  &\gate[wires=2,nwires=2]{\EC} & \gate[wires=2,nwires=2,style={draw, shape=regular polygon, regular polygon sides=3, minimum size=0.5cm, xscale=0.5,yscale=0.5, inner sep=0pt, outer sep=0pt, shape border rotate=-90}]{\DecI^*}& \makeebit{}&\qw&\qw \\
\qw &\qw&\qw&\qw& \qw& & \qw&\qw\\
\end{tikzcd}
\textbf{=}
\begin{tikzcd}
\qw & \gate[wires=2,nwires=2]{\EC} & \gate[wires=2,nwires=2,style={draw, shape=regular polygon, regular polygon sides=3, minimum size=0.5cm, xscale=0.5,yscale=0.5, inner sep=0pt, outer sep=0pt, shape border rotate=-90}]{\DecI^*}& \makeebit{}&\gate{G}&\qw \\
\qw &\qw&\qw& & \qw&\qw  \\
\end{tikzcd}

\begin{tikzcd}
\gate[wires=2,nwires=2,style={draw, shape=semicircle, minimum size=0.5cm,  xscale=1,yscale=1, inner sep=0pt, outer sep=0pt, shape border rotate=80}]{} & \gate[wires=2,nwires=2]{\EC}& \gate[wires=2,nwires=2,style={draw, shape=regular polygon, regular polygon sides=3, minimum size=0.5cm, xscale=0.5,yscale=0.5, inner sep=0pt, outer sep=0pt, shape border rotate=-90}]{\DecI^*}&
\makeebit{} &\qw  \\
&\qw&\qw& &\qw\\
\end{tikzcd}
\textbf{=}
\begin{tikzcd}
\gate[style={draw, shape=semicircle, minimum size=0.5cm,  xscale=0.7,yscale=0.7, inner sep=0pt, outer sep=0pt, shape border rotate=90}]{} &\qw &\qw   \\
\gate[style={draw, shape=semicircle, minimum size=0.5cm,  xscale=0.7,yscale=0.7, inner sep=0pt, outer sep=0pt, shape border rotate=90,fill=gray, fill opacity=0.5}]{} &\qw&\qw&\\
\end{tikzcd}

 \begin{tikzcd}
\qw & \gate[wires=2,nwires=2]{\EC}& \gate[wires=2,nwires=2,style={draw, shape=semicircle, minimum size=0.5cm,  xscale=1,yscale=1, inner sep=0pt, outer sep=0pt, shape border rotate=270}]{} \\
\qw &\qw&\qw\\
\end{tikzcd}\hspace{0.5cm}\textbf{=}
\begin{tikzcd}
\qw & \gate[wires=2,nwires=2]{\EC}& \gate[wires=2,nwires=2,style={draw, shape=regular polygon, regular polygon sides=3, minimum size=0.5cm, xscale=0.5,yscale=0.5, inner sep=0pt, outer sep=0pt, shape border rotate=-90}]{\DecI^*}&
\makeebit{}  &\gate[style={draw, shape=semicircle, minimum size=0.5cm,  xscale=0.7,yscale=0.7, inner sep=0pt, outer sep=0pt, shape border rotate=270}]{}  \\
\qw&\qw&\qw&&\gate[style={fill=gray, fill opacity=0.5}]{\trace} \\
\end{tikzcd}

\caption{Pictorial representation of well-behaved exRecs for cases 1,3 and 4 from Definition~\ref{def-well-behaved}. We highlight the operations on the syndrome space in grey.}
\label{fig-well-behaved}
\end{figure}


\begin{definition}[Well-behaved quantum circuits]
If every exRec in an implementation of a circuit $\Gamma$ in an error correcting code $\mathcal{C}$ which is affected by some Pauli fault pattern $F$ is well-behaved, we call the implementation of the circuit $[\Gamma_{\mathcal{C}}]_{F}$ well-behaved under $F$.
\end{definition}

If a circuit is well-behaved, 
its subparts can be transformed in the manner represented by Figure~\ref{fig-well-behaved}, and the whole circuit can be related to the corresponding ideal circuit decoupled from the syndrome space:

\begin{lemma}[{\cite[Lemma~4]{AGP05} and \cite[Lemma~II.6]{CMH20}}]
 Let $\mathcal{C}\in (\mathbbm{C}^2)^{\otimes K}$ be a stabilizer code with $\dim(\mathcal{C})=2$, and let $\DecI^*$ be the ideal decoding operation for this code. Let $F$ be a Pauli fault pattern. Let $\Gamma:\mathcal{M}_2^{\otimes n}\rightarrow \mathcal{M}_2^{\otimes m}$ be a quantum circuit. If $[\Gamma_{\mathcal{C}} \circ \EC^{\otimes n}]_{F}$ is well behaved under $F$, there exists a quantum channel $S_F:\mathcal{M}_2^{\otimes n(K-1)}\rightarrow \mathcal{M}_2^{\otimes m(K-1)}$ such that 
    \[(\DecI^* )^{\otimes m} \circ [\Gamma_{\mathcal{C}} \circ \EC^{\otimes n} ]_F = (\Gamma \otimes S_F) \circ(\DecI^* )^{\otimes n}\circ [ \EC^{\otimes n} ]_F .\]
\end{lemma}

\begin{example}\label{ex-if-every-exrec-is-wb}
   As an example of applying the transformation rules for exRecs in Definition~\ref{def-well-behaved} to a well-behaved circuit, we choose the circuit from Example~\ref{ex-error-corr-implementation} and illustrate each step of the transformation in Figure~\ref{fig-if-every-exrec-is-wb}.
\begin{figure}[htbp]
\centerfloat
       \includegraphics[width=1.25\textwidth,angle=270]{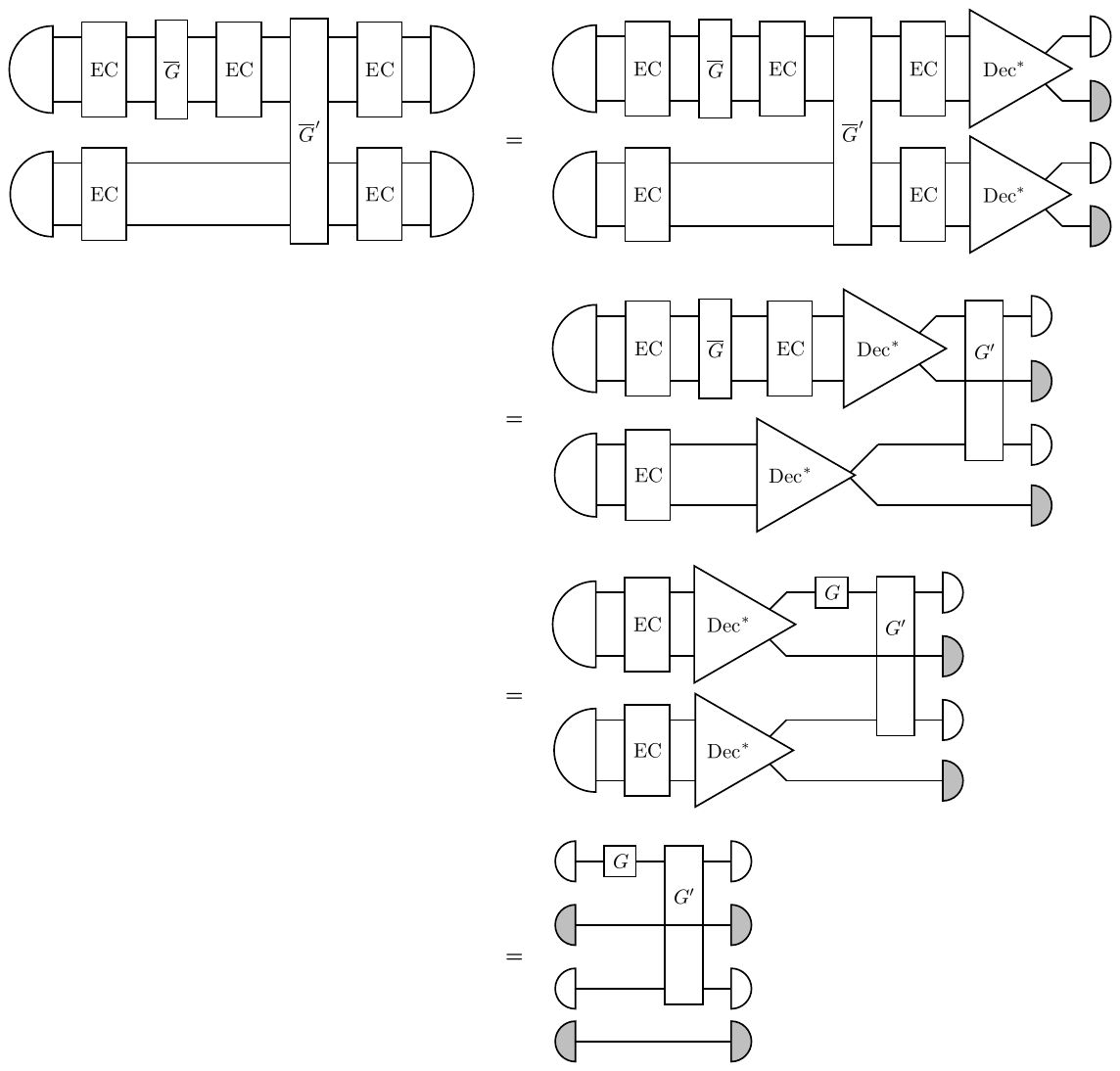}
 \caption{If every exRec in a circuit is well-behaved, the transformation rules ensure that all of the circuits above are equivalent. Then, the encoded circuit simulation in the beginning is equivalent to the unencoded circuit on the data qubits, and the data qubits are completely decoupled from the syndrome space (highlighted in grey).
 }
 \label{fig-if-every-exrec-is-wb}
\end{figure}
 
\end{example}


For the above implementation in terms of gadgets and rectangles, the probability of a given quantum circuit being fault-tolerant can be bounded by a union bound over the probability that each exRec in the circuit is well-behaved. This probability, if it is high enough for a specific implementation of exRecs, may then allow us to formulate a threshold theorem, e.g. Theorem~\ref{thm-threshold} for the concatenated 7-qubit Steane code.

\subsection{Examples of error correcting codes}
\label{sec-error-correction-codes-examples}

A large variety of quantum error correcting codes has been proposed based on existing classical schemes and on the unique properties of quantum systems. It is an object of highly active research to develop these codes and adapt them to the needs of near-term quantum hardware. The examples of error correcting codes in this section and many others can be found in \cite{Steane06,Gottesman09,ErrorCorrectionZoo}, and an overview over state-of-the-art experimental implementations of different quantum error correction and detection codes can be found in \cite[Table~IV]{AAA+22}.

\subsubsection{The 3-qubit Repetition code}

The 3-qubit repetition code is associated to the stabilizer generators $Z_1Z_2$ and $Z_2 Z_3$.
Associated with these stabilizers are the following simultaneous eigenstates with eigenvalue +1:
\[\ket{\overline{0}}=\ket{000}, \]
\[\ket{\overline{1}}=\ket{111}, \]
forming a basis of the codespace, which is in this case a 2-dimensional subspace of the 8-dimensional space. Then, for $\alpha,\beta\in \mathbbm{C}$, $|\alpha|^2+|\beta|^2=1$, a state $\alpha\ket{0}+\beta\ket{1}$ would be encoded as $\alpha\ket{\overline{0}}+\beta\ket{\overline{1}}=\alpha\ket{000}+\beta\ket{111}$.

On this code, the logical $X$-gate is the gate  $\overline{X}:\ket{\overline{0}}\mapsto\ket{\overline{1}},\ket{\overline{1}}\mapsto\ket{\overline{0}}$, which would be implemented by the gate $X_1X_2X_3$ applied to the physical qubits.

This code can correct any $X$-error, so long as only one physical qubit is affected. It is therefore not a very practical code, as it is ineffective against phase errors. Nonetheless, we will use it to illustrate some of the abstract objects from Section~\ref{sec-error-correction-intro}.

\begin{table}[htbp]
    \centering
    \begin{tabular}{c|c|c|c}
         Syndrome & $X_1$ & $X_2$ &$X_3$   \\ \hline
         $Z_1Z_2$ & 1 & 1& 0  \\
$Z_2 Z_3$      &0 & 1 & 1 \\  
        \end{tabular}
    \caption{Table outlining the relation between the generators and single-qubit $X$-errors. $X_i$ means that a Pauli $X$-error occurs on the i-th qubit, while $\mathbbm{1}_2$ is applied the other two qubits. A value of 0 signifies that the error and the generator commute, while 1 signifies that they anti-commute.}
    \label{tab-repetition-code}
\end{table}

If a single qubit $X$-error happens on the second physical qubit, our system would be in the state $\alpha\ket{010}+\beta\ket{101}$, and our syndrome measurement would indicate that the parity of the first and second physical qubit (associated to the $Z_1Z_2$ stabilizer, with$(\alpha\bra{010}+\beta\bra{101})Z_1Z_2 (\alpha\ket{010}+\beta\ket{101})=-|\alpha|^2-|\beta|^2=-1$) is $1$ and that the parity of the second and third qubit is $1$, giving a bitstring $11$ as the resulting syndrome.

Table~\ref{tab-repetition-code} indicates that each single qubit $X$-error leads to a distinct syndrome pattern, and thus allows us to produce a guess about the error that happened. If our syndrome measurement returns the bitstring $11$, we would associate this with an $X$-error on the second qubit, and we would apply $X_2$ to the state in order to correct it. This error would thus be a correctable error.


There are different error patterns that could give rise to the same syndrome string - for example, $X_1X_3$ would also give rise to the syndrome $11$. We thus have some freedom in associating a Pauli error $E_s$ to each syndrome, as outlined in Eq.~\eqref{eq-choosing-Paulis}. Depending on our choice, different errors become classified as correctable. Due to our assumption of locality, it follows that an error on two qubits (in this case qubit 1 and qubit 3) is less likely ($\sim p^2$) than an error on one single qubit (in this case qubit 2, with probability $\sim p$), and would thus be inclined to associate the latter error with the syndrome $11$. In an alternative error model, where we might suppose that errors on 2 physical qubits are much more likely than single-qubit errors, the opposite choice would be preferable.

For each of the four possible syndromes, we thus associate them with the following $2$-dimensional subspaces:
\begin{itemize}
    \item $W_{00}=\{ \ket{\overline{0}},\ket{\overline{1}} \}$, 
    \item $W_{10}=\{ X_1 \ket{\overline{0}}, X_1 \ket{\overline{1}} \}$, 
    \item $W_{11}=\{ X_2 \ket{\overline{0}}, X_2 \ket{\overline{1}} \}$, 
    \item $W_{01}=\{ X_3 \ket{\overline{0}}, X_3 \ket{\overline{1}} \}$.
\end{itemize}

The space $\bigoplus_{i,j} W_{ij}$ contains every 3-qubit state.

We define a unitary change of basis such that \begin{equation*}
    \begin{split}
      D: &  \ket{000}\mapsto\ket{000},\ket{100}\mapsto\ket{001},\ket{010}\mapsto\ket{011},\ket{001}\mapsto\ket{001},\\& \ket{111}\mapsto\ket{100},\ket{011}\mapsto\ket{101},\ket{101}\mapsto\ket{111},\ket{110}\mapsto\ket{101}
    \end{split}
\end{equation*}
which can be used to construct the decoder in Eq.~\eqref{eq-choosing-Dec}.

Consider now the error $X_1 X_3$ happening on the state $\ket{\overline{0}}$ (logical $\ket{0}$). Then, the map from Eq.~\eqref{eq-choosing-Dec} would map this to $D(X_1X_3 \ket{\overline{0}})=D(\ket{101})=D(X_2\ket{\overline{1}})=\ket{1}\otimes \ket{11}$. Now, we would observe the syndrome $11$, supposing that an $X_2$ error happened. To correct this error, we would thus apply $X_2$ - leading to an overall logical error, where our system would be in state $\ket{\overline{1}}$ after the correction, when it should be in $\ket{\overline{0}}$. This error would still be referred to "correctable", and the corrected state is still in the code space. We would thus like this error (and other two-qubit or three-qubit $X$-errors) to be sufficiently unlikely to occur, and the performance of our error correcting code is linked to the likelihood of such logical errors.



Any $Z$-error on one, two or three physical qubits commutes with the generators, and would thus be an uncorrectable error.

In a model where phase flips occur with the same probability as bit flips, this code is therefore not an ideal candidate for achieving fault-tolerance. By change of basis, we can also define a 3-qubit phase-flip repetition code that corrects all single qubit $Z$-errors, but no $X$-errors; to correct both at the same time, more qubits are needed. Among the first codes to be proposed which can correct all single qubit errors were the 7-qubit Steane code [[7,1,3]] \cite{Steane96} (see below) and the 9-qubit Shor code [[9,1,3]] \cite{Shor95}, followed by the 5-qubit code [[5,1,3]] \cite{LMPZ96,BDSW96}.



\subsubsection{The [[4,2,2]] error detection code}

The [[4,2,2]] error detection code \cite{VGW96,GBP97} encoded two logical qubits into 4 physical qubits and can detect all single-qubit errors. The code states are given by:

\[\ket{\overline{00}}= \frac{1}{\sqrt{2}} (\ket{0000}+\ket{1111}),\]
\[\ket{\overline{01}}= \frac{1}{\sqrt{2}} (\ket{1100}+\ket{0011}),\]
\[\ket{\overline{10}}= \frac{1}{\sqrt{2}} (\ket{1010}+\ket{0101}),\]
\[\ket{\overline{11}}= \frac{1}{\sqrt{2}} (\ket{0110}+\ket{1001}).\]

These states span a subspace of dimension 4 of the 4-qubit space of dimension 16, which is the subspace associated to the $+1$-eigenvectors of the operators $X_1X_2X_3X_4$ and $Z_1Z_2Z_3Z_4$, which form this code's stabilizer group for error detection.

Logical $\overline{X}_1$ on the first qubit is implemented by $X_1 X_3$ on the physical qubits, and logical $\overline{X}_2$ acting the second qubit is implemented by $X_1 X_2$ on the physical qubits. Since there are at least two qubits involved in these logical gates, the code distance is $2$; it can thus not correct any errors, and detect exactly one.

Because of its low qubit number, and also because implementing error detection is currently more feasible than implementing error correction, the [[4,2,2]] code is particularly well-suited for demonstrations of the principles of fault-tolerance on real-life current quantum devices. In Chapter~\ref{chapter-aws}, the state-of-the-art of implementations of this code is discussed in more detail and supplemented with our own
investigations for trapped ion quantum computers.

\begin{remark} A [[4,1,2]] Bacon-Shor subsystem error detection code has been proposed in \cite{SPS21}. While stabilizer codes rely on measurements of the stabilizer group, subsystem codes construct a gauge group with the stabilizer group as its center. For complicated stabilizers, this can simplify error detection or correction because these gauge operators can be easier to measure in terms of how many gates and qubit connections are needed. For example, the parity measurements proposed in the [[4,2,2]] code cannot be performed fault-tolerantly; after the measurement has been performed, the state is not guaranteed to be in the code space. Based on the [[4,2,2]] Bacon-Shor error detection circuit, the [[4,1,2]] Bacon-Shor subsystem code can only protect one logical qubit instead of two, but has fault-tolerant error detection circuits \cite{SPS21}. Due to this, these error detection circuits can be performed intermittently, for example after every logical gate, which would prohibit some combinations of errors that would not be detected in the [[4,2,2]] code. \end{remark} 


\subsubsection{The 7-qubit Steane code}

\begin{figure}[htpb]
    \centering
    \includegraphics[width=8cm]{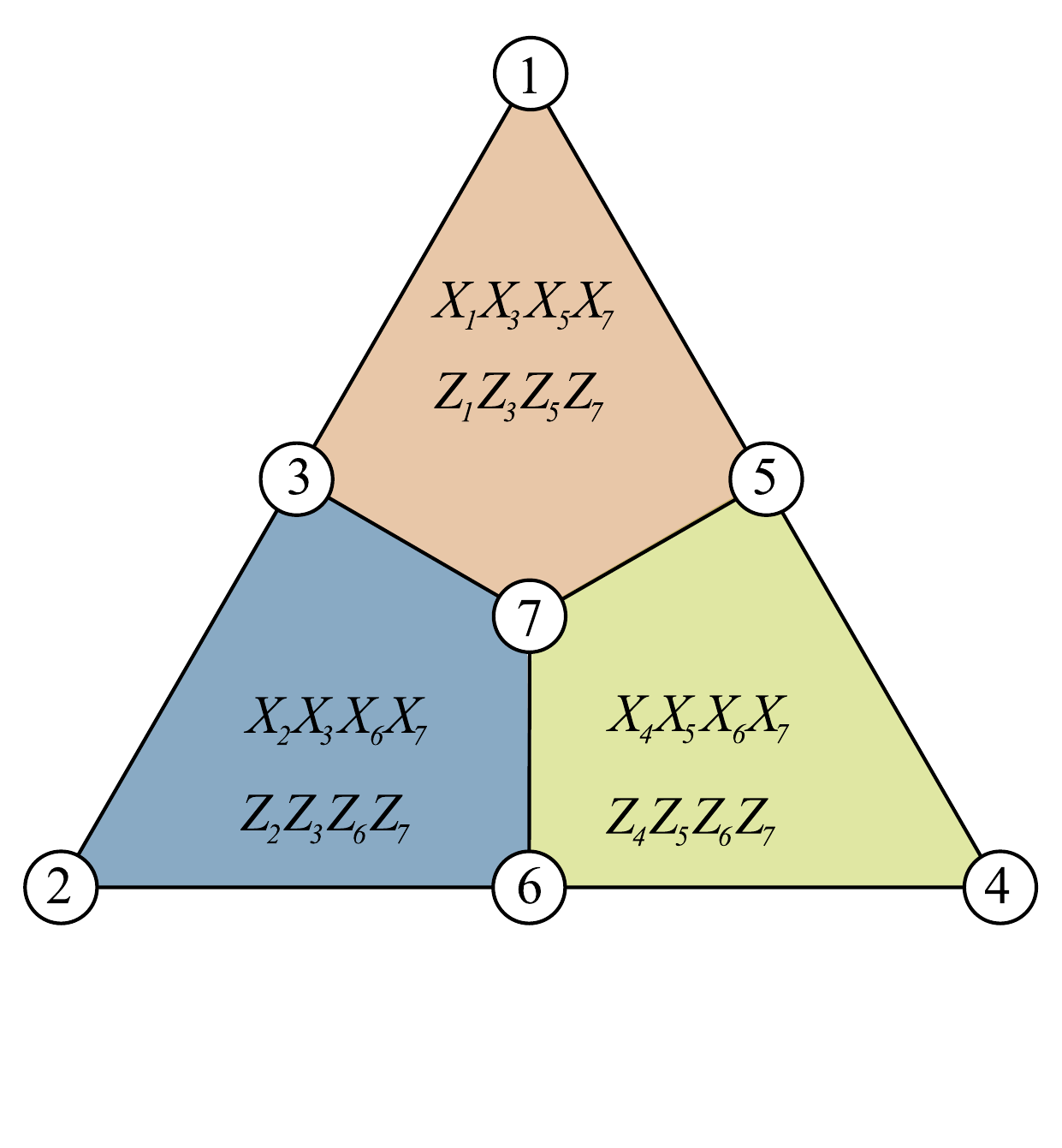}
    \caption{Illustration of the connectivity layout for the 7-qubit Steane code because of the qubits that have to interact with the same ancilla for the stabilizer measurement.}
    \label{fig:enter-label}
\end{figure}

The 7-qubit Steane code [[7,1,3]] \cite{Steane96} encodes one logical qubit into 7 physical qubits, and can correct all single-qubit errors. The stabilizer generators for the code are given by the following set:
\[g_1=X_4 X_5 X_6 X_7, \]
\[g_2=X_2 X_3 X_6 X_7, \]
\[g_3=X_1 X_3 X_5 X_7, \]
\[g_4= Z_4 Z_5 Z_6 Z_7, \]
\[g_5=Z_2 Z_3 Z_6 Z_7, \]
\[g_6 =Z_1 Z_3 Z_5 Z_7 .\]

The stabilizers mutually commute, $[g_i,g_j]=0$ $\forall i, j$, and have the joint eigenvectors $\ket{\bar{0}}$ and $\ket{\bar{1}}$ with eigenvalue $+1$, where
\begin{equation*}\begin{split}  \ket{\bar{0}}_{7} &=
 \frac{1}{8} \Big(\ket{0000000}+\ket{0001111} + \ket{0110011}+\ket{1010101} \\& \hspace{1cm} +\ket{0111100}+\ket{1011010}+\ket{1100110} +\ket{1101001}\Big) \end{split}\end{equation*}
and
\begin{equation*}\begin{split}   \ket{\bar{1}}_{7} 
&= \frac{1}{8} \Big(\ket{1111111}+\ket{1110000} + \ket{1001100}+\ket{0101010} \\& \hspace{1cm} +\ket{1000011}+\ket{0100101}+\ket{0011001} +\ket{0010110}\Big) .\end{split}\end{equation*}
These two states thus span the code space of the 7-qubit Steane code. Any $7$-qubit state can then be written as a combination in the basis of the code space and the subspaces from Eq.~\eqref{eq-choosing-Paulis}.


Logical Pauli gates are given by $\overline{X}=X^{\otimes 7}$ and $\overline{Z}=Z^{\otimes 7}$.

Any single qubit error applied to these states can be corrected, and the syndromes associated with $X$ and $Z$-errors are listed in Table~\ref{tab-steane-x} and \ref{tab-steane-z}. 

\begin{table}[h!]
    \centering
    \begin{tabular}{c|c|c|c|c|c|c|c}
         Syndrome & $X_1$ & $X_2$ &$X_3$ &$X_4$ &$X_5$ &$X_6$ &$X_7$  \\ \hline
                                   $ g_4$ & 0 & 0& 0& 1&1&1&1 \\
                                             $g_5$ & 0 & 1& 1& 0& 0&1&1 \\
                                                      $g_6$ & 1 & 0& 1& 0& 1&0&1 \\
                                                      \end{tabular}
    \caption{Table outlining the relation between the generators and various single qubit errors. $X_i$ means that a Pauli $X$-error occurs on the i-th qubit, while $\mathbbm{1}$ is applied the other (6) remaining qubits. A value of 0 signifies that the error and the generator commute, while 1 signifies that they anti-commute.}
\label{tab-steane-x}
\end{table}

\begin{table}[h!]
    \centering
    \begin{tabular}{c|c|c|c|c|c|c|c}
         Syndrome & $Z_1$ & $Z_2$ &$Z_3$ &$Z_4$ &$Z_5$ &$Z_6$ &$Z_7$  \\ \hline
                                   $ g_1$ & 0 & 0& 0& 1&1&1&1 \\
                                             $g_2$ & 0 & 1& 1& 0& 0&1&1 \\
                                                      $g_3$ & 1 & 0& 1& 0& 1&0&1 \\
                                                      \end{tabular}
    \caption{Table outlining the relation between the generators and various single qubit errors. $Z_i$ means that a Pauli $Z$-error occurs on the i-th qubit, while $\mathbbm{1}$ is applied the other (6) remaining qubits. A value of 0 signifies that the error and the generator commute, while 1 signifies that they anti-commute.}
    \label{tab-steane-z}
\end{table}

\subsubsection{The concatenated 7-qubit Steane code}

The 7-qubit Steane code can be concatenated in order to tolerate higher amounts of noise \cite{KL96}. In essence, concatenation refers to a single qubit being encoded into 7 qubits, and each of the 7 qubits being encoded again into 7 qubits each, so that a single logical state is encoded into $7^l$ qubits at concatenation level $l$. For a concatenated code with level $l$, the ideal decoder described in Eq.~\eqref{eq-choosing-Dec} is defined recursively as $\DecI_l^*:\mathcal{M}_2^{7^l}\rightarrow \mathcal{M}_2^{7^l}$ with \[\DecI_l^* = (\DecI_1^*\otimes \id_2^{\otimes 7(7^{l-1}-1)} )\circ (\DecI_{l-1}^*)^{\otimes 7}\]
where $\DecI_1^*=\DecI^*$ is the ideal decoder associated to the 7-qubit Steane code. An analogous recursive construction can be found for the encoder $\EncI_l^*:\mathcal{M}_2^{7^l}\rightarrow \mathcal{M}_2^{7^l}$ with 
\[\EncI_l^* =  (\EncI_{1}^*)^{\otimes 7^{l-1}} \circ (\EncI_{l-1}^*\otimes \id_2^{\otimes (7-1)(7^{l-1}-1)} ) \]
and $\EncI_1^*=\EncI^*$.
Figure~\ref{fig-dec} illustrates this concatenation structure further.

\begin{figure}[h!]
    \centering
\includegraphics[width=12cm]{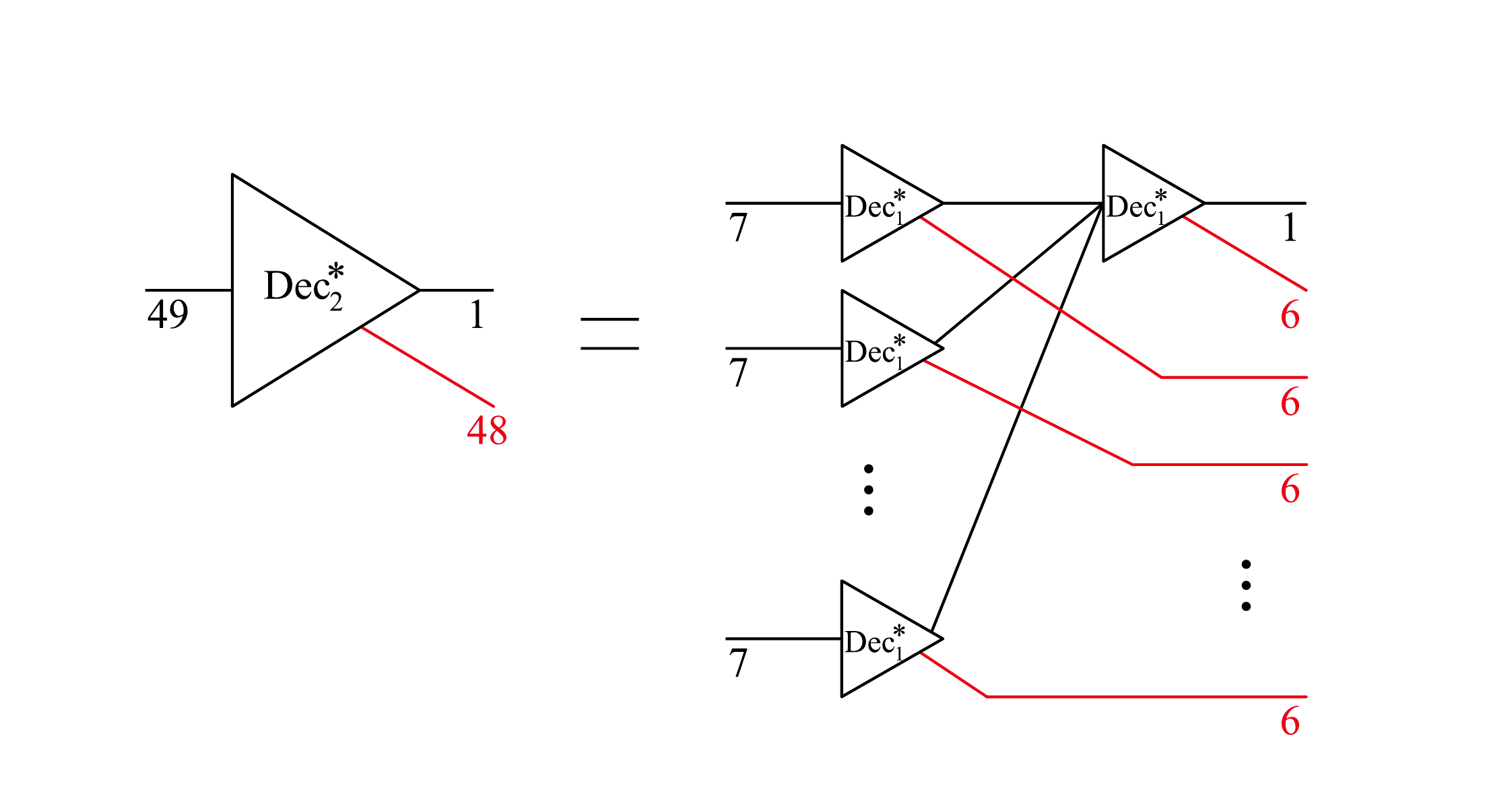}
    \caption{Sketch of the blockwise application of $\DecI^*$ for the concatenated 7-qubit Steane code in order to mathematically analyze the circuit. Here, we draw the construction for the second level $l=2$, where $7^2=49$ qubits are employed to encode one qubit. 
    }
    \label{fig-dec}
\end{figure}

As discovered in \cite{AGP05}, the exRecs in the concatenated 7-qubit Steane code are likely to be well-behaved, making a circuit implemented in the concatenated 7-qubit Steane code fault-tolerant with a high degree of protection if the level of concatenation is sufficient. This is made more precise by the following threshold theorem:

\begin{theorem}[{Threshold theorem for the concatenated 7-qubit Steane code, \cite{AGP05,CMH20}}]
\label{thm-threshold} For each $l\in \mathbbm{N}$, let ${\mathcal{C}_l}$ denote the $l$-th level of the concatenated 7-qubit Steane code with threshold $p_0$. For any quantum circuit $\Gamma:\mathcal{M}_{d_A} \rightarrow \mathcal{M}_{d_B}$, for any level $l$ and any $0\leq p <p_0$, we have 
\[P(\text{An exRec in $[\Gamma_{\mathcal{C}_l}\circ \EC]_{\mathcal{F}(p)}$ is not well-behaved}) \leq2p_0(\frac{p}{p_0})^{2^l} |\Loc(\Gamma)| . \]
The probability is taken over the distribution of $F$ according to the fault model $\mathcal{F}(p)$.

For a circuit with classical input and output, we have
\[ | \Gamma-[\Gamma_{\mathcal{C}_l}]_{\mathcal{F}(p)} | \leq 2p_0(\frac{p}{p_0})^{2^l} |\Loc(\Gamma)|.\]
\end{theorem}

\section{Fault-tolerance and communication}
\label{sec-ft-and-comm}

In Shannon theory, it is an important objective to find the best possible encoder and decoder for a given classical channel in order to achieve the capacity as the communication rate. These encoder and decoder maps can be implemented in circuits, and it is usually assumed that these circuits consist of noise-free gates. This is largely justified for modern classical computers, where gate errors are negligible, in particular at the timescales relevant for communication protocols. In quantum computers, as we have discussed in Section~\ref{sec-ft-intro}, it is not expected that the error-rates of quantum gates will vanish in the near- and even longer term \cite{Preskill18,Terhal23}. Due to this feature of quantum noise, the same assumption as in the classical case of error-free gates, which has been made in most communication scenarios studied in quantum Shannon theory so far, may not be realistic.

Furthermore, many techniques from quantum fault-tolerance often cannot directly be applied to the problem of communication, or will only allow for significantly reduced communication rates. Naive strategies with one (large) fault-tolerant implementation, where the communication channel is considered as part of the circuit noise, will only give rates approaching zero due to their high overhead encodings, and they will only work for channels which are very close to the identity (i.e. with noise below the threshold).

Fortunately, we can indeed construct coding schemes with a fault-tolerant encoder and a fault-tolerant decoder that allow for communication at rates that almost achieve capacity for all quantum channels in the case of classical and quantum capacity \cite{CMH20}, and we extend these results to the case of entanglement-assisted capacity in our work in Chapter~\ref{chapter-fteacap}. This crucially means that the noise in the communication line does not have to be below a threshold and can be treated separately from the noise that locally affects the encoder circuit and the noise that locally affects the decoder circuit. Thereby, these results have implications not only for communication over large distances but also for on-chip communication within a quantum computer, where different parts of the computer may experience varying levels of noise.

A key player in our strategy is an additional circuit which we call \emph{interface}, which is applied in between our fault-tolerant circuit implementations and the communication channel.

\begin{definition}[Interfaces] \label{def-interfaces}
 Let $\mathcal{C}\in (\mathbbm{C}^2)^{\otimes K}$ be a stabilizer code with $\dim(\mathcal{C})=2$, and let $\ketbra{0}{0}\in \mathcal{M}_2^{\otimes K-1}$ denote the state corresponding to the zero-syndrome. Let $\EncI^*:\mathcal{M}_2^{\otimes K}\rightarrow\mathcal{M}_2^{\otimes K}$ and $\DecI^*:\mathcal{M}_2^{\otimes K}\rightarrow\mathcal{M}_2^{\otimes K}$ be the ideal encoding and decoding operations. Then, we have:
\begin{enumerate}
    \item An encoding interface $\EncI:\mathcal{M}_2\rightarrow\mathcal{M}_2^{\otimes K}$ for a code $\mathcal{C}$ is a quantum circuit with an error correction as a final step, and fulfilling \[\DecI^*\circ \EncI =\id_2 \otimes \ketbra{0}{0} .\]
    \item A decoding interface $\DecI:\mathcal{M}_2^{\otimes K}\rightarrow\mathcal{M}_2$ is a quantum circuit fulfilling \[\DecI \circ \EncI^* (\cdot \otimes \ketbra{0}{0}) =\id_2 (\cdot) .\]
    \end{enumerate}
\end{definition}


These interfaces are quantum circuits which are subject to noise in the same way as the encoder and decoder circuit. Since this can lead to faulty inputs to a quantum channel, we will need interfaces that are protected enough against quantum errors. For concatenated codes, such interfaces have been constructed and proven to fail with low enough probability in \cite{MGHHLPP14,CMH20}:

\begin{theorem}[{Correctness of interfaces for the concatenated 7-qubit Steane code, see \cite[Theorem~III.3]{CMH20}}] \label{thm-correct-interfaces}
For each $l\in \mathbbm{N}$, let ${\mathcal{C}_l}$ denote the $l$-th level of the concatenated 7-qubit Steane code with threshold $p_0$. Then,
there exist interface circuits $\EncI_l:\mathcal{M}_2 \rightarrow \mathcal{M}_2^{\otimes 7^l}$ and $\DecI_l:\mathcal{M}_2^{\otimes 7^l}\rightarrow\mathcal{M}_2$  for the $l$-th level of this code such that for any $0\leq p \leq \frac{p_0}{2}$, we have
\begin{enumerate}
    \item \[\Prob( \big[\EncI_l \big]_F \text{ is not correct}) \leq 2cp,\] where $\EncI_l$ is correct under a Pauli fault pattern $F$ if there exists a quantum state $\sigma_S(F)$ on the syndrome space such that $\DecI^*\circ \big[\EncI \big]_F = \id_2 \otimes \sigma_S(F)$. The probability is taken over the distribution of $F$ according to the fault model $\mathcal{F}(p)$.
    \item \[\Prob( \big[\DecI_l \big]_F \text{ is not correct}) \leq 2cp,\] where $\DecI$ is correct under a Pauli fault pattern $F$ if $  \big[\DecI_l\circ \EC_l \big]_F = (\id_2 \otimes \Tr_S)\circ \DecI_l^*\circ  \big[\EC_l \big]_F $ where $\Tr_S$ traces out the syndrome space.
\end{enumerate}
Here, $c=p_0 \max \{|\Loc(\EncI_1)|,|\Loc(\DecI_1\circ \EC)|\}$ is a constant that does not depend on $l$ or $p$.
\end{theorem}

The fact that the interfaces for the concatenated 7-qubit Steane code are correct with high probability is crucial for \cite[Lemma~III.8]{CMH20} as well as Theorem~\ref{thm-effective-encoder}
and \ref{thm-eff-decoder}, which allows us to regard the communication channel and the interfaces together as an effective channel, for which we design a coding scheme in Chapter~\ref{chapter-fteacap}.

This implication is not straightforward because the notions of fault-tolerance of a circuit and correctness of an interface have overlap in an error-correction gadget. More precisely, in the case of the communication scheme's encoder, the well-behavedness condition (according to Definition~\ref{def-well-behaved}) of the last exRec before the interface refers to an error-correction gadget that also appears in the correctness condition for decoding interfaces in Theorem~\ref{thm-correct-interfaces}. This overlap requires careful analysis and prevents us from seeing the fault-tolerant circuit and the interfaces as separate. As a consequence of this overlap, the interactions in the fault-tolerant implementation of the coding scheme's encoder may also lead to correlations between separate decoding interfaces. In Theorem~\ref{thm-effective-encoder}, we see that this can be modelled by an i.i.d. effective channel with correlated input, which is crucial in order to apply known techniques from quantum Shannon theory in our work in Chapter~\ref{chapter-fteacap}.

The potential for correlated syndrome states was first noted in \cite{AGP05,AGP07}, where it was also shown that the correlations do not prevent us from formulating a threshold theorem for implementations of quantum circuits that begin and end in classical information. In our communication setup, our circuit produces quantum states that are sent through the quantum channel, where the quantum channel may now be influenced by these correlations, which has to be studied more closely; to what extent this affects the communication depends on the structure of the effective syndrome state. A highly entangled syndrome state does not necessarily lead to an uncorrectable syndrome, but this depends on the underlying circuits. In our work, we prove a coding theorem under worst-case assumptions for a highly correlated syndrome state; the result could perhaps be improved for more specific noise models or circuits, where more is known about the structure of the syndrome state.

\begin{example}[A correlated syndrome state]
Consider two logical qubits in the state $\ket{\bar{0}} \otimes \ket{\bar{0}}$. Assume a Hadamard error happens on one of the physical qubits, say $H_1=\frac{X_1+Z_1}{2}$, followed by a logical $\CNOT$ gate with the control on the first logical qubit. The $\CNOT$ gate is transversal in the 7-qubit Steane code, and the error would thus spread via the physical $\CNOT$ between the first qubits of the code block to become $\frac{1}{\sqrt{2}} (X_1 \otimes X_1 + Z_1\otimes \mathbbm{1})$. Therefore, the state after the $\CNOT$ gate would be $\frac{1}{\sqrt{2}} (X_1 \otimes X_1 + Z_1\otimes \mathbbm{1})\ket{\bar{0}} \otimes \ket{\bar{0}}$.
An application of $(\DecI^*)^{\otimes 2}$ to this state leads to $\ket{00}\otimes \frac{1}{\sqrt{2}}(\ket{s_{x_1} s_{x_1}}+\ket{s_{z_1} s_{0}} )$, a pure entangled syndrome state.

It is notable that such an error cannot happen in the i.i.d. Pauli error model, which is the main model in our work in Chapter~\ref{chapter-fteacap}, and that our example seems to demand a more exotic noise model. For example, one could think of a model with Hadamard errors or a model in which the execution of gates might fail with some probability, and instead of applying a Hadamard, nothing is applied to the state. To the best of our knowledge, it is an open question whether correlated syndrome states can appear in the i.i.d. Pauli model.
\end{example}

Even though we state our results in Chapter~\ref{chapter-fteacap} with the same fault-tolerant architecture for the encoder and decoder circuit, this is not a necessity for our construction. We could choose to use the same code with different levels of concatenation, or even two different codes (as long as they have the appropriate interfaces). This sort of separation may be particularly crucial for quantum repeaters or larger quantum communication networks. In such setups with multiple parties or long distances, it would perhaps be even more important that the communication lines between devices does not have to be subject to threshold constraints in the same way as the localized computation has to be.


\chapter{Fault-tolerant entanglement-assisted communication}
\label{chapter-fteacap}

This chapter (excluding Section~\ref{sec-extra}, and with minor revisions) is a reproduction of arXiv:2210.02939 \cite{BCMH22}, which has also appeared as a shortened version in the proceedings of the IEEE International Symposium on Information Theory 2023 \cite{BCMH23SHORT}. It explores the question of entanglement-assisted communication in the presence of gate errors during the encoder and decoder circuit. 

In particular, we show that the achievable rates for entanglement-assisted communication with noise-affected gates can be bounded from below in terms of the quantum mutual information reduced by a continuous function in the single gate error $p$, and the usual faultless entanglement-assisted capacity is recovered for small probabilities of local gate error, which confirms and substantiates the practical relevance of quantum Shannon theory. This is not only relevant for communication between spatially separated quantum computers, but also for communication between distant parts of a single quantum computing chip, where the communication line may be subject to a larger amount of noise than the local gates.
\newpage

\noindent\rule[0.5ex]{\linewidth}{1pt}
\vspace{0.5cm}
\textbf{arXiv:2210.02939:}\\
\textbf{Fault-tolerant entanglement-assisted communication}
\vspace{0.5cm}

\emph{Paula Belzig, Matthias Christandl, Alexander Müller-Hermes}

\vspace{0.5cm}
\noindent\rule[0.5ex]{\linewidth}{1pt}



\section{Introduction}

\label{sec-intro}

The successful transfer of information via a communication infrastructure is of crucial importance for our modern, highly-connected world. This process of information transfer, e.g., by wire, cable or broadcast, can be modelled by a communication channel $T$ which captures the noise affecting individual symbols. Instead of sending symbols individually, the sender and receiver typically agree to send messages using codewords made up from many symbols. With a well-suited code, the probability of receiving a wrong message can be made arbitrarily small. How well a given channel $T$ is able to transmit information can be quantified by the asymptotic rate of how many message bits can be transmitted per channel use with vanishing error using the best possible encoding and decoding procedure. This asymptotic rate is a characteristic of the channel, called its \emph{capacity} $C(T)$.

In \cite{Shannon48}, Shannon introduced this model for communication and derived a formula for $C(T)$ in terms of the mutual information between the input and output of the channel:
\[C(T)=\sup_{p_X} I(X:Y).\]
Here, $I(X:Y)=H(X)+H(Y)-H(XY)$ denotes the mutual information between the random variable $X$ and the output $Y=T(X)$, where $H(X)$ is the Shannon entropy of the discrete random variable $X$ with a set of possible values $x$ that is distributed according to a probability distribution $p_X$, given by $H(X)=-\sum_x p_X(x)\log(p_X(x) )$.

Various generalizations of this communication scenario to quantum channels lead to different notions of capacity. 
Two important examples are the classical capacity of a quantum channel \cite{Holevo96,SW97} , which quantifies how well a quantum channel can transmit classical information encoded in quantum states, and the quantum capacity \cite{Lloyd97,Shor02,Devetak03}, where quantum information itself is to be transmitted through the channel. Both of these notions of capacity have entropic formulas. However, they are not known to admit a characterization which is independent of the number of channel copies $n$, a so-called single-letter characterization, which would simplify their calculation.

The \emph{entanglement-assisted capacity}, where the encoding and decoding machines have access to arbitrary amounts of entanglement, does not only admit such a single-letter characterization, but it can in fact be regarded as the only direct formal analogue of Shannon's original formula, since the classical mutual information is simply replaced by its quantum counterpart \cite{BSST02}:
\[C^{ea}(T)= \sup_{\substack{\varphi_{AA'} \in \mathcal{M}_{d_A}\otimes\mathcal{M}_{d_{A'}}  \\  \varphi_{AA'} \text{ a pure quantum state} }} I(A':B)_{(T\otimes \id)(\varphi_{AA'})} .\]
Here, $I(A:B)_{\rho}=H(A)_{\rho}+H(B)_{\rho}-H(AB)_{\rho}$ denotes the quantum mutual information with the von Neumann entropy $H(AB)_{\rho}=-\Tr\left[\rho\log(\rho)\right]$ for a quantum state $\rho$.

In order to communicate with a given channel $T$, the encoding and decoding procedures need to be decomposed into quantum circuits as a sequence of quantum gates. The next step in a real-world scenario would be to implement these circuits on a quantum device so that we can realize an actual quantum communication system. However, this scenario generally does not consider one of the major obstacles of quantum computation: the high susceptibility of quantum circuits to noise and faults. In classical computers, the error rates of individual logical gates are known to be effectively zero in standard settings and at the time-scales relevant for communication \cite{Nicolaidis11}. The assumption of noiseless gates implementing the encoder and decoder circuit is therefore realistic in many scenarios. Real-life quantum gates, however, are affected by non-negligible amounts of noise. This is certainly a problem in near-term quantum devices, and it is generally assumed that it will continue to be a problem in the longer term \cite{Preskill18}.

Considering the encoder and decoder circuits as specific quantum circuits affected by noise therefore leads to potentially more realistic measures of how well information can be transferred via a quantum channel: \emph{fault-tolerant capacities}, which quantify the optimal asymptotic rates of transmitting information per channel use in the presence of noise on the individual gates. To construct suitable encoders and decoders for this scenario, we build on Christandl and Müller-Hermes' work \cite{CMH20}, which has introduced and analyzed fault-tolerant versions of the classical and quantum capacity, combining techniques from fault-tolerant quantum computing \cite{AB99,KLZ98,Kitaev03,AGP05} and quantum communication theory \cite{Wilde13}.


More precisely, we extend their work to entanglement-assisted communication. In particular, we show that entanglement-assisted communication is still possible under the assumption of noisy quantum devices, with achievable rates given by \[C^{ea}_{\mathcal{F}(p)}(T) \geq C^{ea}(T)-f(p)\]
where $C^{ea}_{\mathcal{F}(p)}(T)$ denotes the fault-tolerant entanglement-assisted capacity for gate error probability $p$ below a threshold, and with $\lim_{p\rightarrow 0} f(p) \rightarrow 0$.


In other words, the achievable rates for entanglement-assisted communication with noise-affected gates can be bounded from below in terms of the quantum mutual information reduced by a continuous function in the single gate error $p$. The usual faultless entanglement-assisted capacity is recovered for small probabilities of local gate error, which confirms and substantiates the practical relevance of quantum Shannon theory. 
This is not only relevant for communication between spatially separated quantum computers, but also for communication between distant parts of a single quantum computing chip, where the communication line may be subject to higher levels of noise than the local gates. In particular, the noise level for the communication line does not have to be below the threshold of the gate error.

This chapter is structured around the building blocks needed to achieve this result. In Section~\ref{sec-ft}, we briefly review concepts from fault-tolerance of quantum circuits used for communication. In Section~\ref{sec-ea-cap}, we outline how the fault-tolerant communication setup can be reduced to an information-theoretic problem which generalizes the usual, faultless entanglement-assisted capacity. In Section~\ref{sec-avp}, we prove a coding theorem for this information-theoretic problem. One important facet of communication with entanglement-assistance in our scenario comes in the form of noise affecting the entangled resource states, for which we introduce a scheme of fault-tolerant entanglement distillation in Section~\ref{sec-ft-ent-dist}. Finally, these techniques will be combined to obtain a threshold-type coding theorem for fault-tolerant entanglement-assisted capacity in Section~\ref{sec-coding-thm}.

\section{Fault-tolerant encoder and decoder circuits for communication}
\label{sec-ft}

\noindent Here, we review some aspects of common techniques for fault-tolerance, but for a detailed overview of the relevant concepts, we refer to \cite{AGP05} and \cite{CMH20}, as well as Section~\ref{sec-ft-intro}.

Note that our notation for mathematical objects from quantum theory is the same as in \cite[Section~II.A]{CMH20}. We define quantum channels as completely positive and trace preserving maps $T: \mathcal{M}_{d_A}\rightarrow \mathcal{M}_{d_B}$ where $\mathcal{M}_d$ denotes the matrix algebra of complex $d\times d$-matrices.
Probability distributions of $d$ elements are vectors in $\mathbbm{C}^d$ where each entry is positive and the sum of all entries equals $1$. Channels with classical input are defined as linear maps from $\mathbbm{C}^d$ that yield unit-trace positive semi-definite Hermitian matrices, and channels with classical output map unit-trace positive semi-definite Hermitian matrices to elements of $\mathbbm{C}^d$.

\subsection{Fault-tolerance for quantum circuits}
\label{sec-ft-circuits}

In this work, we will use the notation for circuits, errors and stabilizer codes as introduced in Section~\ref{sec-ft-intro}.

We recall from Section~\ref{sec-intro-noise-and-errors} that quantum circuits are the dense subset of quantum channels which can be written as a composition of the elementary gate operations. It should be emphasized that the linear map realized by a quantum circuit might be written in different ways as a composition of elementary gates. As in \cite{CMH20}, we will assume that each circuit is specified by a particular circuit diagram detailing which elementary gates are to be executed at which time and place in the quantum circuit. Given such a circuit diagram, we can model the noise affecting the resulting quantum circuit. For simplicity, we will always consider the \emph{i.i.d. Pauli noise model}, where one of the Pauli channels (with single Kraus operator $\sigma_x,\sigma_y$ or $\sigma_z$) is applied with probability $\frac{p}{3}$ in between the gates, as introduced in Section~\ref{sec-intro-noise-and-errors}.

To protect against noise, a quantum circuit can be implemented in a stabilizer error correcting code, where single, potentially fault-affected qubits (\emph{logical qubits}) can be encoded in a quantum state of $K$ physical qubits for each logical qubit, as explored in Section~\ref{sec-error-correction-intro}. Using this formalism, we can choose a unitary transformation in relation to a convenient basis for decoupling noise and data, which specifies an ideal decoder for the code, as described in more detail in Eq.~\eqref{eq-choosing-Dec} and Section~\ref{sec-error-correction-intro}.

Throughout this chapter, we will frequently use notation where an operation marked with a star should be considered an ideal operation that is useful for circuit analysis, and not a fault-location. In particular, we will sometimes write $\id_2^*$ to denote an identity map between qubits, which should not be taken to be a fault-affected storage.


In this work, like in \cite{CMH20}, we will consider implementations in the concatenated 7-qubit Steane code. The 7-qubit Steane code introduced in \cite{Steane96} is an error correcting code that can correct all single-qubit errors, and that can be concatenated to improve protection against errors \cite{KL96}. For the concatenated 7-qubit Steane code, as shown in \cite{AGP05}, an implementation where the error correction's gadget is performed between each operation minimizes the accumulation of errors. Under this implementation, the concatenated 7-qubit Steane code fulfills a threshold theorem for computation, see also Theorem~\ref{thm-threshold}. 

Fault-tolerance can in principle be achieved by other quantum error correcting codes \cite{AB99,KLZ98,Kitaev03,AGP05}. One could also consider using two different quantum error correcting codes for the encoder and decoder circuit in our setup. For simplicity, we restrict ourselves to using the concatenated 7-qubit Steane code \cite{Steane96,KL96} with the same level of concatenation for both circuits, but our definitions for fault-tolerant communication can straightforwardly be extended to the more general case.

\subsection{Fault-tolerance for communication}
\label{sec-ft-interfaces}

\noindent By performing error correction, a quantum circuit with classical input and output that is affected by faults at a low rate can thus be implemented in a way such that it behaves like an ideal circuit (i.e. a circuit without faults) by threshold-type theorems. These code implementations cannot, however, be directly used in the encoding and decoding circuits for communication, as they require classical input and output, whereas the encoder's output in our communication setup, for instance, serves as input into the noisy quantum channel. The fault-tolerant implementation of an encoder and decoder in a communication setting therefore leads to the message being encoded in the corresponding code space. In the case of the concatenated 7-qubit Steane code, the number of physical qubits increases by a factor of $7$ for each level of concatenation. To obtain our results for communication rates, we therefore perform an additional circuit mapping information in the code space to the physical system where the quantum channel acts. This circuit will be referred to as \emph{decoding interface} $\DecI$. Similarly, another circuit can be performed to transfer the channel's output into the code space where it can be processed by the fault-tolerantly implemented decoder. This circuit is called \emph{encoding interface} $\EncI$. These circuits, introduced in Definition~\ref{def-interfaces}, and discussed further in Section~\ref{sec-ft-and-comm}, are also affected by faults.

In contrast to $\DecI^*$ and $\EncI^*$, which are objects used for the mathematical analysis of the circuits and not implemented in practice, the interfaces $\EncI$ and $\DecI$ are quantum circuits consisting of gates that can be affected by faults. Since this can lead to faulty inputs to a quantum channel, we will need interfaces that are tolerant against such faults. Unfortunately, it is impossible to make the overall failure probability of interfaces arbitrarily small, since they will always have a first (or last) gate that is executed on the physical level and not protected by an error correcting code, resulting in a failure with a probability of at least gate error $p$. Fortunately, it is possible to construct qubit interfaces for concatenated codes which fail with a probability of at most $2cp$ for some constant $c$, which are good enough for our purposes, as found in~\cite{MGHHLPP14}, \cite[Theorem~III.3]{CMH20} and restated in Theorem~\ref{thm-correct-interfaces}.

In combination with the threshold theorem from \cite{AGP05}, this can be used to prove the following extensions of Lemma III.8 from \cite{CMH20} for the combination of circuit and interface with additional quantum input (cf. Figure~\ref{fig-eff-channel}):

\begin{lemma}[Effective encoding interface]
\label{thm-effective-encoder}
Let $m,n,k\in \mathbbm{N}$ and let $\Gamma:\mathcal{M}_2^{\otimes n+k} \rightarrow \mathbbm{C}^{2^m}$ be a quantum circuit with quantum input and classical output.
For each $l\in \mathbbm{N}$, let ${\mathcal{C}_l}$ denote the $l$-th level of the concatenated 7-qubit Steane code with threshold $p_0$. Moreover, let $\EncI_l: \mathcal{M}_2 \rightarrow  \mathcal{M}_2^{\otimes 7^l}$ 
be the encoding interface circuit for the $l$-th level of the concatenated 7-qubit Steane code with threshold $p_0$.

Then, for any $0\leq p \leq \frac{p_0}{2}$ and any $l\in\mathbbm{N}$, there exists a quantum channel $N_l:\mathcal{M}_2\rightarrow \mathcal{M}_2$, which only depends on $l$ and the interface circuit $\EncI_l$, such that:
\begin{align*}
 &\|  \big[\Gamma_{\mathcal{C}_l} \circ (\EncI_l^{\otimes n} \otimes \EC_l^{\otimes k}) \big]_{\mathcal{F}(p)} \\ &\hspace{1cm}  -(\Gamma\otimes \Tr_S) \circ \left(N_{enc,p,l}^{\otimes n} \otimes (\DecI_l^*\circ [\EC_l]_{\mathcal{F}(p)} )^{\otimes k})\right) \|_{1 \rightarrow 1}  \\& \hspace{1cm}\hspace{1cm}   \leq 2 p_0 \Big(\frac{p}{p_0}\Big)^{2^l} |\Loc(\Gamma)| 
\end{align*}
with
\[N_{enc,p,l} = (1-2cp) \id_2 +2cpN_l\]
where $c=p_0 \max \{|\Loc(\EncI_1)|,|\Loc(\DecI_1\circ \EC)|\}$.
\end{lemma}

\begin{lemma}[Effective decoding interface]
\label{thm-eff-decoder}
Let $m,n,k\in \mathbbm{N}$ and $\Gamma: \mathbbm{C}^{2^m} \otimes \mathcal{M}_2^{\otimes k}\rightarrow \mathcal{M}_2^{\otimes n} $ be a quantum circuit with quantum and classical input and quantum output.
For each $l\in \mathbbm{N}$, let ${\mathcal{C}_l}$ denote the $l$-th level of the concatenated 7-qubit Steane code with threshold $p_0$.  Moreover, let 
$\DecI_l: \mathcal{M}_2^{\otimes 7^l} \rightarrow \mathcal{M}_2$ be the decoding interface circuit for the $l$-th level of the concatenated 7-qubit Steane code with threshold $p_0$.

Then, for any $0\leq p \leq \frac{p_0}{2}$ and any $l\in\mathbbm{N}$, there exists a quantum channel $N_l:\mathcal{M}_2 \otimes \mathcal{M}_2^{\otimes(7^l-1)} \rightarrow \mathcal{M}_2$, which only depends on $l$ and the interface circuit $\DecI_l$, such that:
    \begin{align*}
      &\| \big[ \DecI_l^{\otimes n}  \circ \Gamma_{\mathcal{C}_l} \circ (\id_{cl}^{\otimes m} \otimes \EC_l^{\otimes k}) \big]_{\mathcal{F}(p)}  \\ &\hspace{1cm} -N_{dec,p,l}^{\otimes n} \circ \left(\Gamma\otimes S_{S} \right)\circ (\id_{cl}^{\otimes m} \otimes (\DecI^* \circ [\EC_l]_{\mathcal{F}(p)})^{\otimes k}) \|_{1 \rightarrow 1} \\&\hspace{2cm} \leq  2 p_0 \Big(\frac{p}{p_0}\Big)^{2^l} |\Loc(\Gamma)| + 2np_0|\Loc(\EncI_1)| \Big(\frac{p}{p_0}\Big)^{2^l-1}  ,
    \end{align*}
where $S_S: \mathcal{M}_2^{\otimes k(7^l-1) }\rightarrow  \mathcal{M}_2^{\otimes n(7^l-1) }$ is some quantum channel on the syndrome space, and with
\[N_{dec,p,l} = (1-2cp) \id_2 \otimes \Tr_{S} +2cp N_l\]
where $c=p_0 \max \{|\Loc(\EncI_1)|,|\Loc(\DecI_1\circ \EC)|\}$.
\end{lemma}

Here, $\id_{cl}:\mathbbm{C}^2\rightarrow \mathbbm{C}^2$ denotes the identity map on a classical bit, and $\| T-S \|_{1 \rightarrow 1} := \sup \{ \| (T-S)(\rho) \|_{\Tr}   | \rho\in \mathcal{M}_{d_A}\text{a quantum state}\}$ denotes the 1-to-1 distance of two quantum channels $T:\mathcal{M}_{d_A}\rightarrow \mathcal{M}_{d_B} $ and $S:\mathcal{M}_{d_A}\rightarrow \mathcal{M}_{d_B} $, where $ \| \rho - \sigma \|_{\Tr} $ denotes the trace distance induced by the trace norm $\| \rho \|_{\Tr} := \frac{1}{2} \| \rho\|_{1} =\frac{1}{2}  \Tr(\sqrt{ \rho ^{\dagger} \rho}) $.

\begin{figure*}[!t]
  \centering
       \includegraphics[width=\textwidth]{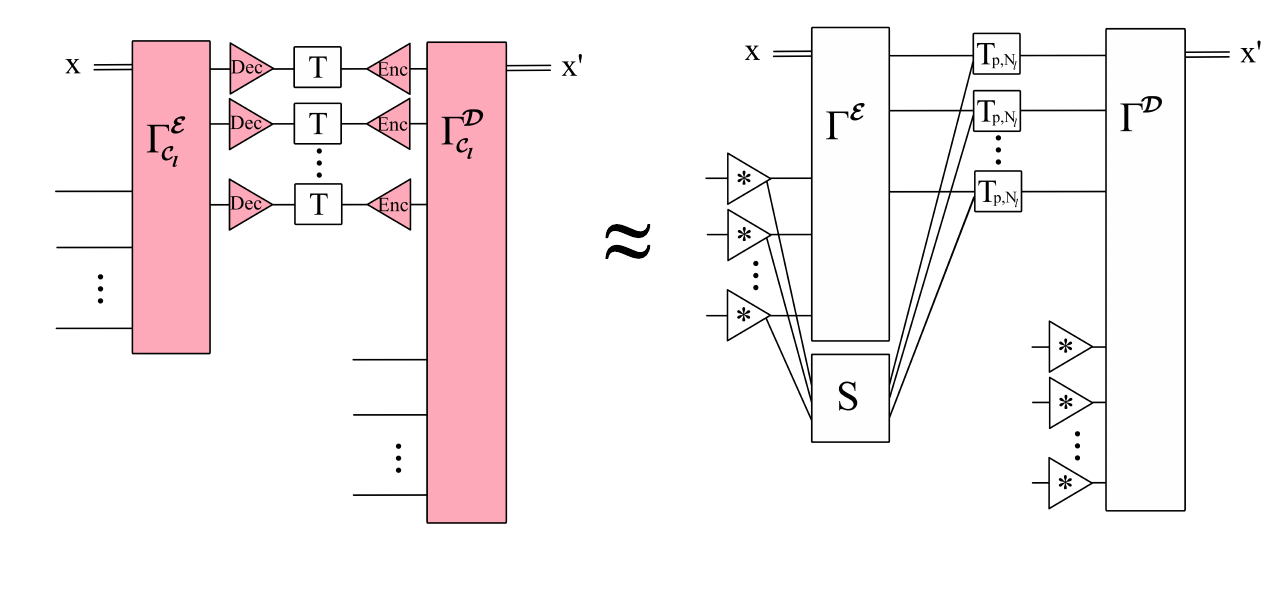}
\caption{\textbf{Sketch of the setup for the effective channel.} The fault-tolerantly implemented encoder $\Gamma_{\mathcal{C}_l}^{\mathcal{E}}$ takes input in the form of $m$ classical bits $x$ and $r$ physical qubits which are encoded in the code space. The resulting codewords are sent through $n$ copies of the quantum channel $T$, preceeded by the decoding interface. The output of the channel is fed into the encoding interface, whose output serves as input to the fault-tolerantly implemented decoder $\Gamma_{\mathcal{C}_l}^{\mathcal{D}}$, which also receives additional quantum input in the form of $s$ qubits in the code space. Theorem~\ref{thm-eff-channel} shows that this setup is very close to a faultless setup with an effective channel, where the quantum systems are transformed by the perfect decoding operation $\DecI^*$, represented by a triangle marked with a star. The effective channel $T_{p,N_l}$ receives input in the form of data qubits and a potentially correlated syndrome state.
}
\label{fig-eff-channel}
\end{figure*}

Lemma~\ref{thm-effective-encoder} and Lemma~\ref{thm-eff-decoder} can be combined to obtain Theorem~\ref{thm-eff-channel}, which is a modified version of Theorem III.9 from \cite{CMH20} that we will use in our analysis of entanglement-assisted capacity.
This theorem links the fault-affected scenario to a communication problem with faultless encoder and decoder circuits connected by an effective noisy channel of a special form, as illustrated in Figure~\ref{fig-eff-channel}.

It is important to note that our setup for fault-tolerant communication considers the operation of encoding information into a quantum channel $T$ and subsequent decoding of this information as two fault-affected circuits connected by $T$. The channel $T$ itself can be taken to be a model a noisy communication channel, however, we do not consider it as a noise-affected circuit with well-defined fault-locations (hence its white color in the figures). In particular, the noise affecting $T$ can be very different from the noise affecting the encoding and decoding circuits, and does not have to be below threshold. 

\begin{theorem}[Effective channel with quantum input] \label{thm-eff-channel}
Let $T:\mathcal{M}_2^{\otimes j_1} \rightarrow \mathcal{M}_2^{\otimes j_2}$ be a quantum channel, and let $\Enc: \mathbbm{C}^{2^m} \otimes \mathcal{M}_2^{\otimes r}\rightarrow \mathcal{M}_2^{\otimes nj_1}$ be a quantum circuit with $m$ bits of classical input and $r$ qubits of quantum input and let $\Dec:\mathcal{M}_2^{\otimes nj_2} \otimes \mathcal{M}_2^{\otimes s} \rightarrow \mathbbm{C}^{2^m}$ be a quantum circuit with classical output of $m$ bits. For each $l\in\mathbbm{N}$, let $\mathcal{C}_l$ denote the $l$-th level of the concatenated 7-qubit Steane code with threshold $0<p_0\leq 1$. Let $\EncI_l: \mathcal{M}_2 \rightarrow  \mathcal{M}_2^{\otimes 7^l}$ and $\DecI_l: \mathcal{M}_2^{\otimes 7^l} \rightarrow \mathcal{M}_2$ be the interface circuits for the $l$-th level of the concatenated 7-qubit Steane code with threshold $p_0$.

Then, for any $l\in\mathbbm{N}$ and any $0\leq p \leq \min\{p_0/2,1/4c\}$, there exists a quantum channel $N_l: \mathcal{M}_2^{\otimes j_1 7^l }\rightarrow \mathcal{M}_2^{\otimes j_2}$ and a quantum channel $S_S: \mathcal{M}_2^{\otimes j_1(7^l-1) }\rightarrow  \mathcal{M}_2^{\otimes j_1(7^l-1) }$ on the syndrome space such that 
\begin{align*}
&\big\| \big[\Gamma_{\mathcal{C}_l}^{\mathcal{D}} \circ \bigg( \Big( \big( \EncI_l^{\otimes nj_2} \circ T^{\otimes n} \circ \DecI_l^{\otimes nj_1} \circ\Gamma_{\mathcal{C}_l}^{\mathcal{E}} \big) \circ   (\id_{cl}^{\otimes m} \otimes\EC_l^{\otimes r}) \Big)\otimes\EC_l^{\otimes s}
\bigg) \big]_{\mathcal{F}(p)}  \\  &\hspace{1cm}  - ({\Dec}\otimes \Tr_S) \circ \Big(\big( T_{p,N_l}^{\otimes n} \circ ({\Enc}\otimes  S_S )\big)\otimes \id_2^{\otimes s} \Big)  \circ \Compactcdots \\&\hspace{5cm} \Compactcdots  \circ \Big(\id_{cl}^{\otimes m} \otimes (\DecI_l^*\circ [\EC_l]_{\mathcal{F}(p)})^{\otimes (r+s)}\Big)\big\|_{1 \rightarrow 1}  \\
&\hspace{1cm}  \leq 2 p_0 \Big(\frac{p}{p_0}\Big)^{2^l} (|\Loc(\Enc)|+|\Loc(\Dec)|)  +2p_0|\Loc(\EncI_1)| \Big(\frac{p}{p_0}\Big)^{2^l-1} j_1n
  \end{align*}
with
\[T_{p,N_l}=(1-2(j_1+j_2) cp) (T\otimes \Tr_S) + 2(j_1+j_2)cpN_l\]
with $c=p_0 \max \{|\Loc(\EncI_1)|,|\Loc(\DecI_1\circ EC)|\}$.
$S_S$ may depend on $l, \Enc$ and $\DecI_l$, while $N_l$ may depend on $ l, \EncI_l$ and $\DecI_l$.
\end{theorem}

 This theorem is formulated for quantum channels which map from a quantum system composed of $j_1$ qubits to a quantum system composed of $j_2$ qubits because we consider interfaces between qubits. However, any quantum channel can always be embedded into a quantum channel between systems composed of qubits, such that Theorem~\ref{thm-eff-channel} and subsequent results apply to general quantum channels.

\section{Entanglement-assisted communication with faultless or faulty devices}
\label{sec-ea-cap}

\noindent When a sender and a receiver are connected by many copies of a quantum channel $T$ and have access to entanglement, they can use this setup to transmit a classical message via entanglement-assisted communication. Then, one can identify the best possible operations for the sender and receiver to perform in order to maximize their transmission rate. This section includes a short introduction into entanglement-assisted communication in Section~\ref{sec-ea-cap-normally}, which will serve as a basis for the coding scheme in our main result, followed by a description of the setup for fault-tolerant entanglement-assisted communication and our strategy for its analysis in Section~\ref{sec-ft-cap-setup}. 

\subsection{The entanglement-assisted capacity}
\label{sec-ea-cap-normally}

\noindent Using the superdense coding protocol \cite{BW92}, two classical bits can be communicated by sending only one qubit over a noiseless quantum channel assisted by entanglement. It is therefore natural to study a noisy channel's classical capacity with entanglement assistance \cite{BSST02}, which we introduce in detail in Section~\ref{sec-quantum-shannon-theory}.

To model entanglement-assisted classical communication, recall that we consider a scheme with classical input and output, where quantum entanglement is available to the sender and the receiver.
As sketched in Figure~\ref{fig-ea-cap}, the sender encodes a classical message of $m$ bits into a quantum state of $n$ qudits by performing an encoding map $\mathcal{E}$. The resulting quantum state serves as input into the tensor product of $n$ copies of a quantum channel $T$, which is equivalent to $n$ independent uses of a quantum wire modelled by $T$. Then, the transformed quantum state is decoded by the receiver applying a decoding map $\mathcal{D}$ which converts the channel's output back into a bit string of length $m$. The performance of such a scheme can be quantified by the probability that this resulting bit string and the original message are identical, as formalized in Definition~\ref{def-ea-coding-scheme}. 
Because of superdense coding \cite{BW92} and teleportation \cite{BBCJPW93}, the classical entanglement-assisted capacity of a channel is exactly double its quantum entanglement-assisted capacity.

We will need a more explicit characterization of the communication error that can be reached by using certain entanglement-assisted coding schemes achieving rates close to capacity. The specific bound can be obtained from \cite[Section~20.4]{Wilde13}, and its error term follows from the packing lemma \cite{HDW08}, notions from weak typicality, and Hoeffding's bound \cite{Hoeffding63}.

\subsection{The fault-tolerant entanglement-assisted capacity}
\label{sec-ft-cap-setup}

\noindent In Section~\ref{sec-ea-cap-normally}, the encoder and decoder are assumed to be ideal quantum channels. In order to perform these channels on some given quantum device, they have to be implemented by quantum circuits, i.e., compositions of finitely many elementary gates.
It is well known that quantum devices (unlike classical computers) are notoriously susceptible to faults at the single-gate level which can have devastating effects on the whole computation. This is also true for the circuits encoding and decoding the information that we want to send between different devices or computers. Through clever and protective implementation, the computation within the encoding and decoding devices can be made robust against such faults, raising the question of a channel's fault-tolerant entanglement-assisted capacity.

A coding scheme for a setup affected by noise is defined as follows:

\begin{definition}[Fault-tolerant entanglement-assisted coding scheme]\label{defn:FTEACS}
Let $T:\mathcal{M}_{d_A} \rightarrow \mathcal{M}_{d_B}$ be a quantum channel, and let $n,m \in \mathbbm{N}$, $R_{ea}\in\mathbbm{R}^+$ and $\epsilon>0$.
For $0\leq p \leq 1$, let $\mathcal{F}(p)$ denote the i.i.d. Pauli noise model.

Then, an $(n,m,\epsilon,R_{ea})$-coding scheme for fault-tolerant entanglement-assisted communication consists of quantum circuits $\mathcal{E}:\mathbbm{C}^{2^m}\otimes\mathcal{M}_2^{ \otimes  \lfloor nR_{ea}\rfloor }\rightarrow\mathcal{M}_{d_A}^{\otimes n}$ and $\mathcal{D}:\mathcal{M}_{d_A}^{\otimes n}\otimes \mathcal{M}_2^{ \otimes \lfloor nR_{ea}\rfloor } \rightarrow\mathbbm{C}^{2^m}$ such that
\begin{align*}
F\big(X,\big[\mathcal{D} \circ  \big( (T^{\otimes n} \circ \mathcal{E}\big) \otimes \id_2^{*\otimes  \lfloor nR_{ea}\rfloor }  \big)  \big]_{\mathcal{F}(p)} (X \otimes \phi_+^{\otimes  \lfloor nR_{ea}\rfloor })\big) \geq 1-\epsilon\end{align*}
where $X=\ketbra{x}$, for all bit strings  $x\in \{0,1\}^m$.
\end{definition}

\begin{remark} If the entanglement resource in our setup was arbitrary, we could consider a setup assisted by an arbitrary pure entangled state in the code space.  In this scenario, the entanglement resource would directly be available to the fault-tolerantly implemented encoding and decoding circuits, without being corrupted by noisy encoding interfaces, and without necessitating the additional step of entanglement distillation. Achievable rates for such a scenario can be inferred from the expression from our main result (Theorem~\ref{thm-final-coding-thm}) with $f_1(p)=0$.
Under this model, we would assume that the state was prepared and stored (for the duration of the encoding and decoding circuits) within the code space without incurring any faults through the preparation gates and time steps. Here, we choose to consider the more practically relevant scenario where the entanglement resource becomes subject to noise as soon as it enters the code space of the encoding and decoding circuits. The form and amount of entanglement-assistance we consider is as in the standard setup, given by $nR_{ea}$ copies of physical maximally entangled qubits. This scenario also covers situations where $\phi_+$ may be prepared and stored in some highly noise-tolerant and well-suited way until it is needed for computation.
\end{remark}

The asymptotically best possible fault-tolerantly achievable rate defines the channel's fault-tolerant entanglement-assisted capacity, fundamentally characterizing how much information the channel can transmit under this noise model.

\begin{definition}[Fault-tolerant entanglement-assisted capacity]
Let $T:\mathcal{M}_{d_A} \rightarrow \mathcal{M}_{d_B}$ be a quantum channel, and let $R_{ea}\in \mathbbm{R}_+$ be the rate of entanglement-assistance. 
For $0\leq p \leq 1$, let $\mathcal{F}(p)$ denote the i.i.d. Pauli noise model.

If, for some $R_{ea}$ and for every $n\in \mathbbm{N}$, there exists an $(n,m(n),\epsilon(n),R_{ea})$-coding scheme for fault-tolerant entanglement-assisted communication under the noise model $\mathcal{F}(p)$, then a rate $R\geq 0$ is called achievable for fault-tolerant entanglement-assisted communication via the quantum channel $T$ if

\begin{align*}
   R & \leq 
   \liminf_{n\rightarrow \infty} \Big\{ 
    \frac{m(n)}{n} \Big\}
    \end{align*}
and
\begin{align*}
   \lim_{n\rightarrow \infty} \epsilon(n) \rightarrow 0
    \end{align*}
The fault-tolerant entanglement-assisted capacity of $T$ is given by
  \begin{equation*}
\begin{split}
  &C^{ea}_{\mathcal{F}(p)}(T) =\sup \{R | R \text{ achievable rate for fault-tolerant}  \\ &\hspace{5cm} \text{ entanglement-assisted communication via $T$}\}.
  \end{split}
\end{equation*}
\end{definition}

\begin{figure}[htbp]
  \centering
       \includegraphics[width=10cm]{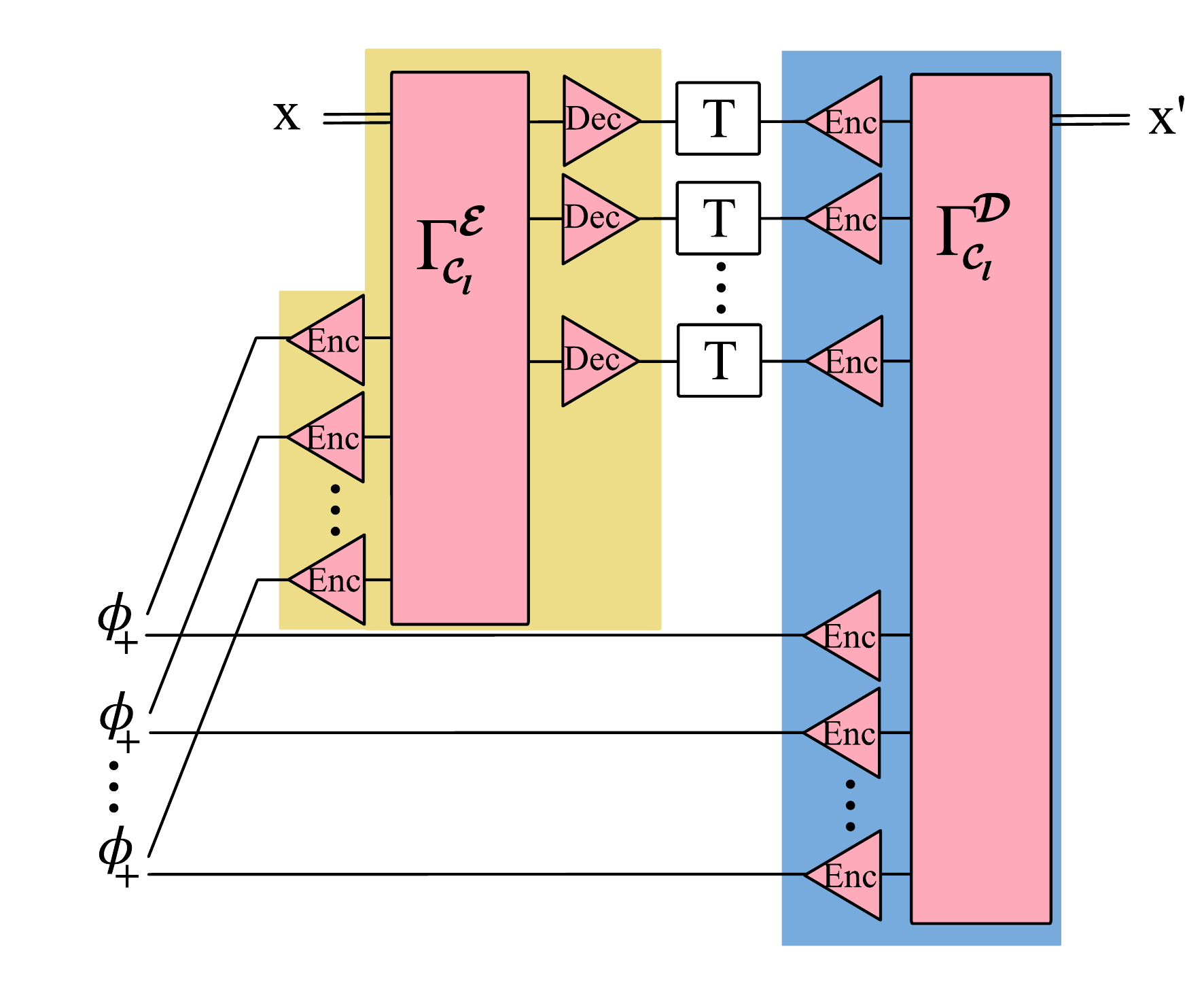}
\caption{\textbf{Basic setup for our coding scheme for fault-tolerant entanglement-assisted communication, see Definition~\ref{defn:FTEACS}.} The encoding map $\mathcal{E}$ (yellow) encodes a bit string of length $m$ into a quantum state that serves as input into $n$ copies of the quantum channel $T$, and the decoding map $\mathcal{D}$ (blue) decodes the received quantum state back to a bit string. In the entanglement-assisted scenario, $\mathcal{E}$ and $\mathcal{D}$ are connected by $\sim nR_{ea}$ maximally entangled states. To make the communication fault-tolerant, the encoding and decoding circuits are implemented fault-tolerantly in an error correcting code $\mathcal{C}_l$ as $\Gamma_{\mathcal{C}_l}^{\mathcal{E}}$ and $\Gamma_{\mathcal{C}_l}^{\mathcal{D}}$, and combined with interfaces $\EncI$ and $\DecI$ mapping between the quantum states serving as input and output for $T$, and the quantum states being transformed in the fault-tolerantly implemented encoding and decoding circuits.
}
\label{fig-ft-ea-cap-setup}
\end{figure}

Formally, any quantum circuits $\mathcal{E}$ and $\mathcal{D}$ may be chosen in Definition~\ref{defn:FTEACS}, leading to a coding scheme for fault-tolerant entanglement-assisted classical communication. To prove lower bounds to the fault-tolerant capacity $C^{ea}_{\mathcal{F}(p)}(T)$ for a quantum channel $T$ in terms of the capacity $C^{ea}(T)$, we will use a particular construction that is similar to constructions in \cite{CMH20}. 

Consider some coding scheme not affected by noise for entanglement-assisted classical communication over the channel $T$. We can turn this coding scheme into a fault-tolerant coding scheme by first approximating it by quantum circuits, and then implementing these quantum circuits in a high level of the concatenated $7$-qubit Steane code. Crucially, we will use the interface circuits from \cite{MGHHLPP14,CMH20} to convert between physical qubits and logical qubits in the code space, e.g., when qubits from the output of the channel $T$ are brought into the code space. Unfortunately, these interfaces fail with a probability $2cp$, where $p$ is the gate error probability of the noise model and $c$ some interface-dependent constant (from Theorem~\ref{thm-correct-interfaces}) and the fault-tolerant implementation of the coding scheme affected by faults will not be equivalent to the original coding scheme for the quantum channel $T$. Instead it will be equivalent to the coding scheme for a certain effective quantum channel $T_{p,N}$ as in Theorem~\ref{thm-eff-channel}. 

Our strategy starts by considering a coding scheme for entanglement-assisted classical communication for channels that include our effective communication channels $T_{p,N}$. We refer to this channel model as \emph{arbitrarily varying perturbation} (AVP) and we will discuss it in detail in Section~\ref{sec-avp}. This model has been introduced in \cite{CMH20} in the cases of unassisted classical and quantum communication, and it is closely related to the fully-quantum arbitrarily varying channels studied in \cite{BDNW18}. As described in the preceeding paragraphs, we then obtain a fault-tolerant coding scheme by implementing the coding scheme under AVP in a high level of the concatenated $7$-qubit Steane code. For the fault-tolerant entanglement-assisted capacity, the setup crucially includes a supply of maximally entangled states that are connected to the fault-tolerantly implemented encoder and decoder circuit via additional interfaces, as illustrated in Figure~\ref{fig-ft-ea-cap-setup}. Because of the effective probability of failure of these interfaces, when transferring the maximally entangled states into the code space, they are only correctly transmitted with a probability of approximately $1-4cp$ (since there is one interface for each qubit). Subsequently, the entanglement inserted into the code space is noisy and in a mixed state. To counteract this, we show that this entanglement can still be made usable by transforming it back into pure state entanglement in the code space by performing (fault-tolerant) entanglement distillation in Section~\ref{sec-ft-ent-dist}. Since entanglement distillation requires classical communication, we will need to use a subset of the channels $T$ to run the fault-tolerant protocol from \cite{CMH20} to send classical information. Thereafter, with slightly fewer copies of $T$ remaining, an analysis similar to \cite{CMH20} is carried out in Section~\ref{sec-coding-thm} to arrive at a coding theorem describing rates of fault-tolerant entanglement-assisted communication that are achievable with vanishing error.

\section{Entanglement-assisted communication under arbitrarily varying perturbation}
\label{sec-avp}

\noindent As described in Section~\ref{sec-ft-cap-setup}, we find a correspondence between the capacity of a fault-affected setup and an information-theoretic communication setup under non-i.i.d. perturbations which we outline in Section~\ref{sec-avp-model}. Based on similar channel models in \cite{CMH20}, we introduce a generalized version of an entanglement-assisted capacity which allows for arbitrarily varying syndrome input and prove a coding theorem for this model in Section~\ref{sec-avp-coding-thm}.

\subsection{The entanglement-assisted capacity under arbitrarily varying perturbation}\label{sec-avp-model}

\noindent One key feature of the communication problem emerging from Theorem~\ref{thm-eff-channel} is that the effective channel takes input from the space of channel coding symbols as well as the syndrome space. Since the syndrome state can be correlated across different channel uses, the effective communication problem is not covered by standard i.i.d. communication scenarios, but instead defines a communication problem of its own as introduced in \cite{CMH20} and with similarities to communication scenarios studied in \cite{BDNW18}.
 
For $T:\mathcal{M}_{d_A}\rightarrow \mathcal{M}_{d_B}$ and any quantum channel $N:\mathcal{M}_{d_A}\otimes \mathcal{M}_{d_S} \rightarrow \mathcal{M}_{d_B}$ let  $T_{p,N}:\mathcal{M}_{d_A}\otimes \mathcal{M}_{d_S} \rightarrow \mathcal{M}_{d_B}$ denote the quantum channel $T_{p,N}=(1-p) (T\otimes \Tr_S) + pN $. Here, we consider the problem of communicating via a channel of the form
$T_{p,N}^{\otimes n} (\cdot\otimes \sigma_S):\mathcal{M}_{d_A}^{\otimes n} \rightarrow \mathcal{M}_{d_B}^{\otimes n}$ for arbitrary syndrome states $\sigma_S$ and arbitrary quantum channels $N$. We refer to this model as \emph{communication under arbitrarily varying perturbation} (AVP). 

Note that the effective channel model emerging from Theorem~\ref{thm-eff-channel} takes the form $T_{p,N}^{\otimes n} (\cdot\otimes \sigma_S)$ for some syndrome state $\sigma_S$ and some quantum channels $N$, where the dimension $d_S$ depends on the level of concatenation. Since the level of concatenation has to increase with the number $n$ of channel uses if we want the overall communication error to vanish in a fault-tolerant communication scenario, we will allow $d_S$ to be arbitrary and possibly dependent on $n$ in the definition of the capacity under AVP. Previous results in \cite{BDNW18} for models with fixed syndrome state dimension are therefore not directly applicable in this setting. Notably, $\sigma_S$ can be arbitrarily entangled between the $n$ spaces. If $\sigma_S$ is a separable state, this communication model is a special case of a channel model studied in \cite{ABBN12}. Our definitions here consider general syndrome states and can thereby be taken to apply to the worst-case scenario with arbitrarily correlated syndrome states.

\begin{definition}[Entanglement-assisted coding scheme under arbitrarily varying perturbation]
\label{def-avp-coding-scheme}
Let $T:\mathcal{M}_{d_A} \rightarrow \mathcal{M}_{d_B}$ be a quantum channel, and let $n,m \in \mathbbm{N}$, $R_{ea}\in\mathbbm{R}^+$ and $\epsilon>0$.

Then, an $(n,m,\epsilon,R_{ea})$-coding scheme for entanglement-assisted communication under AVP of strength $p$ consists of
a pure bipartite quantum state $\varphi\in \mathcal{M}_2^{\otimes 2}$ and the quantum channels $\mathcal{E}:\mathbbm{C}^{2^m}\otimes\mathcal{M}_2^{ \otimes nR_{ea}}\rightarrow\mathcal{M}_{d_A}^{\otimes n}$ and $\mathcal{D}:\mathcal{M}_{d_A}^{\otimes n}\otimes \mathcal{M}_2^{ \otimes nR_{ea}} \rightarrow\mathbbm{C}^{2^m}$ such that
\begin{align*} \inf F\Big(X, \mathcal{D} \circ  \big( (T_{p,N}^{\otimes n}  \circ (\mathcal{E} \otimes \sigma_S) ) \otimes \id_2^{\otimes nR_{ea}}  \big) (X \otimes \varphi^{\otimes nR_{ea}}) \Big) \geq 1-\epsilon\end{align*}
where $X=\ketbra{x}$, for all bit strings  $x\in \{0,1\}^m$, where
\[T_{p,N}= (1-p) (T\otimes \Tr_S ) + pN . \]
The infimum goes over the dimension $d_S\in \mathbbm{N}$, quantum states $\sigma_S\in\mathcal{M}_{d_S}^{\otimes n}$, and quantum channels $N:\mathcal{M}_{d_A}\otimes \mathcal{M}_{d_S} \rightarrow \mathcal{M}_{d_B}$.
\end{definition}

\begin{remark}
Here, we consider copies of arbitrary bipartite pure entangled states instead of maximally entangled states as entanglement resource, as the latter would require an extra step of entanglement dilution for the coding scheme from \cite[Theorem~21.4.1]{Wilde13}. In order to ensure exponential decay of the entanglement dilution error that we require in our proof, the communication rate would be reduced by some function linear in perturbation strength $p$. However, this is not necessary in the context of fault-tolerant coding because the extra classical communication can be performed together with the entanglement distillation, leading to an overall better bound on the achievable rate in our main result, Theorem~\ref{thm-final-coding-thm}.
\end{remark}

Naturally, the best possible rate of information transfer using such a coding scheme is a channel's entanglement-assisted capacity under AVP:

\begin{definition}[Entanglement-assisted capacity under arbitrarily varying perturbation]
Let $T:\mathcal{M}_{d_A} \rightarrow \mathcal{M}_{d_B}$ be a quantum channel, and let $R_{ea}\in\mathbbm{R}^+$.

If, for some $R_{ea}$ and for every $n\in \mathbbm{N}$, there exists an $(n,m(n),\epsilon(n),R_{ea})$-coding scheme for entanglement-assisted communication under AVP of strength $p$, then a rate $R\geq 0$ is called achievable for entanglement-assisted communication under AVP via the quantum channel $T$ if 
\begin{align*}
   R & \leq 
   \liminf_{n\rightarrow \infty} \Big\{ 
    \frac{m(n)}{n} \Big\}
    \end{align*}
and
\begin{align*}
   \lim_{n\rightarrow \infty} \epsilon(n) \rightarrow 0
    \end{align*}
The entanglement-assisted capacity of $T$ under AVP is given by
\begin{align*}
  & C^{ea}_{AVP}(p,T) =\sup \{R | R \text{ achievable rate for } \\& \hspace{3cm}\text{entanglement-assisted communication under AVP via $T$}\}.
   \end{align*}
\end{definition}

This version of entanglement-assisted capacity thus characterizes how well one can communicate with non-i.i.d. channel input, which may be interesting for various communication problems, and appears in particular in our study of fault-tolerant communication in Section~\ref{sec-coding-thm}.

\subsection{A coding theorem for entanglement-assisted communication under arbitrarily varying perturbation} \label{sec-avp-coding-thm}

\noindent The information-theoretic model outlined in Section~\ref{sec-avp-model} naturally raises the question of how much such AVP can hinder communication. Here, we show that communication is still possible in this scenario, and that achievable rates are given by the following theorem:

\begin{theorem}[Lower bound on the entanglement-assisted capacity under AVP] \label{thm-avp-coding-thm}
For any quantum channel $T:\mathcal{M}_{d_A} \rightarrow \mathcal{M}_{d_B}$, and for any $0\leq p \leq 1$, we have an entanglement-assisted capacity under AVP with 
\[C^{ea}_{AVP}(p,T) \geq  C^{ea}(T)-g(p) \]
where

\begin{align*}
g(p)&=  
2(d_Ad_B\log(d_Ad_B) +1)\sqrt{2{\log(d_B)p}} |\log( \frac{p^2}{d_Ad_B} )|\\& +2h\Big(d_Ad_B\sqrt{2{\log(d_B)p}} |\log(\frac{p^2}{d_Ad_B})|\Big) \\&
+5p\log(d_B)+ 2(1+2p)h\Big(\frac{2p}{1+2p}\Big)  =\mathcal{O}(p\log(p))\\
\end{align*}
\end{theorem}

\begin{proof}[Proof of Theorem~\ref{thm-avp-coding-thm}]
In this proof, we find a coding scheme for entanglement-assisted communication under AVP with strength $p\in\left[ 0,1\right]$ by constructing an encoder channel $\mathcal{E}$ and a decoder channel $\mathcal{D}$. Achievable rates are rates for which the fidelity in Definition~\ref{def-avp-coding-scheme} goes to $1$, which corresponds to rates $R$ for which the following expression goes to zero (which is a consequence of Fuchs-van-de-Graaf-inequality \cite{FvdG97}):

\begin{align}\begin{split}
    \label{this-has-to-go-to-zero-avp}
&\sup\{
\| \id_{cl}^{\otimes nR} - \mathcal{D} \circ  (T_{p,N}^{\otimes n} \circ (\mathcal{E} \otimes \sigma_S) \otimes  \id_2^{\otimes nR_{ea}}) \circ   (\id_{cl}^{\otimes nR} \otimes \varphi^{\otimes nR_{ea}}) \|_{1 \rightarrow 1}  \}
 \overset{n\rightarrow \infty}{\rightarrow}  0 \end{split}
\end{align}
where the supremum goes over the dimension $d_S\in \mathbbm{N}$, quantum states $\sigma_S\in\mathcal{M}_{d_S}^{\otimes n}$, and quantum channels $N:\mathcal{M}_{d_A}\otimes \mathcal{M}_{d_S} \rightarrow \mathcal{M}_{d_B}$.

Our construction makes use of a particular coding scheme for entanglement-assisted communication: For any quantum channel $T$, we consider the quantum channel $T_p=(1-p) T + p \frac{\mathbbm{1}}{d_B} \Tr(\cdot )$ with $p$ being the strength of the AVP. Using the coding scheme from \cite[Theorem~21.4.1]{Wilde13}, for any quantum channel $T_p$, and any pure bipartite quantum state $\varphi\in \mathcal{M}_{d_A}\otimes \mathcal{M}_{d_A}$, there exists an encoder $\mathcal{E}$ and a decoder $\mathcal{D}$ for $T_p$ such that
\begin{align}\begin{split} \label{eq-fid} \Xi( T_p^{\otimes n }):= F\Big(X,\mathcal{D} \circ  \big( ( T_p^{\otimes n } \circ \mathcal{E} ) \otimes \id_{2}^{\otimes nR_{ea}} \big) (X \otimes \varphi^{\otimes nR_{ea} })\Big)  \geq 1-{\epsilon_{ea}} \end{split}\end{align}
for any classical message $x$ with the corresponding quantum state $X=\ketbra{x}$ of length $nR'$, where 
\begin{align}\begin{split} \label{eq-coding-error} \epsilon_{ea}&\leq 12 e^{-\frac{n\delta^2}{2(\log(\lambda_{\min}))^2}} + 8 \cdot 2^{-n(I(A':B)_{(T_p\otimes \id_2)(\varphi)} -\eta(\delta,d_A,d_B)-d_Ad_B\frac{\log(n+1)}{n} -R')} \end{split}\end{align}
with the function $\eta(\delta,d_A,d_B)=2 (d_Ad_B\log(d_Ad_B) +1)  \delta+2h(d_Ad_B\delta)$ and with the smallest non-vanishing eigenvalue $\lambda_{\min}=\min\{\lambda\in \Spec((T_p\otimes\id_2)(\varphi))| \lambda>0 \}$. Here, $h(x)=-x\log(x)-(1-x)\log(1-x)$ denotes the binary entropy and $I(A:B)_{\rho}=H(A)_{\rho}+H(B)_{\rho}-H(AB)_{\rho}$ denotes the quantum mutual information.

To apply this coding scheme to the original channel $T$ under AVP, we will use the postselection-type result in \cite[Lemma~IV.10]{CMH20}: For any $\tilde{\delta}>0$, we have
\[ T_{p,N}^{\otimes n} (\cdot \otimes \sigma_S) \leq d_B^{n(p+\tilde{\delta})} T_p^{\otimes n} + e^{-\frac{n\tilde{\delta}^2}{3p}} S \]
for some quantum channel $S:\mathcal{M}_{d_A}^{\otimes n}\rightarrow \mathcal{M}_{d_B}^{\otimes n}  $. Here, we write $S_1\leq S_2$ for completely positive maps $S_1$ and $S_2$ if the difference $S_2-S_1$ is completely positive. Using a simple monotonicity property of the fidelity, we have: 
\begin{align*}
&\|\id_{cl}^{\otimes nR}- \Big( \mathcal{D} \circ (T_{p,N}^{\otimes n} \circ (\mathcal{E} \otimes \sigma_S) \otimes  \id_2^{\otimes nR_{ea}}) \circ (\id_{cl}^{\otimes nR} \otimes {\varphi}^{\otimes  nR_{ea}} ) \Big) \|_{1 \rightarrow 1}  \\
 &\leq 2 \sqrt{1-\Xi( {T}_{p,N}^{\otimes n }) }\\
 & \leq2 \sqrt{ d_B^{(p+\tilde{\delta})n} \Big(1- \Xi(  {T}_p^{\otimes n }) \Big)- e^{-\frac{n \tilde{\delta}^2}{3p}} } \\
 & \leq2\sqrt{d_B^{(p+\tilde{\delta})n} \epsilon_{ea}-e^{-\frac{n \tilde{\delta}^2}{3p}}  }\\
\end{align*}
where we make use of \cite[Proposition~4.3]{KW04} for the first inequality, \cite[Lemma~IV.10]{CMH20} for the second inequality, and Eq.~\eqref{eq-fid} in the last inequality. Clearly, we have Eq.~\eqref{this-has-to-go-to-zero-avp} if

\begin{align*} d_B^{(p+\tilde{\delta})n} \epsilon_{ea} & \overset{n\rightarrow \infty}{\rightarrow}  0\end{align*}
with $\epsilon_{ea}$ from Eq.~\eqref{eq-coding-error}.


For any  $\tilde{\delta}>0$ sufficiently small, we thus obtain a bound on the choices of $\delta$ and $R$, where the choice of $\delta$ should guarantee that
\begin{equation}\label{eq:deltaErrorBound}
d_B^{(p+\tilde{\delta})n} e^{-\frac{n\delta^2}{2(\log(\lambda_{\min}))^2}}\rightarrow 0 ,
\end{equation}
while the bound on $R$ should guarantee that
\begin{align}\begin{split} \label{eq:restErrorBound}
d_B^{(p+\tilde{\delta})n}  2^{-n(I(A':B)_{(T_p\otimes \id_2)(\varphi)} -\eta(\delta,d_A,d_B)-\frac{d_Ad_B}{n}\log(n+1) -R)} \rightarrow 0.\end{split}
\end{align}
To guarantee that Eq.~\eqref{eq:deltaErrorBound} holds, we choose $\delta$ as
\[\delta =\sqrt{2{\log(d_B)p}} |\log(\lambda_{min})| ,\]
and $\tilde{\delta}>0$ sufficiently small. We then find the bound
\begin{align}\begin{split} \label{eq-r-avp}
    R &<  I(A':B)_{( T_p \otimes \id_2)(\varphi)} - \log(d_B)p
     -\eta(\sqrt{2{\log(d_B)p}} |\log(\lambda_{min})|,d_A,d_B)
    \end{split}
\end{align}
such that Eq.~\eqref{eq:restErrorBound} holds. Thereby, we obtain
\begin{align*}
&C^{ea}_{AVP}(p,T)   \\&\hspace{1cm}\geq
I(A':B)_{(T_p \otimes \id_2)(\varphi)}- \log(d_B)p   -\eta(\sqrt{2{\log(d_B)p}} |\log(\lambda_{min}) |,d_A,d_B)  
\end{align*}
for any pure quantum state $\varphi$, where $\lambda_{\min}=\min\{\lambda\in \Spec((T_p\otimes\id_2)(\varphi))| \lambda>0 \} \}$.

To get an expression in terms of the usual entanglement-assisted capacity $C^{ea}(T)$, we use continuity estimates in the following way: Consider the pure quantum state $\varphi^*=\text{argmax}_{\varphi} I(A':B)_{( T\otimes \id_2)(\varphi)}$ which achieves the maximum for the quantum mutual information for the channel $T$.
 Then, consider a quantum state $\varphi_p=(1-p)\varphi^* + p\phi_+$.
 Then, $\|(T_p\otimes \id_2)(\varphi_p)-(T_p\otimes \id_2)(\varphi^*) \|_{\Tr} \leq p$, and the minimum eigenvalue of $( T_p\otimes \id_2)(\varphi_p)$ is lower bounded as $\lambda_{\min} \geq   \frac{p^2}{d_Ad_B}$. In addition, we know that $ \|T(\rho)-T_p(\rho) \|_{\Tr} \leq p $ for all quantum states $\rho$.
  By triangle inequality, we therefore find
  \[  \|(T_p\otimes \id_2)(\varphi_p)-(T\otimes \id_2)(\varphi^*) \|_{\Tr} \leq 2p \]
 Then, using the continuity of mutual information \cite[Corollary~1]{Shirokov17}, we find that
 \begin{align*}
     | I(A':B)_{( T_p \otimes \id_2)(\varphi_p)} - I(A':B)_{( T \otimes \id_2)(\varphi*)} | \leq 4p\log(d_B)+ 2(1+2p)h\Big(\frac{2p}{1+2p}\Big) 
 \end{align*} 
In total, we thereby obtain the bound \[C^{ea}_{AVP}(p,T) \geq  C^{ea}(T)-g(p) \]
where
\begin{equation*}\begin{split}
g(p)&=     \eta(\sqrt{2{\log(d_B)p}} |\log(\frac{p^2}{d_Ad_B})|,d_A,d_B) 
\\&+5p\log(d_B)+ (1+2p)h\Big(\frac{2p}{1+2p}\Big) \\& =
2(d_Ad_B\log(d_Ad_B) +1)\sqrt{2{\log(d_B)p}} |\log( \frac{p^2}{d_Ad_B} )|\\& +  2h\Big(d_Ad_B\sqrt{2{\log(d_B)p}} |\log(\frac{p^2}{d_Ad_B})|\Big) \\& 
+5p\log(d_B)+ 2(1+2p)h\Big(\frac{2p}{1+2p}\Big) .
\end{split}\end{equation*}
\end{proof}
As a consequence of this result, we find the following continuity in perturbation strength $p$ of the entanglement-assisted capacity under AVP. Moreover, the usual notion of entanglement-assisted capacity is recovered for vanishing perturbation probability $p$.

\begin{theorem}
For every $\eta>0$ and $d_A,d_B\in\mathbbm{N}$ there exists a $p(\eta,d_A,d_B)\in \left[ 0,1\right]$ such that 
\[
C^{ea}_{AVP}(p,T) \geq C^{ea}(T) - \eta,
\]
for every $p\leq p(\eta,d_A,d_B)$ and every quantum channel $T:\mathcal{M}_{d_A} \rightarrow \mathcal{M}_{d_B}$.
\end{theorem}

\begin{corollary}
Let $p\geq0$. Then, for every quantum channel $T:\mathcal{M}_{d_A}\rightarrow \mathcal{M}_{d_B}$, we have 
\[\lim_{p\rightarrow 0} C^{ea}_{AVP}(p,T) = C^{ea}(T).\]
\end{corollary}

\section{Fault-tolerant entanglement distillation} \label{sec-ft-ent-dist}

\noindent As outlined in Section~\ref{sec-ea-cap}, the entanglement-assisted capacity considers encoders and decoders as general quantum channels that have access to entanglement.
In a fault-tolerant setup, framing the encoder and decoder as circuits with an implementation in a fault-tolerant code means that the entanglement has to be transferred into the code space through an interface as explained in Section~\ref{sec-ft-cap-setup}. Naturally, this interface is in itself a fault-affected circuit, and can produce a noisy mixed state in the code space. 

\begin{figure}[htbp]
  \centering
       \includegraphics[width=10cm]{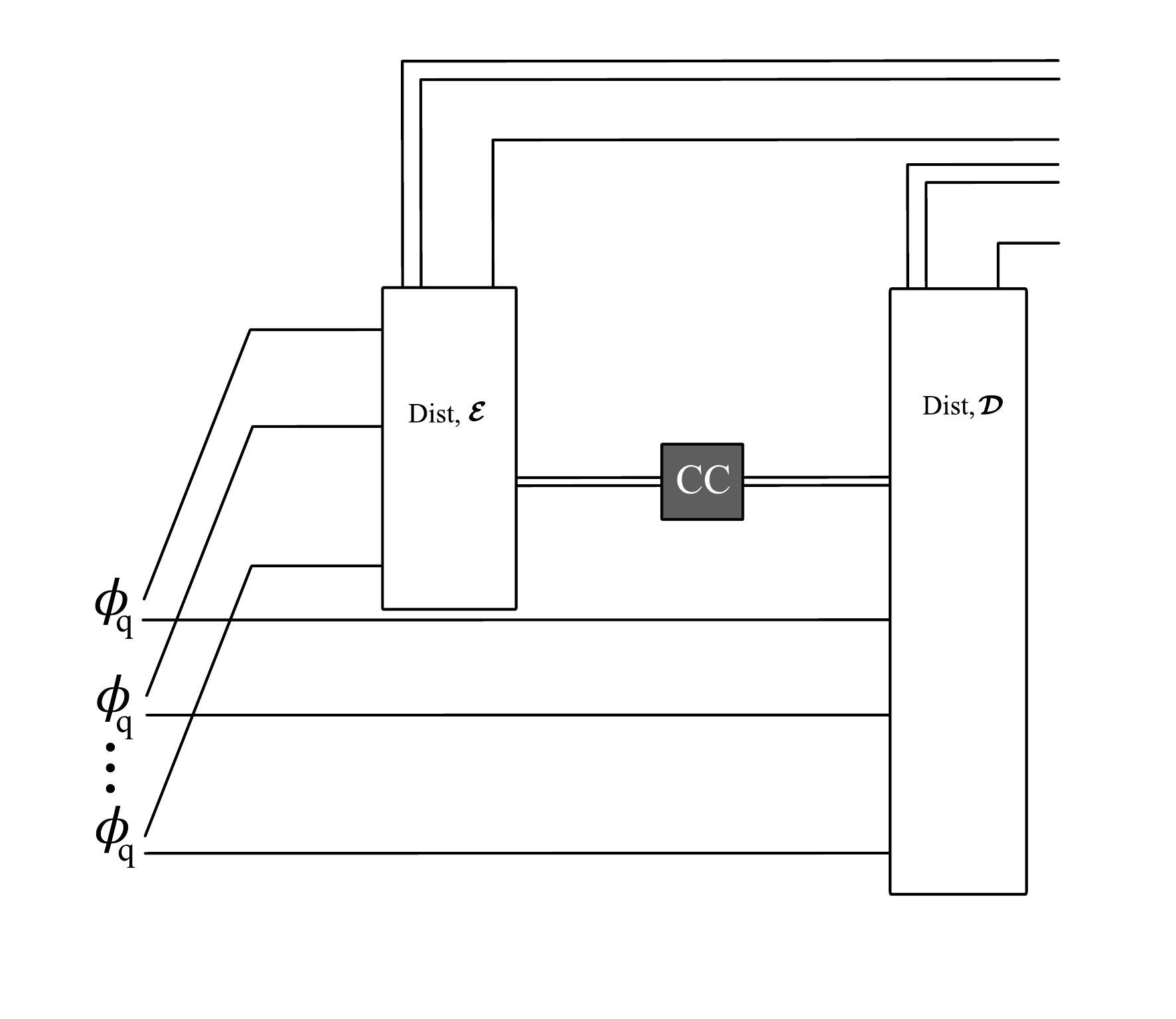}
\caption{\textbf{Setup for entanglement distillation based on the protocol in \cite{DW03}.} Two parties each have access to one part of $k$ noisy entangled states $\phi_q$. One party performs local operations ${\mathcal{E}}^{\Dist}$ and sends one-way classical communication to the other, who performs local operations ${\mathcal{D}}^{\Dist}$. The output state of this scheme is close in fidelity to $(1-H(\phi_q))k$ copies of the maximally entangled state (cf. Theorem~\ref{thm-ent-dist-original}).
}
\label{fig-ent-dist-normally}
\end{figure}

Fortunately, we can asymptotically carry out entanglement distillation with one-way classical communication \cite{DW03}, thereby transforming many copies of a noisy entangled state into fewer copies of a perfectly maximally entangled state.

\begin{theorem}[Entanglement distillation, see Theorem~10 in \cite{DW03}] \label{thm-ent-dist-original}
Let $\{\phi_+,\phi_-,\psi_+,\psi_-\}$ be the Bell basis of the space $\mathcal{M}_2^{\otimes 2}$.

For sufficiently large $k\in\mathbbm{N}$, there exists a $\delta\geq0$ and a quantum channel $\Dist:\mathcal{M}_2^{\otimes 2k}\rightarrow \mathcal{M}_2^{\otimes 2(1-H(\phi_q)))k}$ consisting of local operations and $H(\phi_q) k$ bits of one-way classical communication, such that $k$ copies of the state $\phi_{q}=(1-q) \phi_+ +\frac{q}{3} (\phi_- + \psi_+ + \psi_-)$ are mapped to $(1-H(\phi_q))k$ copies of the maximally entangled state $\phi_+$
with the following fidelity:
\[ F\big((\phi_+)^{\otimes (1-H(\phi_{q}))k}, \Dist(\phi_q^{\otimes k}) \big) \geq 1- \epsilon_{dist}(q)\]
with 
\[\epsilon_{dist}(q) \leq 
2e^{-\frac{k \delta^2}{\log(q/3)^2}}+\sqrt{2\sqrt{3}e^{-k \frac{\delta^2}{2\log(q/3)^2}}}\]
\end{theorem}

\begin{remark}
The von Neumann entropy of the state $\phi_{q}$ is $ H(\phi_{q})=-(1-q)\log(1-q) - q \log(\frac{q}{3})=h(q)+q\log (3)$ where $h$ denotes the binary entropy. We restrict ourselves to using states of this form because we consider the i.i.d. Pauli fault model, but similar considerations can be made for general noisy input states. In that case, the amount of maximally entangled states that can be obtained from copies of the state $\rho$ is given by its distillable entanglement $H(A)_{\rho}-H(AB)_{\rho}$ per copy, and the amount of classical communication that has to be performed amounts to $H(A)_{\rho}-H(B)_{\rho}+H(R)_{\rho}$ bits per copy, where $R$ denotes a purifying system \cite[Remark~11]{DW03}. In other words, our results extend to fault-tolerant entanglement distillation from arbitrary states, where the noisy state $\rho$ is additionally transformed by the noisy effective interfaces such that the entanglement effectively has to be distilled from the state $(N_{enc,p,l} \otimes N_{enc,p,l})(\rho)$. Fault-tolerant entanglement distillation could still be performed, but may require a higher number of copies of $T$ and may lead to fewer perfect maximally entangled pairs in the code space.
\end{remark}

\begin{remark}
Assuming faultless encoder and decoder circuits, but noisy mixed entangled states, using a subset of the channel copies for entanglement distillation implies a notion of classical capacity with assistance by noisy states, which may be of independent interest. 
For $R_{ea}\geq \frac{\log(2)}{1-H(\phi_q)}$, we see that the capacity $C_{\phi_q}^{ea}$ with assistance by $nR_{ea}$ copies of the state $\phi_{q}=(1-q) \phi_+ +\frac{q}{3} (\phi^- + \psi^+ + \psi^-)$ is given by $C_{\phi_q}^{ea}(T)\geq C_{\phi_+}^{ea}(T)- 
H(\phi_q) R_{ea} \frac{C_{\phi_+}^{ea}(T)}{C(T)} $.
Note that this is related to the scheme of special dense coding, a generalized version of superdense coding where the sender and receiver have access to arbitrary pairs of qubits and are connected by a noiseless, perfect quantum channel. Using purification procedures \cite{BPV98,Bowen01,ZZYS17} or directly finding a coding scheme \cite{HHHLT01}, coding protocols have been proposed and achievable rates have been computed for this scenario, where the latter also shows that states with bounded (i.e. non-distillable) entanglement do not enhance communication via a perfect quantum channel at all.
\end{remark}

Implementing the circuits for the distillation machines fault-tolerantly requires physical states to be inserted into the code space via an interface.
This interface is also subject to the fault model and only correct with a certain probability which cannot be made arbitrarily small by increasing the concatenation level of the concatenated $7$-qubit Steane code. Effectively, this leads to noisy states in the code space, which the distillation tries to counteract.

Note that this means that the input state into the whole protocol, the original sea of maximally entangled states can still be assumed to be noiseless; the input to our protocol for entanglement-assisted communication will be in the form of perfectly maximally entangled physical qubits that become noisy because of the fault-affected interface, as sketched in Figure~\ref{fig-ft-ea-cap-setup}. In summary, this proposed scheme for fault-tolerant distillation takes perfectly maximally entangled physical qubits as an input, and the desired output is in the form of perfectly maximally entangled states in the code space.

In order to show that a fault-tolerant distillation protocol with the concatenated $7$-qubit Steane code can transform physical maximally entangled states such that they are very close to maximally entangled states in the code space, we are going to make use of Lemma~\ref{thm-effective-encoder}. Because of the fault-affected interface, any quantum state that serves as input into a fault-tolerantly implemented circuit is transformed by the interface into an effective input state which is a mixture of the original state (with weight of approximately $1-4cp$, where $c$ is the constant from Theorem~\ref{thm-correct-interfaces}) and a noisy state (with weight of $4cp$), serving as input into the perfect circuit. This is true in particular for $k$ copies of the maximally entangled state. Then, we can employ fault-tolerant circuits implementing the protocol from \cite{DW03} to fault-tolerantly restore maximal entanglement for $(1-h(4cp)-4cp\log(3))k$ qubit states in the code space. Our setup for fault-tolerant distillation is sketched in Figure~\ref{fig-ft-ent-dist}.

\begin{figure}[htbp]
  \centering
      \includegraphics[width=10cm]{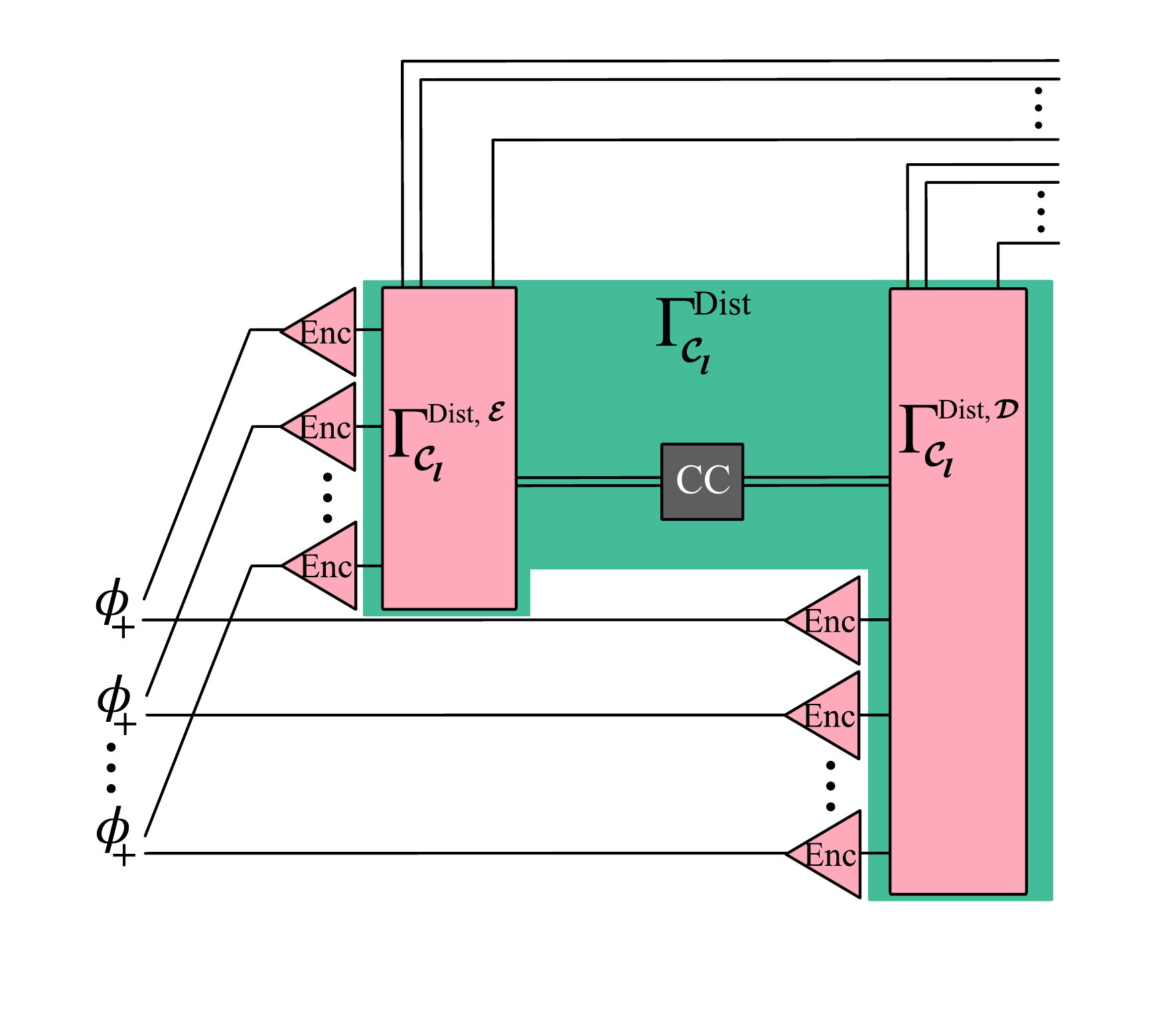}
\caption{\textbf{Setup for fault-tolerant entanglement distillation.} The local operations performed in Figure~\ref{fig-ent-dist-normally} are implemented in an error correcting code $\mathcal{C}_l$, and interfaces map the $k$ logical states into effective mixed states in the code-space (cf. Theorem~\ref{thm-ft-ent-dist-normally}).}
\label{fig-ft-ent-dist}
\end{figure}

\begin{theorem}[Fault-tolerant entanglement distillation] 
\label{thm-ft-ent-dist-normally}
For each $l\in \mathbbm{N}$, let $\mathcal{C}_l$ denote the $l$-th level of the concatenated 7-qubit Steane code with threshold $p_0$.
For any $0\leq p \leq \frac{p_0}{2}$ and for all $k\in \mathbbm{N}$ large enough, there exists a circuit 
$\Gamma^{\Dist}:\mathcal{M}_2^{\otimes 2k}\rightarrow \mathcal{M}_2^{\otimes 2\beta(4cp)k}$
 using $(1-\beta(4cp))k$ bits of classical communication, and two quantum states $\sigma^{\mathcal{E}}_S \in \mathcal{M}_{d_{S_E}}$ and $\sigma^{\mathcal{D}}_S \in \mathcal{M}_{d_{S_D}}$ such that
    \begin{align*}
   & \|(\DecI_l^*)^{\otimes 2\beta(4cp)k}\circ \big[ \Gamma^{\Dist}_{\mathcal{C}_l}  \circ \EncI_l^{\otimes 2k} \big]_{\mathcal{F}(p)} (\phi_+)^{\otimes k}  -(\phi_+)^{\otimes \beta(4cp)k}  \otimes \sigma^{\mathcal{E}}_S \otimes \sigma^{\mathcal{D}}_S\|_{\Tr}   \\ &\hspace{1cm} \leq  p_0 \left(\frac{p}{p_0}\right)^{2^l} |\Loc(\Gamma^{\Dist})|  +  \sqrt{ \epsilon_{dist}(4cp) } + \frac{2}{k} \end{align*}
with the constant $c$ from Theorem~\ref{thm-correct-interfaces}, $\epsilon_{dist}(q)$ the function from Theorem~\ref{thm-ent-dist-original}, and $\beta(q)=1-h(q)-q\log(3)$.
\end{theorem}

The proof employs techniques from \cite{AGP05} to relate the fault-affected circuit implementations sketched in Figure~\ref{fig-ft-ent-dist} to the ideal circuits via the threshold theorem and choosing a high enough concatenation level $l$. Due to Lemma~\ref{thm-effective-encoder}, the physical input states (which are maximally entangled states here) are acted upon by the effective interface.
Thereby, they are effectively transformed into the noisy mixed states of the form $(N_{enc,p,l}\otimes N_{enc,p,l}) (\phi_+)= (1-4cp) \phi_+ + 4cp \tilde{\sigma}_l $ for some quantum state $\sigma_l$, which are twirled into Bell-diagonal form $\phi_{4cp}=(1-4cp)\phi_+ + \frac{4cp}{3} (\phi_- +\psi_+ + \psi_-)$ by the first step of the distillation protocol. For states of this form, the results from Theorem~\ref{thm-ent-dist-original} apply.

While the apparatus described in Theorem~\ref{thm-ft-ent-dist-normally} performs fault-tolerant encoding and decoding, it still requires one-way classical communication between two parties, which is not allowed in the communication setup we investigate in the next section. For the purposes of fault-tolerant entanglement-assisted capacity, we therefore combine fault-tolerant implementation of these circuits with fault-tolerant classical communication via the channel $T$ to distill perfect maximal entanglement in the code space.
In this process, a fraction of the available channel copies is used to transmit classical communication. The protocol for this is essentially the same as the protocol from Theorem~\ref{thm-ft-ent-dist-normally}, where the classical communication between sender and receiver is not modelled by transmission over copies of the channel $\id_{cl}$, but instead transmitted by using the coding scheme from \cite[Theorem~V.8]{CMH20} as a subroutine on $\frac{h(4cp)+4cp\log(3)}{C_{\mathcal{F}(p)}(T)}$ copies of the channel $T$. For completeness, this process is sketched in Figure~\ref{fig-ft-ent-dist-via-t}.

\begin{figure}[htpb]
    \centering
    \includegraphics[width=10cm]{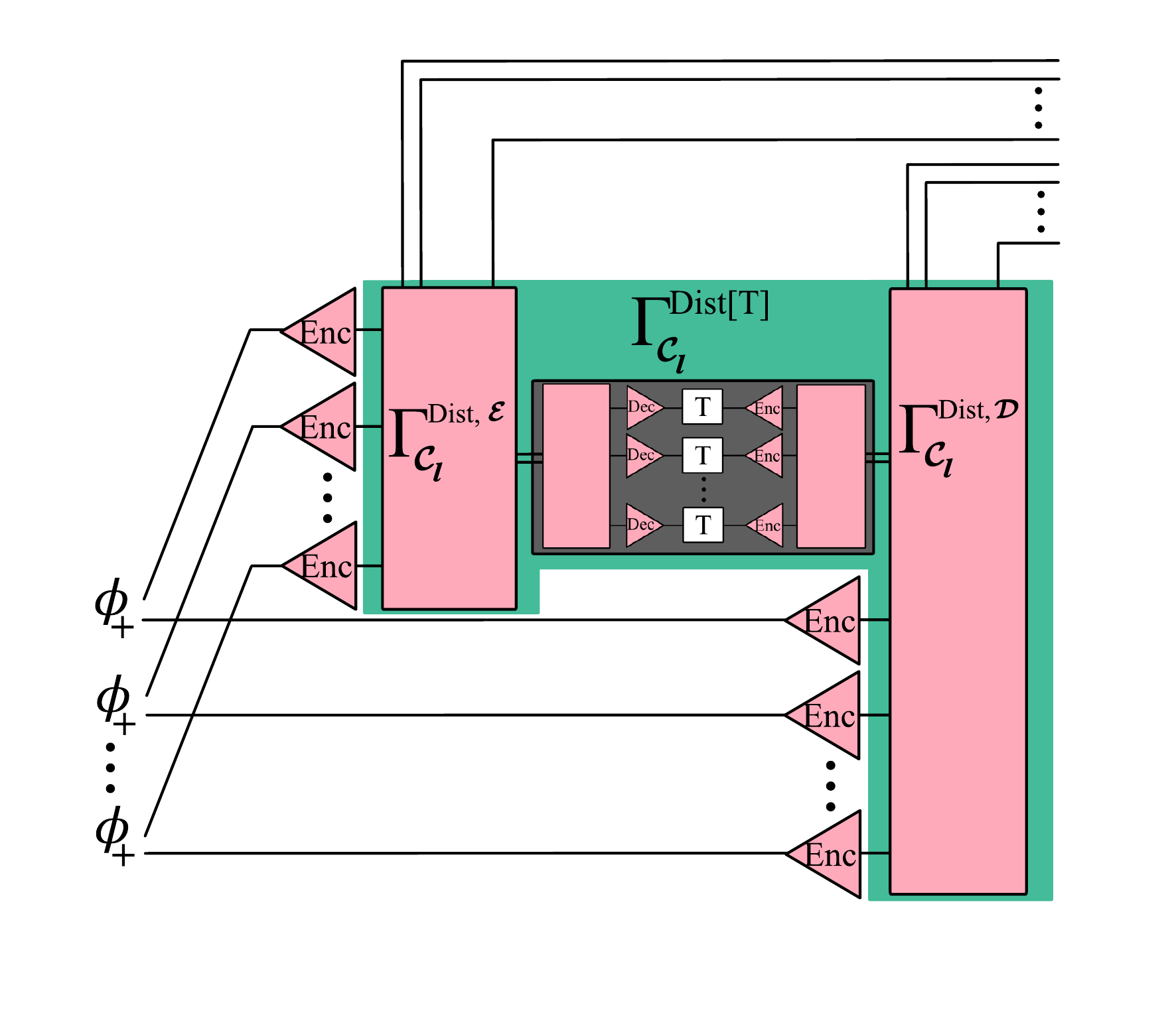}
    \caption{\textbf{Setup for fault-tolerant entanglement distillation with communication via a channel $T$.} The local operations performed in Figure~\ref{fig-ent-dist-normally} are implemented in an error correcting code $\mathcal{C}_l$, and interfaces map the $k$ logical states into effective mixed states in the code-space. The classical communication is performed as a subroutine using the coding scheme of \cite{CMH20}. Since the transmitted information is only classical, the syndrome states of these subroutines will not be correlated, allowing us to treat them as separate. 
    }
    \label{fig-ft-ent-dist-via-t}
\end{figure}

\section{A coding theorem for fault-tolerant entanglement-assisted capacity} \label{sec-coding-thm}

\noindent By performing fault-tolerant entanglement distillation as a subroutine and using a subset of the channels to convert entanglement, we can thus obtain the resource entanglement for entanglement-assisted communication in the code space.
Then, the remainder of channel copies and the recovered pure state entanglement can be used for information transfer via a coding scheme for entanglement-assisted communication, contributing to the channel capacity. This subroutine can be analyzed separately from the distillation part, as sketched in Figure~\ref{fig-allparts}. If the combined coding scheme is fault-tolerant, then the information transfer is fault-tolerant.

As outlined in Section~\ref{sec-ft-cap-setup}, the coding scheme we will use after the distillation in the fault-tolerant setting is based on the scheme used for an effective noisier channel in the faultless setting which takes a correlated syndrome state as part of its input. More precisely, we use the coding scheme for the effective channel to prepare the codeword states in the logical subspace. Then, our results on entanglement-assisted capacity under AVP apply in order to obtain bounds on achievable rates in the presence of correlated syndrome states.

Notably, the upper bound \[C^{ea}(T)\geq C^{ea}_{\mathcal{F}(p)}(T) \] trivially holds for any channel $T$. Here, we derive a lower bound in the form of a threshold theorem for fault-tolerant entanglement-assisted communication for any quantum channel $T$, where the fault-tolerant entanglement-assisted capacity approaches the usual, faultless case for vanishing gate error probability:

\begin{theorem}[Threshold theorem for fault-tolerant entanglement-assisted communication] \label{thm-main-result-1} 
For every quantum channel $T:\mathcal{M}_{d_A}\rightarrow \mathcal{M}_{d_B}$, and any $\eta>0$, there exists a threshold $p_{th}(\eta,T)> 0$ such that, for any $0\leq p\leq p_{th}$, we have
\[C^{ea}_{\mathcal{F}(p)}(T) \geq C^{ea}(T)-\eta\]
\end{theorem}

\begin{corollary} \label{thm-main-result-2}
Let $p\geq0$. Then, for every quantum channel $T:\mathcal{M}_{d_A}\rightarrow \mathcal{M}_{d_B}$, we have 
\[\lim_{p\rightarrow 0} C^{ea}_{\mathcal{F}(p)}(T) = C^{ea}(T) \]
\end{corollary}

This is a consequence of the following result (noting that $C(T)=0$ implies $C^{ea}(T)=0$):

\begin{theorem}[Lower bound on the fault-tolerant entanglement-assisted capacity] \label{thm-final-coding-thm}
For $0\leq p \leq 1$, let $\mathcal{F}(p)$ denote the
i.i.d. Pauli noise model and let $0\leq p_0 \leq 1$ denote the threshold of the concatenated 7-qubit Steane code.
For any quantum channel $T:\mathcal{M}_2^{\otimes j_1} \rightarrow \mathcal{M}_2^{\otimes j_2}$ with classical capacity $C(T)>0$ and for any $0\leq p \leq \min\{p_0/2,1/(2c(j_1+j_2)\}$, we have a fault-tolerant entanglement-assisted capacity with 
\[C^{ea}_{\mathcal{F}(p)}(T) \geq  C^{ea}(T)- 4f_1(p)\frac{C^{ea}(T)}{C(T)} -f_2(p) \]
where 
\begin{equation*}\label{f_1}
    f_1(p)= \frac{(h(4cp)+4cp\log(3))j_2 }{1-h(4cp)-4cp\log(3)},
\end{equation*}
and


\begin{align*}
f_2(p)
& = 2\sqrt{2{j_2p}} \big(2^{j_1+j_2} (j_1+j_2) +1\big) \big|2\log( \frac{2(j_1+j_2)cp }{2^{j_1j_2}})\big|  \\&  +2h\big(\sqrt{2{j_2p}} 2^{j_1+j_2} |2\log( \frac{2(j_1+j_2)cp }{2^{j_1j_2}})| \big) \\& 
 +(1+4(j_1+j_2)cp)h\big(\frac{4(j_1+j_2)cp}{1+4(j_1+j_2)cp}\big) \\ &   +10 (j_1+j_2)cpj_2 ,
\end{align*}
and with $c$ being the constant from Theorem~\ref{thm-correct-interfaces}.
\end{theorem}

\begin{proof}[Proof of Theorem~\ref{thm-final-coding-thm}]
In this proof, we construct a fault-tolerant coding scheme for entanglement-assisted communication affected by the i.i.d. Pauli noise model $\mathcal{F}(p)$ by proposing an encoder circuit $\Enc$ and a decoder circuit $\Dec$ which are implemented in the concatenated 7-qubit Steane code ${\mathcal{C}_l}$ with threshold $p_0$ and for some level $l$. With a rate of entanglement assistance $R_{ea}$, we will obtain a bound on rates $R$ that fulfill
\begin{align}\begin{split}
    \label{this-has-to-go-to-zero}
    &\|\id_{cl}^{\otimes nR}-   \big[ \Gamma_{\mathcal{C}_l}^{\mathcal{D}} \circ  ((\EncI_l \circ T \circ \DecI_l)^{\otimes n}\otimes \id_2^{*\otimes nR_{ea}})  \circ    (\Gamma_{\mathcal{C}_l}^{\mathcal{E}} \otimes \id_2^{*\otimes nR_{ea}} ) \circ\Compactcdots   \\&\hspace{1cm}\Compactcdots \circ (\id_{cl}^{\otimes nR} \otimes (\EncI^{\otimes 2}(\phi_+))^{\otimes nR_{ea}}) \big]_{\mathcal{F}(p)}\|_{1 \rightarrow 1} \overset{n\rightarrow \infty}{\rightarrow}  0 ,
\end{split}\end{align}
showing which rates $R$ are achievable.

Our proof will progress according the following strategy:

\begin{enumerate}
\item Construct the coding scheme out of the relevant subcircuits for distillation and coding under arbitrarily varying perturbations, as illustrated in Figure~\ref{fig-allparts}.
\item Choose the concatenation level $l$ corresponding to the number of locations in the entire coding scheme, including all subcircuits, in Eq.~\eqref{eq-level-choice}.
\item Bound the expression in Eq.~\eqref{this-has-to-go-to-zero} in terms of the effective channel and the distilled state using Theorem~\ref{thm-eff-channel} and Theorem~\ref{thm-ft-ent-dist-normally}. 
\item Invoke our results on entanglement-assisted capacity under arbitrarily varying perturbations of strength $2(j_1+j_2)cp$ from Theorem~\ref{thm-avp-coding-thm} to obtain a bound on the achievable rates.
\end{enumerate}

The coding scheme for fault-tolerant communication is based on the coding scheme for communication  at a rate $R_{AVP}$ under arbitrarily varying perturbation.
For each $n\in\mathbbm{N}$, using the Solovay-Kitaev theorem \cite{NC00}, we may choose specific quantum circuits $\Gamma^{\AVP,\mathcal{E}}$ and $\Gamma^{\AVP,\mathcal{D}}$ implementing the encoder $\mathcal{E}$ and decoder $\mathcal{D}$ used for communication under AVP, such that
\begin{equation}\label{eq-Enc}
\|\Gamma^{\AVP,\mathcal{E}} - \mathcal{E}\|_{1 \rightarrow 1} \leq \frac{1}{ n} \end{equation}
\begin{equation}\label{eq-Dec} \|\Gamma^{\AVP,\mathcal{D}}- \mathcal{D}\|_{1 \rightarrow 1} \leq \frac{1}{ n}\end{equation}

In addition, let $\Gamma^{\Dist[T]}$ be the circuit performing entanglement distillation, which is based on the circuit of Theorem~\ref{thm-ft-ent-dist-normally} with classical communication via a subset of the copies of the channel $T$, and an additional step of entanglement dilution to distill the bipartite pure entangled state $\varphi$ using the protocol from \cite{LP99}, which requires additional one-way classical communication at an asymptotically negligible rate so long as the state is pure (see also \cite[Footnote~1]{BSST02}). Thereby, similar to the sketch in Figure~\ref{fig-ft-ent-dist-via-t}, the fault-tolerant implementation of this distillation circuit distills $\varphi$ in the code space, whereby it is made available for the fault-tolerantly implemented communication setup.

Let $\Gamma^{\AVP,\mathcal{E}}_{\mathcal{C}_l} $, $\Gamma^{\AVP,\mathcal{D}}_{\mathcal{C}_l} $ and $\Delta_{\mathcal{C}_l}^{[T]}$ denote the implementations of these circuits in the 7-qubit Steane code with concatenation level $l$. 
The circuits $ \Gamma_{\mathcal{C}_l}^{\mathcal{E}}$ and $ \Gamma_{\mathcal{C}_l}^{\mathcal{D}}$ which implement our proposed fault-tolerant coding scheme are then constructed by the local parts of fault-tolerant entanglement distillation, followed by the fault-tolerant implementation of the coding scheme for arbitrarily varying perturbations, as sketched in Figure~\ref{fig-allparts}.

In total, the maximally entangled resource states are transformed by the noisy interface into effective noisy states in the code space. Thereafter, entanglement distillation is performed to restore pure state entanglement in the code space, using up a subset of the copies of $T$ for classical communication. Subsequently, the remaining copies of $T$ are used for entanglement-assisted communication.

By our construction, we obtain the following expression for Eq.~\eqref{this-has-to-go-to-zero} in terms of the two subroutines in Figure~\ref{fig-allparts}:

\begin{align}\begin{split} \label{eq-the-very-long-calculation1}
&  \| \id_{cl}^{\otimes nR } -  \big[ \Gamma_{\mathcal{C}_l}^{\mathcal{D}} \circ  \Big( \big( (\EncI_l \circ T \circ \DecI_l)^{\otimes n} \circ \Gamma_{\mathcal{C}_l}^{\mathcal{E}} \big)\otimes \id_2^{*\otimes nR_{ea}} \Big)  \circ \Compactcdots \\& \hspace{1cm} \Compactcdots\circ \big(\id_{cl}^{\otimes nR} \otimes ( \EncI_l^{\otimes 2}(\phi_+))^{\otimes nR_{ea}} \big) \big]_{\mathcal{F}(p)}\|_{1 \rightarrow 1} 
\\
  & = \| \id_{cl}^{\otimes nR }-     \big[ \Gamma^{\AVP,\mathcal{D}}_{\mathcal{C}_l}  
\circ  \Big( \big( (\EncI_l \circ T \circ \DecI_l)^{\otimes (1-\alpha_p)n} \circ \Gamma^{\AVP,\mathcal{E}}_{\mathcal{C}_l} \big)\otimes \id_2^{*\otimes \beta_{p}n} \Big) \circ \Compactcdots \\& \hspace{1cm}
\Compactcdots\circ  \big(\id_{cl}^{\otimes nR} \otimes \Gamma^{\Dist[T]}_{\mathcal{C}_l}  ( \EncI_l^{\otimes 2}(\phi_+))^{\otimes nR_{ea}} \big) \big]_{\mathcal{F}(p)}\|_{1 \rightarrow 1} 
\end{split}
\end{align}
where \[\alpha_{p}=\frac{h(4cp)+4cp\log(3)}{C_{\mathcal{F}(p)}(T)} R_{ea}\] and \[\beta_{p}=\frac{1-h(4cp)-4cp\log(3)}{H(A)_{\varphi}}R_{ea} .\]
\begin{figure}[htbp] 
  \centering
       \includegraphics[width=10cm]{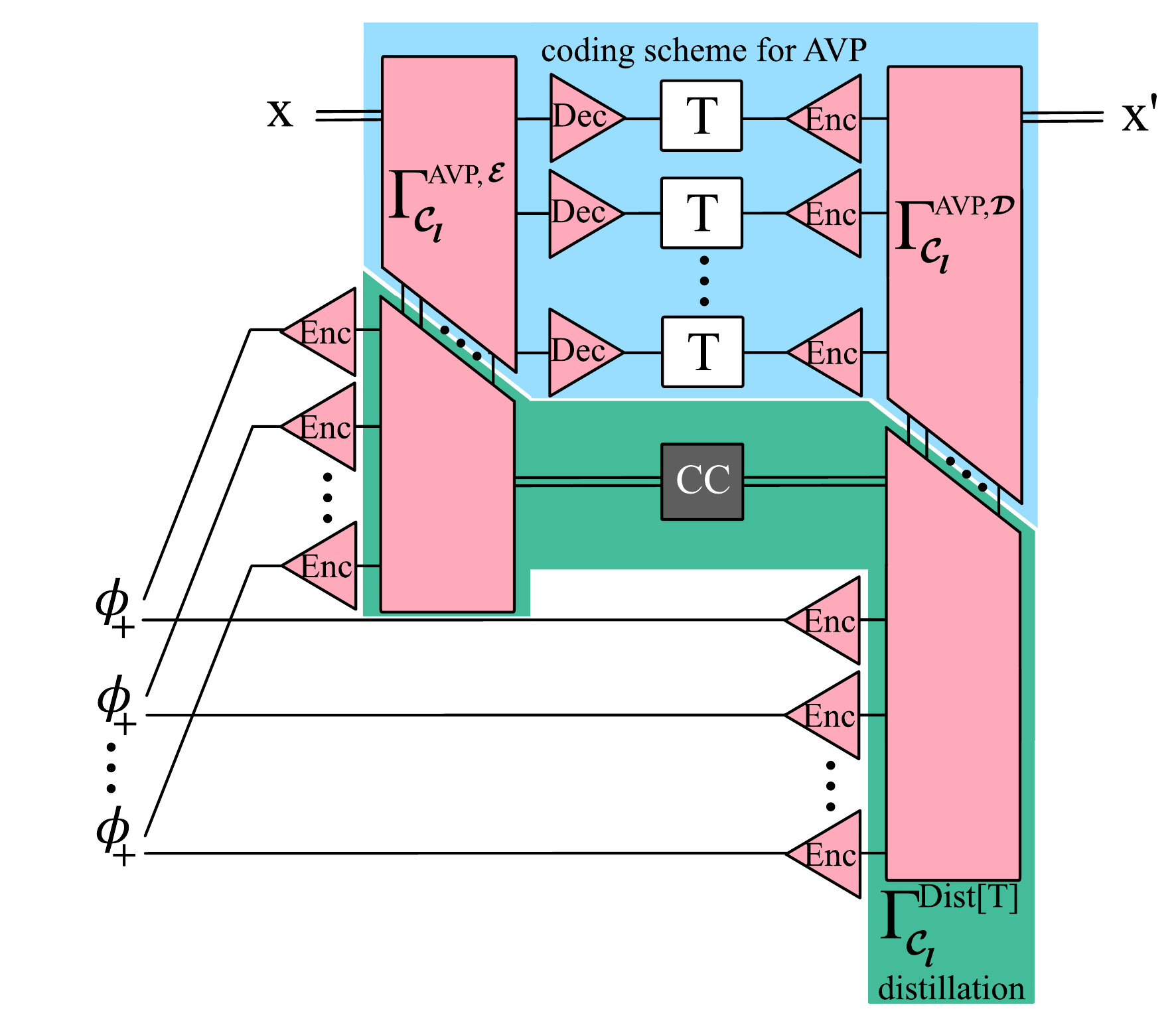}
\caption{\textbf{An illustration of the building blocks in our analysis.} The circuits for encoding and decoding in our fault-tolerant entanglement-assisted communication setup, as illustrated in Figure~\ref{fig-ft-ea-cap-setup}, are constructed out of 3 individual subroutines. Firstly, one subroutine (indicated in green) consists of the local operations performed during entanglement distillation. This part of the analysis will be based on our results in Theorem~\ref{thm-ft-ent-dist-normally}. The classical communication between the local distillation parts makes up another subroutine (indicated by a black box), and is performed using the coding scheme proposed in \cite{CMH20}, as sketched in Figure~\ref{fig-ft-ent-dist-via-t}.
The final subroutine (indicated in blue) is given by the fault-tolerant implementation of the coding scheme for 
entanglement-assisted communication under AVP, as described in Theorem~\ref{thm-avp-coding-thm}. Note that the distillation subroutine and the classical communication subroutine are only connected by classical information, while the coding scheme and the distillation subroutine are connected by quantum states in the code. In total, these subroutines and their analysis are combined in Eq.~\eqref{eq-the-very-long-calculation3}.}
\label{fig-allparts}
\end{figure}
For any sequence of circuits, we choose the Steane code concatenation level $l=l_{ n}$ high enough such that the implementations above fulfill:
\begin{equation}\label{eq-level-choice}
    \left(\frac{p}{p_0}\right)^{2^{l_{n}-1}} (|\Loc(\Gamma^{\AVP,\mathcal{E}})|+|\Loc(\Gamma^{\AVP,\mathcal{D}})|+j_1  n + |\Loc(\Gamma^{\Dist})|)  \leq \frac{1}{n}
\end{equation}
Using Theorem~\ref{thm-eff-channel}, Eq.~\eqref{eq-the-very-long-calculation1} can be bounded in terms of the effective channel $T_{p,N_l}=(1-2(j_1+j_2) cp) (T\otimes \Tr_S) + 2(j_1+j_2)cpN_l$ for some quantum channel $N_l$, and thereby connected to our results on capacity under AVP for perturbation probability $2(j_1+j_2) cp$:
\begin{align}\begin{split}\label{eq-the-very-long-calculation2}
  &  \| \id_{cl}^{\otimes nR }-     \big[ \Gamma^{\AVP,\mathcal{D}}_{\mathcal{C}_l}  
\circ  \Big( \big( (\EncI_l \circ T \circ \DecI_l)^{\otimes (1-\alpha_p)n} \circ \Gamma^{\AVP,\mathcal{E}}_{\mathcal{C}_l} \big)\otimes \id_2^{*\otimes \beta_{p}n} \Big)  \circ \Compactcdots \\&\hspace{1cm}  \Compactcdots
\circ \big(\id_{cl}^{\otimes nR} \otimes \Gamma^{\Dist[T]}_{\mathcal{C}_l}  ( \EncI_l^{\otimes 2}(\phi_+))^{\otimes nR_{ea}} \big) \big]_{\mathcal{F}(p)}\|_{1 \rightarrow 1} 
\\&
 \leq\| \id_{cl}^{\otimes nR } -   (\Gamma^{\AVP,\mathcal{D}}\otimes \Tr_S)  \circ  \Big( T_{p,N_l}^{\otimes (1-\alpha_p)n} \circ (\Gamma^{\AVP,\mathcal{E}} \otimes S_S )\big)\otimes \id_2^{*\otimes \beta_{p}n} \Big) \circ \Compactcdots \\&\hspace{1cm}  \Compactcdots \circ  \big(\id_{cl}^{\otimes nR} \otimes (\DecI_l^*)^{\otimes 2\beta_{p}n} \circ  \big[ \Gamma^{\Dist[T]}_{\mathcal{C}_l} ( \EncI_l^{\otimes 2}(\phi_+))^{\otimes nR_{ea}}\big]_{\mathcal{F}(p)}  \big) \|_{1 \rightarrow 1}  \\ &\hspace{2cm} +
2p_0 \left(\frac{p}{p_0}\right)^{2^{l_{n}-1}} (|\Loc(\Gamma^{\AVP,\mathcal{E}})|+|\Loc(\Gamma^{\AVP,\mathcal{D}})|+j_1  n)
\end{split}\end{align}
Note that during this inequality, the final error correction in the distillation part of the circuit is used for the transformation, but remains unchanged and is recombined with the circuit immediately.
Note also that $\Tr_S(\sigma_S)=1$ for any syndrome state.

Then, Theorem~\ref{thm-ft-ent-dist-normally} is employed to perform entanglement distillation of the bipartite entangled states $\varphi$ in the code space, leading to the following transformation:
\begin{align}\begin{split}\label{eq-the-very-long-calculation3}
 &
\| \id_{cl}^{\otimes nR } -   (\Gamma^{\AVP,\mathcal{D}}\otimes \Tr_S)  \circ  \Big( T_{p,N_l}^{\otimes (1-\alpha_p)n} \circ (\Gamma^{\AVP,\mathcal{E}} \otimes S_S )\big)\otimes \id_2^{*\otimes \beta_{p}n} \Big) \circ \Compactcdots \\&\hspace{1cm}  \Compactcdots \circ  \big(\id_{cl}^{\otimes nR} \otimes (\DecI_l^*)^{\otimes 2\beta_{p}n} \circ  \big[ \Gamma^{\Dist[T]}_{\mathcal{C}_l} ( \EncI_l^{\otimes 2}(\phi_+))^{\otimes nR_{ea}}\big]_{\mathcal{F}(p)}  \big) \|_{1 \rightarrow 1}  \\ &\hspace{1cm} +
2p_0 \left(\frac{p}{p_0}\right)^{2^{l_{n}-1}} (|\Loc(\Gamma^{\AVP,\mathcal{E}})|+|\Loc(\Gamma^{\AVP,\mathcal{D}})|+j_1  n) \\
 &\leq \| \id_{cl}^{\otimes nR } -   (\Gamma^{\AVP,\mathcal{D}}\otimes \Tr_S) \circ  \Big( T_{p,N_l}^{\otimes (1-\alpha_p)n} \circ (\Gamma^{\AVP,\mathcal{E}} \otimes S_S )\big)\otimes \id_2^{*\otimes\beta_{p}n } \Big)\circ \Compactcdots \\&\hspace{1cm}  \Compactcdots \circ (\id_{cl}^{\otimes nR} \otimes \varphi^{\otimes \beta_{p}n} \otimes \sigma_S)  \|_{1 \rightarrow 1}  +2p_0 (\frac{p}{p_0})^{2^{l_{n}-1}} |\Loc(\Gamma^{\Dist})|)+ \frac{2}{nR_{ea}}
  \\
 &\hspace{1cm} +\sqrt{\epsilon_{dist}(4cp)} +2p_0 (\frac{p}{p_0})^{2^{l_{n}-1}} (|\Loc(\Gamma^{\AVP,\mathcal{E}})|+|\Loc(\Gamma^{\AVP,\mathcal{D}})|+j_1  n)
\\
 &\leq \| \id_{cl}^{\otimes nR } -    \Gamma^{\AVP,\mathcal{D}}  \circ  \Big( T_{p,N_l}^{\otimes (1-\alpha_p)n} \circ  (\Gamma^{\AVP,\mathcal{E}} \otimes S_S(\sigma_S) )\big) \otimes \id_2^{*\otimes \beta_{p}n} \Big) \circ \Compactcdots \\&\hspace{1cm}  \Compactcdots \circ  (\id_{cl}^{\otimes nR} \otimes \varphi^{\otimes \beta_{p}n} ) \|_{1 \rightarrow 1}  +
\frac{1}{n}+\frac{2}{nR_{ea}}+\sqrt{\epsilon_{dist}(4cp)}
\end{split}
\end{align}
For the first inequality, we used Theorem~\ref{thm-ft-ent-dist-normally} to perform entanglement distillation on the $k=nR_{ea}$ entangled states $\big[ \EncI_l^{\otimes 2}\big]_{\mathcal{F}(p)} (\phi_+)
$ that have been affected by the noisy interface, obtaining the bipartite pure entangled state $\varphi$ in the code space. In total, this distillation process uses up $\alpha_{p}n$ copies of $T$ in the process to perform classical communication and produces $\beta_{p}n$ copies of $\varphi$ in the code space.

The second inequality is a consequence of our choice of concatenation level in Eq.~\eqref{eq-level-choice}.

We now use Eq.~\eqref{eq-Enc} and Eq.~\eqref{eq-Dec} to relate the circuits for the coding scheme in Eq.~\eqref{eq-the-very-long-calculation3} to the ideal operations:

\begin{align}\begin{split}
     \label{eq-the-very-long-calculation}
 & \| \id_{cl}^{\otimes nR } -    \Gamma^{\AVP,\mathcal{D}}  \circ  \Big( T_{p,N_l}^{\otimes (1-\alpha_p)n} \circ (\Gamma^{\AVP,\mathcal{E}} \otimes S_S(\sigma_S) )\big)\otimes \id_2^{*\otimes \beta_{p}n} \Big)  \circ \Compactcdots \\&\hspace{1cm}  \Compactcdots \circ  (\id_{cl}^{\otimes nR} \otimes \varphi^{\otimes \beta_{p}n} ) \|_{1 \rightarrow 1}  +
\frac{1}{n}+\frac{2}{nR_{ea}}+\sqrt{\epsilon_{dist}(4cp)}\\
& \leq  \| \id_{cl}^{\otimes nR } -  \mathcal{D}  \circ \Big( \big( T_{p,N_l}^{\otimes (1-\alpha_p)n} \circ (\mathcal{E} \otimes S_S(\sigma_S) )\big)\otimes \id_2^{*\otimes \beta_{p}n} \Big)   \circ \Compactcdots \\&\hspace{1cm}  \Compactcdots \circ (\id_{cl}^{\otimes nR} \otimes \varphi^{\otimes \beta_{p}n} )  \|_{1 \rightarrow 1}   +\frac{2}{nR_{ea}}+\sqrt{\epsilon_{dist}(4cp)}+\frac{3}{n}
\end{split}
\end{align}
Finally, we note that $\epsilon_{dist}(4cp)$ goes to zero as $n\rightarrow \infty$, and so do $\frac{3}{ n}$ and $\frac{2}{nR_{ea}}$. The remaining summand goes to zero for all achievable rates $R_{AVP}$ of entanglement-assisted communication under AVP with perturbation probability $2(j_1+j_2)cp$ and quantum channel $N_l$. For any entanglement-assistance at rate $R_{ea}'\geq 1$, these achievable rates are described by the expression in Eq.~\eqref{eq-r-avp}. Since a fraction of the copies of $T$ are used in the distillation, the communication rate is thus reduced. In total, we find the following bound on the achievable rate of fault-tolerant entanglement-assisted communication:
\[R < (1-\alpha_p) R_{AVP}\]
The bound on $R_{ea}'$ also automatically implies that $R_{ea}\geq \frac{1}{\beta_{p}}= \frac{H(A)_{\varphi}}{1-h(4cp)-4cp\log(3)}$. In order to simplify notation in the main theorem, and since additional entanglement does not increase capacity, we will henceforth set this rate to  \[R_{ea} := \frac{\log(2)}{1-h(4cp)-4cp\log(3)}\]
This leads to a fault-tolerant entanglement-assisted capacity of
\begin{align*}
C^{ea}_{\mathcal{F}(p)}(T) &\geq (1-\alpha_p) C_{AVP}((j_1+j_2)cp,T)\\& \geq  C^{ea}(T)- f_1(p)\frac{C^{ea}(T)}{C_{\mathcal{F}(p)}(T)} -f_2(p) \end{align*}
where
\begin{equation*}
    f_1(p)= \frac{(h(4cp)+4cp\log(3))\log(2) }{1-h(4cp)-4cp\log(3)}
\end{equation*}
and
\begin{align*}
f_2(p)
& = 2\sqrt{2{j_2p}} \big(2^{j_1+j_2} (j_1+j_2) +1\big) \big|2\log( \frac{2(j_1+j_2)cp }{2^{j_1j_2}})\big|  \\&+2h\big(\sqrt{2{j_2p}} 2^{j_1+j_2} |2\log( \frac{2(j_1+j_2)cp }{2^{j_1j_2}})| \big) \\& 
 +(1+4(j_1+j_2)cp)h\big(\frac{4(j_1+j_2)cp}{1+4(j_1+j_2)cp}\big) \\ &   +10 (j_1+j_2)cpj_2 ,
\end{align*}
Using \cite[Theorem~V.8]{CMH20}, for any channel $T$ with classical capacity $C(T)>0$, we find an explicit function $0\leq f(p)\leq C(T)$ such that we have
\begin{align*}
    C^{ea}_{\mathcal{F}(p)}(T)& \geq  C^{ea}(T)- f_1(p)\frac{C^{ea}(T)}{C(T)-f(p)} -f_2(p)\\ & 
     \geq C^{ea}(T)- f_1(p)\frac{C^{ea}(T)}{C(T)} \frac{1}{1-\frac{f(p)}{C(T)}} -f_2(p)\\&
    \geq C^{ea}(T)- 2f_1(p)\frac{C^{ea}(T)}{C(T)}\Big(1+\frac{f(p)}{C(T)}\Big) -f_2(p)\\&
    \geq C^{ea}(T)- 4f_1(p)\frac{C^{ea}(T)}{C(T)} -f_2(p)
\end{align*}
\end{proof}

Theorem~\ref{thm-main-result-1} is obtained as a direct consequence of this result:

\begin{proof}[Proof of Theorem~\ref{thm-main-result-1}]
For a given quantum channel $T$, we have 
$C^{ea}_{\mathcal{F}(p)}(T) \geq  C^{ea}(T)- 4f_1(p)\frac{C^{ea}(T)}{C(T)} -f_2(p) $ with the functions from Theorem~\ref{thm-final-coding-thm}.

Then, for any $\epsilon>0$, we can find a $p_0(T,\epsilon)$ such that $4f_1(p)\frac{C^{ea}(T)}{C(T)} +f_2(p)\leq\epsilon $ for all $0\leq p\leq p_0(T,\epsilon)$.
\end{proof}

It should be noted that the bound in Theorem~\ref{thm-final-coding-thm} is dependent on the individual channel $T$ and not only on its dimension, which would lead to a uniform convergence statement. Uniformity would follow if the quotient of classical and entanglement-assisted capacity were bounded for a given dimension, as has been conjectured in \cite{BSST02}.

\section{Conclusion and open problems}

\noindent The usual notion of capacity of a channel considers a perfect encoding of information into the channel, transfer via the (noisy) channel, and subsequent decoding. In real-world devices, this process of encoding and decoding the information cannot be assumed to be free of faults, which suggests the necessity of a modified notion of capacity.
Here, we show that entanglement-assisted transfer of information is possible at almost the same rates for fault-affected devices as long as the probability for gate error is below a threshold. 

Coding theorems can be understood as a conversion between resources, where quantum channels, entangled states and classical channels are used to simulate one another. Based on the notation from \cite{DHW08}, we say $\alpha \geq_{FT(p)} \beta$ if there exists a fault-tolerant transformation from a resource $\alpha$ to a resource $\beta$ at gate error $p$ with asymptotically vanishing overall error. 
Then, our coding theorem in Theorem~\ref{thm-final-coding-thm} for fault-tolerant entanglement-assisted communication via a quantum channel $T:\mathcal{M}_{d_{A}} \rightarrow \mathcal{M}_{d_{B}}$ corresponds to the following resource inequality: for any pure state $\varphi$ on $\mathcal{M}_{d_A}\otimes \mathcal{M}_{d_{A'}}$, we have
\begin{align*} \langle T \rangle + &\big(H(A)_{(T\otimes \id_{A'})(\varphi)}+ \mathcal{O}(p) \big) [qq] \\ &\geq_{FT(p)} \big(I(A':B)_{(T\otimes \id_{A'})(\varphi)}+\mathcal{O}(p \log(p)) \big) [c\rightarrow c]  \end{align*}
which specifies the asymptotic resource trade-off for fault-tolerant entanglement-assisted communication.
For vanishing gate error $p$, this reduces to the standard resource inequality from \cite[Eq.~54]{DHW08}.

The presented results can be understood as a further development of the toolbox of quantum communication with noisy encoding and decoding devices. Even though we chose to present our work in the frame of an explicitly chosen setup for fault-tolerant computation (i.e. 7-qubit Steane code and i.i.d. Pauli noise), the buildup is modular in nature, allowing for the adaptation to other fault-tolerant scenarios. 

We envision that our treatment of entanglement distillation with noisy devices will find application in other quantum communication contexts (e.g. connecting quantum computers, quantum repeaters, multiparty quantum communication and quantum cryptography).

As it is not covered by the presented techniques, we leave the study of fault-tolerant communication via infinite-dimensional quantum channels for future work. 

\noindent\rule[0.5ex]{\linewidth}{1pt}

\noindent\emph{End of arXiv:2210.02939 - Fault-tolerant entanglement-assisted communication}

\section{Some notes beyond the arXiv-manuscript}
\label{sec-extra}

The field of quantum fault-tolerance is rapidly evolving, highlighting the practical relevance of our results regarding communication in the presence of noise. In Section~\ref{sec-alternative-interface} we give an alternative theorem for effective interfaces which may be useful in other communication setups. In Section~\ref{sec-extra-1}, we comment in more detail on the scaling of our achievable rates from Theorem~\ref{thm-final-coding-thm} and how it originates from our proof strategy in Section~\ref{sec-coding-thm}. In addition, we highlight some proposals for other error correcting codes and outline how our results may be generalized beyond the 7-qubit Steane code in Section~\ref{sec-overhead-and-scaling}.

\subsection{An alternative effective interface theorem}\label{sec-alternative-interface}

In Section~\ref{sec-ft-interfaces}, we refined the effective interface theorems from \cite{CMH20} for the purpose of allowing for an additional quantum input, in order to construct a strategy for entanglement-assisted coding in Theorem~\ref{thm-effective-encoder} and \ref{thm-eff-decoder}.

Here, we briefly present an alternative theorem to \cite[Theorem~III.8]{CMH20} which combines the result for quantum circuits with classical output (and quantum input), and quantum circuits with classical input (and quantum output) into a more general result for quantum circuits with quantum input and quantum output. This result can be used for an alternative method for combining several quantum circuits with fault-tolerant implementations. For example, this would in principle also apply to the circuit for entanglement-distillation and for entanglement-assisted coding in our work, which could then be understood as one bigger circuit in a fault-tolerant implementation. However, we chose to present our results in such a way that the modular nature of our combinations of quantum circuit is highlighted. In particular, the approach in the alternative theorem below does not facilitate concatenation of two circuits with implementations in different error correction codes.

Nonetheless, this theorem is a generalization of results from \cite{CMH20} that could find application in other fault-tolerant communication scenarios, and therefore, we choose to state it here and give the proof.


\begin{theorem}[Effective encoder and decoder]
\label{thm-extra-effective-encoder}
Let $\Gamma:\mathcal{M}_2^{\otimes n} \rightarrow \mathcal{M}_2^{\otimes m}$ be a quantum circuit. For each $l\in \mathbbm{N}$, let the circuit $C_l$ denote the $l$-th level of the concatenated 7-qubit Steane code with threshold $p_0$. Moreover, let $\EncI_l:$ and $\DecI_l$  be the interface circuits for the $l$-th level of the concatenated 7-qubit Steane code with threshold $p_0$.

Then, for any $0\leq p \leq \frac{p_0}{2}$ and any $l\in\mathbbm{N}$, there exists a quantum channel $N_l:\mathcal{M}_2\rightarrow \mathcal{M}_2$, which only depends on $l$ and the interface circuit $\EncI_l$, a quantum channel $\tilde{N}_l:\mathcal{M}_2^{\otimes 7^l} \rightarrow \mathcal{M}_2$, which only depends on $l$ and the interface circuit $\DecI_l$, and a quantum state $\sigma_S$ on the syndrome space such that:
\begin{align*}
    &\| [\DecI_l^{\otimes m} \circ \Gamma_{\mathcal{C}_l} \circ \EncI_l^{\otimes n} ]_{\mathcal{F}(p)} -  (N_p^{dec,l})^{\otimes m} \circ (\Gamma \otimes \sigma_S) \circ  (N_p^{enc,l})^{\otimes n} \| \\& \leq  2n p_0 \big(\frac{p}{p_0}\big)^{2^{l-1}} + 2 p_0\big(\frac{p}{p_0}\big)^{2^l} |\Loc(\Gamma)|  
\end{align*} 
with
\[N_{enc,l,p} = (1-2cp) \id_2 +2cpN_l\]
and
\[N_{dec,l,p} = (1-2cp) \id_2 \otimes \trace_{S_l} +2cp \tilde{N}_l\]
where $c=p_0 \max \{|\Loc(\EncI_1)|,|\Loc(\DecI_1\circ EC)|\}$.
\end{theorem}

\begin{proof}
We decompose a single fault pattern $F=F_1\oplus F_2 \oplus F_3$, such that $F_1$ affects the circuit $\DecI_l^{\otimes m}$, $F_2$ affects $\Gamma_{\mathcal{C}_l}$, and $F_3$ affects $\EncI_l^{\otimes m}$. These three fault-patterns are independent of each other.

Furthermore, for the fault-tolerant implementation of the circuit $\Gamma_{\mathcal{C}_l}$ ends with an error correction gadget, and we can thus define a circuit $ \tilde{\Gamma}_{\mathcal{C}_l}$ such that $\Gamma_{\mathcal{C}_l}=\EC_l^{\otimes m} \circ \tilde{\Gamma}_{\mathcal{C}_l}$. We will denote the fault-pattern affecting $ \tilde{\Gamma}_{\mathcal{C}_l}$ by $F'_2$ and the fault-pattern affecting $\EC_l^{\otimes m}$ by $G(F'_2)$.
The same applies to the encoder circuit, and we define a circuit $ \tilde{\EncI}_{l}$ such that ${\EncI}_{l}=\EC_l\circ \tilde{\EncI}_{l}$. We denote the fault pattern on $\tilde{\EncI}_{l}^{\otimes n}$ by $F'_3$, and the fault-pattern on the corresponding $\EC_l^{\otimes n}$ by $H(F'_3)$.
In summary, we write 
\begin{equation*}
\begin{split}
&[\DecI_l^{\otimes m} \circ \Gamma_{\mathcal{C}_l} \circ \EncI_l^{\otimes n} ]_F\\ &= [\DecI_l^{\otimes m}]_{F_1} \circ [\Gamma_{\mathcal{C}_l}]_{F_2} \circ [\EncI_l^{\otimes n} ]_{F_3}\\
& =[\DecI_l^{\otimes m}]_{F_1} \circ [\EC_l^{\otimes m}]_{G(F'_2)}\circ [\tilde{\Gamma}_{\mathcal{C}_l}]_{F'_2} \circ [\EC_l^{\otimes n} ]_{H(F'_3)} \circ [\tilde{\EncI}_l^{\otimes n} ]_{F'_3}.\\
\end{split}
\end{equation*}
The $m$ tensor products of the decoding interface $\DecI_l$ are subject to independent faults, and therefore we can decompose the fault-pattern as $F_1=F_{1,1}\oplus F_{1,2} \oplus\Compactcdots\oplus F_{1,m}$. The same is true for the $m$ tensor products of $\EC$ and the fault-pattern $G(F'_2)$. Then, for any $k\in\{1,2,...,m\}$, consider the function 
\[I_{\alpha(F_{1,k},G_k(F'_2))} = \begin{cases}
\id_2\otimes\trace_S & \text{ if } \alpha(F_{1,k},G_k(F'_2))=0 \\
[\DecI_l]_{F_{1,k}}\circ \EncI^*_l &\text{ if } \alpha(F_{1,k},G_k(F'_2))=1
\end{cases}\]
with \[{\alpha(F_{1,k},G_k(F'_2))} = \begin{cases}
0 & \text{ if } \DecI_l\circ\EC_l \text{ is well-behaved under }  F_{1,k}\oplus G_k(F'_{2})\\
1 & \text{else.}
\end{cases} \]
With this function, we have 
\[[\DecI_l\circ \EC_l ]_{F_{1,k}\oplus G_k(F'_2)}=I_{\alpha(F_{1,k},G_k(F'_2))} \circ \DecI_l^* \circ [\EC_l]_{G_k(F'_2)}.\]
For a fault pattern $F$, we use the following decomposition:
\begin{equation*}
\begin{split}
&[\DecI_l^{\otimes m} \circ \Gamma_{\mathcal{C}_l} \circ \EncI_l^{\otimes n} ]_F 
\\& =[\DecI_l^{\otimes m}]_{F_1} \circ [\EC_l^{\otimes m}]_{G(F'_2)}\circ [\tilde{\Gamma}_{\mathcal{C}_l}]_{F'_2} \circ [\EC_l^{\otimes n} ]_{H(F'_3)} \circ [\tilde{\EncI}_l^{\otimes n} ]_{F'_3}\\
& =\bigotimes^m_{k=1}(I_{\alpha(F_{1,k},G_k(F'_2))}) \circ (\DecI_l^*)^{\otimes m} \circ [\EC_l^{\otimes m}]_{G(F'_2)}\circ [\tilde{\Gamma}_{\mathcal{C}_l}]_{F'_2}  \circ [\EC_l^{\otimes n} ]_{H(F'_3)} \circ [\tilde{\EncI}_l^{\otimes n} ]_{F'_3} \\
& =\bigotimes^m_{k=1}(I_{\alpha(F_{1,k},G_k(F'_2))}) \circ (\DecI_l^*)^{\otimes m} \circ [{\Gamma}_{\mathcal{C}_l}]_{F_2}  \circ [\EC_l^{\otimes n} ]_{H(F'_3)} \circ [\tilde{\EncI}_l^{\otimes n} ]_{F'_3} .\\
\end{split}
\end{equation*}
The sum over all fault patterns $F$ can be split into fault patterns under which $\Gamma$ is well-behaved (i.e. all exRecs in $\Gamma$ are well-behaved), and fault patterns under which $\Gamma$ is not well-behaved. We name the set of fault patterns under which $\Gamma$ is well-behaved $\mathcal{A}\subseteq \mathcal{F}$. Due to the threshold theorem in Theorem~\ref{thm-threshold}, the probability $\epsilon_A$ that a given fault-pattern is not well-behaved ($F\notin \mathcal{A}$) is upper-bounded by $\epsilon_A\leq p_0 (\frac{p}{p_0})^{2^l} |\Loc(\Gamma)|$. Thus there exists a quantum channel $E$ such that
\begin{equation*}
\begin{split}
&[\DecI_l^{\otimes m} \circ \Gamma_{\mathcal{C}_l} \circ \EncI_l^{\otimes m} ]_{\mathcal{F}(p)} \\&=\sum_{F} P(F) [\DecI_l^{\otimes m} \circ \Gamma_{\mathcal{C}_l} \circ \EncI_l^{\otimes n} ]_F  \\ &=\sum_{F} P(F) \bigotimes^m_{k=1}(I_{\alpha(F_{1,k},G_k(F'_2))}) \circ (\DecI_l^*)^{\otimes m} \circ [{\Gamma}_{\mathcal{C}_l}]_{F_2}  \circ [\EC_l^{\otimes n} ]_{H(F'_3)} \circ [\tilde{\EncI}_l^{\otimes n} ]_{F'_3} \\
&=\sum_{F\in \mathcal{A}} P(F) \bigotimes^m_{k=1}(I_{\alpha(F_{1,k},G_k(F'_2))}) \circ (\Gamma\otimes S_{F_2}) \circ (\DecI_l^*)^{\otimes n} \circ [\EC_l^{\otimes n}]_{H(F'_3)} \circ [\tilde{\EncI}_l^{\otimes n} ]_{F'_3} \\&\qquad \qquad\qquad\qquad\qquad\qquad\qquad\qquad\qquad\qquad\qquad\qquad\qquad\qquad\qquad\qquad+ \epsilon_A E\\
& =\sum_{F\in \mathcal{A}} P(F) \bigotimes^m_{k=1}(I_{\alpha(F_{1,k},G_k(F'_2))}) \circ (\Gamma\otimes S_{F_2}) \circ (\DecI_l^*)^{\otimes n}  \circ [\EncI_l^{\otimes n} ]_{F_3}+\epsilon_A E . \\
\end{split}
\end{equation*}
The faults affect the circuit independently, and the probability of a given fault-pattern $F=F_1\bigoplus F_2 \bigoplus F_3$ can thus be expressed as $P(F)=P_1(F_1)P_2(F_2)P_3(F_3)$. Therefore,
\begin{equation}\label{eq-weird-interface-proof}
\begin{split}
&[\DecI_l^{\otimes m} \circ \Gamma_{\mathcal{C}_l} \circ \EncI_l^{\otimes n} ]_{\mathcal{F}(p)}\\ & = \big( \sum_{F_1,F_2\in \mathcal{A}'} P_1(F_1)P_2(F_2) \bigotimes^m_{k=1}(I_{\alpha(F_{1,k},G_k(F'_2))}) \circ (\Gamma\otimes S_{F_2}) \big)\circ \Compactcdots \\&\hspace{1cm} \Compactcdots \circ\big( \sum_{F_3} P_3(F_3) (\DecI_l^*)^{\otimes n}  \circ [\EncI_l^{\otimes n} ]_{F_3} \big)+\epsilon_A E .\\
\end{split}
\end{equation}
Now, we can regard the part of the circuit that is affected by $F_3$ separately, in order to prove that it effectively behaves like the channel $N_{q,enc,l}$. For any distribution of fault patterns $J$ with probability $P(J)$, we have
\begin{equation*}
\begin{split}
&\sum_{J} P(J) (\DecI_l^*)  \circ [\EncI_l ]_{J} \\&=\sum_{J \text{ s.t. $[\EncI_l ]_{J}$ correct}} P(J) (\DecI_l^*)  \circ [\EncI_l ]_{J} +\sum_{J \text{ s.t. $[\EncI_l ]_{J}$ not correct}} P(J) (\DecI_l^*)  \circ [\EncI_l ]_{J}\\
&= (1-q) \id_2\otimes \sigma +q N_l\\
&:=N_{q,enc,l}.
\end{split}
\end{equation*}
where $N_l=\frac{1}{q_2} \sum_{F\text{ s.t. $[\EncI_l ]_{J}$ not correct}} P(J) (\DecI_l^*)  \circ [\EncI_l ]_{J}$, and $q\leq 2cp$. 
Because the interfaces are independent of each other, when taking the tensor product, we have
\begin{equation*}
\begin{split}
N_{q,enc,l}^{\otimes n}  &=( \sum_{J} P(J) (\DecI_l^*)  \circ [\EncI_l ]_{J})^{\otimes n}\\& = \sum_{J_i, \ i=1,...,n} P(J_1)P(J_2)...P(J_n)  (\DecI_l^*)^{\otimes n} \circ ([\EncI_l ]_{J_1} \otimes \Compactcdots \otimes [\EncI_l ]_{J_n})\\
&= \sum_{J'} P(J') (\DecI_l^*)^{\otimes n}  \circ [\EncI_l^{\otimes n}]_{J'} .
\end{split}
\end{equation*}
For $J'=F_3$, we insert this into Eq.~\eqref{eq-weird-interface-proof} and obtain
\begin{equation*}
\begin{split}
&[\DecI^{\otimes m} \circ \Gamma_{\mathcal{C}_l} \circ \EncI^{\otimes m} ]_{\mathcal{F}(p)}  \\&= \Big( \sum_{F_1,F_2\in \mathcal{A}'} P_1(F_1)P_2(F_2) \bigotimes^m_{k=1}(I_{\alpha(F_{1,k},G_k(F'_2))}) \circ (\Gamma\otimes \sigma_{F_2}) \Big) \circ\big( N_{q,enc,l}^{\otimes m}  \big)+\epsilon_A E .\\
\end{split}
\end{equation*}
For the encoding interface, our result above is thus already in the desired shape. What remains to be proven is that the decoding interface can also be expressed as a similar effective channel.

For any $k\in\{1,2,...,m\}$, consider the function
\[\tilde{I}_{\beta(F_{1,k})} = \begin{cases}
\id_2\otimes\trace_S & \text{ if } \beta(F_{1,k})=0 \\
[\DecI_l]_{F_{1,k}}\circ \EncI^*_l &\text{ if } \beta(F_{1,k})=1 
\end{cases}\]
with \[{\beta(F_{1,k})} = \begin{cases}
0 & \text{ if there exists $G_k(F'_{2})$ such that } \DecI_l\circ\EC_l \text{ is well-behaved under }  F_{1,k}\oplus G_k(F'_{2})\\
1 & \text{else.}
\end{cases} \]
We have \begin{equation*}
\begin{split}
& \sum_{F_2\in \mathcal{A}} P(F) \| \bigotimes^m_{k=1}I_{\alpha(F_{1,k},G_k(F'_2))} - \bigotimes^m_{k=1}\tilde{I}_{\beta(F_{1,k})} \|_{1\rightarrow 1} \\&\leq 2 \Prob[\exists k: \beta(F_{1,k})\neq \alpha(F_{1,k},G_k(F'_2))]\\&\leq 2 p_0 (\frac{p}{p_0})^{2^{l-1}}, \end{split}
\end{equation*}
and thus
\begin{equation*}
\begin{split}
& \sum_{F_1,F_2\in \mathcal{A}'} P_1(F_1)P_2(F_2) \bigotimes^m_{k=1}(I_{\alpha(F_{1,k},G_k(F'_2))}) \circ (\Gamma\otimes S_{F_2})  \\& =  \sum_{F_2\in \mathcal{A}} P_2(F_2) (\sum_{F_1} P_1(F_1) \bigotimes^m_{k=1}(I_{\alpha(F_{1,k},G_k(F'_2))}) ) \circ (\Gamma\otimes S_{F_2}) \\&\leq  \sum_{F_2\in \mathcal{A}} P_2(F_2) (\sum_{F_1} P_1(F_1) \bigotimes^m_{k=1}(\tilde{I}_{\beta(F_{1,k})}) ) \circ (\Gamma\otimes S_{F_2}) +2 p_0 (\frac{p}{p_0})^{2^{l-1}} \\&=  \sum_{F_2\in \mathcal{A}} P_2(F_2) ( \bigotimes^m_{k=1}( \sum_{F_{1,k}} P_1(F_{1,k})\tilde{I}_{\beta(F_{1,k})}) ) \circ (\Gamma\otimes S_{F_2}) +2 p_0 (\frac{p}{p_0})^{2^{l-1}} \\&=   ( \bigotimes^m_{k=1}( \sum_{F_{1,k}} P_1(F_{1,k})\tilde{I}_{\beta(F_{1,k})}) ) \circ (\Gamma\otimes (\sum_{F_2\in \mathcal{A}} P_2(F_2)S_{F_2}) ) +2 p_0 (\frac{p}{p_0})^{2^{l-1}} \\&=   ( \bigotimes^m_{k=1}( \sum_{F_{1,k}} P_1(F_{1,k})\tilde{I}_{\beta(F_{1,k})}) ) \circ (\Gamma\otimes S_S) +2 p_0 (\frac{p}{p_0})^{2^{l-1} }\\&=   N_{dec,r,l}^{\otimes m} \circ (\Gamma\otimes S_S) +2 p_0 (\frac{p}{p_0})^{2^{l-1}} ,\\&
\end{split}
\end{equation*}
where we define the quantum channel $S_S:=\sum_{F_2} P(F_2) S_{F_2}$, and the quantum channel $N_{dec,r,l}= (1-r) \id_2 \otimes \trace_S + r \tilde{N}_l$, which corresponds to splitting the sum over $F_{1,k}$ into fault-patterns under which the circuit is well-behaved and fault-patterns under which it is not well-behaved, and we have $r\leq 2cp$. For this channel, we have:
$\sum_{F_1} P_1(F_1) \bigotimes^m_{k=1}(\tilde{I}_{\beta(F_{1,k})}) )  =N_{dec,r,l}^{\otimes m}$.
Then, a triangle inequality gives the desired result.\end{proof}

\subsection{Scaling of our capacity bound}\label{sec-extra-1}

Our bound in Theorem~\ref{thm-final-coding-thm} should be regarded mostly as a proof-of-principle bound, and we place our focus on the scaling in the gate error $q\in [0,1]$.
For a real-world implementation of a communication protocol, the dependence of the fault-tolerant rate on the channel dimension strongly limits the achievable rate. Perhaps this could in part be remedied by the use of tighter bounds and better approximations, particularly for the cases of specific channels and specific error correction codes. This is further illustrated in Example~\ref{ex-howbaditis}, where we look at the scaling of our bound for a simple example channel.

\begin{example}[Fault-tolerant entanglement-assisted capacity of the qubit depolarizing channel]\label{ex-howbaditis}

Here, we illustrate the scaling of our bound on fault-tolerantly achievable rates from Theorem~\ref{thm-final-coding-thm} for the qubit depolarizing channel $T_p$ with depolarizing probability $p\in[0,1]$ as introduced in Example~\ref{ex-depo-channel}. For this channel, we have a closed-form expression for the entanglement-assisted capacity, and an upper bound on the quotient between entanglement-assisted capacity and classical capacity due to \cite{BSST99}. Since it is a qubit channel, the dimension is $d=2$.
We plot our bounds on the fault-tolerant entanglement-assisted capacity for various combinations of gate error probability $q$ and depolarizing probability $p$.

In Figure~\ref{fig-howbaditis1}, we plot the fault-tolerant entanglement-assisted capacity as a function of the depolarizing probability $p$ for various values of gate error probability $q$. It is clearly visible that the capacity is reduced by an amount that depends on $q$ (not on $p$), and that the reduction is increased for higher gate error probabilities $q$, where the scaling is dominated by $q\log(q)$.

\begin{figure}[htbp]
    \centering
    \includegraphics[width=10cm]{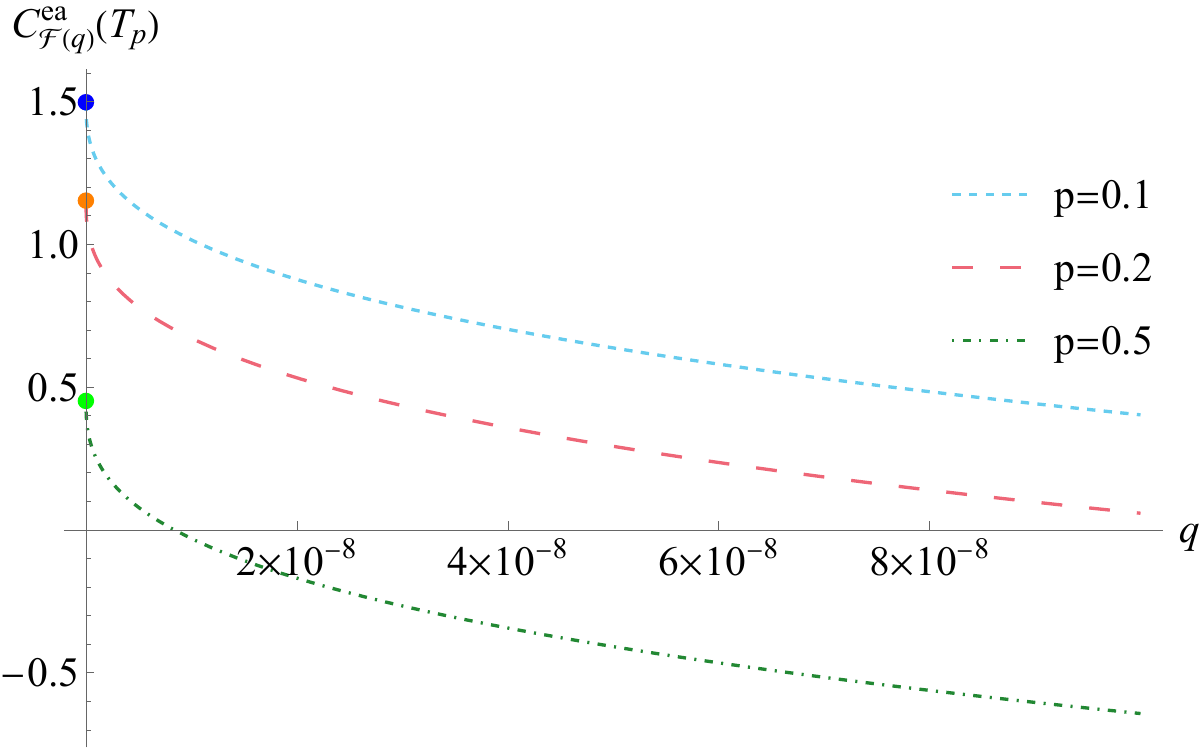}
    \caption{Plot of the fault-tolerant entanglement-assisted capacity of the qubit depolarizing channel as a function of the depolarizing probability $p$. The separation increases as the gate error $q$ increases.}
    \label{fig-howbaditis1}
\end{figure}

In Figure~\ref{fig-howbaditis2}, we plot the fault-tolerant entanglement-assisted capacity as a function of the gate error probability $q$ for some values of $p$. We see that the achievable rates decrease sharply for small gate error probabilities. Our findings are consistent with the observation that the difference between the fault-less and the fault-affected capacity behaves as $q\log(q)$.

\begin{figure}[htbp]
    \centering
    \includegraphics[width=10cm]{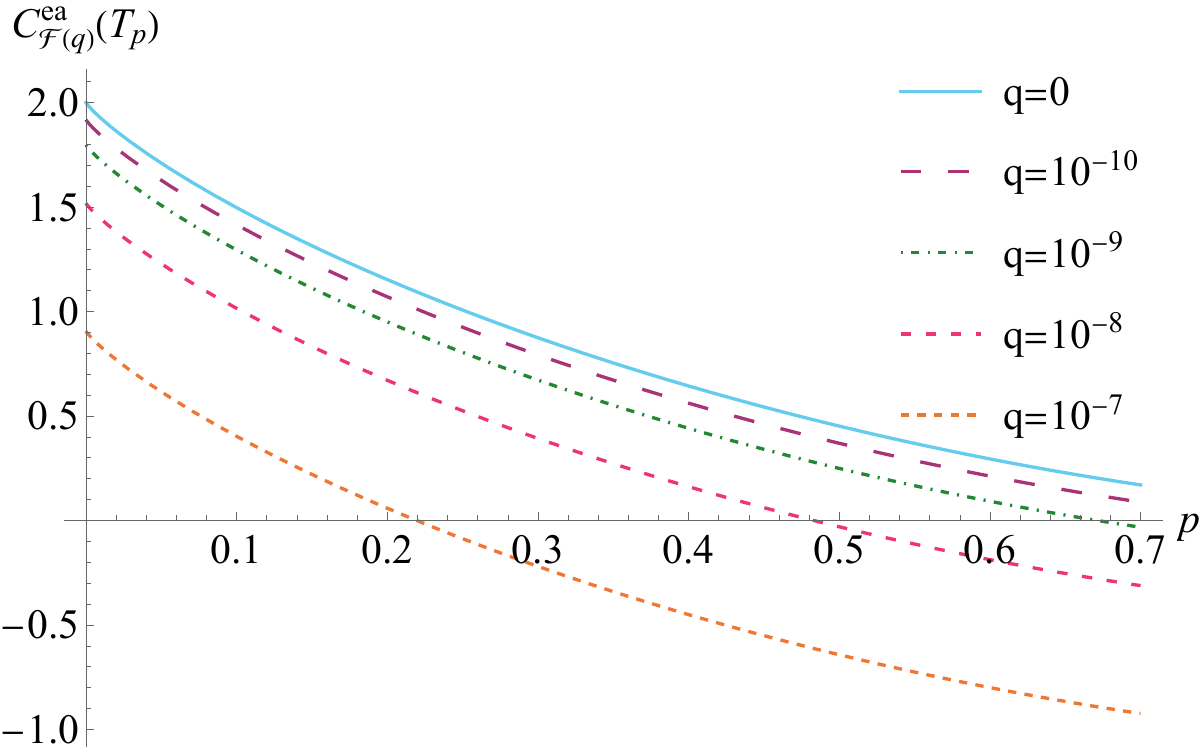}
    \caption{Plot of the lower bound on the fault-tolerant entanglement-assisted capacity of the qubit depolarizing channel as a function of the gate error $q$ for three different values of depolarizing probability $p$. The fault-less entanglement-assisted capacity is given at $q=0$.}
    \label{fig-howbaditis2}
\end{figure}
\end{example}

The function $f(q)$ which features in our bound on the achievable rate in Theorem~\ref{thm-final-coding-thm} is a sum of several terms which can be traced back to steps in our proof strategy. We can distinguish three main contributions:  
\begin{itemize}
\item A contribution due to our use of entanglement distillation: \[f_{DIST}(q)=4 \frac{C^{ea}(T)}{C(T)} \frac{(h(4cq)+4cq\log(3))\log(2) }{1-h(4cq)-4cq\log(3)}\]
    \item A contribution from the continuity relation of quantum mutual information: 
    \[f_{CONT}(q)= (1+4(j_1+j_2)cq)h\big(\frac{4(j_1+j_2)cq}{1+4(j_1+j_2)cq}\big)  +8 (j_1+j_2)cqj_2 \]
    \item A contribution due to our use of the AVP coding scheme:
    \begin{align*}
        f_{AVP}(q)&=2\sqrt{2{j_2q}} \big(2^{j_1+j_2} (j_1+j_2) +1\big) \big|2\log( \frac{2(j_1+j_2)cq }{2^{j_1j_2}})\big|  \\&  +2h\big(\sqrt{2{j_2q}} 2^{j_1+j_2} |2\log( \frac{2(j_1+j_2)cq }{2^{j_1j_2}})| \big)+ 2 (j_1+j_2)cqj_2 
    \end{align*}
\end{itemize}
resulting in a total reduction of the rate by $f(q)=f_{DIST}(q)+f_{CONT}(q)+f_{AVP}(q)$.

We compare these contributions briefly in Figure~\ref{fig-mainoffenders} for the depolarizing channel from Example~\ref{ex-depo-channel}. For small gate error $q$, the contribution from the arbitrarily varying coding scheme dominates and behaves similar to $q\log(q)$. For higher gate errors, the contribution due to entanglement distillation becomes more and more relevant.

\begin{figure}[htbp]
     \centering
     \begin{subfigure}[b]{0.7\textwidth}
         \centering
         \includegraphics[width=\textwidth]{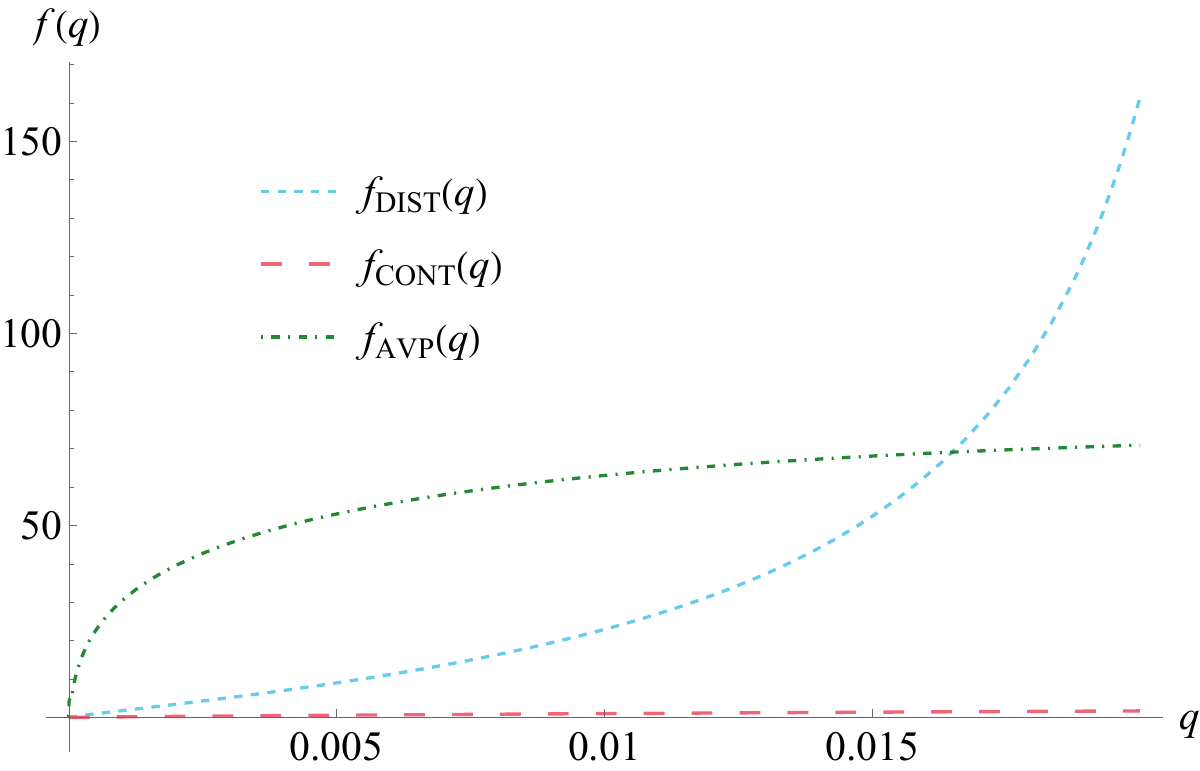}
         \caption{Comparison of the three contributing functions}
     \end{subfigure}
     \hfill
     \begin{subfigure}[b]{0.7\textwidth}
         \centering
         \includegraphics[width=\textwidth]{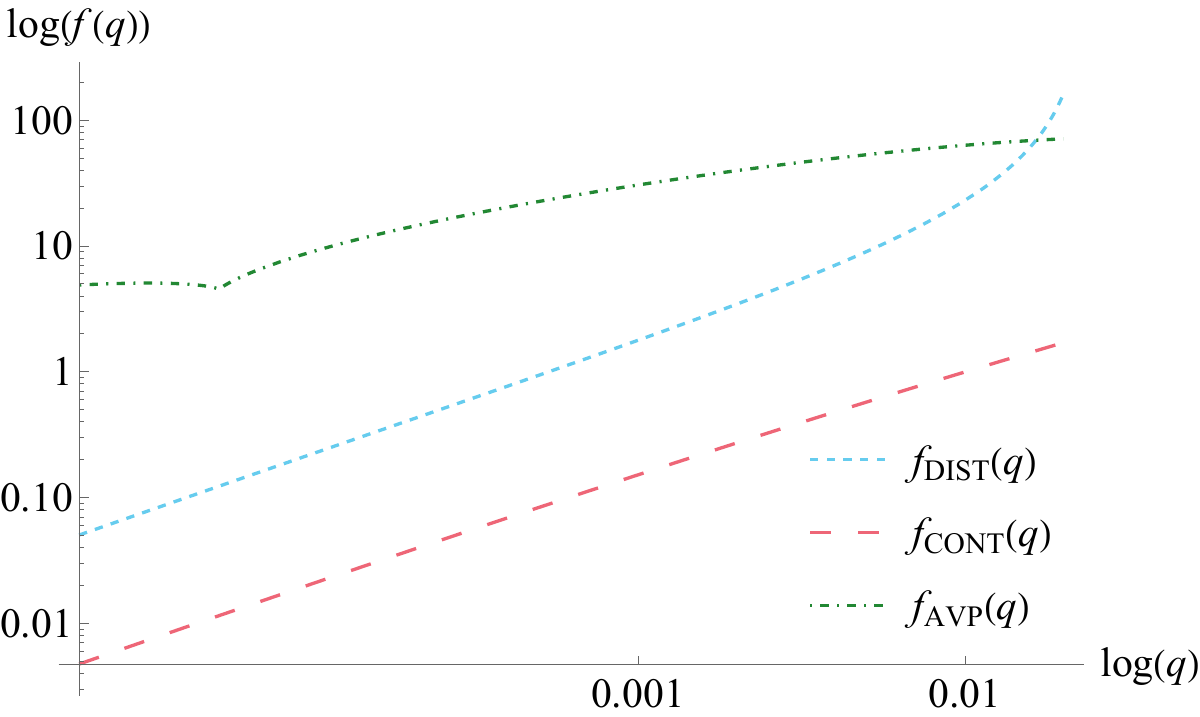}
         \caption{Comparison of the three contributing functions in a logarithmic plot}
     \end{subfigure}
        \caption{Plot of the scaling of the different contributions to the function $f(q)$ by which the capacity is reduced through gate errors.}
        \label{fig-mainoffenders}
\end{figure}

\subsection{About noise models and codes}
\label{sec-overhead-and-scaling}

The Pauli noise model is commonly used in theoretical work \cite{AGP05,Gottesman16} because of its comparitive simplicity, and in experimental work \cite{Google21,ZLZCLZK20} because of its decent accuracy in modelling current quantum computers. Additionally, \cite{WWDHOE22} found that all non-Markovian noise can be transformed into Pauli noise by randomized compiling, further justifying the use of this simplified model. Beyond this, we believe that our results in Section~\ref{sec-coding-thm} can be expected to hold more generally for different types of i.i.d. noise and even more complicated noise models, in part due to our worst-case assumption on the syndrome state.



While we place focus on the concatenated 7-qubit Steane code, we expect that other codes may also be used for the implementation of the coding scheme, and we briefly give some examples and outline strategies in the remainder of this section.

For example, recent work \cite{YK22} proposes a concatenated code with quantum Hamming codes where the code size and distance increases at every level. This proposal can achieve a constant space overhead, while many desireable properties of quantum Hamming codes are preserved, and many properties are the same as for the concatenation of 7-qubit Steane codes. For example, the codes have the same transversal gate set and similar scaling of the error threshold. We therefore conjecture that similar theorems to the results in \cite{CMH20} and our work in Section~\ref{sec-coding-thm} also hold for encodings in this code. The level of concatenation would still have to be selected according to Eq.~\eqref{eq-level-choice}, but for high enough concatenation levels, the fault-tolerant implementation would require $32$ qubits per logical qubit instead of $7^l$. While it is clear that the results in \cite{CMH20,BCMH22},  Theorem~\ref{thm-final-coding-thm} are asymptotic in nature, the feasibility of implementing a fault-tolerant coding scheme (or a finite-blocklength version) would be greatly increased by such a reduction in overhead.

Furthermore, there exist various other families of error correction codes that are candidates for long-term fault-tolerance. In particular, much attention has been paid to topological codes. While these codes have similarities to the 7-qubit Steane code, there are also some key differences. In particular, the block-like structure of concatenation of codes limits the spread of error during decoding and encoding interfaces, which is a key feature of our construction. There exist proposals for concatenation schemes for topological codes \cite{BFBD11,CT16,LKH22}, which may provide a similar structure and thereby allow us to extend our results.

Quantum LDPC codes \cite{Gottesman14,FGL18,BE21} may also be candidates for generalizing our results; for example, in the case of quantum expander codes, these codes have constant space overhead \cite{Gottesman14}, efficient decoders \cite{FGL18} and utilize a concatenated structure for state preparation.

For some more specific channel families, it may also be possible to make assumptions on the coding scheme which limit the spread of errors. Since the correlations in the channel originate in the encoding circuit, we expect that our analysis may be simplified and that the achievable rates may be improved if we could construct a coding scheme that has low circuit depth or very limited connectivity in the encoder circuit. 

\chapter{Towards small-scale fault-tolerance on trapped ion platforms}\label{chapter-aws}


\section{Introduction}


Despite recent advances in scale and error-resistance of quantum hardware, it remains an experimental and theoretical challenge to protect computations against noise and decoherence. A protective implementation based on techniques from quantum error correction \cite{KLZ98,AB99,Kitaev03,AGP05} can serve to enhance the time span in which a quantum system remains stable and useful for computation. This is especially  important because the noise in quantum systems is chiefly believed to occur as a result of the time-span it takes to execute quantum gates, which is likely to play a role beyond near-term devices \cite{Preskill18}.

Quantum error-correcting codes typically work by encoding the quantum information into a larger quantum system, known as a code space or a code word. Based on this encoding, part of the system can be analyzed in order to detect and correct errors. These codes have been developed and optimized to provide a high level of protection against errors while maintaining a relatively low overhead in terms of the number of resources required.

While many candidates for fault-tolerant quantum error correcting codes exist, the performance and practical relevance of a code depends on several factors, such as the type and strength of the noise affecting the qubits, the resources available for implementing the code, and the specific requirements of the quantum computation being performed. In particular, the limits of current quantum hardware emphasize a need for quantum error correction at small scales, with low qubit numbers and/or limited connectivity. Because of this, the list of near-term attractive error-correcting codes is comprised of the 5-qubit code, the 7-qubit Steane code, the 9-qubit Shor code and small topological codes, all of which can correct up to one single qubit error. Fault-tolerant implementations in these codes often need additional qubits acting as auxilliary systems and classical control. In \cite{Google21}, experimental demonstration of applying the logical gates in a 5-qubit surface code was performed and supplemented with classical simulations. Crucially, their work is focused on the application of gates to states that are already encoded, i.e. faultlessly prepared in the code space, disregarding the error involved in creating the encoded state.
For practical purposes, this encoding should also be fault-tolerant.
Like many state-of-the-art experiments on real devices, the recent experiments performed in \cite{Google21,HRO+17,AAA+22} do not, in fact, employ error correction in the spirit of on-the-fly, real-time correction. Rather, they often implement methods of error detection in order to gauge the probability of logical error and emphasize the need for error correction.
This is due to the fact that error detection is much more feasible on current hardware.
It often needs a lower number of qubits, and - unlike error correction - does not require the implementation of conditional measurements or re-initialization of the auxilliary systems. Both of these operations are not available on any current cloud-based quantum platform. 

Error detection itself can also be used in order protect a computation against faults. In a circuit implemented in an error detection code, if an error is detected during a run of the experiment, instead of providing a method for correcting it, the run should be aborted and restarted. If we can detect all occasions where one single error (i.e. one gate is corrupted during the whole computation) occurs, we are effectively post-selecting the runs without error and a fraction of the runs with more than one error (which we suppose to be unlikely enough).

One such error detection code with a low qubit number is the [[4,2,2]] Bacon-Shor code proposed in \cite{VGW96,GBP97} and highlighted in \cite{Gottesman16}. This code encodes logical two-qubit states in terms of four physical qubits, using one auxilliary system in the process. Notably, this code allows for a low-depth, comparatively simple preparation of an encoded maximally entangled state, which is an important resource state for computations and communication protocols. Due to the low qubit number, this code is a prime candidate for near-term demonstrations of the principles of fault-tolerance. The preparation of  some selected logical two-qubit states in this code has been demonstrated for superconducting qubits in \cite{TCCCG17} and for trapped ion qubits in \cite{LGLFDBM17}. 

With recent leaps in technology and precision, it has also become increasingly common for the builders of quantum computers to offer remote access to their devices, which allows users to submit a description of a quantum circuit and obtain the corresponding experimental data via a cloud-computing based infrastructure. This access allows anyone to investigate whether methods from fault-tolerance can improve computations on current devices. Using remote quantum computing setups for superconducting qubits, the [[4,2,2]] error detection code has been shown to lead to improved quantum state lifetimes and lower errors \cite{OPH+16,Vuillot18,HF19,HLS18,WWJdRM18}.%

In this work, we use this feature to test the limits and possibilities of current trapped ion quantum computers. In particular, we investigate whether the [[4,2,2]] code can protect quantum information on such devices.
Due to the inherent characteristics of trapped ion systems, it becomes easier to perform two-qubit gates, which are traditionally difficult on super-conducting devices, manifesting in their limited connectivity. However, it is possible that this comparative easiness of manipulating multiple qubits might also increase the likelihood of correlated errors, against which the proposed methods for fault-tolerance commonly do not protect. In fact, \cite{HF19} shows that correlated errors could in principle render an encoding of any computation in the [[4,2,2]] code useless, as it may perform \emph{worse} than the unencoded computation if the correlated noise is sufficiently prevalent.

Here, we will say that the encoding of a circuit in the [[4,2,2]] code \emph{demonstrates fault-tolerance} if the encoding performs better than the unencoded version when both circuits are performed on the same device under the same conditions. This is different from our notion of fault-tolerance as introduced in Section~\ref{sec-ft-intro}, which compares a circuit which is affected by noise with an ideal, noiseless counterpart. While devices are moving towards physical error rates that would lie below a threshold for true fault-tolerance, they do not yet have sufficiently low errors. It is therefore common \cite{Google21,HLS18} and perhaps more promising to instead compare implementations of the same circuit in different error correcting codes. Then, we may observe whether they demonstrate improvements in logical error rates (consistent with the theory of fault-tolerance), even if both implementations are not truly fault-tolerant. In the most basic such test, we can compare an implementation without any extra protection ("in a [[2,2,1]] code") to an implementation with some degree of protection, e.g. the [[4,2,2]] code, which is what we compare here.

For one of the quantum devices we investigated, the 11-qubit IonQ Harmony, our observations seem to demonstrate the benefit of protecting a computation by implementing it in an error detection code and are consistent with the theory of quantum fault-tolerance. Our observations suggest that this trapped ion quantum computer does not suffer from excessively highly correlated noise. However, although the implementation in the [[4,2,2]] code improves the life-time of quantum states in current devices, we clearly observe that the computation as a whole - unencoded or encoded - is not fault-tolerant for longer circuit depths.

For the newer trapped ion quantum device, the 21-qubit IonQ Aria, we did not observe a conclusive improvement for circuits implemented in the [[4,2,2]] code. This is surprising because this device has much lower gate error rates, which we expect to be below the threshold. This suggests that our assumed error model does not capture the noise affecting this device, and that there may indeed be highly correlated errors occuring.

First, we give a proof of a threshold theorem for demonstrating the principles of fault-tolerance under an i.i.d. Pauli error model for any error detection code that eliminates all single-qubit gate errors in Section~\ref{sec-error-detecting-threshold-thm}. Then, we give a more detailed description of the [[4,2,2]] error detection code and its implementation in terms of the natural gate set for trapped ion quantum devices in Section~\ref{sec-gottesman-code}. In Section~\ref{sec-gottesman-results}, we compare two-qubit circuits without error detection and two-qubit circuits implemented in the [[4,2,2]] error detection code with respect to their performance on IonQ's quantum computing platforms.



\section{A threshold theorem for demonstrations of fault-tolerance}
\label{sec-error-detecting-threshold-thm}

A quantum circuit is a CPTP map resulting from a concatenation of elementary quantum gates in a gate set as described in Section~\ref{sec-intro-noise-and-errors}. In this work, we will consider two gate sets, the more common gate set from \cite{Gottesman16} and the gate set that is native to trapped ion computers.
The model for errors in a quantum computer that we will use here is the i.i.d. Pauli noise model, which is also considered in \cite{AGP05,Gottesman16,CMH20,BCMH22} and introduced in Section~\ref{sec-intro-noise-and-errors}. \cite{ZLZCLZK20} uses the i.i.d. Pauli error model for an trapped ion quantum computer and suggests that it is justified by comparing their experimental data to simulations.


We will mostly consider circuits $\Gamma:\mathbbm{C}^{2^B}\rightarrow \mathbbm{C}^{2^B}$ with classical input and output where $B$ refers to the number of classical bits the circuit encodes. We denote by $\Gamma_{\mathcal{C}}:\mathbbm{C}^{2^B}\rightarrow \mathbbm{C}^{2^{B}}$ the simulation of $\Gamma$ in a quantum error detection code $\mathcal{C}$ (including the classical post-processing). We say that a \emph{logical error} occurs when the output of our noisy circuits ($[\Gamma]_{\mathcal{F}}$ or $[\Gamma_{\mathcal{C}}]_{\mathcal{F}}$) which are subject to a noise model $\mathcal{F}$ does not match the output that we would obtain from a noiseless version of $\Gamma$.

A quantum error detection code is associated to a specific list of errors which it can reliably detect. This is implemented by a check -  in the case of the [[4,2,2]] code, a check of the global parity - which detects this list of errors with certainty. Runs where a detection occurs are then discarded, and statistics are formed from the data after postselecting runs where the detection event did not occur. A useful error detection code should not abort with certainty, i.e. the probability to survive postselection should not be 0, and the probability of an error outside the list leading to a logical error should be sufficiently low.

Assuming a fault model with restricted locality, like the i.i.d. Pauli model, means that fault patterns where one location has an error occur are the most likely, and more errors become increasingly unlikely. Therefore, in this error model, we want the postselection to detect and discard all runs with such fault patterns which would lead to a logical error. Then, for an encoding in which the most likely set of logical errors are eliminated by the postselection, we can hope to demonstrate the principles of fault-tolerance. 

For longer circuits, the benefit of postselection is lessened, as it becomes more likely that fault patterns with more faults appear. However, postselection is often utilized to implement the preparation of auxilliary qubits as a subroutine in larger fault-tolerant architechtures, see for example \cite{PR13,ZLB18}, and demonstrations of fault-tolerance with an error detecting code could (at the very least) be regarded as demonstrations of fault-tolerant preparations of auxilliary qubits.

We compare the noisy circuit's output distribution to the output distribution of the corresponding noiseless state preparation circuit, using use total variation distance as a measure of error.

\begin{definition}
    Let $p_X$ and $q_X$ be probability distributions over a set $\{x\}_{0,1,...,D-1}$ with $D$ elements. Then, the total variation distance (TV-distance) of $p_X$ and $q_X$ is given by
    \[TV_{dist} (p_X,q_X)= \frac{1}{2} \sum_{x=0}^{D-1} |p_X(x)-q_X(x)|.\]
\end{definition}

Then, as proposed in \cite{Gottesman16}, we define a criterion for demonstrating fault-tolerance with an error detection code:

\begin{definition}[Fault-tolerance demonstration criterion] \label{def-ft-criterion}
Let $\mathcal{C}$ be a quantum error detection code, and let $\mathcal{F}$ denote a noise model. Let $\Gamma:\mathbbm{C}^{2^B}\rightarrow \mathbbm{C}^{2^B}$ be a quantum circuit which prepares the quantum state ${\ket{\phi}}$ and outputs the probability distribution $P_{\ket{\phi}}$ resulting from ${\ket{\phi}}$ by Born's rule.

Let $P_{\ket{\phi}}^{unenc}$ be the probability distribution that is output by $[\Gamma]_{\mathcal{F}}$, and let $P_{\ket{\phi}}^{enc}$ be the probability distribution that is output by $[\Gamma_{\mathcal{C}}]_{\mathcal{F}}$.

Then, we call $\Gamma_{\mathcal{C}}$ an implementation of $\Gamma$ that demonstrates fault-tolerance if 
\[TV_{dist}(P_{\ket{\phi}}^{unenc},P_{\ket{\phi}}) > TV_{dist}(P_{\ket{\phi}}^{enc},P_{\ket{\phi}}).\]
\end{definition}

For any code where single gate errors cannot lead to a logical error, this can be reformulated as a threshold theorem which determines the probability of physical single-qubit error below which the code should demonstrate fault-tolerance:

\begin{lemma}[Threshold for demonstrating fault-tolerance]
\label{thm-thresholdcriterion}
Let $\mathcal{C}$ be a quantum error detection code, and let $\mathcal{F}(p)$ denote the i.i.d. Pauli noise model with $p\in[0,1]$.
Let $\Gamma:\mathbbm{C}^{2^B}\rightarrow \mathbbm{C}^{2^B}$ be a quantum circuit with $L_u$ locations, and let $\Gamma_{\mathcal{C}}:\mathbbm{C}^{2^B}\rightarrow \mathbbm{C}^{2^{B}}$ be the simulation of $\Gamma$, where $\Gamma_{\mathcal{C}}$ has $L_e$ locations. Then, for any $p<p_{th}$ with
\[  p_{th}  = \frac{L_u}{ L_e(L_e-1) + L_uL_e   +  L_u^2 } , \]
the implementation in the code demonstrates fault-tolerance according to the criterion in Definition~\ref{def-ft-criterion}.
\end{lemma}

\begin{proof}
We call the event of surviving the postselection "acceptance". The probability of the encoded circuit with postselection to be subject to a logical error is given by the conditional probability \begin{equation*}
    \begin{split}
       & \Prob [\text{Logical error in $[\Gamma_{\mathcal{C}}]_{\mathcal{F}(p)}$}|\text{acceptance}] \\&= \frac{\Prob[\text{Logical error in $[\Gamma_{\mathcal{C}}]_{\mathcal{F}(p)}$}   \cap \text{acceptance}]}{\Prob[\text{acceptance}]} .
    \end{split}
\end{equation*}
If this probability is strictly smaller than the probability of failure in the unencoded circuit, i.e. if
\[ \frac{\Prob[\text{Logical error in $[\Gamma_{\mathcal{C}}]_{\mathcal{F}(p)}$}   \cap \text{acceptance}]}{\Prob[\text{acceptance}]} < \Prob [\text{Logical error in }[\Gamma]_{\mathcal{F}(p)} ] , \]
then the encoded simulation outperforms the unencoded version, and we may say that the encoding demonstrates fault-tolerance in accordance with the criterion from \ref{def-ft-criterion}.

If there exists an upper bound $p_L\geq \Prob[\text{Logical error in $[\Gamma_{\mathcal{C}}]_{\mathcal{F}(p)}$}   \cap \text{acceptance}] $, and lower bounds $p_S \leq \Prob[\text{acceptance}] $ and $p_U \leq \Prob [\text{Logical error in }[\Gamma]_{\mathcal{F}(p)} ] $, then $\frac{\Prob[\text{Logical error in $[\Gamma_{\mathcal{C}}]_{\mathcal{F}(p)}$}   \cap \text{acceptance}]}{\Prob[\text{acceptance}]}\leq  \frac{p_L}{p_S}$, and the proposed code demonstrates fault-tolerance if \begin{equation} \label{eq-thresholdcondition} \frac{p_L}{p_S} < p_U.\end{equation}
We now proceed to find upper and lower bounds on these probabilities.

For any circuit $\Gamma':\mathbbm{C}^{2^B}\rightarrow \mathbbm{C}^{2^B}$ with $L$ locations, and an error occuring at each location with a probability $p$, the overall probability that errors occur is bounded by the following combinatorial expression, due to the union bound:
\begin{align*}
    \Prob [\text{Logical error in }[\Gamma']_{\mathcal{F}(p)} ] &\leq \Prob [\text{At least one error occurs in }[\Gamma']_{\mathcal{F}(p)} ] \\&\leq\sum_{k=1}^{L} \binom{L}{k} p^k (1-p)^{L-k} .
\end{align*}
This is because the probability of $k$ locations having an error, and all other $L-k$ locations having no error, is given by $p^k (1-p)^{L-k}$, and there are $\binom{L}{k}$ possible fault patterns where exactly $k$ faults occur in the circuit.

All terms where exactly one location is faulty occurs are postselected or do not lead to a logical error in the encoded circuit, and the probability of errors and acceptance is thus upper bounded by
\begin{align*}
&\Prob[\text{Logical error in $[\Gamma_{\mathcal{C}}]_{\mathcal{F}(p)}$}   \cap \text{acceptance}] \\&  \leq  \sum_{k=2}^{L_e} \binom{L_e}{k} p^k (1-p)^{L_e-k}\\& 
= (1-p)^{L_e}  \sum_{k=2}^{L_e} \binom{L_e}{k} (\frac{p}{1-p})^k \\&  
=(1-p)^{L_e} (   (1+\frac{p}{1-p})^{L_e} -L_e \frac{p}{1-p}-1) \\&   =  (1-p)^{L_e} (1+\frac{p}{1-p})^{L_e} - L_e p (1-p)^{L_e-1}- (1-p)^{L_e}\\&  = (1-p+p)^{L_e} - L_e p (1-p)^{L_e-1}- (1-p)^{L_e} \\&  \leq 1- L_ep (1- (L_e-1)p) - (1-L_ep)
\\& = L_e(L_e-1)p^2 :=p_L.
\end{align*}
The first inequality is an upper bound because there may be additional error combinations (besides the single-error terms) which would not survive postselection and therefore not be accepted. 

Postselection can only occur if at least one error happened during the circuit, and it can not occur if there is no error in the circuit. Therefore, the probability of acceptance, i.e. surviving postselection, is bounded by
\begin{equation*} \begin{split}\Prob[\text{acceptance}] &\geq \Prob[\text{no error occurs}]\\& \geq (1-p)^{L_e} = 1 + \sum_{k=1}^{L_e} \binom{L_e}{k} (-p)^k\\& \geq 1-L_ep:=p_S.\end{split}\end{equation*}
To show the last inequality, we group together $k$ and $k+1$ for all even $k$, and show that the summand corresponding to $k$ (appearing with positive sign) is always larger than the summand corresponding to $k+1$ (appearing with negative sign in the sum): $\binom{N}{k} p^k \geq \binom{N}{k+1} p^{k+1} \iff p \leq \frac{\binom{L_e}{k}}{\binom{L_e}{k+1}}= \frac{k+1}{N-k} $. This is true for all even $k$ if it is true for $k=2$, because the expression increases for increasing $k$, which means that it is true for  $p<\frac{3}{L_e-3}$. (Note that this threshold is only reasonable and/or non-trivial for $L_e>3$, which is true by default for this code, and most codes.)

Similarly, we have the following lower bound for the unencoded circuit:
\begin{align*}
  \Prob [\text{Logical error in }[\Gamma]_{\mathcal{F}(p)} ]&\geq \Prob [\text{Exactly one error occurs in }[\Gamma]_{\mathcal{F}(p)} ] \\&= L_u p (1-p)^{L_u-1} \\&\geq L_u p (1-p)^{L_u}\\& = L_u p (\sum_{k=1}^{L_u} \binom{L_u}{k} (-p)^k+1) \\& = L_u p +L_u p \sum_{k=1}^{L_u} \binom{L_u}{k} (-p)^k  :=p_U.
\end{align*}
In total, we obtain the following chain of relations:
\begin{align*}
  & \frac{  L_e(L_e-1) p^2  }{(1-L_e p)} < L_u p +L_u p \sum_{k=1}^{L_u} \binom{L_u}{k} (-p)^k \hspace{1cm}\Leftarrow \\  & L_e(L_e-1) p^2  < L_u p +L_u p \sum_{k=1}^{L_u} \binom{L_u}{k} (-p)^k - L_uL_e p^2 -L_uL_e p^2 \sum_{k=1}^{L_u} \binom{L_u}{k} (-p)^k.
  \end{align*}
  Then, we collect the terms that are associated to $p^2$:
  \begin{align*}
       &\Big( L_e(L_e-1) + L_uL_e   +  L_u^2  \Big) p^2 \\& < L_u p +L_u p \sum_{k=2}^{L_u} \binom{L_u}{k} (-p)^k -L_uL_e p^2 \sum_{k=1}^{L_u} \binom{L_u}{k} (-p)^k  ,
    \end{align*}
and divide both sides by $p$:
    \begin{align*}
  &  \Big( L_e(L_e-1) + L_uL_e   +  L_u^2  \Big) p  \\& < L_u  +L_u  \sum_{k=2}^{L_u} \binom{L_u}{k} (-p)^k +L_uL_e  \sum_{k=1}^{L_u} \binom{L_u}{k} (-p)^{k+1}  .
\end{align*}
We note that the sum of the second and third term is always positive: \begin{align*}    &   L_u  \sum_{k=2}^{L_u} \binom{L_u}{k} (-p)^k +L_uL_e  \sum_{k=1}^{L_u} \binom{L_u}{k} (-p)^{k+1} \\ &  =L_uL_e p (1-(1-p)^{L_u}) +L_u (-1 + (1-p)^{L_u}+L_up)>0.
   \end{align*}
   Therefore, we choose the threshold to be 
   \[  p_{th}  = \frac{L_u}{ L_e(L_e-1) + L_uL_e   +  L_u^2 }  .\]
Then, for all $p\leq p_{th}$, the condition in Eq.~\eqref{eq-thresholdcondition} is fulfilled, and therefore the encoding in the error detection code combined with postselection performs better than the unencoded original circuit.
\end{proof}

\section{The [[4,2,2]] error detection code}
\label{sec-gottesman-code}

The [[4,2,2]] error detection code as introduced in Section~\ref{sec-error-correction-codes-examples} is the smallest known quantum error detection code that can detect any single-qubit error. It was proposed in \cite{VGW96,GBP97} and highlighted in \cite{Gottesman16} for fault-tolerance experiments within the reach of current and near-term quantum hardware. This code was originally designed for restricted connectivity layouts, where 5 qubits are arranged in a ring.

One particularity of this code is that it encodes two-qubit states directly.  Within most quantum error correction codes, a preparation  of some two-qubit state would be executed by a preparation of two one-qubit states, and an application of fault-tolerant quantum gates. However, the [[4,2,2]] code accomplishes this more directly, i.e. by applying comparitively few gates, and with only one additional physical auxilliary qubit.

This is especially relevant in the case of the maximally entangled state of two qubits, which can be prepared in the code space by a low-depth 4-qubit circuit. This state acts as a resource in many quantum computing and quantum communication protocols, such superdense coding \cite{BW92}, teleportation \cite{BBCJPW93} and gate teleportation \cite{GC99}. In our work in Chapter~\ref{chapter-fteacap}, for example, the preparation of maximally entangled states within a code space in order to utilize them for entanglement-assisted communication plays an important role.

Beyond its use for devices with a low qubit number, the [[4,2,2]] code can also be concatenated with other codes in order to provide further protection \cite{BFBD11,CT16}.

Recall that the code states of the [[4,2,2]] code are given by:
\[\ket{\overline{00}}= \frac{1}{\sqrt{2}} (\ket{0000}+\ket{1111}),\]
\[\ket{\overline{01}}= \frac{1}{\sqrt{2}} (\ket{1100}+\ket{0011}),\]
\[\ket{\overline{10}}= \frac{1}{\sqrt{2}} (\ket{1010}+\ket{0101}),\]
\[\ket{\overline{11}}= \frac{1}{\sqrt{2}} (\ket{0110}+\ket{1001}).\]
We will therefore consider circuits $\Gamma$ acting on 2 qubits, which output a probability distribution obtained from computational basis measurements, and $\Gamma_{\mathcal{C}}$ (where $\mathcal{C}$ denotes the [[4,2,2]] code) acting on 4 physical qubits (+1 auxilliary qubit), which outputs the postselected probability distribution. 

The preparation of these states has been demonstrated in \cite{TCCCG17,LGLFDBM17}, and circuits with higher depths have been implemented in this code on superconducting qubit platforms in \cite{OPH+16,Vuillot18,HF19,HLS18,WWJdRM18}. These experiments generally find that the implementation of a circuit in this code does indeed improve the accuracy of the computation for short circuits, serving as an indication that the techniques from fault-tolerance may eventually become useful. To our knowledge, this code has not been studied for trapped ion quantum computers beyond the analysis of the preparation of logical computational basis states in \cite{LGLFDBM17}, perhaps due to the focus on connectivity restrictions that do not apply to trapped ion qubits. Here, we attempt to study if an implementation in the [[4,2,2]] code can improve the life-time of the prepared states on a trapped ion quantum computer.

\subsection{Thresholds in the proposed gate set of \cite{Gottesman16}}

In order to demonstrate fault-tolerance for implementations in this code, we need to be able to both prepare states in this subspace and to perform gates in such a way that every single qubit error that leads to a logical error is detected. In \cite{Gottesman16}, Gottesman proposed state preparation circuits for three distinct two-qubit states in this code: $\ket{\overline{00}}$, $\ket{\overline{0+}}=\frac{1}{\sqrt{2}} (\ket{\overline{00}}+\ket{\overline{01}})$ and $\ket{\overline{\phi_+}}=\frac{1}{\sqrt{2}} (\ket{\overline{00}}+\ket{\overline{11}})$.
The encoded preparation circuits are given in Table~\ref{tab-ft-preps}.

The set of logical gates in this code that preserve the property that single-qubit errors cannot lead to a logical error is limited. Therefore, not all states can be prepared and interconverted within the code. For example, the logical Hadamard gate would in principle map between the states $\ket{\overline{00}}$ and $\ket{\overline{0+}}$. However, this gate is not known to have the desired property, and this is why a different circuit is proposed to prepare $\ket{\overline{0+}}$. In total, the set of states that we can prepare in this construction is given by the three states listed above, and all states that are obtained by performing one of the gates from the logical gate set on one of them, where the known such gates are listed in Eq.~\eqref{eq-ft-gates}.

\begin{table}[htbp] 
  \centering
  \caption{State preparation circuits for [[4,2,2]] Bacon-Shor error detection code as proposed in \cite{Gottesman16}.}
    \begin{tabular}{!{\color{gray}\vrule}l!{\color{gray}\vrule}l!{\color{gray}\vrule}}
       \arrayrulecolor{gray}\hline
        Logical state  & State preparation circuit \\ 
        \hline \hline
         $\ket{\overline{00}}$ &\hspace{1cm} \begin{minipage}{.5\textwidth} \vspace{0.5cm}
         \begin{tikzcd}
 \gate[style={draw, shape=semicircle, minimum size=0.5cm, xscale=0.7,yscale=0.7, inner sep=0pt, outer sep=0pt, shape border rotate=90}]{} &\qw &\qw& \targ{}& \qw & \ctrl{4} & \gate[wires=4,nwires=4][1cm]{U} & \gate[style={draw, shape=semicircle, minimum size=0.5cm,  xscale=0.7,yscale=0.7, inner sep=0pt, outer sep=0pt, shape border rotate=270}]{} \\
  \gate[style={draw, shape=semicircle, minimum size=0.5cm, xscale=0.7,yscale=0.7, inner sep=0pt, outer sep=0pt, shape border rotate=90}]{}&\gate{H} &\ctrl{1}& \ctrl{-1}& \qw & \qw & \qw & \gate[style={draw, shape=semicircle, minimum size=0.5cm,  xscale=0.7,yscale=0.7, inner sep=0pt, outer sep=0pt, shape border rotate=270}]{}\\
 \gate[style={draw, shape=semicircle, minimum size=0.5cm, xscale=0.7,yscale=0.7, inner sep=0pt, outer sep=0pt, shape border rotate=90}]{} &\qw& \targ{}  & \ctrl{1}  & \qw & \qw& \qw & \gate[style={draw, shape=semicircle, minimum size=0.5cm,  xscale=0.7,yscale=0.7, inner sep=0pt, outer sep=0pt, shape border rotate=270}]{}\\
   \gate[style={draw, shape=semicircle, minimum size=0.5cm, xscale=0.7,yscale=0.7, inner sep=0pt, outer sep=0pt, shape border rotate=90}]{}& \qw  &\qw& \targ{} & \ctrl{1} & \qw& \qw & \gate[style={draw, shape=semicircle, minimum size=0.5cm,  xscale=0.7,yscale=0.7, inner sep=0pt, outer sep=0pt, shape border rotate=270}]{} \\
     \gate[style={draw, shape=semicircle, minimum size=0.5cm, xscale=0.7,yscale=0.7, inner sep=0pt, outer sep=0pt, shape border rotate=90}]{}& \qw & \qw & \qw &  \targ{} & \targ{} &\gate[style={draw, shape=semicircle, minimum size=0.5cm,  xscale=0.7,yscale=0.7, inner sep=0pt, outer sep=0pt, shape border rotate=270}]{} \\
\end{tikzcd}
         \end{minipage}
         \\ \hline
         $\ket{\overline{0+}}$ & \hspace{1cm}
         \begin{minipage}{.5\textwidth}  \vspace{0.5cm}
\begin{tikzcd}
 \gate[style={draw, shape=semicircle, minimum size=0.5cm, xscale=0.7,yscale=0.7, inner sep=0pt, outer sep=0pt, shape border rotate=90}]{}&\gate{H} & \ctrl{1}  & \gate[wires=4,nwires=4][1cm]{U}  &\gate[style={draw, shape=semicircle, minimum size=0.5cm,  xscale=0.7,yscale=0.7, inner sep=0pt, outer sep=0pt, shape border rotate=270}]{}\\
 \gate[style={draw, shape=semicircle, minimum size=0.5cm, xscale=0.7,yscale=0.7, inner sep=0pt, outer sep=0pt, shape border rotate=90}]{} & \qw  & \targ{}  & \qw  &\gate[style={draw, shape=semicircle, minimum size=0.5cm,  xscale=0.7,yscale=0.7, inner sep=0pt, outer sep=0pt, shape border rotate=270}]{}\\
  \gate[style={draw, shape=semicircle, minimum size=0.5cm, xscale=0.7,yscale=0.7, inner sep=0pt, outer sep=0pt, shape border rotate=90}]{}&\gate{H} & \ctrl{1} & \qw & \gate[style={draw, shape=semicircle, minimum size=0.5cm,  xscale=0.7,yscale=0.7, inner sep=0pt, outer sep=0pt, shape border rotate=270}]{}\\
   \gate[style={draw, shape=semicircle, minimum size=0.5cm, xscale=0.7,yscale=0.7, inner sep=0pt, outer sep=0pt, shape border rotate=90}]{}& \qw  & \targ{} & \qw & \gate[style={draw, shape=semicircle, minimum size=0.5cm,  xscale=0.7,yscale=0.7, inner sep=0pt, outer sep=0pt, shape border rotate=270}]{}\\
\end{tikzcd}
    \end{minipage}
    \\\hline
          $\ket{\overline{\phi_+}}$ & \hspace{1cm}
          \begin{minipage}{.75\textwidth} \vspace{0.5cm}
 \begin{tikzcd}
 \gate[style={draw, shape=semicircle, minimum size=0.5cm, xscale=0.7,yscale=0.7, inner sep=0pt, outer sep=0pt, shape border rotate=90}]{}&\gate{H} &\qw& \ctrl{3}  & \gate[wires=4,nwires=4][1cm]{U} &\gate[style={draw, shape=semicircle, minimum size=0.5cm,  xscale=0.7,yscale=0.7, inner sep=0pt, outer sep=0pt, shape border rotate=270}]{} \\
 \gate[style={draw, shape=semicircle, minimum size=0.5cm, xscale=0.7,yscale=0.7, inner sep=0pt, outer sep=0pt, shape border rotate=90}]{} &\qw& \targ{}  & \qw  & \qw & \gate[style={draw, shape=semicircle, minimum size=0.5cm,  xscale=0.7,yscale=0.7, inner sep=0pt, outer sep=0pt, shape border rotate=270}]{}\\
  \gate[style={draw, shape=semicircle, minimum size=0.5cm, xscale=0.7,yscale=0.7, inner sep=0pt, outer sep=0pt, shape border rotate=90}]{}&\gate{H} &\ctrl{-1}& \qw & \qw&\gate[style={draw, shape=semicircle, minimum size=0.5cm,  xscale=0.7,yscale=0.7, inner sep=0pt, outer sep=0pt, shape border rotate=270}]{}\\
   \gate[style={draw, shape=semicircle, minimum size=0.5cm, xscale=0.7,yscale=0.7, inner sep=0pt, outer sep=0pt, shape border rotate=90}]{}& \qw  &\qw& \targ{} &  \qw&\gate[style={draw, shape=semicircle, minimum size=0.5cm,  xscale=0.7,yscale=0.7, inner sep=0pt, outer sep=0pt, shape border rotate=270}]{} \\
\end{tikzcd}
    \end{minipage}\\ \hline
    \end{tabular}
    \label{tab-ft-preps}
\end{table}



\begin{lemma}
Let $\mathcal{C}$ denote the [[4,2,2]] code and let $\mathcal{F}(p)$ denote the i.i.d. Pauli noise model with $p\in [0,1]$. Let $\Gamma:\mathbbm{C}^4\rightarrow \mathbbm{C}^4$ be the circuit which prepares $\ket{\phi_+}$. Then, no single-qubit gate error in $[\Gamma_{\mathcal{C}}]_{\mathcal{F}(p)}$ causes a logical error.
\end{lemma}

\begin{proof}
This claim can be verified by analyzing all possible cases of a single-qubit gate error. For illustration, we will give the details of this process here.

The maximally entangled state of two logical qubits is encoded as $\ket{\overline{\phi_+}}=\frac{1}{2}(\ket{0000}+\ket{0110}+\ket{1001}+\ket{1111})$ on the four physical qubits. This has even parity (i.e. will pass through the parity check for the postselection) and will be decoded into $\ket{00}$ with 50\% probability and $\ket{11}$ with 50\% probability, and therefore lead to the correct statistics.

Note that errors can spread in the following ways in the gate set used in the proposed circuits:
\begin{align*}   
&H X=ZH \\& SX=YS\\&  [\CNOT_{21},X_1 ] =0\\&  [\CNOT_{12},Z_1 ] =0\\& X_1X_2 \CNOT_{12}=\CNOT_{12} X_1\\& Z_1Z_2 \CNOT_{21}=\CNOT_{21} Z_1
 \end{align*}
 Here, we use $X_i/Z_i$ to indicate an $X/Z$ gate on the $i$-th qubit and $\CNOT_{ij}$ to
denote a $\CNOT$ gate where the $i$-th qubit acts as the control and the $j$-th qubit acts as the target.

We use Q1, Q2, Q3, Q4 to denote the four qubits in our circuit, where we assign the labels from top to bottom in the circuit diagram in Table~\ref{tab-ft-preps}. With this, we can go through all the cases of single-qubit $X,Y$ or $Z$-errors occuring, and show that they are either detected or do not change the decoded state:
 
\begin{itemize}
    \item Because of the connectivity, an $X$-error can only lead to one $X$-error or two $X$-errors, which can only happen on two qubits that are connected by CNOTs.
    \begin{itemize}
         \item One $X$-error after the preparation of Q1 is equivalent to a $Z$-error after the first Hadamard gate, leading to a change in relative phase that does not affect the decoded state. 
         \item An $X$-error on Q1 after the Hadamard gate on Q1 will spread via the $\CNOT$ gate between Q1 and Q4, which does not affect the decoded state.
         \item If an $X$-error occurs on Q4 after the preparation gate on Q4, it commutes with the $\CNOT$ gate on Q1 and Q4. This error changes the (global) parity of the state, and would thus be detected and rejected.
         \item If an $X$-error occurs on Q1 after the $\CNOT$ gate, before the measurement, this would change the global parity and would thus be detected. The same holds true for an $X$-error on Q5 before the measurement.
         \item Due to the circuit construction, the gates applied to Q2 and Q3 are the same, and in the same order, as the gates applied to Q4 and Q1. Thus, single gate errors on Q2 and Q3 give rise to the same cases and patterns, and would thus be detected or not affect the decoded state.
    \end{itemize}
    \item $Z$-errors affect the relative phase, which does not affect the decoding here. However, $Z$-errors are transformed into bit-flip errors by the Hadamard gate $H$. This situation is equivalent to an $X$-error happening after $H$, and thus covered in the above point relating to all possible $X$-errors.
    \item For $Y$-errors, we have $H XZ=ZH Z=ZXH$. For $\CNOT$, depending on the direction, either $X$ or $Z$ commute while the other one "spreads" to both qubits. For example, $\CNOT_{12} X_1Z_1 = X_1 Z_1 X_2 \CNOT_{12} $. If this occurs on Q2 and Q3, the relative phase changes, and the parities do not change (i.e. it is postselected and decoded correctly). If it occurs on Q1, it will spread an $X$-error (only) to Q4 and be covered by the cases before. 
\end{itemize}

Thus, all single-qubit gate errors are detected or do not affect the decoded state.

It is also notable that the same can be said if the circuit is implemented on a 5-qubit ring, as would be necessary for a device that could prepare both $\ket{\overline{\phi_+}}$ and $\ket{\overline{00}}$. Then, the $\CNOT$ gate between Q1 and Q4 would instead be implemented by a $\CNOT$ gate between Q1 and Q5, followed by a $\SWAP$ gate on Q5 and Q4. A $\SWAP$ gate can be implemented by three $\CNOT$ gates, and this circuit can lead to one $X$-error (detected on Q4, irrelevant on Q5) or two $X$-errors, i.e. on Q4 and Q5, which would be detected.
\end{proof}

\begin{corollary}
Let $\mathcal{C}$ denote the [[4,2,2]] code and let $\mathcal{F}(p)$ denote the i.i.d. Pauli noise model with $p\in [0,1]$. Let $\Gamma:\mathbbm{C}^{4}\rightarrow \mathbbm{C}^4$ be a circuit that consists of the proposed state preparations. Like \cite{Gottesman16}, let us assume that qubits are initialized only when they are needed, and do not have wait locations before the first gate is applied to them.
\begin{itemize}
    \item For $\ket{\overline{\phi_+}}$, $\Gamma$ has $L_u=6$ locations, and $\Gamma_{\mathcal{C}}$ has  $L_e=12$ locations, and therefore, $[\Gamma_{\mathcal{C}}]_{\mathcal{F}(p)}$ demonstrates fault-tolerance for $p<  1/40 =0.025$.
    \item For $\ket{\overline{00}}$, the location counts are $L_u=6$ and $L_e=22$, and therefore, the circuit demonstrates fault-tolerance for $p<2/283=0.00706714$.
    \item For $\ket{\overline{0+}}$, the location counts are $L_u=5$ and $L_e=12$, and therefore, the circuit demonstrates fault-tolerance for $p<  5/217=0.0230415$.
\end{itemize}
\end{corollary}

\begin{proof}
We go through the location counting for the preparation of the maximally entangled state, illustrated also in Figure~\ref{fig-maxent-counting}.

The unencoded circuit consists of 2 preparations, 1 single-qubit gate, 1 two-qubit gate and 2 measurements, leading to a total location count of $L_u=6$.

The locations of the encoded circuit are given by 4 preparations, 2 single-qubit gates, 2 two-qubit gates
and 4 measurements, i.e. a total of $L_e=12$ locations.

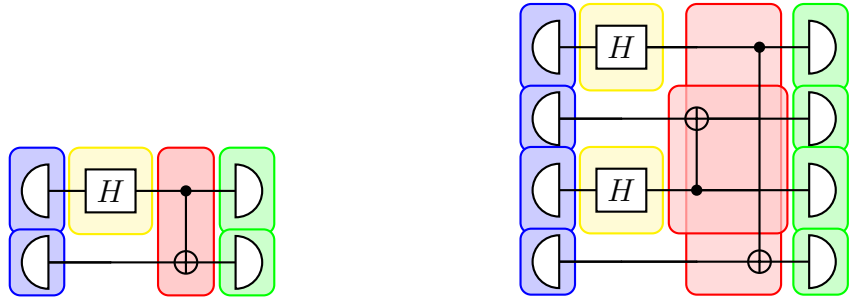
\begin{figure}[htbp]
     \centering
     \begin{subfigure}[b]{0.49\textwidth}
         \centering
         
\begin{tikzcd}
 \gate[style={draw, shape=semicircle, minimum size=0.5cm, xscale=0.7,yscale=0.7, inner sep=0pt, outer sep=0pt, shape border rotate=90}]{}\gategroup[wires=1,steps=1,style={draw=blue,rounded
corners,fill=blue!20, inner
xsep=2pt},background]{} & \gate{H}\gategroup[wires=1,steps=1,style={draw=yellow,rounded corners,fill=yellow!20, inner xsep=2pt},background]{} & \ctrl{1}\gategroup[wires=2,steps=1,style={draw=red,rounded
corners,fill=red!20, inner
xsep=2pt},background]{}  &\gate[style={draw, shape=semicircle, minimum size=0.5cm,  xscale=0.7,yscale=0.7, inner sep=0pt, outer sep=0pt, shape border rotate=270}]{}\gategroup[wires=1,steps=1,style={draw=green,rounded corners,fill=green!20, inner xsep=2pt},background]{}  \\
 \gate[style={draw, shape=semicircle, minimum size=0.5cm, xscale=0.7,yscale=0.7, inner sep=0pt, outer sep=0pt, shape border rotate=90}]{}\gategroup[wires=1,steps=1,style={draw=blue,rounded
corners,fill=blue!20, inner
xsep=2pt},background]{} & \qw  & \targ{} &\gate[style={draw, shape=semicircle, minimum size=0.5cm,  xscale=0.7,yscale=0.7, inner sep=0pt, outer sep=0pt, shape border rotate=270}]{}\gategroup[wires=1,steps=1,style={draw=green,rounded corners,fill=green!20, inner xsep=2pt},background]{}
\end{tikzcd}
         \caption{Preparation of $ \ket{\phi_+}$. The total number of circuit locations is $L_u=6$.}
     \end{subfigure}
     \hfill
     \begin{subfigure}[b]{0.49\textwidth}
         \centering
 \begin{tikzcd}
 \gate[style={draw, shape=semicircle, minimum size=0.5cm, xscale=0.7,yscale=0.7, inner sep=0pt, outer sep=0pt, shape border rotate=90}]{}\gategroup[wires=1,steps=1,style={draw=blue,rounded
corners,fill=blue!20, inner
xsep=2pt},background]{} &\gate{H}\gategroup[wires=1,steps=1,style={draw=yellow,rounded corners,fill=yellow!20, inner xsep=2pt},background]{} &\qw& \ctrl{3}\gategroup[wires=4,steps=1,style={draw=red,rounded
corners,fill=red!20, fill opacity=0.7,inner
xsep=2pt,xscale=1.7,xshift=-0.2cm},background]{}  &\gate[style={draw, shape=semicircle, minimum size=0.5cm,  xscale=0.7,yscale=0.7, inner sep=0pt, outer sep=0pt, shape border rotate=270}]{}\gategroup[wires=1,steps=1,style={draw=green,rounded corners,fill=green!20, inner xsep=2pt},background]{}  \\
 \gate[style={draw, shape=semicircle, minimum size=0.5cm, xscale=0.7,yscale=0.7, inner sep=0pt, outer sep=0pt, shape border rotate=90}]{}\gategroup[wires=1,steps=1,style={draw=blue,rounded
corners,fill=blue!20, inner
xsep=2pt},background]{} &\qw& \targ{}\gategroup[wires=2,steps=2,style={draw=red,rounded
corners,fill=red!20,fill opacity=0.7, inner
xsep=2pt},background]{}   & \qw & \gate[style={draw, shape=semicircle, minimum size=0.5cm,  xscale=0.7,yscale=0.7, inner sep=0pt, outer sep=0pt, shape border rotate=270}]{}\gategroup[wires=1,steps=1,style={draw=green,rounded corners,fill=green!20, inner xsep=2pt},background]{} \\
  \gate[style={draw, shape=semicircle, minimum size=0.5cm, xscale=0.7,yscale=0.7, inner sep=0pt, outer sep=0pt, shape border rotate=90}]{}\gategroup[wires=1,steps=1,style={draw=blue,rounded
corners,fill=blue!20, inner
xsep=2pt},background]{} &\gate{H}\gategroup[wires=1,steps=1,style={draw=yellow,rounded corners,fill=yellow!20, inner xsep=2pt},background]{} &\ctrl{-1}& \qw & \gate[style={draw, shape=semicircle, minimum size=0.5cm,  xscale=0.7,yscale=0.7, inner sep=0pt, outer sep=0pt, shape border rotate=270}]{}\gategroup[wires=1,steps=1,style={draw=green,rounded corners,fill=green!20, inner xsep=2pt},background]{} \\
   \gate[style={draw, shape=semicircle, minimum size=0.5cm, xscale=0.7,yscale=0.7, inner sep=0pt, outer sep=0pt, shape border rotate=90}]{}\gategroup[wires=1,steps=1,style={draw=blue,rounded
corners,fill=blue!20, inner
xsep=2pt},background]{} & \qw    &  \qw&   \targ{}&\gate[style={draw, shape=semicircle, minimum size=0.5cm,  xscale=0.7,yscale=0.7, inner sep=0pt, outer sep=0pt, shape border rotate=270}]{}\gategroup[wires=1,steps=1,style={draw=green,rounded corners,fill=green!20, inner xsep=2pt},background]{}  
\end{tikzcd}
         \caption{Preparation of $\ket{\overline{\phi_+}}$. The total number of circuit locations is $L_e=12$.}
     \end{subfigure}
        \caption{Circuits for preparing $\ket{\phi_+}$ and $\ket{\overline{\phi_+}}$. Preparation gates are highlighted in blue, measurement gates are highlighted in green, two qubit gates are highlighted in red and single-qubit gates are highlighted in yellow.}
        \label{fig-maxent-counting}
\end{figure}

Inserting this into the criterion for demonstrating fault-tolerance \ref{thm-thresholdcriterion}, we get a threshold of 
\[ p<  \frac{6}{36+6\times 12 + 11 \times 12}= \frac{6}{6\times(6+12+22)} =1/40 .\]
The location counts in the other two state preparation circuits, as illustrated in Figure~\ref{fig-state-prep-counts-2} and \ref{fig-state-prep-counts-3}, analogously lead to their respective thresholds.

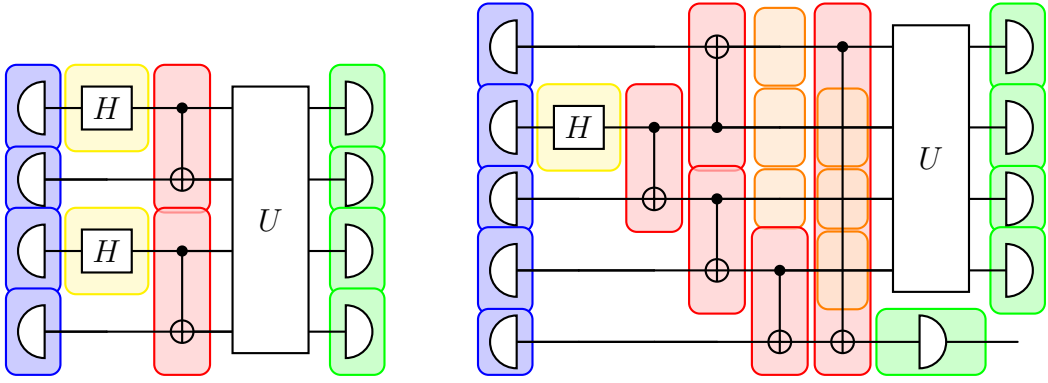
\begin{figure}[htbp]
     \centering
     \begin{subfigure}[b]{0.49\textwidth}
         \centering

\begin{tikzcd}
 \gate[style={draw, shape=semicircle, minimum size=0.5cm, xscale=0.7,yscale=0.7, inner sep=0pt, outer sep=0pt, shape border rotate=90}]{}\gategroup[wires=1,steps=1,style={draw=blue,rounded
corners,fill=blue!20, inner
xsep=2pt},background]{}&\gate{H}\gategroup[wires=1,steps=1,style={draw=yellow,rounded corners,fill=yellow!20, inner xsep=2pt},background]{} & \ctrl{1}\gategroup[wires=2,steps=1,style={draw=red,rounded
corners,fill=red!20, fill opacity=0.7, inner
xsep=2pt},background]{}  & \gate[wires=4,nwires=4][1cm]{U}  &\gate[style={draw, shape=semicircle, minimum size=0.5cm,  xscale=0.7,yscale=0.7, inner sep=0pt, outer sep=0pt, shape border rotate=270}]{}\gategroup[wires=1,steps=1,style={draw=green,rounded corners,fill=green!20, inner xsep=2pt},background]{} \\
 \gate[style={draw, shape=semicircle, minimum size=0.5cm, xscale=0.7,yscale=0.7, inner sep=0pt, outer sep=0pt, shape border rotate=90}]{}\gategroup[wires=1,steps=1,style={draw=blue,rounded
corners,fill=blue!20, inner
xsep=2pt},background]{} & \qw  & \targ{}  & \qw  &\gate[style={draw, shape=semicircle, minimum size=0.5cm,  xscale=0.7,yscale=0.7, inner sep=0pt, outer sep=0pt, shape border rotate=270}]{}\gategroup[wires=1,steps=1,style={draw=green,rounded corners,fill=green!20, inner xsep=2pt},background]{} \\
  \gate[style={draw, shape=semicircle, minimum size=0.5cm, xscale=0.7,yscale=0.7, inner sep=0pt, outer sep=0pt, shape border rotate=90}]{}\gategroup[wires=1,steps=1,style={draw=blue,rounded
corners,fill=blue!20, inner
xsep=2pt},background]{} &\gate{H}\gategroup[wires=1,steps=1,style={draw=yellow,rounded corners,fill=yellow!20, inner xsep=2pt},background]{} & \ctrl{1}\gategroup[wires=2,steps=1,style={draw=red,rounded
corners,fill=red!20, fill opacity=0.7, inner
xsep=2pt},background]{} & \qw & \gate[style={draw, shape=semicircle, minimum size=0.5cm,  xscale=0.7,yscale=0.7, inner sep=0pt, outer sep=0pt, shape border rotate=270}]{}\gategroup[wires=1,steps=1,style={draw=green,rounded corners,fill=green!20, inner xsep=2pt},background]{} \\
   \gate[style={draw, shape=semicircle, minimum size=0.5cm, xscale=0.7,yscale=0.7, inner sep=0pt, outer sep=0pt, shape border rotate=90}]{}\gategroup[wires=1,steps=1,style={draw=blue,rounded
corners,fill=blue!20, inner
xsep=2pt},background]{}& \qw  & \targ{} & \qw & \gate[style={draw, shape=semicircle, minimum size=0.5cm,  xscale=0.7,yscale=0.7, inner sep=0pt, outer sep=0pt, shape border rotate=270}]{}\gategroup[wires=1,steps=1,style={draw=green,rounded corners,fill=green!20, inner xsep=2pt},background]{} 
\end{tikzcd}
         \caption{Preparation of $\ket{\overline{0+}}$. This circuit has $L_e=12$ locations, whereas the unencoded circuit has $L_u=5$ locations.}
         \label{fig-state-prep-counts-2}
     \end{subfigure}
     \hfill
     \begin{subfigure}[b]{0.49\textwidth}
         \centering
\begin{tikzcd}
 \gate[style={draw, shape=semicircle, minimum size=0.5cm, xscale=0.7,yscale=0.7, inner sep=0pt, outer sep=0pt, shape border rotate=90}]{}\gategroup[wires=1,steps=1,style={draw=blue,rounded
corners,fill=blue!20, inner
xsep=2pt},background]{} &\qw &\qw& \targ{}\gategroup[wires=2,steps=1,style={draw=red,rounded
corners,fill=red!20, fill opacity=0.7, inner
xsep=2pt},background]{}& \qw\gategroup[wires=1,steps=1,style={draw=orange,rounded
corners,fill=orange!20, xscale=0.9, yscale=0.9, fill opacity=0.7, inner
xsep=2pt},background]{} & \ctrl{4}\gategroup[wires=5,steps=1,style={draw=red,rounded
corners,fill=red!20, fill opacity=0.7, inner
xsep=2pt},background]{} & \gate[wires=4,nwires=4][1cm]{U} & \gate[style={draw, shape=semicircle, minimum size=0.5cm,  xscale=0.7,yscale=0.7, inner sep=0pt, outer sep=0pt, shape border rotate=270}]{}\gategroup[wires=1,steps=1,style={draw=green,rounded corners,fill=green!20, inner xsep=2pt},background]{}  \\
  \gate[style={draw, shape=semicircle, minimum size=0.5cm, xscale=0.7,yscale=0.7, inner sep=0pt, outer sep=0pt, shape border rotate=90}]{}\gategroup[wires=1,steps=1,style={draw=blue,rounded
corners,fill=blue!20, inner
xsep=2pt},background]{} &\gate{H}\gategroup[wires=1,steps=1,style={draw=yellow,rounded corners,fill=yellow!20, inner xsep=2pt},background]{} &\ctrl{1}\gategroup[wires=2,steps=1,style={draw=red,rounded
corners,fill=red!20, fill opacity=0.7, inner
xsep=2pt},background]{}& \ctrl{-1}& \qw\gategroup[wires=1,steps=1,style={draw=orange,rounded
corners,fill=orange!20, xscale=0.9, yscale=0.9, fill opacity=0.7, inner
xsep=2pt},background]{} & \qw\gategroup[wires=1,steps=1,style={draw=orange,rounded
corners,fill=orange!20, xscale=0.9, yscale=0.9, fill opacity=0.7, inner
xsep=2pt},background]{} & \qw & \gate[style={draw, shape=semicircle, minimum size=0.5cm,  xscale=0.7,yscale=0.7, inner sep=0pt, outer sep=0pt, shape border rotate=270}]{}\gategroup[wires=1,steps=1,style={draw=green,rounded corners,fill=green!20, inner xsep=2pt},background]{} \\
 \gate[style={draw, shape=semicircle, minimum size=0.5cm, xscale=0.7,yscale=0.7, inner sep=0pt, outer sep=0pt, shape border rotate=90}]{}\gategroup[wires=1,steps=1,style={draw=blue,rounded
corners,fill=blue!20, inner
xsep=2pt},background]{} &\qw& \targ{}  & \ctrl{1}\gategroup[wires=2,steps=1,style={draw=red,rounded
corners,fill=red!20, fill opacity=0.7, inner
xsep=2pt},background]{}   & \qw\gategroup[wires=1,steps=1,style={draw=orange,rounded
corners,fill=orange!20, xscale=0.9, yscale=0.9, fill opacity=0.7, inner
xsep=2pt},background]{}&\qw\gategroup[wires=1,steps=1,style={draw=orange,rounded
corners,fill=orange!20, xscale=0.9, yscale=0.9, fill opacity=0.7, inner
xsep=2pt},background]{} &\qw& \gate[style={draw, shape=semicircle, minimum size=0.5cm,  xscale=0.7,yscale=0.7, inner sep=0pt, outer sep=0pt, shape border rotate=270}]{}\gategroup[wires=1,steps=1,style={draw=green,rounded corners,fill=green!20, inner xsep=2pt},background]{} \\
   \gate[style={draw, shape=semicircle, minimum size=0.5cm, xscale=0.7,yscale=0.7, inner sep=0pt, outer sep=0pt, shape border rotate=90}]{}\gategroup[wires=1,steps=1,style={draw=blue,rounded
corners,fill=blue!20, inner
xsep=2pt},background]{} & \qw  &\qw& \targ{} & \ctrl{1}\gategroup[wires=2,steps=1,style={draw=red,rounded
corners,fill=red!20, fill opacity=0.7, inner
xsep=2pt},background]{} & \qw\gategroup[wires=1,steps=1,style={draw=orange,rounded
corners,fill=orange!20, xscale=0.9, yscale=0.9, fill opacity=0.7, inner
xsep=2pt},background]{} & \qw & \gate[style={draw, shape=semicircle, minimum size=0.5cm,  xscale=0.7,yscale=0.7, inner sep=0pt, outer sep=0pt, shape border rotate=270}]{}\gategroup[wires=1,steps=1,style={draw=green,rounded corners,fill=green!20, inner xsep=2pt},background]{}  \\
     \gate[style={draw, shape=semicircle, minimum size=0.5cm, xscale=0.7,yscale=0.7, inner sep=0pt, outer sep=0pt, shape border rotate=90}]{}\gategroup[wires=1,steps=1,style={draw=blue,rounded
corners,fill=blue!20, inner
xsep=2pt},background]{} & \qw & \qw & \qw &  \targ{} & \targ{} &\gate[style={draw, shape=semicircle, minimum size=0.5cm,  xscale=0.7,yscale=0.7, inner sep=0pt, outer sep=0pt, shape border rotate=270}]{}\gategroup[wires=1,steps=1,style={draw=green,rounded corners,fill=green!20, inner xsep=2pt},background]{}  &
\end{tikzcd}

         \caption{Preparation of $\ket{\overline{00}}$. This circuit has $L_e=22$ locations, whereas the unencoded circuit has $L_u=4$ locations.}
         \label{fig-state-prep-counts-3}
     \end{subfigure}
        \caption{Preparation circuits for $\ket{\overline{0+}}$ and $\ket{\overline{00}}$, using the same color indicators as in Figure~\ref{fig-maxent-counting}, with the addition of orange markers, which we use to indicate a wait location.}
\end{figure}
\end{proof}

\begin{remark}
    Like \cite{Gottesman16}, we assume here that qubits are only initialized when they are needed, and do not need wait locations before. If we do not make this assumption, and also count the additional wait locations, we obtain slightly lower thresholds. 
    We could also assume a different error model where fault-locations may be counted differently. In particular, the measurement gate is commonly known to take the longest time. We could, for the preparation above, also assume that the measurements are performed immediately after the last gate has been performed to alleviate additional wait locations. However, in our later investigations, where we may apply logical gates or wait for a certain amount of time steps, the qubits will always be measured at the same time, and we therefore use this assumption in our counting.
\end{remark}


The following logical gates are proposed in \cite{Gottesman16} and \cite{CCNH21}:
\begin{equation} \label{eq-ft-gates}\begin{split}
    &\overline{X}_1=X_1 X_3 \\& \overline{X}_2=X_1 X_2 \\ &\overline{Z}_1=Z_1 Z_2 \\ &\overline{Z}_2=X_1 Z_3 \\ &\overline{\SWAP}_{1,2}\circ \overline{H}_1 \overline{H}_2 =H_1H_2H_3 H_4 \\ &\overline{Z}_1 \overline{Z}_2 \circ \overline{\text{CZ}}_{1,2} =S_1S_2S_3 S_4 
\end{split}\end{equation}





\begin{lemma}
Let $\mathcal{C}$ denote the [[4,2,2]] code and let $\mathcal{F}(p)$ denote the i.i.d. Pauli noise model with $p\in [0,1]$. Let $\Gamma:\mathbbm{C}^{4}\rightarrow \mathbbm{C}^4$ be a circuit that consists of the proposed state preparations and the proposed logical gates. Then, no single-qubit gate error in $[\Gamma_{\mathcal{C}}]_{\mathcal{F}(p)}$ causes a logical error.
\end{lemma}

\begin{proof}
This follows because all transformation rules from before hold, and the encoded gates only contain single-qubit gates.
\end{proof}


Effectively, within the [[4,2,2]] code, we can thus prepare any state that can be obtained from applying a combination of gates from the logical gate set 
to one of the three states listed in Table~\ref{tab-ft-preps}, with the threshold decreasing as a function of the number of gates that are applied. Then, if the probability of single-qubit error is below this threshold, we expect the circuit implemented in the [[4,2,2]] code to fulfill the criterion for demonstrating fault-tolerance in Definition~\ref{def-ft-criterion}.


\begin{corollary}
Let $\mathcal{C}$ denote the [[4,2,2]] code and let $\mathcal{F}(p)$ denote the i.i.d. Pauli noise model with $p\in [0,1]$. Let $\Gamma:\mathbbm{C}^{4}\rightarrow \mathbbm{C}^4$ be a circuit that consists of the proposed state preparations and $T$ proposed logical gates.

Then, for any gate sequence length $T$, $[\Gamma_{\mathcal{C}}]_{\mathcal{F}(p)}$ demonstrates fault-tolerance for the following thresholds:

\begin{itemize}
    \item For the initial state $\ket{\overline{\phi_+}}$:
\[p < \frac{1}{40+14T}.\]
\item For the initial state $\ket{\overline{0+}}$:
\[p <  \frac{5+2T}{ 217+156T+28T^2 }  .\]
\item For the initial state $\ket{\overline{00}}$:
\[p <  \frac{2+T}{ 283+124T+14T^2 } . \]
\end{itemize}

\end{corollary}

\begin{proof}
All logical gates are transversal single-qubit gates. This means that each logical gate in the sequence adds 4 locations to the 4-qubit circuit, and 2 locations to the 2-qubit circuit. Therefore, when a gate sequence of length $T$ is applied to $\ket{\overline{\phi_+}}$, the total circuit has  $L_u=12+4T$ locations, and applying the equivalent sequence to the unencoded state $\ket{\phi_+}$, the resulting number of locations is $L_e\geq 6+2T$.
Then, we can use the criterion from Lemma~\ref{thm-thresholdcriterion}, 
\[  p < \frac{L_u}{ L_e(L_e-1) + L_uL_e   +  L_u^2 } , \]
to conclude that the circuit simulation is expected to demonstrate fault-tolerance if 
\[p <  \frac{6+2T}{ (12+4T)(11+4T) + (6+2T)(12+4T) +(6+2T)^2 } = \frac{1}{40+14T}.\]
\end{proof}

The thresholds for different lengths of gate sequences $T$ are plotted in Figure~\ref{fig-tvsthreshold}. As expected, increasing the number of gates leads to lower thresholds.

\begin{figure}[htbp]
    \centering
    \includegraphics[width=8cm]{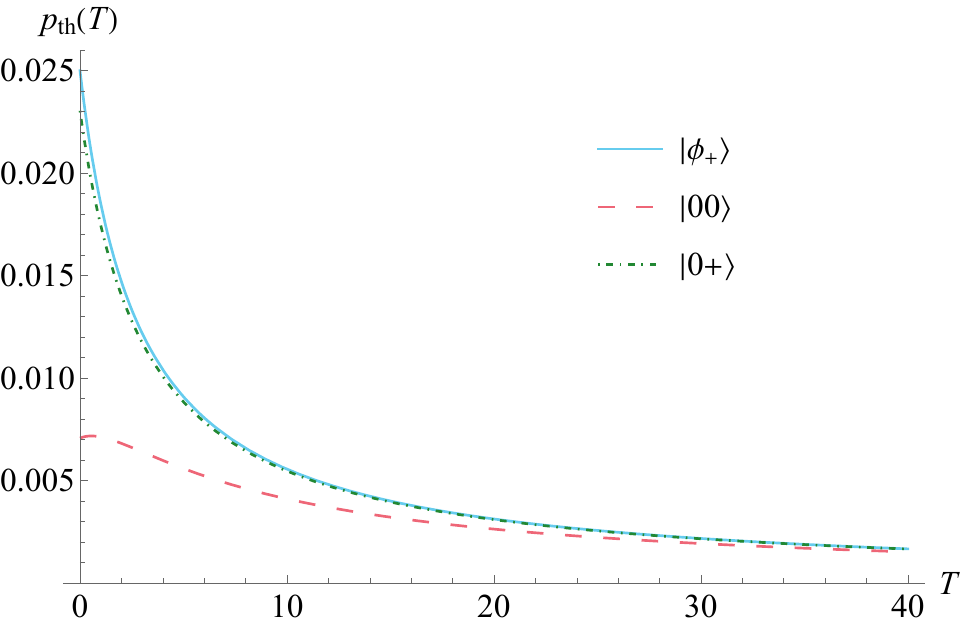}
    \caption{Thresholds $p_{th}(T)$ for demonstrating fault-tolerance for different gate sequence lengths $T$ applied to the preparation circuits listed in Table~\ref{tab-ft-preps}, where locations are counted in terms of the gate set from \cite{Gottesman16}.}
    \label{fig-tvsthreshold}
\end{figure}



\subsection{Thresholds in the native gate set}

The IonQ quantum computers use trapped ion technology, which involves trapping individual ionized ytterbium atoms using electromagnetic fields and then manipulating their quantum states with laser beams to perform a computation. This approach offers high qubit fidelity, long coherence times and flexible qubit connectivity \cite{BenchmarkIonQ19}.


When provided with the high-level description of a quantum circuit, any quantum computer uses several levels of compilation and optimization to improve the performance and efficiency of the computation. As a first step, the circuit is typically converted into a sequence of native gates (if it is provided in terms of different gates). This process often involves optimizing the circuit in order to minimize the number of gates, utilize the highest fidelity qubits, and to reduce the overall error rate of the computation \cite{LSAA22,Maslov17}. The goal of mid-level compilation is then to reduce the circuit depth and improve the parallelism of the computation through techniques such as gate reordering, which can help to protect against the effects of gate errors \cite{VDRF17,JLFCHZHZ22}. Finally, high-level compilation involves optimizing the overall structure of the quantum algorithm itself, which can also involve selecting appropriate error correction codes or other strategies to improve the reliability of the computation.

To truly test a current quantum device, these additional device-level optimizations should be kept in mind when implementing an algorithm. The quantum computers of the present are greatly affected by errors, and therefore, results of a computation are highly dependent on the strategies used for error mitigation and circuit optimization. One of the most basic and yet interesting near-term goals would involve demonstrating that a quantum state can stay alive and usable for an extended period of time despite the high error rates, and that this time duration can be improved by the use of error correction and mitigation \cite{Terhal23}. Typically, this could be investigated in a setup where a quantum state is prepared and then left alone for longer and longer time periods before it is measured (as in \cite{OPH+16} for superconducting qubits); this type of circuit, however, would involve applying an identity gate to the qubits in each time step, which would be removed by the compilation procedure. Experiments such as this are therefore only sensible if one can know and control the level of compilation. In 2022, IonQ became the first trapped ion quantum computer to allow for cloud-based remote-access verbatim compilation, which is a feature ensuring that the gates in the circuit are executed exactly as described. By providing the task in terms of native gates, it can therefore be facilitated that a user knows exactly which gates are executed, as well as their ordering.


For commercial use, IonQ currently offers cloud-based access to an 11 qubit device \emph{Harmony}, and a 21 qubit device \emph{Aria}. Recently, three native gates have become available for commercial use on these devices in the context of verbatim compilation.

 The single qubit gates are implemented using laser pulses that induce transitions between the ion's electronic energy levels. Details of the physical implementation are discussed in \cite{BenchmarkIonQ19}.
\[ \gpi(\phi)=
\begin{pmatrix}
0& e^{-i\phi}\\
e^{i\phi}&0\\
\end{pmatrix}
\]
\[ \gpii(\phi)= \frac{1}{\sqrt{2}}
\begin{pmatrix}
1& -ie^{-i\phi}\\
-ie^{i\phi}&1\\
\end{pmatrix}
\]
As an entangling gate, IonQ employs the two-qubit Mølmer-Sørensen gate \cite{SM99}:
\[ \ms(\phi,\psi)= \frac{1}{\sqrt{2}}
\begin{pmatrix}
1& 0&0&-i e^{-i(\phi+\psi)}\\
0& 1&-i e^{-i(\phi-\psi)}&0\\
0& -ie^{i(\phi-\psi)}&1&0\\
-ie^{i(\phi+\psi)}& 0&0&1\\
\end{pmatrix}
.\]
Techniques for decomposing a given circuit in terms of this gate set are discussed in \cite{Maslov17}. For some of the most common gates, in particular gates that are used in the usual description of the [[4,2,2]] code, their decomposition in terms of IonQ's native gates is given in Table~\ref{tab-ionq-native-gates-decomp}.

\begin{table}[htbp]
  \centering
  \caption{Native gate compositions of common gates}
  \begingroup
\setlength{\tabcolsep}{6pt} 
\renewcommand{\arraystretch}{1.7} 
    \begin{tabular}{!{\color{gray}\vrule}l!{\color{gray}\vrule}l!{\color{gray}\vrule}}
       \arrayrulecolor{gray}\hline 
        Gate & Decomposition in native gates \\ 
      \hline  \hline
      \begin{tikzcd}
&\gate{X} & \qw 
\end{tikzcd} & 
            \begin{tikzcd}
&\gate{\gpi(0)} & \qw 
\end{tikzcd}
         \\ \hline
         \begin{tikzcd}
&\gate{Y} & \qw 
\end{tikzcd} & 
            \begin{tikzcd}
&\gate{\gpi(\pi/2)} & \qw 
\end{tikzcd}
\\\hline
\begin{tikzcd}
&\gate{Z} & \qw 
\end{tikzcd} & 
            \begin{tikzcd}
&\gate{\gpi(0)} &\gate{\gpi(\pi/2)}& \qw 
\end{tikzcd}

\\\hline
\begin{tikzcd}
            &\gate{H} & \qw 
            \end{tikzcd} & 
            \begin{tikzcd}
&\gate{\gpi(0)} & \gate{\gpii(-\pi/2)} &\qw 
\end{tikzcd}
\\\hline
\begin{tikzcd}
& \ctrl{1}  & \qw \\
& \targ{} & \qw 
\end{tikzcd} & 
            \begin{tikzcd}
& \gate{\gpii(\pi/2)}& \gate[wires=2][2cm]{\ms(0,0)} &\gate{\gpii(\pi)} & \gate{\gpii(-\pi/2)} & \qw \\
 & \qw  & \qw &\gate{\gpii(\pi)} & \qw& \qw 
\end{tikzcd}

\\\hline
    \end{tabular}
    \endgroup
    \label{tab-ionq-native-gates-decomp}
\end{table}

We find that the threshold theorem for demonstrating fault-tolerance, Theorem~\ref{thm-thresholdcriterion}, also applies to the code when decomposed in the native gates of IonQ. This can be seen by analyzing the output statistics from every possible single-qubit error in the i.i.d. Pauli fault model, which confirm that such errors will either be detected or give the correct statistics. It is therefore justified to employ this code as an error detection code and expect it to demonstrate fault-tolerance under the criterion in Definition~\ref{def-ft-criterion}.
We obtain the following thresholds for the state preparation:
\begin{itemize}
    \item For $\ket{\overline{00}}$, the location counts are $L_u= 4$ and $L_e=26$, and the threshold for demonstrating fault-tolerance under the i.i.d. Pauli noise model with the preparation circuit is $p_{th}=0.00519481$.
    \item For $\ket{\overline{0+}}$, the location counts are $L_u= 6$, $L_e=20$, and the threshold for the preparation circuit is $p_{th}=3/268=0.011194$.
    \item For $\ket{\overline{\phi_+}}$, the location counts are $L_u= 10$, $L_e=20$, and the threshold for the preparation circuit is $p_{th}=1/68=0.0147059$.
\end{itemize}

The thresholds for sequences of $T$ logical gates for the three state preparations are plotted in Figure~\ref{fig-thresholds-native}. The use of native gates slightly increases the location numbers, which manifests itself in a lower threshold.

\begin{figure}[htbp]
    \centering
    \includegraphics[width=8cm]{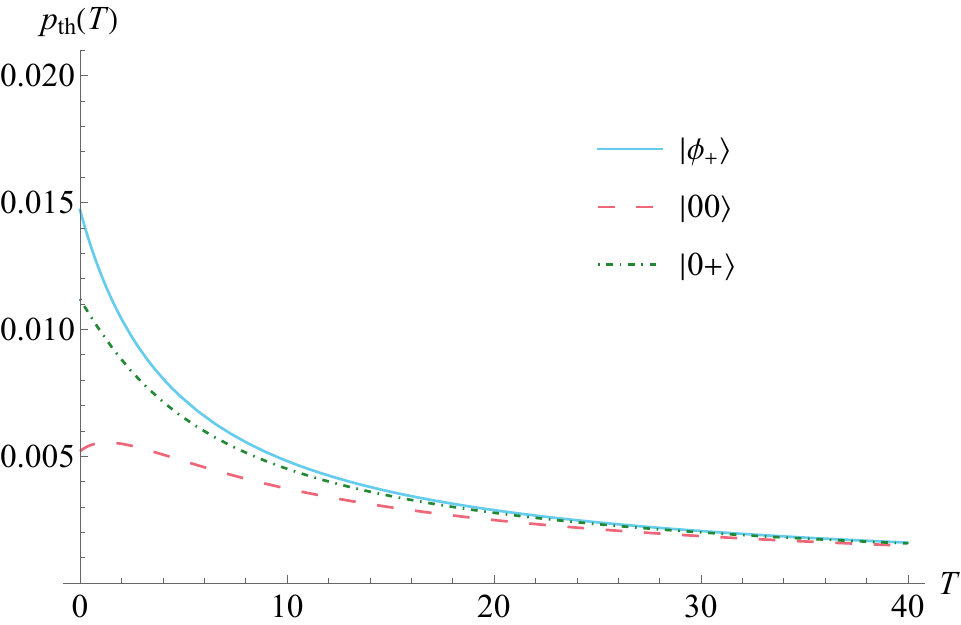}
    \caption{Threshold $p_{th}(T)$ for demonstrating fault-tolerance as a function of the gate sequence length $T$, where locations are counted in terms of the native gate set of IonQ Harmony and Aria.}
    \label{fig-thresholds-native}
\end{figure}



\section{Error detection with the [[4,2,2]] code on current hardware}
\label{sec-gottesman-results}

The cloud-based open access promises a more direct way of testing and implementing small algorithms without a deeper understanding of experimental components and system-based protocol designs. However, for access as a researcher, we found significant obstacles in relation to exploring the capabilities of these systems. While support for verbatim compilation was recently introduced, it has to be noted that the list of native gates that are available does not comprise every possible gate; in particular, an identity gate, which would correspond to doing nothing to the system for a certain amount of time, is not available. An identity gate can be constructed by applying another gate twice, for example $\gpi(0)$, as $\gpi(0)^2=\mathbbm{1}$, which is also what IonQ's documentation suggests. However, it is important to note that this is not exactly the same as doing nothing, as it might lead to higher errors, and introduce bias depending on the gate we apply and the qubit we apply it to. Here, we chose to apply the gate $\gpi(0)$ to one additional qubit, and undoing it immediately after. This qubit does not interact with the other circuit qubits and is disregarded during the analysis. Because our model of an identity means that we are applying a gate twice, each such round constitutes two time steps in our life-time analysis.

For a given circuit that prepares a state $\ket{\psi}$, we denote the ideal output distribution by $P_{\psi}^{theo}$, which is the probability distribution we would expect in the asymptotic limit if there are never any errors, i.e. the probability distribution obtained from $\ket{\psi}$ through Born's rule. $P_{\psi}^{unenc}$ means the output probability distribution obtained from an experiment for the unencoded circuit, whereas $P_{\psi}^{enc}$ means the probability distribution obtained from post-selection of the output distribution of the experiment implementing the encoded circuit. In the case of circuit simulations with depolarizing errors, we denote the obtained probability distributions by $P_{\psi,q}^{\text{unenc}}$ and $P_{\psi,q}^{\text{enc}}$ where $q\in[0,1]$ is an index corresponding to the depolarizing probability used for the simulation. In the case of circuit simulations with amplitude damping noise, we use an index $\gamma\in[0,1]$, which corresponds to the amplitude damping parameter used for the simulation.
For experiments on quantum hardware, we denote the resulting probability distributions by $P_{\psi,D}^{\text{unenc}}$ and $P_{\psi,D}^{\text{enc}}$ where the index $D=\{H,A\}$ indicates whether the experiment was performed on IonQ's 11-qubit quantum computer Harmony, or their 21-qubit quantum computer Aria.

We perform simulations where the errors affecting the qubits are modelled by depolarizing noise and amplitude damping noise. From Figure~\ref{fig-thresholds-native}, we may expect the error parameters in Figure~\ref{fig-maxent-sim-0} to be below threshold for longer times than the parameters in \ref{fig-maxent-sim-2}. However, we observe a separation of the expected error in the unencoded circuit and the encoded circuit in all three cases, even for error parameters well above the threshold. This may very well be a result of the approximations used in the proof of Theorem~\ref{thm-thresholdcriterion} leading to a more pessimistic threshold for observing such an improvement. The simulation seems to confirm that we can indeed expect the [[4,2,2]] error detection code to be useful for computation, in the sense that an encoded state is more resistant against errors over time, at least under these basic error models.

\begin{figure}[htbp]
     \centering
     \begin{subfigure}[b]{0.7\textwidth}
         \centering
         \includegraphics[width=\textwidth]{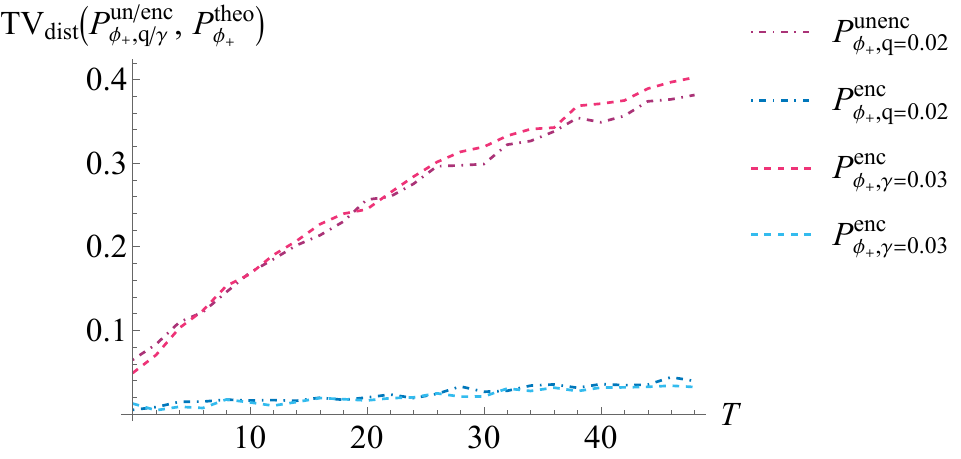}
         \caption{Simulation for error parameters $q=0.02$ and $\gamma=0.03$.}
         \label{fig-maxent-sim-0}
     \end{subfigure}
        \hfill
     \begin{subfigure}[b]{0.7\textwidth}
         \centering
         \includegraphics[width=\textwidth]{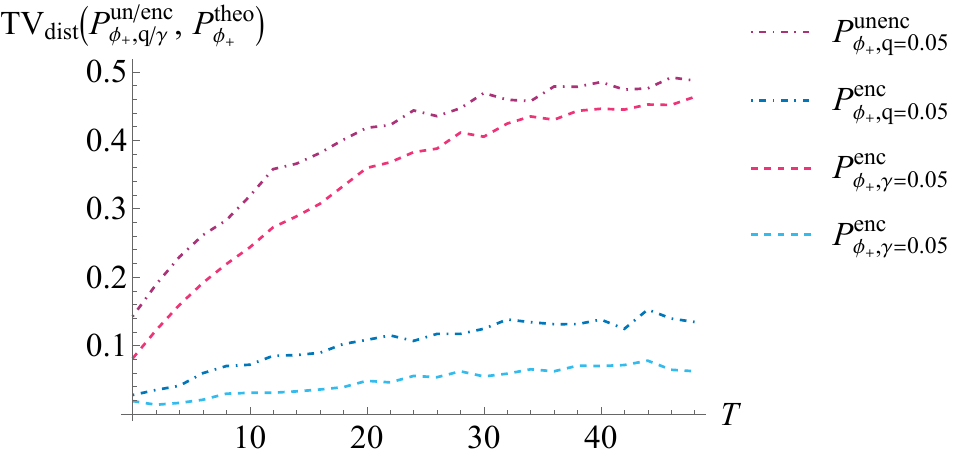}
         \caption{Simulation for error parameters $q=0.05$ and $\gamma=0.05$.}
         \label{fig-maxent-sim-1}
     \end{subfigure}
     \hfill
     \begin{subfigure}[b]{0.7\textwidth}
         \centering
         \includegraphics[width=\textwidth]{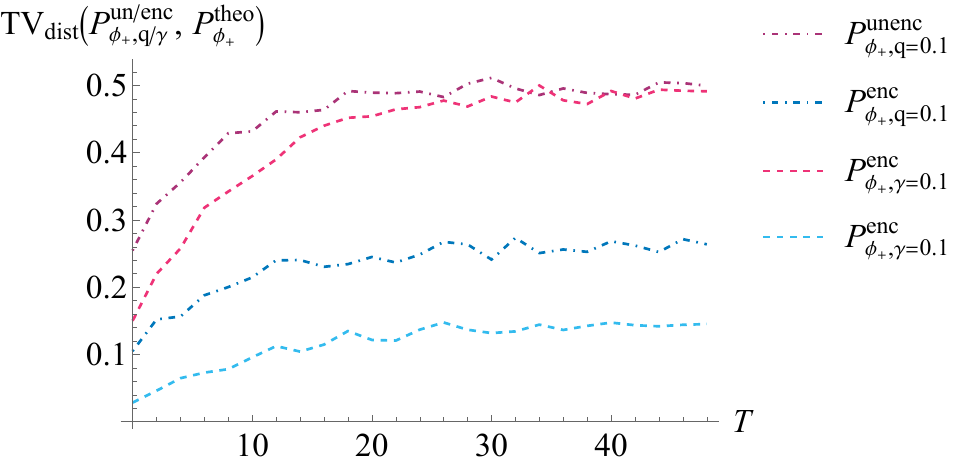}
         \caption{Simulation for error parameters $q=0.1$ and $\gamma=0.1$.}
         \label{fig-maxent-sim-2}
     \end{subfigure}
     \caption{Simulation of the circuits for $\ket{\phi_+}$ and increasing lengths of gate sequences applied to it, with a depolarizing noise model with depolarizing probability $p$ at each gate, and an amplitude damping noise model with damping parameter $\gamma$ at each gate. All simulations were based on 1000 shots for each circuit.}
     \label{fig-maxent-sim}
\end{figure}

On actual hardware, we observe the following data for the maximally entangled state $\ket{\phi_+}$, plotted in Figure~\ref{fig-maxent-exp}. The results for $\ket{00}$ are plotted in Figure~\ref{fig-00-exp}, and the results for $\ket{0+}$ are plotted in Figure~\ref{fig-0p-exp}.

It is notable that the TV-distance seems to vary greatly, in particular for some lengths of gate sequences. This is most likely not a consequence of the number of shots, but rather seems to be an effect of chip-related or daily fluctuations. 

\begin{figure}[htbp]
     \centering
     \begin{subfigure}[b]{0.7\textwidth}
         \centering
         \includegraphics[width=\textwidth]{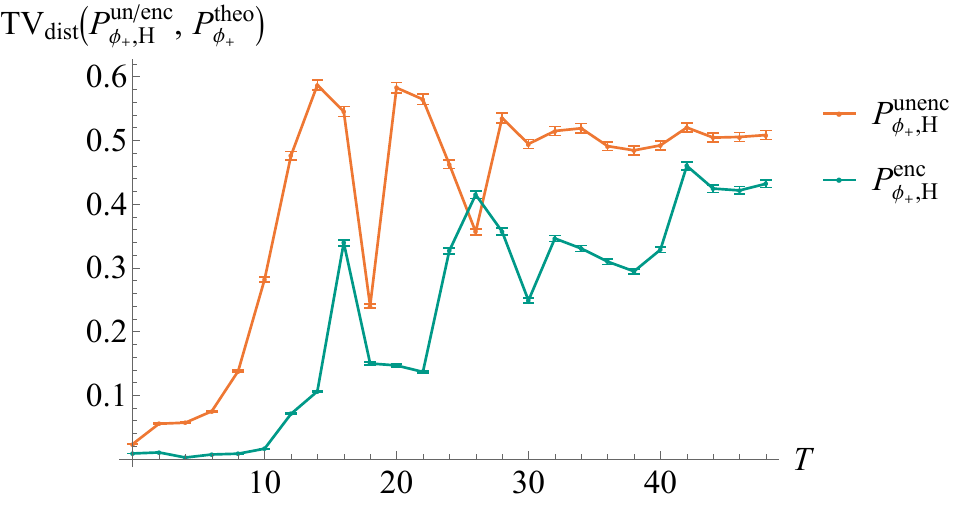}
         \caption{Experimental results on IonQ Harmony for $\ket{\phi_+}$. These results were obtained on the 25 May 2023, with 5000 shots.}
         \label{fig-maxent-h}
     \end{subfigure}
     \hfill
     \begin{subfigure}[b]{0.7\textwidth}
         \centering
         \includegraphics[width=\textwidth]{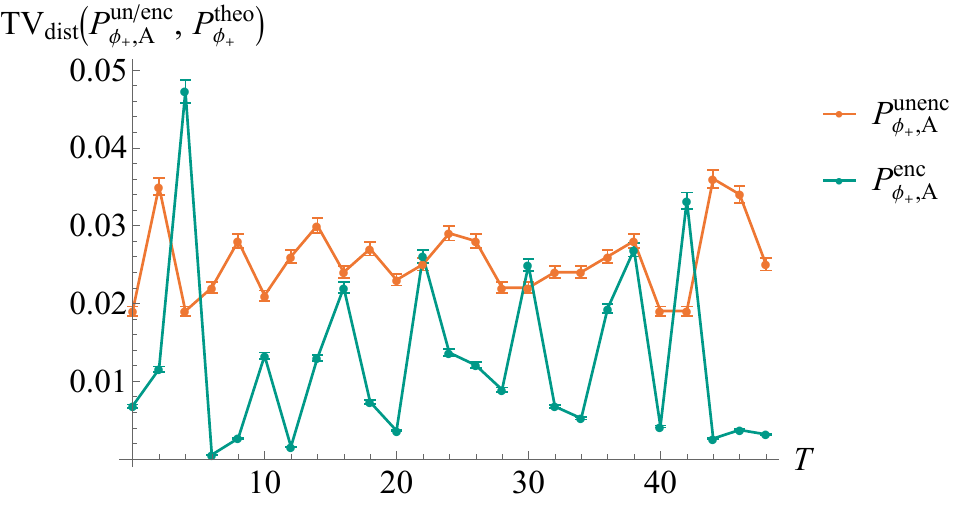}
         \caption{Experimental results on IonQ Aria for $\ket{\phi_+}$. These results were obtained on the 31 May 2023, with 1000 shots.}
         \label{fig-maxent-a}
     \end{subfigure}
        \caption{TV-distance of the probability distributions obtained as an output of current quantum hardware to the expected probability distribution for $\ket{\phi_+}$. One curve in each plot gives the result for the unencoded circuit, and one curve in each plot gives the result for the same circuit encoded in the [[4,2,2]] code for gate sequence lengths between 0 and 50.
        }
        \label{fig-maxent-exp}
\end{figure}



\begin{figure}[htbp]
     \centering
     \begin{subfigure}[b]{0.7\textwidth}
         \centering
         \includegraphics[width=\textwidth]{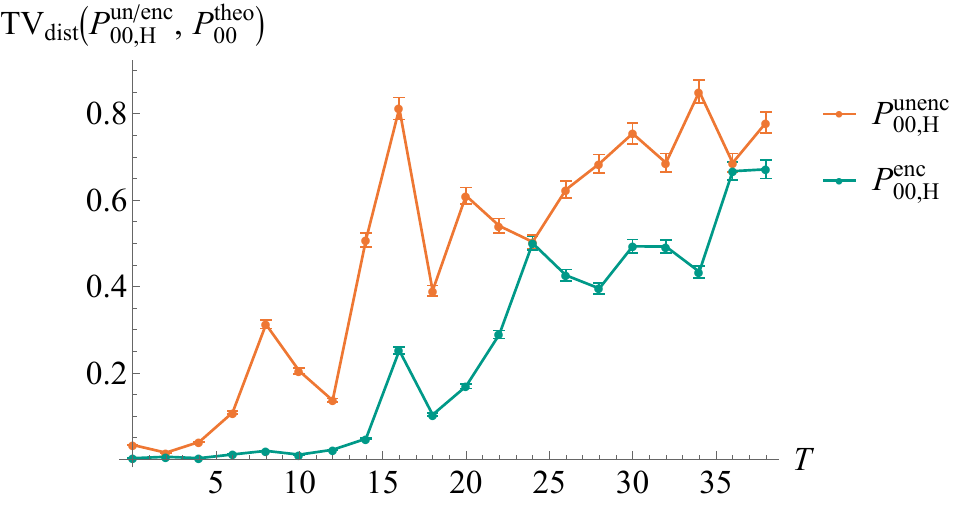}
         \caption{Experimental results on IonQ Harmony for $\ket{00}$. These results were obtained on 29 May 2023.}
     \end{subfigure}
     \hfill
     \begin{subfigure}[b]{0.7\textwidth}
         \centering
         \includegraphics[width=\textwidth]{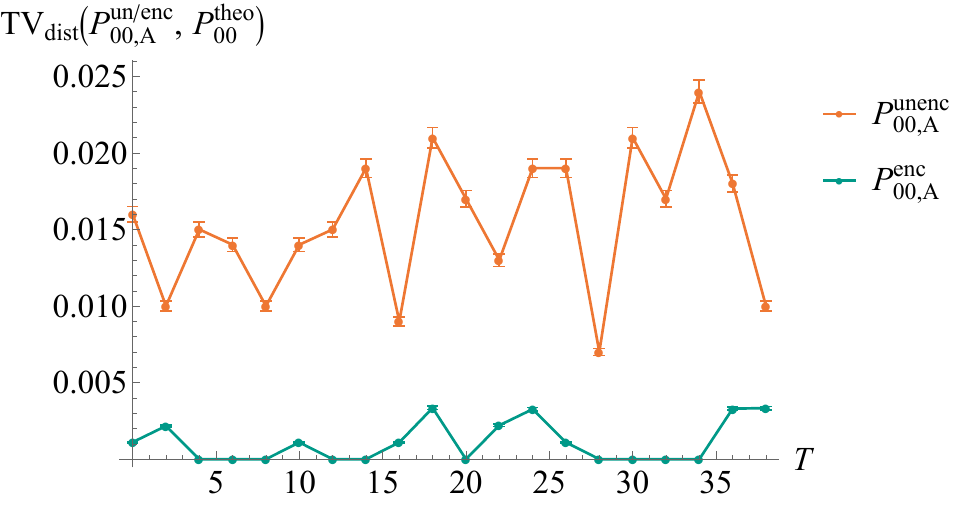}
         \caption{Experimental results on IonQ Aria for $\ket{00}$. These results were obtained on 30 May 2023.}
     \end{subfigure}
        \caption{TV-distance of the probability distributions obtained as an output of current quantum hardware to the expected probability distribution for $\ket{00}$. One curve in each plot gives the result for the unencoded circuit, and one curve in each plot gives the result for the same circuit encoded in the [[4,2,2]] code. Both experiments utilized 1000 shots and were performed for gate sequence lengths ranging between 0 and 40.}
        \label{fig-00-exp}
\end{figure}

\begin{figure}[htbp]
     \centering
     \begin{subfigure}[b]{0.7\textwidth}
         \centering
         \includegraphics[width=\textwidth]{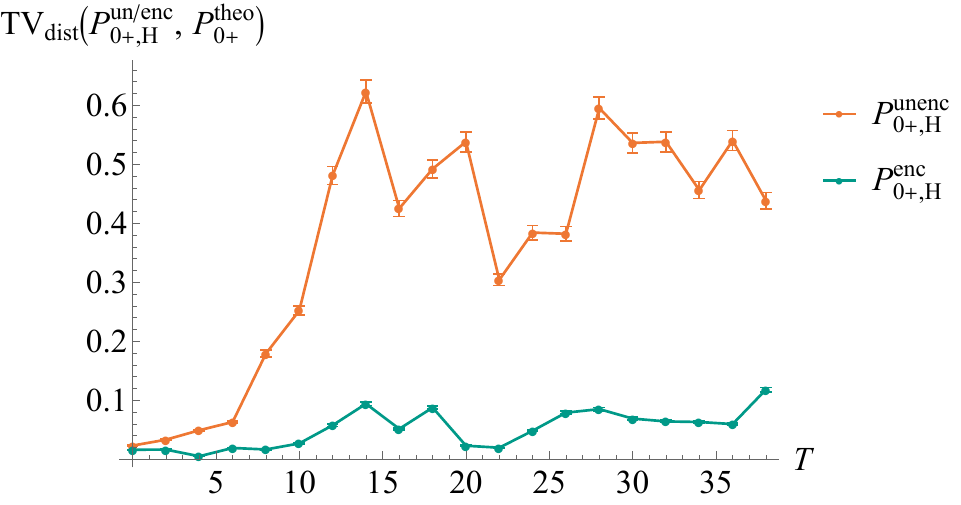}
         \caption{Experimental results on IonQ Harmony for $\ket{0+}$. These results were obtained on 30 May 2023.}
     \end{subfigure}
     \hfill
     \begin{subfigure}[b]{0.7\textwidth}
         \centering
         \includegraphics[width=\textwidth]{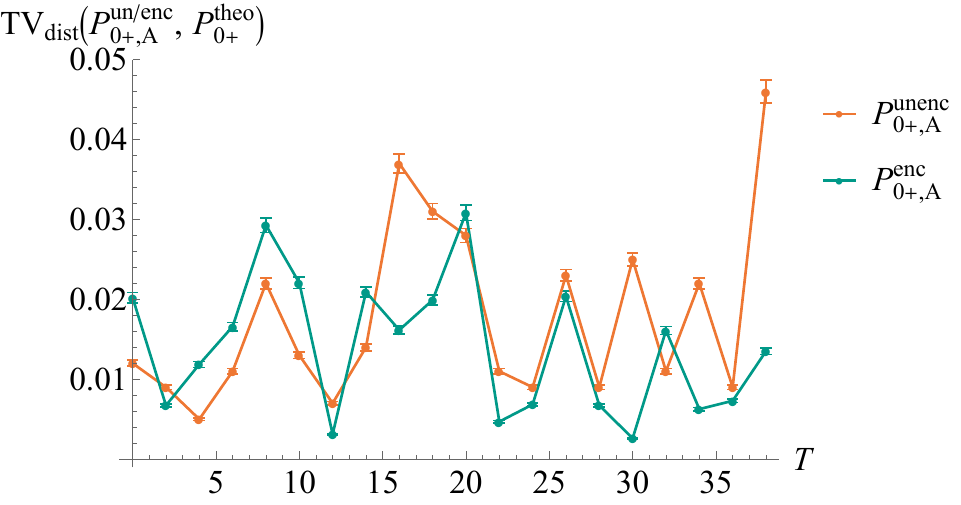}
         \caption{Experimental results on IonQ Aria for $\ket{0+}$. These results were obtained on 31 May 2023.}
     \end{subfigure}
        \caption{TV-distance of the probability distributions obtained as an output of current quantum hardware to the expected probability distribution for $\ket{0+}$. One curve in each plot gives the result for the unencoded circuit, and one curve in each plot gives the result for the same circuit encoded in the [[4,2,2]] code. Both experiments utilized 1000 shots and were performed for gate sequence lengths ranging between 0 and 40.}
        \label{fig-0p-exp}
\end{figure}

It can indeed be observed that encoding the states in the [[4,2,2]] code improves (i.e. lessens) their TV-distance from the ideal state for the experiments performed on IonQ Harmony. At least for short circuit depths (approximately 15 time steps), the encoded states often have a significantly lower probability to lead to a corrupted result. This seems to be in agreement with the reported single-qubit gate error of $0.4\%$ \cite{BenchmarkIonQ19,IonQWebsiteHarm} being below to the thresholds from Figure \ref{fig-thresholds-native} for low circuit depth, and above the thresholds for larger circuit depths. These results can be interpreted as an increased life-span of the information encoded in the quantum state when it is protected by the [[4,2,2]] code. Furthermore, it suggests that the noise is not excessively correlated and that our noise model for the device is fairly accurate. However, even if we generally observe an improvement for the protected state, this effect still does not protect the computation from noise effects that occur at larger circuit depths. In other words, although we manage to keep the quantum information usable for longer times through our implementation in the [[4,2,2]] code, this additional life-time is still far from enough for the circuit depths required for most quantum algorithms. While error detection helps, it can thus not be concluded that the computation overall is fault-tolerant. This is, of course, a consequence of our way defining a threshold, which focusses on a proof-of-principle demonstration of the principles of fault-tolerance. In other words, we are basing our threshold calculation on whether or not the encoded circuit performs better than the unencoded version, while disregarding whether or not either of the circuits performs particularly well. Similar results were observed in \cite{HF19} for superconducting qubits.


On IonQ Aria, we do not observe a conclusive improvement. While the overall errors are significantly lower than for IonQ Harmony for all observed circuit depths, the performance of the encoded circuit and the unencoded circuit is almost indistinguishable in the case of $\ket{0+}$ and $\ket{\phi_+}$ (although there is perhaps a slight tendency towards a lower average error probability for the encoded circuit). Only in the case of the preparation of $\ket{00}$ do we seemingly observe a noticeable separation. However, the corresponding preparation circuit also has an extra qubit which is measured during the error detection phase, which may lead to a lower probability to survive postselection.

Based on the reported single-qubit gate error of $0.05\%$ \cite{IonQWebsiteAria}, we would expect IonQ Aria to have single-qubit gate error rates below the thresholds implied by Figure \ref{fig-thresholds-native}, and certainly lower single-qubit gate error rates than IonQ Harmony, suggesting that the local i.i.d. Pauli noise model may not be accurate for this device. It may be that the effects of correlated errors become more significant in comparison when single-qubit errors are sufficiently unlikely.

\section{Conclusion and open questions}

In this project, we compare the state preparation and lifetime of three two-qubit states with and without a protective implementation in the [[4,2,2]] error detection code on a trapped ion quantum computer. We observe, in particular, that error detection does indeed protect the quantum state to some degree for the 11-qubit device, which suggests that our error model approximates the device well. However, the total computation is far away from being below a threshold for true fault-tolerance. For IonQ's 21-qubit device, the overall computations seem to be much less vulnerable to faults, but the improvement through an error detecting implementation is much less clear, suggesting that our assumed error model for the threshold calculation does not capture the behaviour of the device, which has to be analyzed more carefully in the future.


On IonQ Harmony, we observe that short and medium-depth circuits benefit the most from a protective encoding; for quantum communication setups, particularly for communication over long distances or time durations, it may therefore be beneficial to protect only the encoding and decoding circuit, without imposing threshold constraints on the communication line, which perhaps underlines the relevance of our coding scheme construction in Chapter~\ref{chapter-fteacap}. Furthermore, the findings for IonQ Aria suggest that it may be practically relevant to extend our work on communication with noisy encoding and decoding circuits in Chapter~\ref{chapter-fteacap} to more complicated and perhaps correlated noise models, or to consider circuits with randomized compiling.

\chapter{The limits of enhancing communication rates with entanglement}
\label{chapter-divrad}

This chapter is based on a work which is in preparation for publication.

Based on a conjecture by Bennett et al. in one of the earliest papers on entanglement-assisted communication \cite{BSST02}, we show that the quotient $\frac{C^{ea}(T)}{C(T)}$ is bounded by an expression that only depends on the channel output dimension, not on the channel itself, for many quantum channels $T$.

For these channels, we thus have uniform convergence in Chapter~\ref{chapter-fteacap}, Theorem~\ref{thm-final-coding-thm}.



\section{Introduction}

Information theory, as pioneered by Claude Shannon in his seminal work \cite{Shannon48}, revolutionized communication theory and laid the groundwork for much of our modern communication infrastructure. Mathematically, the process of communicating between a sender and a receiver is modelled by a communication channel $T$, which takes into account the noise affecting individual symbols during transmission. The ability of a specific channel $T$ to transmit information is quantified by its capacity, which is given by the asymptotic number of message bits that can be successfully transmitted per channel usage while achieving vanishing error. This capacity is in fact equal to the mutual information between the channel input and the channel output \cite{Shannon48}.

In quantum information theory, the generalization of this concept is not straightforward, and several non-equivalent generalizations exist. Each notion captures different aspects of quantum information transmission and adresses communication setups in the presence of different resources or constraints from quantum theory.


The \emph{classical capacity} $C(T)$ of a quantum channel $T$ refers to the maximum rate at which classical information can be reliably transmitted through the quantum channel without any additional quantum or classical resources. Nonetheless, this capacity setup introduces effects that only appear in quantum theory, most notably non-additivity in tensor products \cite{Hastings09}. To study this quantity, Holevo \cite{Holevo96} introduced the \emph{Holevo capacity}, which is an entropic expression that bounds the classical capacity of a given channel.

In contrast, the superdense coding protocol \cite{BW92} demonstrates that quantum entanglement between a sender and a receiver can improve their communication rate and acts as a resource for communication. Beyond this example, \cite{BSST99,BSST02} defines an entanglement-assisted communication setup for any quantum channel in order to characterize the maximum rate of transmitting information via the channel when the sender and receiver have access to an unlimited amount entanglement. This \emph{entanglement-assisted capacity} $C^{ea}(T)$ turns out to be equal to the (maximal) quantum mutual information \cite{BSST02} and is therefore seen as a quantum analogue of the classical concept.

The entanglement-assisted capacity is always larger or equal to the classical capacity of a given channel. Here, we investigate how much larger it can be, or if the improvement can be bounded by some finite expression for any channel that does not depend on the channel, but only on its dimension. In \cite[Section~III.3]{BSST02}, a bound on this quotient was conjectured by Bennet, Shor, Smolin and Thalapiyal based on results from Holevo in \cite{Holevo02}.
An upper-bound on the quotient is known for some well-known channels with closed form capacity expressions, such as the erasure channel and the depolarizing channel \cite{BSST99}, and the qubit amplitude damping channel \cite{BSST02}.
Beyond these channel families, we are not aware of any other known upper bounds.

Here, we show that the quotient is bounded by a dimension-dependent expression for a large class of channels. More precisely, for many quantum channels $T:\mathcal{M}_{d_{A'}}\rightarrow \mathcal{M}_{d_{B}}$, we show that there exists a function $f(d_B)$ of the channel output dimension such that
\[\frac{C^{ea}(T)}{C(T)}\leq f(d_B) \]
where $f(d_B)$ does not depend on the channel, only on the output dimension. We achieve this by bounding the quotient of the entanglement-assisted capacity and the Holevo capacity of the same channel.


In order to show this, we give an interpretation of the entanglement-assisted capacity in terms of a relative entropy between the channel $T$ and a replacer channel that replaces any input with a state that depends on the channel image in Section~\ref{sec-stab-div-rad}. Then, we show a novel upper bound on the entanglement-assisted capacity which is quadratic in the induced trace distance between the two channels in Section~\ref{sec-taylors-thm}. Finally, making use of Pinsker's inequality to lower bound the Holevo capacity, we show a dimension-dependent bound on the quotient of capacities in Section~\ref{sec-bounding}.
We furthermore derive a tight bound on the quotient for unital qubit channels in Section~\ref{sec-pauli}.







\section{Capacities as quantum channel divergences}
\label{sec-stab-div-rad}

In this section, we recall that the Holevo capacity of a quantum channel $T$ can be formulated as a channel divergence between the channel and a replacer channel \cite{SW01,TT15}. We go on to show that the entanglement-assisted capacity can be formulated as a stabilized channel divergence in Theorem~\ref{thm-stab-div-rad}. 

We refer to the definitions in Section~\ref{sec-quantum-shannon-theory} for quantum divergences and quantum channel divergences.

\subsection{The Holevo capacity as a channel divergence}

In the case of classical communication via a quantum channel $T:\mathcal{M}_{d_{A'}}\rightarrow \mathcal{M}_{d_B}$, we consider an encoding map that maps classical messages to a quantum state which is sent through the channel, and a decoding map that takes the channel's output and outputs a guess of the classical message. Then, the ability of a quantum channel to transmit information is quantified by the asymptotic rate of message bits per channel use that is achieved by the best possible encoding and decoding operation with vanishing probability to decode the wrong message.


In an effort to study how well a quantum channel transmits classical information, we recall that Holevo \cite{Holevo96} introduced the Holevo capacity, as can be found in Definition~\ref{def-holevo-capacity}. 
This Holevo capacity lower bounds the channel's classical capacity, and thus a lower bound on the Holevo capacity also serves as a lower bound on the classical capacity, see also Theorem~\ref{thm-hsw-thm}.

In order to study the set of states that optimize the Holevo capacity more closely, we introduce the following notion:

\begin{definition}[Holevo divergence radius]
    Let $\mathcal{S}_0\subseteq \mathcal{M}_d$ be a closed subset of the set of quantum states. The (Holevo) divergence radius of $\mathcal{S}_0$ is given by \[\chi(\mathcal{S}_0)=\inf_{\sigma\in \mathcal{M}_d} \Big\{ \sup_{\rho \in \mathcal{S}_0} D(\rho||\sigma) \Big\}.\]
\end{definition}

With this formulation, the Holevo capacity of a channel $T$ can be expressed as the divergence radius of the channel image.

\begin{theorem}[{\cite[Section~V]{SW01}}]\label{thm-class-cap-div-rad}
Let $T:\mathcal{M}_{d_A'}\rightarrow \mathcal{M}_{d_B}$ be a quantum channel. Then, we have
\[C_{H}(T)=\chi(\Im(T))=\inf_{\sigma_B \in \mathcal{M}_{d_B}} \Big\{ \sup_{\rho_B \in \Im(T)} D(\rho_B||\sigma_B)\Big\}. \]
\end{theorem}

This was originally proven in \cite{SW01}; we also give the proof in greater detail in Section~\ref{sec-shannon-theory}, see Theorem~\ref{thm-class-cap-2}.

This expression for the Holevo capacity allows for a geometric interpretation of the capacity by analyzing the channel image \cite{TT15}. The minimum of the optimization problem is achieved by a unique state $\sigma\in\mathcal{M}_{d_B}$, which we will refer to as the \emph{divergence center}. Furthermore, the states achieving the maximum of the optimization problem are referred to as \emph{peripheral states}. The center lies in the convex hull of the peripheral states, and Carathéodory's convex hull theorem ensures that the center can be written as a convex combination of at most $d+1$ peripheral points.

For any quantum state $\sigma\in\mathcal{M}_{d_B}$, we define an associated \emph{replacer channel} $R_{\sigma}:\mathcal{M}_{d_A'}\rightarrow \mathcal{M}_{d_B}$, by which we mean a constant channel that always outputs the quantum state $\sigma$ for any input, i.e. $R_{\sigma}(\rho)=\sigma \forall \rho \in \mathcal{M}_{d_{A'}}$. 

The Holevo capacity of a quantum channel $T:\mathcal{M}_{d_A'}\rightarrow \mathcal{M}_{d_B}$ can further be seen as a channel divergence quantifying the distance between $T$ and a particular replacer channel $\mathcal{R_{\sigma}}:\mathcal{M}_{d_A'}\rightarrow \mathcal{M}_{d_B}$. 
To be precise, we compare the channel $T$ and the replacer channel which is closest to $T$ in channel divergence, which is the replacer channel that replaces any input with the divergence center of $T$:

\begin{corollary}
For any quantum state $\sigma_B\in \mathcal{M}_{d_B}$, let $R_{\sigma}:\mathcal{M}_{d_{A'}}\rightarrow \mathcal{M}_{d_B}$ be a quantum channel such that $R_{\sigma}(\rho_{A'}) = \sigma_B\forall \rho_{A'}\in\mathcal{M}_{d_{A'}}$. Then, for any quantum channel $T:\mathcal{M}_{d_{A'}} \rightarrow \mathcal{M}_{d_B}$, we have
\[C_{H}(T) = \inf_{\sigma_B\in \mathcal{M}_{d_B}} D(T||R_{\sigma}).\]
\end{corollary}

The proof directly follows by the definition of the replacer channel $R_{\sigma}$ and the channel divergence.



\subsection{The entanglement-assisted capacity as a stabilized channel divergence}

When the communication via a quantum channel is assisted by entanglement, the channel's ability to transmit quantum information is given by the quantum mutual information between the channel's input and output, see Theorem~\ref{thm-ent-ass-cap-bsst}.


Similarly to the expression in Theorem~\ref{thm-class-cap-div-rad} for the Holevo capacity, we can find an expression for the entanglement-assisted capacity in terms of a quantum divergence. For this capacity, the expression depends on the image of the channel $T:\mathcal{M}_{d_{A'}}\rightarrow \mathcal{M}_{d_B}$ tensored with an identity on an extra reference system $A$, i.e. the image of $\id \otimes T: \mathcal{M}_{d_A}\otimes \mathcal{M}_{d_{A'}} \rightarrow \mathcal{M}_{d_A}\otimes \mathcal{M}_{d_B}$. In line with works in \cite{Watrous04}, we refer to the related objects as a \emph{stabilized divergence radius} and \emph{stabilized channel divergence}. A similar representation was derived for $\alpha$-entropies and $\alpha$-mutual information in \cite{GW14} for $\alpha \in (1,\infty)$, where our special case can be obtained by taking the limit $\alpha\rightarrow1$.


\begin{definition}[Stabilized divergence radius]
    Let $\mathcal{S}_0\subseteq \mathcal{M}_{d_A} \times \mathcal{M}_{d_B} $ be a closed subset of the set of quantum states. The stabilized divergence radius wrt subsystem $B$ of $\mathcal{S}_0$ is given by \[\eta_B(\mathcal{S}_0)=\inf_{\sigma_B\in \mathcal{M}_{d_B}} \Big\{ \sup_{\rho_{AB}\in \mathcal{S}_0} D(\rho_{AB}||\rho_A \otimes \sigma_B )\Big\}. \]
\end{definition}

\begin{theorem}\label{thm-stab-div-rad}
For a quantum channel $T:\mathcal{M}_{d_{A'}}\rightarrow \mathcal{M}_{d_B}$, we have that its entanglement-assisted capacity is given by
\[C^{ea}(T)=\eta_B(\Im(\id\otimes T)) =\inf_{\sigma_B \in \mathcal{M}_{d_B} } \Big\{ \sup_{\tau_{AB}\in \Im(\id \otimes T)} D(\tau_{AB}||\tau_A \otimes \sigma_B ) \Big\}. \]
\end{theorem}

\begin{proof}
Let $T:\mathcal{M}_{d_{A'}}\rightarrow \mathcal{M}_{d_B}$ be a quantum channel. \cite{BSST02} gives a one-letter formula for the entanglement-assisted capacity of $T$ in terms of the mutual information, which we can formulate as a relative entropy in the following way:
\begin{equation*}
\begin{split}
  C^{ea}(T)&=\sup_{\rho_{AA'}\text{ pure}} I(A':B)_{(\id \otimes T) \rho_{AA'} }
\\&=\sup_{\rho_{AA'}\text{ pure}} D((\id \otimes T) \rho_{AA'}|| \rho_A \otimes T(\rho_{A'} ) ) 
\\&=\sup_{\rho_{AA'}\text{ pure}} \Big\{ \inf_{\sigma_B \in \mathcal{M}_{d_B}} D((\id \otimes T) \rho_{AA'}|| \rho_A \otimes \sigma_B  ) \Big\} .
\end{split}
\end{equation*}

This divergence can be rewritten as follows:
\begin{equation*}
\begin{split}
   & D((\id \otimes T) (\rho_{AA'})|| \rho_A \otimes \sigma_B ) )  \\& = 
   \trace \Big( (\id \otimes T) (\rho_{AA'}) \big( \log ( (\id \otimes T) (\rho_{AA'}) )-\log ( \rho_A \otimes \sigma_B) \big) \Big) 
   \\& 
   =\trace \Big( (\id \otimes T) (\rho_{AA'})   \log ( (\id \otimes T) (\rho_{AA'}) ) \Big) \\&\qquad - \trace \Big( \rho_A \log( \rho_A )   \Big) - \trace\Big( T (\rho_{A'}) \log ( \sigma_B )\Big)
   \\& =  - H(B|A)_{(\id \otimes T) (\rho_{AA'}) } - \trace\Big( T(\rho_{A'}) \log ( \sigma_B )\Big).
\end{split}
\end{equation*}

The first equality follows from the definition of the relative entropy, the second inequality holds because the state in the second logarithm is separable, and therefore the subparts commute. The third equality follows from the linearity of the trace and the last inequality follows from the definition of the conditional quantum entropy.

This expression is convex in $\sigma_B$, as the matrix logarithm is concave.
The expression $-\trace\Big( T(\rho_{A'}) \log ( \sigma _B)\Big)$ is linear in $\rho_{A'}$. The negative conditional entropy $-H(B|A)_{(\id \otimes T) (\rho_{AA'})}$ is concave in $\rho_{A'}$, as the conditional quantum entropy is convex in the marginal state, see Lemma~\ref{thm-convex-cond-entropy}.

Therefore, we can exchange the supremum and the infimum, and obtain
\begin{equation*}
\begin{split}
C^{ea}(T)&=\inf_{\sigma_B \in \mathcal{M}_{d_B}}\Big\{ \sup_{\rho_{AA'}\in \mathcal{M}_{d_A}\otimes \mathcal{M}_{d_{A'}} }   D((\id \otimes T) \rho_{AA'} ||  \rho_A\otimes \sigma_B ) \Big\}\\&= \inf_{\sigma_B \in \mathcal{M}_{d_B}} \Big\{ \sup_{\tau_{AB}\in \Im(\id \otimes T)} D(\tau_{AB} ||   \tau_{A}\otimes\sigma_B)\Big\} 
\\&=\eta_B(\Im(\id\otimes T)).
\end{split}
\end{equation*}
\end{proof}

\begin{corollary}
For any quantum state $\sigma_B\in \mathcal{M}_{d_B}$, let $R_{\sigma}:\mathcal{M}_{d_{A'}}\rightarrow \mathcal{M}_{d_B}$ be a quantum channel such that $R_{\sigma}(\rho_{A'}) = \sigma_B\forall \rho_{A'}\in\mathcal{M}_{d_{A'}}$. Then, for any quantum channel $T:\mathcal{M}_{d_{A'}} \rightarrow \mathcal{M}_{d_B}$, we have
\[C^{ea}(T)=\inf_{\sigma_B\in\mathcal{M}_{d_B}} D_{stab}(T||R_{\sigma}).\]
\end{corollary}

The divergence of two channels is known to be additive for replacer channels in the second argument \cite[Lemma~11]{CMW16}.

\subsection{The divergence center and the stabilized divergence center for some families of channels}

Due to data processing inequality, for all channels $T$, we recover the well-known fact that the entanglement-assisted capacity is greater or equal to the Holevo capacity:
\begin{equation*}
    \begin{split}
       C^{ea}(T) & = \sup_{\rho_{AB}\in \Im(\id\otimes T)} D(\rho_{AB}|| \rho_A \otimes \sigma_B^*)\\&\geq \sup_{\rho_{AB}\in \Im(\id\otimes T)}  D(\rho_B   || \sigma_B^*) \\&  = \sup_{\rho_B \in \Im(T) }  D(\rho_B || \sigma_B^*)  \\& \geq \inf_{\sigma\in\mathcal{M}_{d_B}} \sup_{\rho_B \in \Im(T) }  D(\rho_B || \sigma_B)  \\&=C_H(T).
    \end{split}
\end{equation*}

Beyond this observation, there are some families of quantum channels for which we can make more precise statements about the centers. For some quantum channels, e.g. depolarizing qubit channels in Example~\ref{ex-depo-channel-center}, the divergence center and the stabilized divergence center are given by the same state. However, there are also examples of quantum channels where the two centers are not the same. One family of channels where this is the case is the qubit amplitude damping channel for damping parameters $\gamma\in (0,1)$, which we examine in more detail in Example~\ref{ex-ampdamp-center}.

\begin{example}[Depolarizing qubit channels]\label{ex-depo-channel-center}
    Let $T_p(\rho)=(1-p)\rho+\frac{p}{4}(\rho+X\rho X +Y\rho Y+ Z \rho Z)$ denote the qubit depolarizing channel with depolarizing parameter $p$, which is completely positive for $p\in [0,\frac{4}{3}]$.
    
    In the Bloch sphere, the channel image of the qubit depolarizing channel is a contracted sphere where the radius decreases for increasing parameter $p$. For all $p$, the center of this sphere is the center of the Bloch sphere, i.e. the maximally mixed state $\frac{\mathbbm{1}}{2}$. The peripheral points are the points on the surface of the contracted ball. One example of a peripheral point is the image of $\ket{0}\bra{0}$, given by $(1-\frac{p}{2})\ket{0}\bra{0}+\frac{p}{2}\ket{1}\bra{1}$. 
    Inserting these two points, the divergence radius is given by $\chi(T_p)=D((1-\frac{p}{2})\ket{0}\bra{0}+\frac{p}{2}\ket{1}\bra{1}||\frac{\mathbbm{1}}{2})= (1-\frac{p}{2})\log(2(1-\frac{p}{2}))+\frac{p}{2}\log(2\frac{p}{2}) =1+(1-\frac{p}{2})\log(1-\frac{p}{2}) + \frac{p}{2}\log(\frac{p}{2})$, which is precisely the classical capacity of the qubit depolarizing channel as found in \cite{BSST99}.

    For the entanglement-assisted capacity, consider the image of $\id_2 \otimes T_p$. Due to symmetry, the center has to be $\frac{\mathbbm{1}}{2}$. 
    One peripheral point is obtained by inputting the maximally entangled state, which is mapped to the following state: $(\id \otimes T_p) (\phi_+)=(1-\frac{3p}{4})\phi_+ + \frac{p}{4} (\phi_- +\psi_+ + \psi_-)$ with the Bell basis states. The marginals of this state are $\rho_A=\rho_B=\frac{\mathbbm{1}}{2}$, and therefore $C^{ea}(T_p)=D((\id \otimes T_p) (\phi_+)||\frac{\mathbbm{1}}{2}\otimes \frac{\mathbbm{1}}{2})=(1-\frac{3p}{4})\log((1-\frac{3p}{4})/(1/4))+\frac{3p}{4}\log(4\frac{3p}{4})=2+(1-\frac{3p}{4})\log(1-\frac{3p}{4}) + \frac{3p}{4}\log(\frac{p}{4})$, which is precisely the entanglement-assisted capacity of the qubit depolarizing channel as found in \cite{BSST99}.
\end{example}


There also exist channels where the divergence center and the stabilized divergence center do not agree:

\begin{example}[Amplitude damping qubit channels]\label{ex-ampdamp-center}
By symmetry, the divergence center and the stabilized divergence center are on the $Z$-axis of the Bloch sphere. We parameterize the divergence center by $\sigma_{div}=\begin{pmatrix}
s_{div} & 0\\ 0&1-s_{div }\end{pmatrix}$ for $s_{div}\in [0,1]$, and the stabilized divergence center by $\sigma_{stab}=\begin{pmatrix}
s_{stab} & 0\\ 0&1-s_{stab}\end{pmatrix}$ for $s_{stab}\in[0,1]$.


For the Holevo divergence center, we parameterize the channel's output states for pure input states $\rho=\begin{pmatrix} a & \sqrt{a(1-a)}\\ \sqrt{a(1-a)} &1-a\end{pmatrix}$ by \[T_{\gamma}(\rho)=\begin{pmatrix}a+\gamma (1-a)&-\sqrt{(1-\gamma)a(1-a)}\\-\sqrt{(1-\gamma)a(1-a)} & (1-\gamma)(1-a) \end{pmatrix}.\]  

For the entanglement-assisted capacity, we parametrize the channel's output for pure input states with Schmidt coefficients $a$ and $\sqrt{1-a^2}$ by 
\[(\id\otimes T_{\gamma}(\rho))=\begin{pmatrix} a^2 & \gamma a \sqrt{(1-a^2)} & 0 & a\sqrt{(1-\gamma)(1-a^2)} \\ 0&0&0&0\\ 0&0& \gamma(1-a^2) &0\\a\sqrt{(1-\gamma)(1-a^2)}&0&0 &(1-\gamma)(1-a^2)\end{pmatrix} .\]

For a fixed damping parameter, we minimze $s_{div}$ and $s_{stab}$ for all $a\in [0,1]$ to describe the divergence center and the stabilized divergence center, respectively. The results are plotted in Figure~\ref{fig-ampdamp-center} and clearly show that the center and the stabilized center are only in agreement for damping parameter $\gamma\rightarrow 0$ (where the center for both would be the maximally mixed state, where $s_{div}=s_{stab}=\frac{1}{2}$) and $\gamma\rightarrow 1$ (for which the center becomes $\ketbra{0}{0}$, i.e. $s_{div}=s_{stab}=1$). For damping parameters $\gamma\in(0,1)$, they clearly differ, and the divergence center consistently has a larger largest eigenvalue than the stabilized divergence center. In a Bloch sphere picture, this corresponds to the divergence center lying higher up on the z-axis than the stabilized center. Once $s_{div/stab}$ is found for a given damping parameter, we can find the corresponding parameter $a$ and obtain the capacities - the results are plotted in Figure~\ref{fig-ampdamp-cap} and are consistent with previous results in \cite{BSST02}.

\begin{figure}
\centering
\begin{subfigure}{.7\textwidth}
  \centering
  \includegraphics[width=\linewidth]{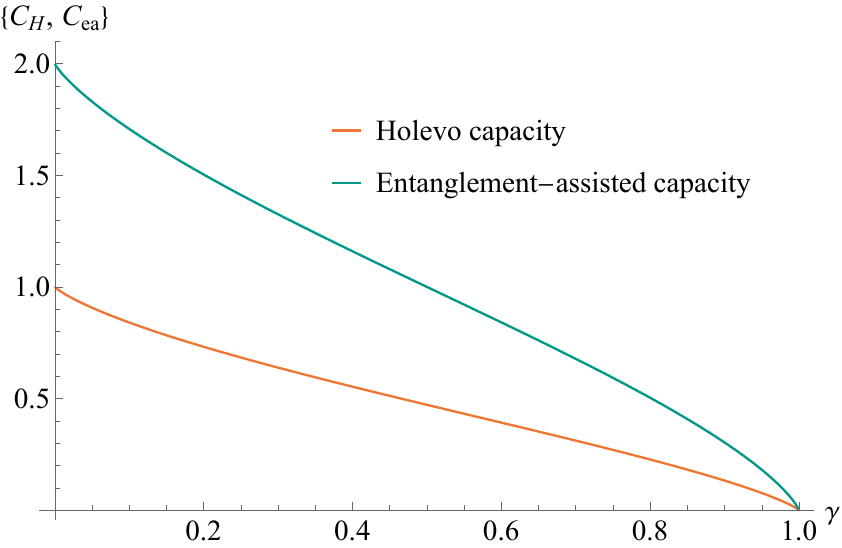}
  \caption{Holevo capacity and entanglement-assisted capacity for the qubit amplitude damping channel. These findings match with the previous known results from \cite{BSST02}.}
  \label{fig-ampdamp-cap}
\end{subfigure}%
\hspace{1cm}
\begin{subfigure}{.7\textwidth}
  \centering
  \includegraphics[width=\linewidth]{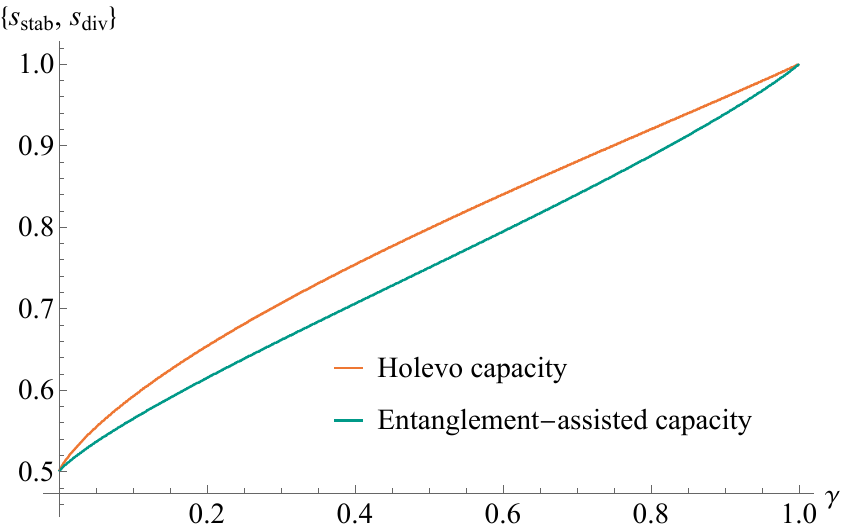}
  \caption{Largest eigenvalue of the divergence center and the stabilized divergence center for the qubit amplitude damping channel. In the limit $\gamma\rightarrow 1$, both approach the state $\ketbra{0}{0}$, i.e. $s\rightarrow 1$.}
  \label{fig-ampdamp-center}
\end{subfigure}
\caption{Numerical results for the qubit amplitude damping channel.}
\label{fig:test}
\end{figure}
\end{example}

\begin{example}[Entanglement-breaking channels]
   
    Let $T:\mathcal{M}_{d_{A'}}\rightarrow \mathcal{M}_{d_B}$ be an entanglement-breaking channel.
    For entanglement-breaking channels, for any $\rho_{AA'}\in\mathcal{M}_{d_A}\otimes \mathcal{M}_{d_{A'}}$, there exist a probability distribution $p_k$ and $k$ quantum states $\tau_k\in \mathcal{M}_{d_A}$ and $k$ quantum states $\omega_k\in \mathcal{M}_{d_B}$ such that $\id \otimes T (\rho_{AA'})=\sum_k p_k \tau_k\otimes \omega_k$.
Then, we have
\begin{equation*}\begin{split}
       &  D( \sum_k p_k \tau_k\otimes \omega_k  ||\sum_k p_k   \tau_k \otimes \sigma )\\& 
       \leq \sum_k p_k D(\tau_k\otimes \omega_k || \tau_k \otimes \sigma ) \\&= \sum_k p_k D( \omega_k  || \sigma ).
    \end{split}
    \end{equation*}
The states $\omega_k$ and probability distributions $p_k$ depend on the input state that maximizes the expression for a given channel. The state $\sum_k p_k \omega_k$ is in the channel image, but not necessarily its summands. If all $\omega_k$ (for which $p_k>0$) are states in the channel image, then the above cannot exceed $\sup_{\rho \in \Im(T)} D(\rho||\sigma )= C_H(T)$, and we have equality of the Holevo capacity and the entanglement-assisted capacity. In fact, since the Holevo capacity is equal to the classical capacity for entanglement-breaking channels \cite{Shor02_2}, we would have equality of the classical capacity and the entanglement-assisted capacity.

Note therefore that there exist entanglement-breaking channels with entanglement-assisted capacity that is not equal to the Holevo quantity; one example is the depolarizing channel from Example~\ref{ex-depo-channel}, which is entanglement breaking for depolarizing probability $p\geq \frac{d}{d+1}$ \cite[Appendix]{Hol12} and has $C_{ea}(T_p)=(d+1) C_H(T_p)$ for $p=1$ \cite{BSST99}. A more detailed discussion and further examples of this phenomenon can be found in \cite{Hol12}.
\end{example}

\begin{example}[Classical-quantum channels]
Classical-quantum channels are entanglement-breaking. For classical-quantum channels, the probability distribution for the channel applied to $\rho$ is given by $p_k=\trace( \ketbra{k}{k} \rho_B \ketbra{k}{k})=\bra{k}\rho_B \ket{k}$. Therefore, the channel image contains all probability distributions (the eigenvalues of any $\rho$), including the distribution $\{1, 0,0,...,0\}$. In other words, let $\rho=\ketbra{0}{0}$; then the channel maps this to $ \omega_0$, and similarly for all other $k$. Therefore, for all $k$, $\omega_k$ is indeed in the channel image, 
and we have $C^{ea}(T)=C_H(T)$, which reproduces the results from
\cite{Hol12} and \cite{Shi12}.
\end{example}

\section{Lower bounds and upper bounds on quantum channel divergences}
\label{sec-bounding}

As we have found in Section~\ref{sec-stab-div-rad}, the Holevo capacity and the entanglement-assisted capacity have expressions in terms of a channel divergence and a stabilized channel divergence with a replacer channel in the second argument.
In this section, we show that the (stabilized) channel divergence in question can be bounded in terms of the squared $1\rightarrow 1$-distance between the channel and the associated replacer channel.
This norm fulfills a stabilization property which allows us to relate the stabilized distance measure to the non-stabilized version.

At the basis of our analysis in Section~\ref{sec-taylors-thm} lies a continuity bound on the entanglement-assisted capacity in the vicinity of a replacer channel. For any replacer channel, its entanglement-assisted capacity is zero. Furthermore, for a sequence of channels containing the replacer channel, the capacity has a minimum at the replacer channel, and the first derivative of the capacity in that point is zero. Then, we show that we can find an upper bound on the second derivative that does not depend on the channel input, and only on the replacer channel output, in Theorem~\ref{thm-second-deriv}. This allows us to make use of Taylor's theorem in order to upper bound the entanglement-assisted capacity in the vicinity of the replacer channel by an expression that is quadratic in the distance between the channels in Theorem~\ref{thm-taylor}.

\subsection{Stabilization of the induced Schatten norm and a lower bound on the channel divergence}



We recall the definitons of the Schatten $p$-norm for matrices and the induced Schatten $q\rightarrow p$-norm for linear maps introduced in Section~\ref{sec-measures-of-distance-and-information}. For some combinations of $p$ and $q$, the induced Schatten norm fulfills a special property:

\begin{theorem}[{Stabilization, \cite[Theorem~4]{Watrous04}}]
For any linear map $\Lambda:\mathcal{M}_{d_{A'}}\rightarrow \mathcal{M}_{d_B}$, and any $p\geq 2$ and $q\leq 2$ we have
    \[\|\Lambda\otimes \id_R\|_{q\rightarrow p} = \|\Lambda\|_{q\rightarrow p} \]
    where $\id_R$ denotes the identity map on $ \mathcal{M}_{d_R}$.
\end{theorem}

and subsequently, because of the relations between $p$-norms on finite dimensional spaces \cite{Johnson12},

\begin{corollary}
\label{thm-norm-stab}
 For any linear map $\Lambda:\mathcal{M}_{d_{A'}}\rightarrow \mathcal{M}_{d_B}$, we have
\[\|\Lambda\otimes \id_{d_R}\|_{1\rightarrow 1} \leq d_B\|\Lambda\|_{1\rightarrow 1}. \]
\end{corollary}

This property enables us to connect a stabilized distance to a non-stabilized distance, but an analogue is not known to hold for quantum channel divergences. However, it is known that the Schatten 1-norm is related to the quantum divergence of two states by quantum Pinsker's inequality in Theorem~\ref{thm-pinsker-ineq}. This implies a lower bound on the channel divergence of two channels in terms of the squared induced $1\rightarrow 1$-norm.

In Section~\ref{sec-taylors-thm}, we now prove a matching upper bound for the stabilized channel divergence.

\subsection{Taylor's theorem for the stabilized channel divergence}
\label{sec-taylors-thm}

Let $f(x):\mathbbm{R}^n \rightarrow \mathbbm{R}$ be a convex function that is twice differentiable. If the second derivative can be upper-bounded in all points i.e. if $\exists M<\infty$ where $\frac{\dd^2 f(x)}{\dd^2 x}\leq M$ for all $x$, then Taylor's theorem guarantees that
\[f(x)\leq f(x_0) + \frac{\dd f(x)}{\dd x}\Big|_{x=x_0} (x-x_0) + M (x-x_0)^2 .\]

Here, we use Taylor's theorem to bound the mutual information as a function of the quantum channel output, and evaluated in the point of a replacer channel. Thereby, we obtain an upper bound on the mutual information that is quadratic in the distance between the replacer channel $R_{\sigma}$ and any quantum channel in the vicinity of this replacer channel in Theorem~\ref{thm-taylor}.

In order to show this, we compute the first derivative and upper-bound the second derivative of the related stabilized channel divergence.

\begin{remark}
There have been other works on quadratic upper bounds of the relative entropy, most notably \cite[Theorem~2]{AE05}, which provides a quantum analogue of the work of Sason \cite{Sason15} on classical reverse Pinsker inequalities. However, these bounds depend on the minimal eigenvalue of the second argument of the relative entropy. In the case of the stabilized divergence radius, this would be a product $\rho_{A}\otimes \sigma_B$ of the marginal of the best input ensemble state $\rho_{AB}$ and the stabilized divergence center $\sigma$. Due to the maximization over $\rho_{AB}$, no guarantee can be given on the smallest eigenvalue of its marginal. In the case of mutual information, Theorem~\ref{thm-taylor} therefore proves an alternative formulation of a "reverse Pinsker inequaltiy" for mutual information which only depends on the minimal eigenvalue of the stabilized divergence center.
\end{remark}



 

\begin{lemma}[First derivative of the channel divergence] \label{thm-first-deriv}
For a quantum state $\sigma_B\in \mathcal{M}_{d_B}$, let $R_{\sigma}:\mathcal{M}_{d_{A'}}\rightarrow \mathcal{M}_{d_B}$ be a quantum channel such that $R_{\sigma}(\rho_{A'}) = \sigma_B\forall \rho_{A'}\in\mathcal{M}_{d_{A'}}$. For any quantum channel $S:\mathcal{M}_{d_{A'}}\rightarrow \mathcal{M}_{d_B}$, we define a quantum channel $T_t:\mathcal{M}_{d_{A'}}\rightarrow \mathcal{M}_{d_B}$ with $T_t=(1-t)R_{\sigma}+t S$ for any $t\in[0,1]$, and we have
\[\frac{\dd}{\dd t} D_{stab}(T_t ||R_{\sigma})\Big|_{t=0} =0.\]
\end{lemma}

\begin{proof}
We compute the derivative using Theorem~\ref{thm-entropy-derivatives}, and obtain 
\begin{equation*}
    \begin{split}
      &  \frac{\dd}{\dd t}  D(\id\otimes T_t(\rho_{AA'})||\id\otimes R_{\sigma}(\rho_{AA'}))\\
        &= \frac{\dd}{\dd t} H(A)_{\rho_{A}}+\frac{\dd}{\dd t} \trace\big( T_t(\rho_{A'}) \log(R_{\sigma}(\rho_A')) \big)-\frac{\dd}{\dd t}H(AB)_{ (\id\otimes T_t) (\rho_{AA'}) }  \\&
        = -\trace\Big( (S-R_{\sigma})(\rho_{A'}) \log (R_{\sigma}(\rho_{A'})) \Big) \\&\qquad \qquad +\trace\Big( (\id\otimes(S -R_{\sigma}))(\rho_{AA'}) \log ((\id\otimes T_t)(\rho_{AA'})) \Big)
        \\&
        = \trace\Big( (\id\otimes(S -R_{\sigma}))(\rho_{AA'}) \big(  \log ((\id\otimes T_t)(\rho_{AA'})) -\mathbbm{1}\otimes \log (R_{\sigma}(\rho_{A'}))  \big) \Big)
        \\&=\trace\Big( (\id\otimes S(\rho_{AA'} )-\rho_A\otimes \sigma_B) \big(\mathbbm{1}\otimes \log ((1-t)\sigma_B +tS(\rho_{A'})) \\&\qquad \qquad -\log ((1-t)\rho_A\otimes\sigma_B+t(\id\otimes S)(\rho_{AA'})) \big) \Big) .
 \end{split}
\end{equation*}
For $t\rightarrow 0$, and for all quantum states $\rho_{AA'}\in\mathcal{M}_{d_A}\otimes \mathcal{M}_{d_{A'}}$, we have 
\[\frac{\dd}{\dd t} D(\id\otimes T_t(\rho_{AA'})||\id\otimes R_{\sigma}(\rho_{AA'}))\Big|_{t=0} =0\] because $\id\otimes (S-T_0)(\rho_{AA'})$ is traceless $\forall \rho_{AA'}$.
This is also true for any state $\rho_{AA'}$ which maximizes the expression, and therefore also true for the stabilized channel divergence.
\end{proof}

\begin{lemma}[Second derivative of the stabilized divergence]
\label{thm-second-deriv} 
For a quantum state $\sigma_B\in \mathcal{M}_{d_B}$, let $R_{\sigma}:\mathcal{M}_{d_{A'}}\rightarrow \mathcal{M}_{d_B}$ be a quantum channel such that $R_{\sigma}(\rho_{A'}) = \sigma_B\forall \rho_{A'}\in\mathcal{M}_{d_{A'}}$. For any quantum channel $S:\mathcal{M}_{d_{A'}}\rightarrow \mathcal{M}_{d_B}$, we define a quantum channel $T_t:\mathcal{M}_{d_{A'}}\rightarrow \mathcal{M}_{d_B}$ with $T_t=(1-t)R_{\sigma}+t S$ and we have
\[\frac{\dd^2}{\dd t^2} D_{stab}(T_t||R_{\sigma})
        \leq \frac{2d_B^3}{(1-t)\lambda_{min}(\sigma)} .\] 
\end{lemma}

\begin{proof}
By definition,
\begin{equation*}
    \begin{split}
      & \frac{\dd^2}{\dd t^2} D(\id\otimes T_t(\rho_{AA'})||\id\otimes R_{\sigma}(\rho_{AA'}))\\&=\frac{\dd^2}{\dd t^2}  H(A)_{\rho_A} + \frac{\dd^2}{\dd t^2}   \trace\Big( T_t(\rho_{A'})\log\big(R_{\sigma}(\rho_{A'})\big)\Big)- \frac{\dd^2}{\dd t^2}  H(AB)_{(\id\otimes T_t)(\rho_{AA'})} .
    \end{split}
\end{equation*}

The second derivative of $H(A)_{\rho_{A}}$ is zero as the expression does not depend on $t$.
The second derivative of $\trace\Big( T_t(\rho_{A'})\log\big(R_{\sigma}(\rho_{A'})\big)\Big)$ is zero, as the expression is linear in $t$. Therefore, the second derivative of the relative entropy above is given by the second derivative of $-H(AB)_{(\id\otimes T_t)(\rho_{AA'})}$, which is computed in Lemma~\ref{thm-entropy-derivatives}.

For the replacer channel, we thus have
\begin{equation*}
    \begin{split}
       &\frac{\dd^2}{\dd t^2} D(\id\otimes T_t(\rho_{AA'})||\id\otimes R_{\sigma}(\rho_{AA'})) \\&
\leq \trace\Big( \big( (\id\otimes S)(\rho_{AA'})-\rho_A\otimes\sigma_B\big)   \Compactcdots\\&  \qquad \Compactcdots \int_0^{\infty} ((1-t)\rho_A\otimes\sigma_B+t(\id\otimes S)(\rho_{AA'})+z\mathbbm{1})^{-1} \big( (\id\otimes S)(\rho_{AA'})-\rho_A\otimes\sigma_B\big)  \Compactcdots\\&  \qquad \Compactcdots  ((1-t)\rho_A\otimes\sigma_B+t(\id\otimes S)(\rho_{AA'})+z\mathbbm{1})^{-1} \dd z    \Big) 
        \\& \leq 2d_B^2 \trace\Big( \big( \rho_A\otimes\mathbbm{1}\big)  \int_0^{\infty} ((1-t)\rho_A\otimes\sigma_B+t(\id\otimes S)(\rho_{AA'})+z\mathbbm{1})^{-1}\Compactcdots\\& \qquad \qquad \Compactcdots \big( \rho_A\otimes\mathbbm{1}\big) ((1-t)\rho_A\otimes\sigma_B+t(\id\otimes S)(\rho_{AA'})+z\mathbbm{1})^{-1} \dd z    \Big) 
        \\& \leq 2d_B^2 \trace\Big( \big( \rho_A\otimes\mathbbm{1}\big)   \int_0^{\infty} ((1-t)\rho_A\otimes\sigma_B+z\mathbbm{1})^{-1}\Compactcdots\\& \qquad \qquad \Compactcdots \big( \rho_A\otimes\mathbbm{1}\big) ((1-t)\rho_A\otimes\sigma_B+z\mathbbm{1})^{-1} \dd z    \Big) 
        \\& \leq 2d_B^2 \trace\Big( \big( \rho_A\otimes\mathbbm{1}\big)^2   ((1-t)\rho_A\otimes\sigma_B)^{-1}   \Big) \\&
        \leq \frac{2d_B^2}{1-t} \trace\Big( \rho_A\otimes\sigma_B^{-1}   \Big)\\&
        \leq \frac{2d_B^2}{(1-t)\lambda_{min}(\sigma)} \trace\Big( \rho_A\otimes\mathbbm{1}   \Big)
        \\&
        \leq \frac{2d_B^3}{(1-t)\lambda_{min}(\sigma)}. 
    \end{split}
\end{equation*}
For the second inequality, we are upper bounding the terms with negative signs by 0, and upper bounding $(\id\otimes S)(\rho_{AA'})\leq d_B \rho_A \otimes \mathbbm{1}$ and $(\id\otimes R_{\sigma})(\rho_{AA'})\leq d_B \rho_A \otimes \mathbbm{1}$. 
Furthermore, for any two matrices $A$ and $B$ where $B$ is positive semidefinite, we have $(A+B)^{-1}\leq A^{-1}$ as a consequence of the Woodbury matrix identity \cite{HS81}.
Using this, we bound $((1-t)\rho_A\otimes\sigma_B+t(\id\otimes S)(\rho_{AA'})+z\mathbbm{1})^{-1}\leq ((1-t)\rho_A\otimes\sigma_B+z\mathbbm{1})^{-1}$. Then, the states commute and we can solve the integral. Since $\rho_A^{-1}$ and $\rho_A$ cancel, only the smallest eigenvalue of $\sigma_B$ appears in our bound.

Since the bound above holds for any state $\rho_{AA'}$, it also holds for states that maximize the expression, and thus for the stabilized channel divergence.
\end{proof}

With these derivatives, we can upper bound the stabilized channel divergence using Taylor's theorem:

\begin{theorem}[Taylor's theorem in the vicinity of the replacer channel] \label{thm-taylor}
For a quantum state $\sigma_B\in \mathcal{M}_{d_B}$, let $R_{\sigma}:\mathcal{M}_{d_{A'}}\rightarrow \mathcal{M}_{d_B}$ be a quantum channel such that $R_{\sigma}(\rho_{A'}) = \sigma_B\forall \rho_{A'}\in\mathcal{M}_{d_{A'}}$.
For any quantum state $\sigma_B\in\mathcal{M}_{d_B}$ and any quantum channel $T:\mathcal{M}_{d_{A'}}\rightarrow \mathcal{M}_{d_B}$ with $\| \id\otimes ( T-R_{\sigma} ) \|_{1\rightarrow 1} \leq \frac{1}{2}$, we have
    \[ D_{stab}(T|| R_{\sigma})
       \leq \frac{2d_B^3}{\lambda_{min}(\sigma)}  \| \id\otimes ( T-R_{\sigma} ) \|_{1\rightarrow 1}^2 .\]
\end{theorem}

\begin{proof}
Let $T_t=(1-t)R_{\sigma}+tS$ be a family of quantum channels for $0\leq t \leq 1$ and for $S\neq R_{\sigma}$. For all $T_t$, we have that the second derivative of the mutual information is bounded by $\frac{\dd^2}{\dd t^2} D(\id\otimes T_t(\rho_{AA'})||\id\otimes R_{\sigma}(\rho_{AA'}))\leq \frac{2d_B^3}{(1-t)\lambda_{min}(\sigma)}$ from Theorem~\ref{thm-second-deriv}. For any $t\leq \frac{1}{2}$, we thus have $\frac{\dd^2}{\dd t^2} D(\id\otimes T_t(\rho_{AA'})||\id\otimes R_{\sigma}(\rho_{AA'}))\leq \frac{4d_B^3}{\lambda_{min}(\sigma)}$.

Then, we expand the mutual information using Taylor's theorem, and we obtain
\begin{equation*}
    \begin{split}
 & D(\id\otimes T_t(\rho_{AA'})||\id\otimes R_{\sigma}(\rho_{AA'}))
   \\ &
     \leq  D(\id\otimes T_t(\rho_{AA'})||\id\otimes R_{\sigma}(\rho_{AA'}))\Big|_{t=0} \\&\qquad +t \frac{\dd}{\dd t}D(\id\otimes T_t(\rho_{AA'})||\id\otimes R_{\sigma}(\rho_{AA'}))\Big|_{t=0} +\frac{t^2}{2} \frac{4d_B^3}{\lambda_{min}(\sigma)}\\&
       \leq t^2 \frac{2d_B^3}{\lambda_{min}(\sigma)} .
    \end{split}
\end{equation*}
where the last inequality holds because the relative entropy is faithful, and the first derivative of the mutual information is zero at $t=0$ by Theorem~\ref{thm-first-deriv}.

We furthermore have $ t=\| \id\otimes( T_t-R_{\sigma} ) \|_{\infty\rightarrow 1} \leq \| \id\otimes( T_t-R_{\sigma} ) \|_{1\rightarrow 1}$ for the traceless map $T_t-R_{\sigma}$ \cite[Eq.~(10)]{AE11}.



In total, we therefore obtain
\begin{equation*}
    \begin{split}
& \sup_{\rho_{AA'}} D(\id\otimes T_t(\rho_{AA'})||\id\otimes R_{\sigma}(\rho_{AA'}))
    \\& =D_{stab}(T_t||R_{\sigma} )\\& \leq \frac{2d_B^3}{\lambda_{min}(\sigma)}  \| \id\otimes (T_t-R_{\sigma}) \|_{1\rightarrow 1}^2.
    \end{split}
\end{equation*}
\end{proof}

We set $T_t=T$ in the following.

\begin{corollary} By combining Theorem~\ref{thm-taylor} and Lemma~\ref{thm-norm-stab}, for any quantum channel $T:\mathcal{M}_{d_{A'}}\rightarrow \mathcal{M}_{d_B}$ with $\|  T -R_{\sigma} \|_{1\rightarrow 1} \leq \frac{1}{2d_B}$, we have
    \[D_{stab}(T|| R_{\sigma})
       \leq \frac{2d_B^5}{\lambda_{min}(\sigma)}  \|  T-R_{\sigma} \|_{1\rightarrow 1}^2.\]
\end{corollary}

\begin{proof}
By Lemma~\ref{thm-norm-stab}, $\|  T -R_{\sigma} \|_{1\rightarrow 1} \leq \frac{1}{2d_B}$  implies that $ \| \id\otimes ( T-R_{\sigma} ) \|_{1\rightarrow 1} \leq d_B\|  T -R_{\sigma} \|_{1\rightarrow 1} \leq \frac{1}{2}$, and thus the condition of Theorem~\ref{thm-taylor} applies. Then, further using Lemma~\ref{thm-norm-stab} again, we have $ \| \id\otimes (T_t-R_{\sigma}) \|_{1\rightarrow 1}^2 \leq d_B^2 \|  T-R_{\sigma} \|_{1\rightarrow 1}^2$. Then, inserting this into Theorem~\ref{thm-taylor} gives the result.
\end{proof}

\section{An upper bound on the quotient of capacities}
\label{sec-bound-on-the-quotient}


In Section~\ref{sec-taylors-thm}, we show that entanglement-assisted capacity in the vicinity of the replacer channel is bounded by an expression that is quadratic in the stabilized channel distance in Theorem~\ref{thm-taylor}.
This is a key component, as the quadratic behaviour matches the quantum Pinsker inequality in Theorem~\ref{thm-pinsker-ineq}, which is a well-known lower bound on the Holevo capacity, and because the $1\rightarrow 1$-norm fulfills a crucial stabilization property due to Theorem~\ref{thm-norm-stab}. Therefore, we can lower bound the Holevo capacity in terms of a quadratic channel distance, and upper bound the entanglement-assisted capacity in terms of the quadratic same channel distance, which allows us to bound the ratio of capacities by a finite, dimension-dependent function in Theorem~\ref{thm-the-quotient-bound}.

\subsection{A bound for channels that are far from a replacer channel}

We emphasize that the ratio of capacities can only diverge when the classical capacity becomes zero. For non-zero classical capacity, say $C(T)= \epsilon$, the ratio is automatically upper-bounded by $2\log(d)/\epsilon$ because the entanglement-assisted capacity is upper-bounded by $2\log(d)$. Problems arise in the case of $\epsilon\rightarrow 0$. Since the classical capacity is equal to a maximal mutual information with cq-states as input (Theorem~\ref{thm-class-cap-1}), which is faithful, this is only the case for channels where information is completely destroyed by the channel, i.e. channels that replace any input state with the same quantum state $\sigma$, $R_{\sigma} (\rho) = \sigma \tr(\rho)$, which we refer to as a replacer channel. Therefore, the points where the capacity quotient could diverge are channels in the vicinity of a replacer channel. This is made more precise by the following:

\begin{theorem}[{\cite[Eq.~12.103]{NC00}}] For any quantum channel $T:\mathcal{M}_{d_{A'}} \rightarrow \mathcal{M}_{d_B}$, let $\sigma_B^*\in \mathcal{M}_{d_B}$ denote its divergence center. If $\| T-R_{\sigma^*} \|_{1\rightarrow 1} > \frac{1}{2d_B}$, we have that $C_H(T)>\frac{1}{2d_B}$, and therefore
\[  C^{ea}(T) \leq  g(d_B) C_H(T) .\]
with $g(d_B)=4d_B\log(d_B)$.
\end{theorem}

Due to this, we will consider channels $T$ that are close to a replacer channel in the remainder of this section.

\subsection{A bound in the vicinity of a replacer channel} 

Using Theorem~\ref{thm-taylor} and Theorem~\ref{thm-pinsker-ineq}, we can bound the quotient of capacities for channels near the associated replacer channel by the following expression:

\begin{theorem}\label{thm-the-quotient-bound} For any quantum channel $T:\mathcal{M}_{d_{A'}} \rightarrow \mathcal{M}_{d_B}$, let $\sigma_B^*\in \mathcal{M}_{d_B}$ denote its divergence center. If $\| T-R_{\sigma} \|_{1\rightarrow 1} \leq \frac{1}{2d_B}$, we have 
\[  C^{ea}(T) \leq  f(d_B,\lambda_{min}(\sigma^*)) C_H(T) .\]
for a function $f(d_B,\lambda_{min}(\sigma^*))=\frac{4\ln(2)d_B^5}{\lambda_{min}(\sigma^*)} $.
\end{theorem}

\begin{proof}
We use our expression of the entanglement-assisted capacity as a divergence from Theorem~\ref{thm-stab-div-rad}. Since the expression entails a minimization over the marginals that is achieved by the stabilized center, the expression can be upper-bounded by a divergence with respect to the divergence center $\sigma^*_B\in\mathcal{M}_{d_B}$ for the un-assisted divergence radius for the channel $T:\mathcal{M}_{d_{A'}} \rightarrow \mathcal{M}_{d_B}$. Therefore, we have
 \begin{equation*}\begin{split}
    C_{ea} (T) &= \inf_{\sigma\in\mathcal{M}_{d_B}} D_{stab}(T||R_{\sigma})\\&\leq D_{stab}(T||R_{\sigma^*}) \\&
   \leq \frac{2d_B^5}{\lambda_{min}(\sigma^*)}  \|  T-R_{\sigma^*} \|_{1\rightarrow 1}^2 .
    \end{split}
\end{equation*}


On the other hand, we have Pinsker's inequality from Theorem~\ref{thm-pinsker-ineq}:
\begin{align*}
C_H(T) &= \sup_{\rho_{A'}\in\mathcal{M}_{d_{A'}}} D(T(\rho_{A'}) ||\sigma_B^* )\\&\geq  \sup_{\rho_{A'}\in\mathcal{M}_{d_{A'}}} \frac{1}{2\ln(2)} \|T(\rho)-\sigma^* \|_1^2\\&=  \frac{1}{2\ln(2)} \|T-R_{\sigma^*} \|_{1\rightarrow 1}^2 .     
\end{align*}

In total, we obtain
 \begin{equation*}\begin{split}
    C_{ea} (T) &\leq \frac{2d_B^5}{\lambda_{min}(\sigma^*)}  \|  T_t-R_{\sigma^*} \|_{1\rightarrow 1}^2 \\&
    \leq \frac{4\ln(2)d_B^5}{\lambda_{min}(\sigma^*)}  C_H(T).
    \end{split}
\end{equation*}
\end{proof}

In other words, for a divergence center $\sigma^*$ with a bounded minimal eigenvalue $\lambda_{min}(\sigma^*)$, we have that any channel $T$ with $\|T-R_{\sigma^*}\|_{1\rightarrow 1} \leq \frac{1}{2d_B}$ has a bounded quotient with
\[\frac{C^{ea}(T)}{C(T)} \leq \frac{4\ln(2)d_B^5}{\lambda_{min}(\sigma^*)}  . \]

\subsection{A tight bound for qubit unital channels}
\label{sec-pauli}

By our method in Theorem~\ref{thm-taylor}, we can upper bound the quotient of capacities by a dimension-dependent factor. However, this bound is very likely not optimal for most channels, and also lies above the conjectured bound by Bennett et al \cite{BSST02}, suggesting that the bound could be improved significantly. Further illustrating this, we prove a (previously unknown) tight bound for the class of unital qubit channels.

More precisely, since the divergence center of any unital qubit channel $T:\mathcal{M}_2\rightarrow \mathcal{M}_2$ is the maximally mixed state due to symmetry and invariance of the capacities under unitaries, we would have $\lambda_{min}(\sigma_B) =\frac{1}{2}$, and thus $\frac{C^{ea}(T)}{C_H(T)}\leq 4\ln(2)2^6 = 256\ln(2)$. Here, we show that the quotient for unital qubit channels is in fact bounded by $\frac{C^{ea}(T)}{C_H(T)}\leq 5.0798$, and that this novel bound is tight.

It is known that the Holevo capacity for unital channels is additive \cite{King02}, and therefore, for this class of channels, the Holevo capacity is equal to the classical capacity, so that a tight bound on $\frac{C^{ea}(T)}{C_H(T)}$ is also a tight bound on the quotient with the regularized capacity.

\begin{theorem}
For any unital qubit channel $T:\mathcal{M}_2\rightarrow \mathcal{M}_2$, we have
\[\frac{C^{ea}(T)}{C_H(T)} \leq 5.0798 \]
and this bound is tight.
\end{theorem}

\begin{proof}
We will prove this upper bound by proving a sequence of bounds for increasingly more specialized families of quantum channels.

For any unital quantum channel $T:\mathcal{M}_2\rightarrow \mathcal{M}_2$, there exist a qubit Pauli channel $S:\mathcal{M}_2\rightarrow \mathcal{M}_2$ and unitary matrices $U\in \mathcal{M}_2$ and $V\in\mathcal{M}_2$ such that $T=\Lambda_U\circ S\circ \Lambda_V $ where $\Lambda_U:\mathcal{M}_2\rightarrow \mathcal{M}_2$ denotes CPTP the map with a single Kraus operator $U$, and $\Lambda_V:\mathcal{M}_2\rightarrow \mathcal{M}_2$ denotes the CPTP map with a single Kraus operator $V$ \cite{RSW01}. The Holevo capacity and the entanglement-assisted capacity of $T$ and $S$ coincide.
Thus an upper bound on the quotient for all qubit Pauli channels implies that the same upper bound applies to all unital qubit channels.

In Theorem~\ref{thm-pauli3}, we prove that the capacity quotient for qubit Pauli channels is upper-bounded by the quotient for covariant Pauli channels, where two channel parameters coincide.

In Theorem~\ref{thm-pauli2}, we show that the capacity quotient for covariant Pauli channels is upper-bounded by the quotient for a family of channels $S_p\in S_0\subseteq L(\mathcal{M}_2,\mathcal{M}_2)$  for $p\in [0,1]$ where the capacity has a closed form that is well known. For these channels, the quotient is upper bounded by 
 \[ \frac{C^{ea}(S_p)}{C_H(S_p)} \leq \frac{6-3\log(3)}{5-3\log{3}}\leq 5.0798 \]
in Theorem~\ref{thm-pauli1}, giving the upper bound for all unital qubit channels.

This bound is tight because it is achieved by the qubit channel $T_1:\mathcal{M}_2\rightarrow \mathcal{M}_2$ with $T_1(\rho)=\frac{1}{3} (X\rho X+ Y\rho Y + Z\rho Z)$. This was first noted in \cite{Holevo02}.
\end{proof}

\section{Conclusion and open problems}

In this manuscript, we prove that the quotient of the entanglement-assisted capacity and the classical capacity of the same quantum channel is bounded by a dimension-dependent function so long as the channel has a geometric center in terms of the quantum relative entropy which has full rank.

The bound $\frac{C^{ea}(T)}{C(T)} \leq \frac{4\ln(2)d_B^5}{\lambda_{min}(\sigma)}$ that we obtain is (likely) not tight and does not match the conjectured upper bound $\frac{C^{ea}(T)}{C(T)} \leq 2(d_B+1)$ from \cite{BSST02}. It is an open question whether our bound could be improved generally or under certain assumptions on the quantum channel beyond our tight bound for unital qubit channels.

Our bound is finite for sequences of quantum channels that move towards a replacer channel which replaces any input with a full-rank state because of the dependence on the minimal eigenvalue of the divergence center. In other words, the quotient could only diverge for sequences of quantum channels where the divergence center changes rank, i.e. it has $\lambda_{min}(\sigma)\rightarrow 0$. For such channels, it is currently not known whether the quotient is finite or not. We suspect that our result can perhaps be extended to sequences of channels that converge to a lower-rank replacer channel using ideas from continuity bounds on the quantum mutual information.




\newpage
\section{Appendices}
\subsection{Properties of entropy}

Here, we recall and prove some properties of von Neumann entropy.

\subsubsection{Convexity of conditional entropy in the marginal}

\begin{lemma}\label{thm-convex-cond-entropy}
Consider two states $\rho_1$ and $\rho_2$ and a state $\rho_0=\lambda \rho_1 + (1-\lambda) \rho_2$ for some $0\leq \lambda\leq 1$. Let $\phi_0$, $\phi_1$ and $\phi_2$ denote the respective purifications of these states where $\Tr_{A}(\phi_i)=\rho_i$. Then, we have
  \[  \lambda H(B|A)_{(\id \otimes T) \phi_1}+ (1-\lambda) H(B|A)_{(\id \otimes T) \phi_2 } \geq \lambda H(B|A)_{(\id \otimes T) \phi_0}.\]
\end{lemma}

\begin{proof}
 Consider a cqqq-state $\omega_{XBRA}=\lambda \ket{0}\bra{0} \otimes (U_N \otimes \id) ( \phi_1) + (1-\lambda) \ket{1}\bra{1} \otimes (\id \otimes U_T) ( \phi_2 ) $, where $U_T$ is the Stinespring dilation of the channel $T$. Then, because the system $X$ is classical, we have
\begin{equation*}
\begin{split}
 &\lambda H(B|A)_{(\id \otimes T) \phi_1}+ (1-\lambda) H(B|A)_{(\id \otimes T) \phi_2 }  \\&= H(B|AX)_{\omega}  =H(AB|X)_{\omega}-H(A|X)_{\omega} .
\end{split}
\end{equation*}
When conditioning on the system $X$, the state $\omega$ is pure in the system $ABR$, and therefore we have
\begin{equation*}
\begin{split}
H(AB|X)_{\omega}-H(A|X)_{\omega} =H(R|X)_{\omega}-H(BR|X)_{\omega}= -H(B|RX)_{\omega}.
\end{split}
\end{equation*}
The conditional quantum entropy is strongly subadditive, and therefore $H(B|RX)_{\omega}\leq H(B|R)_{\omega}$. 
Let there be a state $\sigma_{BRA}=U_N (\phi_0)$.
Restricted to the subsystems B and R, we have $\omega_{BR}=\lambda U_N(\rho_1)+(1-\lambda) U_N (\rho_2)=U_N(\lambda \rho_1 + (1-\lambda)\rho_2)=U_N(\rho_0)$ and $\sigma_{BR}= U_N (\rho_0) $, and thus $H(B|R)_{\omega}=H(B|R)_{\sigma}$, which implies
\begin{equation*}
\begin{split}
-H(B|RX)_{\omega} \geq -H(B|R)_{\omega} =-H(B|R)_{\sigma} .
\end{split}
\end{equation*}
The state $\sigma_{BRA}$ is pure with respect to the systems $ABR$, and therefore we find
\begin{equation*}
\begin{split}
-H(B|R)_{\sigma}  = -H(BR)_{\sigma} + H(R)_{\sigma}  = -H(A)_{\sigma} +H(AB)_{\sigma} = H(B|A)_{\sigma}  =H(B|A)_{\id \otimes T (\phi_0)}   .
\end{split}
\end{equation*}
In summary,
\[ \lambda H(B|A)_{(\id \otimes T) \phi_1}+ (1-\lambda) H(B|A)_{(\id \otimes T) \phi_2 } \geq  H(B|A)_{\id \otimes T (\phi_0)} . \]

\end{proof}

\subsubsection{Derivatives of the von Neumann entropy}

Here, we give the first and second derivative of the (negative) von Neumann entropy. The first derivative was computed in \cite{LR02}; the second derivative is simply obtained from the derivative of the matrix logarithm from \cite{AdlerNote,HaberNote}.

\begin{theorem} \label{thm-entropy-derivatives}
    Let $\rho\in\mathcal{M}_d$ and $\sigma\in\mathcal{M}_d$ be quantum states, and define the quantum state $\rho_t=(1-t)\rho+t\omega=\rho +t(\omega-\rho)$. Then, we have
    \[ \frac{\dd}{\dd t} \trace\Big( \rho_t \log (\rho_t)\Big) =\trace\Big( (\omega-\rho) \log (\rho + t (\omega-\rho)) \Big).\] 
    Further, we have
    \begin{equation*}
    \begin{split}
    &  \frac{\dd^2}{\dd t^2} \trace\Big( \rho_t \log (\rho_t)\Big)  \\& =\trace\Big( (\omega-\rho)   \int_0^{1} ((1-z) \rho_t+z\mathbbm{1})^{-1} (\omega-\rho) ((1-z) \rho_t+z\mathbbm{1})^{-1} \dd z    \Big).
    \end{split}
\end{equation*}
\end{theorem}

\begin{proof}
The first derivative:
\begin{equation*}
    \begin{split}
       & \frac{\dd}{\dd t} \trace\Big( \rho_t \log (\rho_t)\Big) \\ &=
        \trace \Big(\frac{\dd}{\dd t}(\rho_t) \log (\rho_t) \Big) +  \trace\Big( \rho_t \frac{\dd}{\dd t} (\log (\rho_t) ) \Big)\\&
        = \trace\Big( \frac{\dd}{\dd t}(\rho_t) \log (\rho_t) \Big) +  \\&\qquad \trace\Big( \rho_t \int_0^{1} ((1-z) \rho_t+z\mathbbm{1})^{-1} (\frac{\dd}{\dd t}(\rho_t)) ((1-z) \rho_t+z\mathbbm{1})^{-1} \dd z )   \Big)\\&
     = \trace\Big( \frac{\dd}{\dd t}(\rho_t) \log (\rho_t) \Big) +  \trace \Big( \rho_t \rho_t^{-1} \frac{\dd}{\dd t}(\rho_t)    \Big)\\&
     = \trace\Big( (\omega-\rho) \log (\rho + t (\omega-\rho)) \Big) +  \trace\Big(\omega-\rho\Big) \\&
      = \trace\Big( (\omega-\rho) \log (\rho + t (\omega-\rho)) \Big).
    \end{split}
\end{equation*}

The derivative of the matrix logarithm for positive Hermitian matrices $\rho$ and $\omega$ is derived in \cite{AdlerNote}, or in \cite{HaberNote}. The third equality is possible because $\rho_t$ and $\rho_t+z\mathbbm{1}$ commute, and the trace is taken. In the fourth equality, we insert $\rho_t$. The last equality follows from the fact that $\omega-\rho$ is a difference between quantum states and therefore traceless.

In the point $t=0$, we have \[\frac{\dd}{\dd t} \trace\Big( \rho_t \log (\rho_t)\Big)  \Big|_{t=0} =\trace\Big( (\omega-\rho) \log (\rho) \Big).\]

The second derivative in the same direction is given by
\begin{equation*}
    \begin{split}
      &  \frac{\dd^2}{\dd t^2} \trace\Big( \rho_t \log (\rho_t)\Big) \\ &
     =  \frac{\dd}{\dd t} \trace\Big( (\omega-\rho) \log (\rho + t (\omega-\rho)) \Big) \\& = \trace\Big( (\omega-\rho)  \frac{\dd}{\dd t}(\log (\rho_t) ) \Big)\\&=\trace\Big( (\omega-\rho)   \int_0^{1} ((1-z) \rho_t+z\mathbbm{1})^{-1} (\omega-\rho) ((1-z) \rho_t+z\mathbbm{1})^{-1} \dd z    \Big).
    \end{split}
\end{equation*}
\end{proof}

\subsection{About unital qubit channels}

\begin{theorem} \label{thm-pauli1}
 Let $T_p(\rho)=(1-\frac{d^2}{d^2-1} p)\rho+\frac{d^2}{d^2-1} p\frac{\mathbbm{1}}{d}$ denote a family of quantum channels for $p\in[0,1]$. For all $p\in [0,\frac{d^2-1}{d^2}]$, the quotient of the entanglement-assisted capacity and the Holevo capacity of this channel is bounded by
\[\frac{C^{ea}(T_p)}{C_H(T_p)} \leq \lim_{p\rightarrow \frac{d^2-1}{d^2}}  \frac{C^{ea}(T_p)}{C_H(T_p)}= d+1.\]
For all $p\in [0, 1]$, the quotient of the entanglement-assisted capacity and the Holevo capacity of this channel is bounded by
\[\frac{C^{ea}(T_p)}{C_H(T_p)} \leq \frac{\log(\frac{d^2}{d^2-1})}{\log(d)+\frac{1}{d+1}\log(\frac{1}{d+1})+\frac{d}{d+1}\log(\frac{d}{d^2-1})}.\]
\end{theorem}


\begin{proof}
Recall that the Holevo capacity of the channel
\[T_p(\rho)=(1-\frac{d^2}{d^2-1} p)\rho+\frac{d^2}{d^2-1} p\frac{\mathbbm{1}}{d}\]
is given by \cite[Eq.~1]{BSST99}, \cite[Ex.~8.9]{Holevo13}, \cite[Theorem 20.4.3]{Wilde13} to be \[ C_H(T_p) = \log d + (1-  p \frac{d }{d+1}) \log (1-  p \frac{d }{d+1}) +p \frac{d }{d+1} \log(\frac{dp}{d^2-1}) \]
and that its entanglement-assisted capacity is given by \cite[Ex.~9.4]{Holevo13}, \cite[Eq.~3]{BSST99} to be \[C^{ea}(T_p) = \log d^2 + (1- p) \log (1- p) +p \log(\frac{p}{d^2-1}) . \]

In the limit $p\rightarrow \frac{d^2-1}{d^2}$, where $T_p$ approaches a constant channel which always outputs the maximally mixed state, regardless of the input, we have 
\begin{align*}
  &  \lim_{p\rightarrow \frac{d^2-1}{d^2}} \frac{C^{ea}(T_p)}{C_H(T_p)} \\& = \lim_{p\rightarrow \frac{d^2-1}{d^2}} \Big\{\frac{\log d^2 + (1- \frac{d^2p}{d^2-1} \frac{d^2-1}{d^2}) \log (1- \frac{d^2p}{d^2-1} \frac{d^2-1}{d^2}) +\frac{d^2p}{d^2-1} \frac{d^2-1}{d^2} \log(\frac{d^2p}{(d^2-1)d^2}) }{\log d + (1-  p \frac{d }{d+1}) \log (1-  p \frac{d }{d+1}) +p \frac{d }{d+1} \log(\frac{dp}{d^2-1})} \Big\}
\\& =\lim_{p\rightarrow \frac{d^2-1}{d^2}} \Big\{ \frac{(d+1)}{d} \frac{   \log (\frac{d^2-1}{p} -(d^2-1))  }{   \log (\frac{d^2-1}{dp} -(d-1)) }  \Big\}
\\& =  d+1.
\end{align*}
The second equality is obtained by using L'Hopital's rule, and the third equality is obtained by using L'Hopital's rule a second time.

For all $p\in [0, 1]$, the quotient of the entanglement-assisted capacity and the Holevo capacity of this channel is bounded by
\begin{equation*}
    \begin{split}
        \frac{C^{ea}(T_p)}{C_H(T_p)} &\leq \frac{C^{ea}(T_1)}{C_H(T_1)}\\&=\frac{\log(\frac{d^2}{d^2-1})}{\log(d)+\frac{1}{d+1}\log(\frac{1}{d+1})+\frac{d}{d+1}\log(\frac{d}{d^2-1})}  \\&=\frac{(d+1)\log(\frac{d^2}{d^2-1})}{(2d+1)\log(d)-\log({d+1})-d\log({d^2-1})} 
.    \end{split}
\end{equation*} 
\end{proof}

\begin{theorem}  \label{thm-pauli2}
Let $T_{p_0,p_1,p_2,p_3}:\mathcal{M}_2\rightarrow \mathcal{M}_2$ denote a qubit Pauli channel, and let $p_1=p_2$.

For all $0\leq p_0+p_3\leq 1$, we have that 
$\frac{C^{ea}(T_{p_0,p_3})}{\chi(T_{p_0,p_3})} $ is maximal for $p_1=p_2=p_3$ or for $p_0=p_1=p_2$.

The analogue statement also holds under exchanges of the parameters $p_0,p_1,p_2$ and $p_3$.
\end{theorem}

\begin{proof}
We refer to a qubit Pauli channel with $p_1=p_2= \frac{(1-p_0-p_3)}{2}$ as a covariant Pauli channel which acts on a state $\rho$ as
\[T_{p_0,p_3} (\rho)= p_0 \rho+ \frac{(1-p_0-p_3)}{2} (X\rho X + Y \rho Y) +p_3 Z\rho Z.\] 
This channel is covariant under rotations around the Z-axis; since X and Y have the same coefficient, this channel is described by two parameters.

Based on \cite{Siudzinska20}, \cite{PK22} give explicit expressions for the Holevo capacity and the entanglement-assisted capacity for channels from this family.

The entanglement-assisted capacity of a covariant Pauli channel with parameters $0\leq p_0+p_3\leq 1$ is given by
\[C^{ea}(T_{p_0,p_3})=2+p_0 \log(p_0)+p_3 \log(p_3) + (1-p_0-p_3)\log( \frac{1-p_0-p_3}{2}) .\]

The Holevo capacity is given by
\begin{equation*}
C_{H}(T_{p_0,p_3})= \begin{cases}
C_1(T_{p_0,p_3}) & \text{ if } (p_0-p_3)^2 \geq (2p_0+2p_3-1)^2 \\
C_2(T_{p_0,p_3}) &\text{else}
\end{cases}
\end{equation*}
where
\[C_1(T_{p_0,p_3})=1+\frac{1-p_0+p_3}{2}\log(\frac{1-p_0+p_3}{2}) + \frac{1+p_0-p_3}{2}\log(\frac{1+p_0-p_3}{2})\]
\[C_2(T_{p_0,p_3})=1+(p_0+p_3)\log(p_0+p_3) + (1-p_0-p_3)\log(1-p_0-p_3) .\]

There are two distinct parameter regions defined by the case distinction, which is illustrated in Figure~\ref{fig-map-of-cap} in light and dark grey. We will refer to the region where $C_1$ applies as Case 1, and the region where $C_2$ applies as Case 2.

The border between the two regions is parametrized by the two solutions of the quadratic equation $(p_0-p_3)^2 = (2p_0+2p_3-1)^2$, given by the two lines $p_3=\frac{1}{3}-\frac{1}{3}p_0$ and $p_3=1-3p_0$.

Note that $p_3=\frac{1}{3}-\frac{1}{3}p_0$ means that $p_0=1-3p_3$, and $p_1=p_2=p_3$. This channel would thus correspond to a Pauli channel where three of the four parameters coincide, and only $p_0$ differs, which is exactly the type of channel that appears in Theorem~\ref{thm-pauli1}.
Note that $p_3=1-3p_0$ means that $p_1=p_2=p_0$. This channel would thus correspond to a Pauli channel where three of the four parameters coincide, and only $p_3$ differs.
For the other solution, $p_3=1-3p_0$, we have $p_1=p_2=p_0$. Again, three out of the four parameters coincide, and this channel is related to the channel from Theorem~\ref{thm-pauli1} by unitary transformations. Then, because we will show that the quotient is maximal along these two lines, we can upper bound the quotient for all parameter pairs in terms of the bound in Theorem~\ref{thm-pauli1}.

\begin{figure}[h!]
    \centering
    \includegraphics[width=8cm]{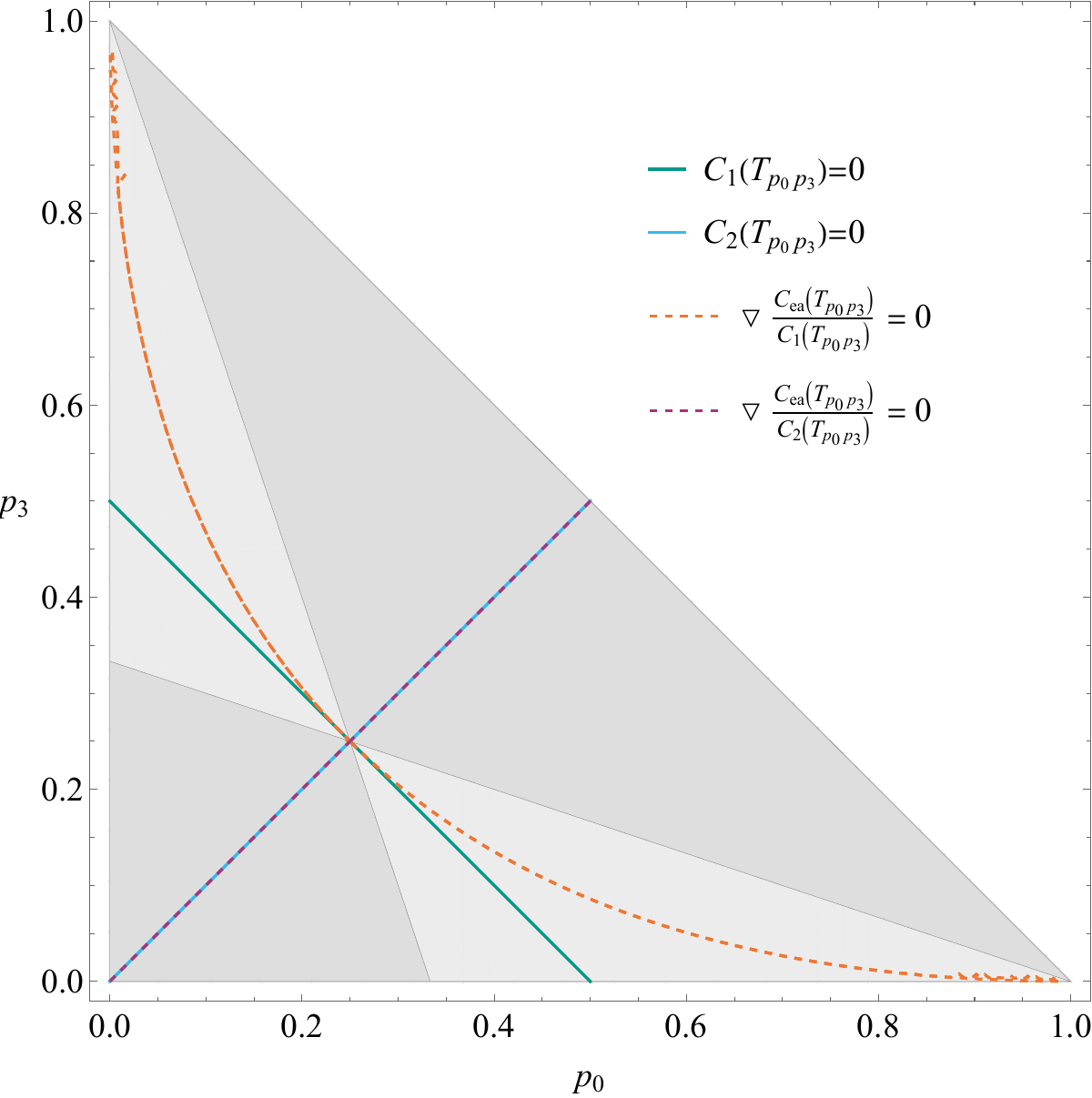}
    \caption{Plot of the parameter region for $p_0$ and $p_3$. The light gray region is the region where case 2 applies. The dark gray region is the region where case 1 applies.}
    \label{fig-map-of-cap}
\end{figure}

\textsc{Case 1:} The quotient is extremal when $\frac{\dd}{\dd{p_0}} \Big(\frac{C^{ea}(T_{p_0,p_3})}{C_1(T_{p_0,p_3})}\Big)=0=\frac{\dd}{\dd{p_3}} \Big(\frac{C^{ea}(T_{p_0,p_3})}{C_1(T_{p_0,p_3})}\Big)$. These conditions turn out to be equivalent. The points where they are equal to zero are described by the solutions to this equation:
\begin{equation*}
    \begin{split}
       &\frac{\log(1+p_0-p_3) \Big(- p_0\log(\frac{(1-p_0-p_3)^2}{4p_0p_3})+\log(4p_3) \Big)}{((1+p_0-p_3)\log(1+p_0-p_3)+(1-p_0+p_3)\log(1-p_0+p_3))^2}\\&- \frac{
       \log(1-p_0+p_3)  \Big( p_0\log(4p_0) +(1-p_0)\log(\frac{(1-p_0-p_3)^2}{p_3}) \Big) }{((1+p_0-p_3)\log(1+p_0-p_3)+(1-p_0+p_3)\log(1-p_0+p_3))^2} \\&=0
.    \end{split}
\end{equation*}

This can be solved numerically, and the solutions are plotted in Figure~\ref{fig-map-of-cap}. Pointwise numerical computation of the quotient along this line shows that it never exceeds 2.1, and pointwise numerical computation of the Hessian shows that the points along this line are minima or saddle points.


Furthemore, we can numerically determine the tangential curve of this curve at $p_0=p_3=\frac{1}{4}$, obtaining $p_3=1/2-p_0$ (which is exactly the curve where $C_2(T_{p_0,p_3})=0$), and we can prove that it is bounded for a channel sequence along this line:
\begin{equation*}\begin{split}
 & \lim_{p_0\rightarrow 1/4} \frac{C^{ea}(T_{p_0,p_3})}{C_1(T_{p_0,p_3})}\Big|_{p_3=1/2-p_0} \\& = \lim_{p_0\rightarrow 1/4}   \frac{1+p_0 \log(p_0)+ (\frac{1}{2}-p_0)\log (\frac{1}{2}-p_0) }{1+  (\frac{3}{4}-p_0)\log (\frac{3}{4}-p_0)+ (\frac{1}{4}+p_0)\log (\frac{1}{4}+p_0)}\\
& =\lim_{p_0\rightarrow 1/4}   \frac{ \log(p_0)+p_0 \frac{1}{p_0} - (\frac{1}{2}-p_0) \frac{1}{(\frac{1}{2}-p_0)} 
-\log (\frac{1}{2}-p_0)}{   -(\frac{3}{4}-p_0) \frac{1}{(\frac{3}{4}-p_0)} - \log(\frac{3}{4}-p_0) + (\frac{1}{4}+p_0) \frac{1}{(\frac{1}{4}+p_0)} + \log(\frac{1}{4}+p_0)}
\\& =\lim_{p_0\rightarrow 1/4}   \frac{ \log(p_0) 
-\log (\frac{1}{2}-p_0)}{  - \log(\frac{3}{4}-p_0)  + \log(\frac{1}{4}+p_0)}
\\&  =\lim_{p_0\rightarrow 1/4} \frac{ \frac{1}{p_0}
+ \frac{1}{\frac{1}{2}-p_0} }{  \frac{1}{\frac{3}{4}-p_0 } + \frac{1}{ \frac{1}{4}+p_0)}}
=\lim_{p_0\rightarrow 1/4} \frac{ 4+4 }{  2+2 } =2.\end{split}
\end{equation*}

The maximum of the quotient is reached on the edge of the domain, which runs along the line which parametrizes a Pauli channel with three equal coefficients.

\textsc{Case 2:}  The quotient is extremal when $\frac{\dd}{\dd {p_0}} \Big(\frac{C^{ea}(T_{p_0,p_3})}{C_2(T_{p_0,p_3})} \Big)=0=\frac{\dd}{\dd {p_3}} \Big(\frac{C^{ea}(T_{p_0,p_3})}{C_2(T_{p_0,p_3})}\Big)$.


Close inspection reveals that the quotients and thereby also their partial derivatives are the same under exchanging $p_0$ and $p_3$, which means they are equal for $p_0=p_3$ (which is the same line along which $C_1(T_{p_0,p_3})$ is zero). This line is sketched in Figure~\ref{fig-map-of-cap} in purple.

Along this line, we have 
\begin{equation*}
    \begin{split}
C^{ea}(T_{p_0,p_3}) \Big|_{p_3=p_0}& = 2+2p_0 \log(p_0)+ (1-2p_0)\log (\frac{1}{2}-p_0) \\&=1+  2p_0 \log (2p_0)+ (1-2p_0)\log (1-2p_0)\\&=C_2 (T_{p_0,p_3}) \Big|_{p_3=p_0}
    \end{split}
\end{equation*}
and the quotient is $1$. This line corresponds to a minimum. 
The quotient reaches its maximum on the edge of the domain, which is along the line where the channel reduces to a Pauli channel with three equal coefficients.
\end{proof}

\begin{figure}
\centering
\begin{subfigure}{.7\textwidth}
  \centering
  \includegraphics[width=\linewidth]{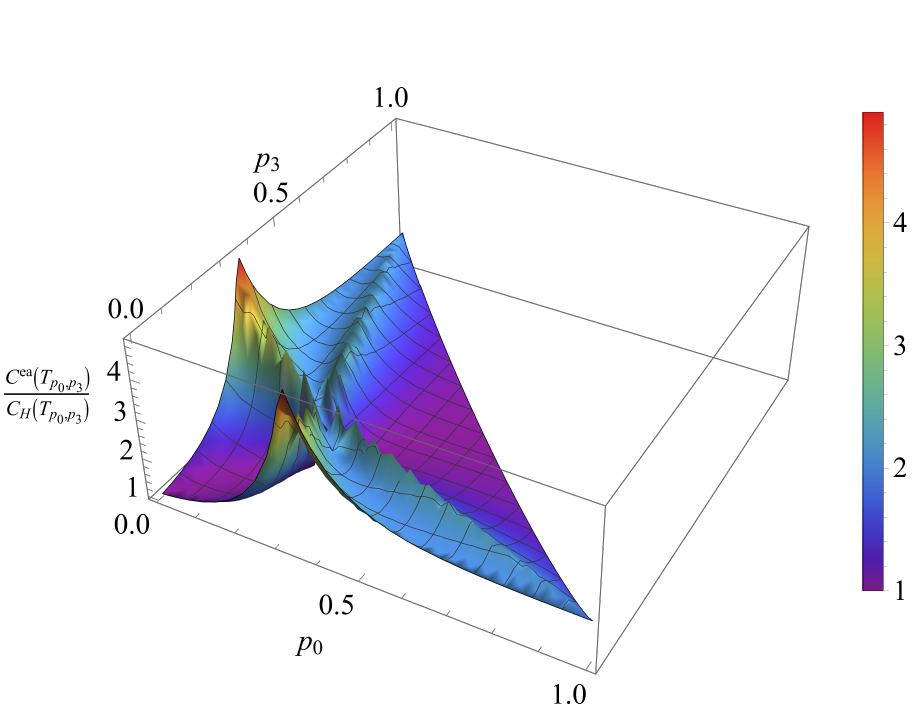}
  \caption{Side view.}
  \label{fig-ampdamp-cap}
\end{subfigure}%
\hspace{1cm}
\begin{subfigure}{.7\textwidth}
  \centering
  \includegraphics[width=\linewidth]{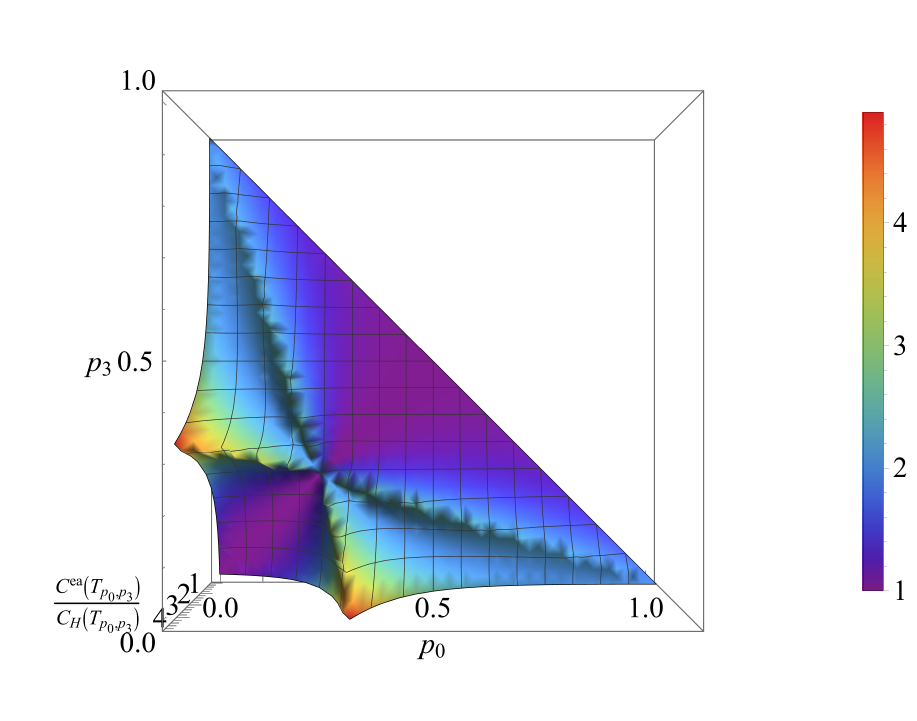}
  \caption{Top view.}
  \label{fig-ampdamp-center}
\end{subfigure}
\caption{For pairs of values $p_0$ and $p_3$, we plot the quotient of $C^{ea}(T_{p_0,p_3})$ and $C_H(T_{p_0,p_3})$. The quotient is maximal on the border between the domains highlighted in Figure~\ref{fig-map-of-cap}.}
\label{fig:test}
\end{figure}

\begin{theorem}
 \label{thm-pauli3}
  Let $T_{p_0,p_1,p_2,p_3}:\mathcal{M}_2\rightarrow \mathcal{M}_2$ denote a Pauli channel. For any $0\leq p_0+p_1+p_2+p_3\leq 1$, we have that $\frac{C^{ea}(T_{p_0,p_1,p_2,p_3})}{C_H(T_{p_0,p_1,p_2,p_3})}$ is maximal if the channel is covariant, i.e. $p_i=p_j$ for some $i\neq j$.
\end{theorem}

\begin{proof}
   Let $T_{p_0,p_1,p_2,p_3}:\mathcal{M}_2\rightarrow \mathcal{M}_2$ be a Pauli channel, acting on a qubit state $\rho$ as
 \[T_{p_0,p_1,p_2,p_3} (\rho)= p_0 \rho+ p_1 X \rho X + p_2 Y \rho Y+p_3 Z\rho Z.\] 

Without loss of generality, let $p_0,p_2$ and $p_3$ be three independent parameters and let $p_1=1-p_0-p_2-p_3$. Note that since $\sum_k p_k =1$ (which is a condition for the channel to be trace preserving), the channel $T_{p_0,p_1,p_2,p_3}$ can be described by these three independent parameters.
 
Recall that this channel's entanglement-assisted capacity is known \cite{CCC21} to be
\[C^{ea}(T_{p_0,p_1,p_2,p_3})=2+p_0\log(p_0)+p_1\log(p_1)+p_2\log(p_2)+p_3\log(p_3)\]
and that its Holevo-capacity is given by \cite{Siudzinska20} as
\[C_H(T_{p_0,p_1,p_2,p_3})=\frac{1+\lambda^*}{2}\log(1+\lambda^*)+\frac{1-\lambda^*}{2}\log(1-\lambda^*)\]
with $\lambda^*=\max\{| \lambda_{min}|
, \lambda_{max}\}= \max \{ |p_0-p_1-p_2-p_3+p_{min} |,p_0-p_1-p_2-p_3+p_{max}\}$.

For a given channel with parameters $p_k$, $\lambda^*$ will only depend on two out of the three parameters. If we are in the case where $\lambda^*=\lambda_{max}=p_0-p_1-p_2-p_3+p_{max}=2p_0-1+p_{max}$, it depends only on $p_0$ and $p_{max}$, and if $\lambda^*=\lambda_{min}$, it depends only on $p_0$ and $p_{min}$.

Let us first assume that we are in the first case, $\lambda^*=\lambda_{max}$, and that $p_{max}=p_3$. Then, $p_2$ does not appear in the expression for the Holevo capacity, which is given by $C_H(T_{p_0,p_1,p_2,p_3})=1+(1-p_0-p_3) \log(2(1-p_0-p_3))+(p_0+p_3)\log(2(p_0+p_3))$. Then, when computing the gradient for the entanglement-assistance factor $\frac{C^{ea}(T_{p_0,p_1,p_2,p_3})}{C_H(T_{p_0,p_1,p_2,p_3})}$, the partial derivative with regard to $p_2$ is given by
\begin{equation*}
    \begin{split}
      &  \frac{\dd}{\dd{p_2}} \Big(\frac{C^{ea}(T_{p_0,p_1,p_2,p_3})}{C_H(T_{p_0,p_1,p_2,p_3})} \Big) \\&
      = \frac{-\log(1-p_0-p_2-p_3)+\log(p_2)}{1+(1-p_0-p_3) \log(2(1-p_0-p_3))+(p_0+p_3)\log(2(p_0+p_3))}.
    \end{split}
\end{equation*}

 This derivative can only be zero as a function of $p_2$ if $p_2=\frac{1}{2}(1-p_0-p_3)$. Thereby, the quotient can only be extremal in the case where $p_2=\frac{1}{2}(1-p_0-p_3)$, which implies $p_1=\frac{1}{2}(1-p_0-p_3)=p_2$, corresponding to a covariant Pauli channel. In this case, this corresponds to a Pauli channel that is covariant with respect to rotations around the $Z$-axis.
 
Note that our assumption on $p_3$ being the largest parameter is without loss of generality because the entanglement-assisted capacity is symmetric in the parameters. If $p_2$ was the largest eigenvalue instead, the Holevo capacity would not depend on $p_3$, and the above argument would yield that the quotient is extremal for $p_3=p_1=\frac{1}{2}(1-p_0-p_2)$, a Pauli channel which is covariant with respect to rotations around the $Y$-axis. The quotient of the capacities does not depend on this axis, as the different covariant Pauli channels are related by conjugation of Pauli matrices, which are unitary, and the capacities are unchanged under unitary transformations.

If the second case applies, i.e. $\lambda^*=|\lambda_{min}|$, the Holevo capacity will still only depend on two out of the three parameters. Assuming now that $p_{min}=p_3$ would lead to the same argument, i.e. that the channel is extremal for a parameter family that describes a covariant Pauli channel.
 \end{proof}

\chapter{Conclusion and open problems}

Quantum effects fascinate us and stump us. They give rise to exciting fundamental notions of communication and information, but they also introduce noise and no-go theorems that require highly specialized techniques. In this thesis, we hope to shed light on this dichotomy in the context of various topics in quantum Shannon theory and quantum fault-tolerance.

As pointed out in \cite{CMH20}, these observations raise a question of how well we can leverage quantum effects for communication in the presence of the noise that typically affects a quantum computer. In Chapter~\ref{chapter-fteacap}, we therefore study entanglement-assisted communication in the presence of noise on the encoding and decoding operation. We construct a coding scheme that can transmit information at almost the same rate as long as the gate error probability is below a threshold.

In Chapter~\ref{chapter-aws}, we study the noise on trapped ion quantum devices, finding that techniques from fault-tolerance can protect some computations, but that current devices cannot be considered fault-tolerant (yet?), in particular for near-term demonstrations of quantum communication. These results highlight the importance of studying the noise in current quantum circuits and coming up with realistic noise models and theoretical as well as experimental mitigation strategies.


In Chapter~\ref{chapter-divrad}, we take a step towards solving an open problem in quantum Shannon theory concerning the ratio of the entanglement-assisted capacity and the classical capacity of the same quantum channel. In particular, we prove that this quotient is bounded by a finite expression for many quantum channels. This implies that there is a fundamental limit on how much the effect of quantum entanglement can improve communication for these channels.



\bibliographystyle{marcotomPB}
\bibliography{phdbib}

\begin{thebibliography}{100}

\bibitem{AAA+22}
R.~Acharya, I.~Aleiner, R.~Allen, et~al.
\newblock {\em ``Suppressing quantum errors by scaling a surface code logical qubit''}, (2023).
\newblock \text{\href{http://dx.doi.org/10.1038/s41586-022-05434-1}{DOI:\,10.1038/s41586-022-05434-1}}.

\bibitem{AdlerNote}
S.~Adler.
\newblock {\em ``Taylor Expansion and Derivative Formulas for Matrix Logarithms''}.
\newblock Available online: \url{https://www.ias.edu/sites/default/files/sns/files/1-matrixlog_tex(1).pdf}.

\bibitem{AB99}
D.~Aharonov and M.~Ben-Or.
\newblock {\em ``Fault-Tolerant Quantum Computation with Constant Error Rate''}.
\newblock \href{http://dx.doi.org/10.1137/S0097539799359385}{SIAM Journal on Computing {\bf 38}(4):\,1207--1282} (2008).

\bibitem{ABBN12}
R.~Ahlswede, I.~Bjelakovi{\'{c}}, H.~Boche, and J.~Nötzel.
\newblock {\em ``Quantum Capacity under Adversarial Quantum Noise: Arbitrarily Varying Quantum Channels''}.
\newblock \href{http://dx.doi.org/10.1007/s00220-012-1613-x}{Communications in Mathematical Physics {\bf 317}(1):\,103--156} (2012).

\bibitem{AF04}
R.~Alicki and M.~Fannes.
\newblock {\em ``Continuity of quantum conditional information''}.
\newblock \href{http://dx.doi.org/10.1088/0305-4470/37/5/l01}{Journal of Physics A: Mathematical and General {\bf 37}(5):\,L55--L57} (2004).

\bibitem{AGP05}
P.~Aliferis, D.~Gottesman, and J.~Preskill.
\newblock {\em ``Quantum accuracy threshold for concatenated distance-3 codes''}.
\newblock \href{http://dx.doi.org/10.48550/ARXIV.QUANT-PH/0504218}{Quantum information and computation {\bf 6}} (2005).

\bibitem{AGP07}
P.~Aliferis, D.~Gottesman, and J.~Preskill.
\newblock {\em ``Accuracy threshold for postselected quantum computation''}, (2007).
\newblock \text{\href{http://arxiv.org/abs/quant-ph/0703264}{arXiv:\,quant-ph/0703264}}.

\bibitem{AE05}
K.~M.~R. Audenaert and J.~Eisert.
\newblock {\em ``Continuity bounds on the quantum relative entropy''}.
\newblock \href{http://dx.doi.org/10.1063/1.2044667}{Journal of Mathematical Physics {\bf 46}(10):\,102104} (2005).

\bibitem{BCMH22}
P.~Belzig, M.~Christandl, and A.~Müller-Hermes.
\newblock {\em ``Fault-tolerant Coding for Entanglement-Assisted Communication''}, (2022).
\newblock \text{\href{http://arxiv.org/abs/2210.02939}{arXiv:\,2210.02939}}.

\bibitem{BCMH23SHORT}
P.~Belzig, M.~Christandl, and A.~Müller-Hermes.
\newblock {\em ``Fault-Tolerant Coding for Entanglement-Assisted Communication''}.
\newblock In {\em 2023 IEEE International Symposium on Information Theory (ISIT)}, pages 84--89, (2023), \text{\href{http://dx.doi.org/10.1109/ISIT54713.2023.10206950}{DOI:\,10.1109/ISIT54713.2023.10206950}}.

\bibitem{BBCJPW93}
C.~H. Bennett, G.~Brassard, C.~Cr\'epeau, R.~Jozsa, A.~Peres, and W.~K. Wootters.
\newblock {\em ``Teleporting an unknown quantum state via dual classical and {Einstein-Podolsky-Rosen} channels''}.
\newblock \href{http://dx.doi.org/10.1103/PhysRevLett.70.1895}{Phys. Rev. Lett. {\bf 70}:\,1895--1899} (1993).

\bibitem{BDHSW14}
C.~H. Bennett, I.~Devetak, A.~W. Harrow, P.~W. Shor, and A.~Winter.
\newblock {\em ``The Quantum Reverse {Shannon} Theorem and Resource Tradeoffs for Simulating Quantum Channels''}.
\newblock \href{http://dx.doi.org/10.1109/tit.2014.2309968}{{IEEE} Transactions on Information Theory {\bf 60}(5):\,2926--2959} (2014).

\bibitem{BDSW96}
C.~H. Bennett, D.~P. DiVincenzo, J.~A. Smolin, and W.~K. Wootters.
\newblock {\em ``Mixed-state entanglement and quantum error correction''}.
\newblock \href{http://dx.doi.org/10.1103/physreva.54.3824}{Physical Review A {\bf 54}(5):\,3824–3851} (1996).

\bibitem{BSST02}
C.~H. Bennett, P.~W. Shor, J.~A. Smolin, and A.~Thapliyal.
\newblock {\em ``Entanglement-assisted capacity of a quantum channel and the reverse {Shannon} theorem''}.
\newblock \href{http://dx.doi.org/10.1109/TIT.2002.802612}{Information Theory, IEEE Transactions on {\bf 48}:\,2637--2655} (2002).

\bibitem{BSST99}
C.~H. Bennett, P.~W. Shor, J.~A. Smolin, and A.~V. Thapliyal.
\newblock {\em ``Entanglement-Assisted Classical Capacity of Noisy Quantum Channels''}.
\newblock \href{http://dx.doi.org/10.1103/physrevlett.83.3081}{Physical Review Letters {\bf 83}(15):\,3081--3084} (1999).

\bibitem{BW92}
C.~H. Bennett and S.~J. Wiesner.
\newblock {\em ``Communication via one- and two-particle operators on {Einstein-Podolsky-Rosen} states''}.
\newblock \href{http://dx.doi.org/10.1103/PhysRevLett.69.2881}{Phys. Rev. Lett. {\bf 69}:\,2881--2884} (1992).

\bibitem{BDNW18}
H.~Boche, C.~Deppe, J.~N\"{o}tzel, and A.~Winter.
\newblock {\em ``Fully Quantum Arbitrarily Varying Channels: Random Coding Capacity and Capacity Dichotomy''}.
\newblock In {\em 2018 IEEE International Symposium on Information Theory (ISIT)}, page 2012–2016, Vail, CO, USA(2018), \text{\href{http://dx.doi.org/10.1109/ISIT.2018.8437610}{DOI:\,10.1109/ISIT.2018.8437610}}.

\bibitem{BPV98}
S.~Bose, M.~B. Plenio, and V.~Vedral.
\newblock {\em ``Mixed state dense coding and its relation to entanglement measures''}, (1999).
\newblock \text{\href{http://arxiv.org/abs/quant-ph/9810025}{arXiv:\,quant-ph/9810025}}.

\bibitem{Bowen01}
G.~Bowen.
\newblock {\em ``Classical information capacity of superdense coding''}.
\newblock \href{http://dx.doi.org/10.1103/PhysRevA.63.022302}{Phys. Rev. A {\bf 63}:\,022302} (2001).

\bibitem{BFBD11}
C.~G. Brell, S.~T. Flammia, S.~D. Bartlett, and A.~C. Doherty.
\newblock {\em ``Toric codes and quantum doubles from two-body Hamiltonians''}.
\newblock \href{http://dx.doi.org/10.1088/1367-2630/13/5/053039}{New Journal of Physics {\bf 13}(5):\,053039} (2011).

\bibitem{BE21}
N.~P. Breuckmann and J.~N. Eberhardt.
\newblock {\em ``Quantum Low-Density Parity-Check Codes''}.
\newblock \href{http://dx.doi.org/10.1103/prxquantum.2.040101}{{PRX} Quantum {\bf 2}(4)} (2021).

\bibitem{BWC+11}
K.~R. Brown, A.~C. Wilson, Y.~Colombe, C.~Ospelkaus, A.~M. Meier, E.~Knill, D.~Leibfried, and D.~J. Wineland.
\newblock {\em ``Single-qubit-gate error below $10^-4$ in a trapped ion''}.
\newblock \href{http://dx.doi.org/10.1103/physreva.84.030303}{Physical Review A {\bf 84}(3)} (2011).

\bibitem{CCNH21}
R.~Cane, D.~Chandra, S.~X. Ng, and L.~Hanzo.
\newblock {\em ``Experimental Characterization of Fault-Tolerant Circuits in Small-Scale Quantum Processors''}, (2021).
\newblock \text{\href{http://arxiv.org/abs/2112.04076}{arXiv:\,2112.04076}}.

\bibitem{CCC21}
D.~Chandra, M.~Caleffi, and A.~S. Cacciapuoti.
\newblock {\em ``The Entanglement-Assisted Communication Capacity over Quantum Trajectories''}, (2021).
\newblock \text{\href{http://arxiv.org/abs/2110.08078}{arXiv:\,2110.08078}}.

\bibitem{Google21}
Z.~Chen, K.~J. Satzinger, J.~Atalaya, et~al.
\newblock {\em ``Exponential suppression of bit or phase flip errors with repetitive error correction''}.
\newblock \href{http://dx.doi.org/10.48550/ARXIV.2102.06132}{Nature {\bf 595}(7867):\,383--387} (2021).

\bibitem{CNHM03}
I.~Chiorescu, Y.~Nakamura, C.~J. P.~M. Harmans, and J.~E. Mooij.
\newblock {\em ``Coherent Quantum Dynamics of a Superconducting Flux Qubit''}.
\newblock \href{http://dx.doi.org/10.1126/science.1081045}{Science {\bf 299}(5614):\,1869--1871} (2003).

\bibitem{CMH20}
M.~Christandl and A.~Müller-Hermes.
\newblock {\em ``Fault-tolerant Coding for Quantum Communication''}.
\newblock \href{http://dx.doi.org/10.1109/TIT.2022.3169438}{IEEE Transactions on Information Theory pages 1--1} (2022).

\bibitem{CMW16}
T.~Cooney, M.~Mosonyi, and M.~M. Wilde.
\newblock {\em ``Strong Converse Exponents for a Quantum Channel Discrimination Problem and Quantum-Feedback-Assisted Communication''}.
\newblock \href{http://dx.doi.org/10.1007/s00220-016-2645-4}{Communications in Mathematical Physics {\bf 344}(3):\,797--829} (2016).

\bibitem{CT16}
B.~Criger and B.~Terhal.
\newblock {\em ``Noise Thresholds for the [[4, 2, 2]]-concatenated Toric Code''}, (2016).
\newblock \text{\href{http://dx.doi.org/10.26421/qic16.15-16}{DOI:\,10.26421/qic16.15-16}}.

\bibitem{Devetak03}
I.~Devetak.
\newblock {\em ``The private classical capacity and quantum capacity of a quantum channel''}.
\newblock \href{http://dx.doi.org/10.1109/TIT.2004.839515}{IEEE Transactions on Information Theory {\bf 51}(1):\,44--55} (2005).

\bibitem{DHW08}
I.~Devetak, A.~W. Harrow, and A.~J. Winter.
\newblock {\em ``A Resource Framework for Quantum {Shannon} Theory''}.
\newblock \href{http://dx.doi.org/10.1109/tit.2008.928980}{{IEEE} Transactions on Information Theory {\bf 54}(10):\,4587--4618} (2008).

\bibitem{DW03}
I.~Devetak and A.~Winter.
\newblock {\em ``Distillation of secret key and entanglement from quantum states''}.
\newblock \href{http://dx.doi.org/10.1098/rspa.2004.1372}{Proceedings of the Royal Society A: Mathematical, Physical and Engineering Sciences {\bf 461}} (2003).

\bibitem{DiVincenzo2000}
D.~P. DiVincenzo.
\newblock {\em ``The Physical Implementation of Quantum Computation''}.
\newblock \href{http://dx.doi.org/10.1002/1521-3978(200009)48:9/11<771::aid-prop771>3.0.co;2-e}{Fortschritte der Physik {\bf 48}(9-11):\,771--783} (2000).

\bibitem{FGL18}
O.~Fawzi, A.~Grospellier, and A.~Leverrier.
\newblock {\em ``Constant Overhead Quantum Fault-Tolerance with Quantum Expander Codes''}.
\newblock In {\em 2018 {IEEE} 59th Annual Symposium on Foundations of Computer Science ({FOCS})}, (2018), \text{\href{http://dx.doi.org/10.1109/focs.2018.00076}{DOI:\,10.1109/focs.2018.00076}}.

\bibitem{FvdG97}
C.~Fuchs and J.~van~de Graaf.
\newblock {\em ``Cryptographic Distinguishability Measures for Quantum-Mechanical States''}.
\newblock \href{http://dx.doi.org/10.1109/18.761271}{Information Theory, IEEE Transactions on {\bf 45}:\,1216 -- 1227} (1999).

\bibitem{FKM10}
M.~Fukuda, C.~King, and D.~K. Moser.
\newblock {\em ``Comments on Hastings' Additivity Counterexamples''}.
\newblock \href{http://dx.doi.org/10.1007/s00220-010-0996-9}{Communications in Mathematical Physics {\bf 296}(1):\,111--143} (2010).

\bibitem{GCS17}
J.~M. Gambetta, J.~M. Chow, and M.~Steffen.
\newblock {\em ``Building logical qubits in a superconducting quantum computing system''}.
\newblock \href{http://dx.doi.org/10.1038/s41534-016-0004-0}{npj Quantum Information {\bf 3}(1)} (2017).

\bibitem{Gottesman97}
D.~Gottesman.
\newblock {\em ``Stabilizer Codes and Quantum Error Correction''}, (1997).
\newblock \text{\href{http://arxiv.org/abs/quant-ph/9705052}{arXiv:\,quant-ph/9705052}}.

\bibitem{Gottesman09}
D.~Gottesman.
\newblock {\em ``An Introduction to Quantum Error Correction and Fault-Tolerant Quantum Computation''}, (2009).
\newblock \text{\href{http://arxiv.org/abs/0904.2557}{arXiv:\,0904.2557}}.

\bibitem{Gottesman14}
D.~Gottesman.
\newblock {\em ``Fault-Tolerant Quantum Computation with Constant Overhead''}, (2014).
\newblock \text{\href{http://arxiv.org/abs/1310.2984}{arXiv:\,1310.2984}}.

\bibitem{Gottesman16}
D.~Gottesman.
\newblock {\em ``Quantum fault tolerance in small experiments''}, (2016).
\newblock \text{\href{http://arxiv.org/abs/1610.03507}{arXiv:\,1610.03507}}.

\bibitem{GC99}
D.~Gottesman and I.~L. Chuang.
\newblock {\em ``Demonstrating the viability of universal quantum computation using teleportation and single-qubit operations''}.
\newblock \href{http://dx.doi.org/10.1038/46503}{Nature {\bf 402}(6760):\,390--393} (1999).

\bibitem{GBP97}
M.~Grassl, T.~Beth, and T.~Pellizzari.
\newblock {\em ``Codes for the quantum erasure channel''}.
\newblock \href{http://dx.doi.org/10.1103/physreva.56.33}{Physical Review A {\bf 56}(1):\,33--38} (1997).

\bibitem{GW14}
M.~K. Gupta and M.~M. Wilde.
\newblock {\em ``Multiplicativity of Completely Bounded p-Norms Implies a Strong Converse for Entanglement-Assisted Capacity''}.
\newblock \href{http://dx.doi.org/10.1007/s00220-014-2212-9}{Communications in Mathematical Physics {\bf 334}(2):\,867--887} (2014).

\bibitem{HaberNote}
H.~E. Haber.
\newblock {\em ``Notes on the Matrix Exponential and Logarithm''}.
\newblock Available online: \url{http://scipp.ucsc.edu/~haber/webpage/MatrixExpLog.pdf}.

\bibitem{HF19}
R.~Harper and S.~T. Flammia.
\newblock {\em ``Fault-Tolerant Logical Gates in the {IBM} Quantum Experience''}.
\newblock \href{http://dx.doi.org/10.1103/physrevlett.122.080504}{Physical Review Letters {\bf 122}(8)} (2019).

\bibitem{HL04}
A.~Harrow and H.-K. Lo.
\newblock {\em ``A Tight Lower Bound on the Classical Communication Cost of Entanglement Dilution''}.
\newblock \href{http://dx.doi.org/10.1109/TIT.2003.822597}{IEEE Transactions on Information Theory {\bf 50}:\,319--327} (2004).

\bibitem{Hastings09}
M.~B. Hastings.
\newblock {\em ``Superadditivity of communication capacity using entangled inputs''}.
\newblock \href{http://dx.doi.org/10.1038/nphys1224}{Nature Physics {\bf 5}(4):\,255--257} (2009).

\bibitem{HRO+17}
R.~W. Heeres, P.~Reinhold, N.~Ofek, L.~Frunzio, L.~Jiang, M.~H. Devoret, and R.~J. Schoelkopf.
\newblock {\em ``Implementing a universal gate set on a logical qubit encoded in an oscillator''}.
\newblock \href{http://dx.doi.org/10.1038/s41467-017-00045-1}{Nature Communications {\bf 8}(1)} (2017).

\bibitem{HS81}
H.~V. Henderson and S.~R. Searle.
\newblock {\em ``On Deriving the Inverse of a Sum of Matrices''}.
\newblock \href{http://dx.doi.org/10.1137/1023004}{SIAM Review {\bf 23}(1):\,53--60} (1981).

\bibitem{HOT81}
F.~Hiai, M.~Ohya, and M.~Tsukada.
\newblock {\em ``Sufficiency, KMS condition and relative entropy in von Neumann algebras''}.
\newblock \href{http://dx.doi.org/10.2140/pjm.1981.96.99}{Pacific Journal of Mathematics {\bf 96}:\,99--109} (1981).

\bibitem{Hoeffding63}
W.~Hoeffding.
\newblock {\em ``Probability Inequalities for Sums of Bounded Random Variables''}.
\newblock \href{http://dx.doi.org/10.1080/01621459.1963.10500830}{Journal of the American Statistical Association {\bf 58}(301):\,13--30} (1963).

\bibitem{Holevo96}
A.~S. Holevo.
\newblock {\em ``The capacity of the quantum channel with general signal states''}.
\newblock \href{http://dx.doi.org/10.1109/18.651037}{IEEE Transactions on Information Theory {\bf 44}(1):\,269--273} (1998).

\bibitem{Holevo02}
A.~S. Holevo.
\newblock {\em ``On entanglement-assisted classical capacity''}.
\newblock \href{http://dx.doi.org/10.1063/1.1495877}{Journal of Mathematical Physics {\bf 43}(9):\,4326--4333} (2002).

\bibitem{Hol12}
A.~S. Holevo.
\newblock {\em ``Information capacity of a quantum observable''}.
\newblock \href{http://dx.doi.org/10.1134/s0032946012010012}{Problems of Information Transmission {\bf 48}(1):\,1--10} (2012).

\bibitem{Holevo13}
A.~S. Holevo.
\newblock {\em Quantum Systems, Channels, Information}.
\newblock \href{http://dx.doi.org/doi:10.1515/9783110273403}{De Gruyter} (2013).

\bibitem{HHHLT01}
M.~Horodecki, P.~Horodecki, R.~Horodecki, D.~W. Leung, and B.~M. Terhal.
\newblock {\em ``Classical Capacity of a Noiseless Quantum Channel Assisted by Noisy Entanglement''}.
\newblock \href{http://dx.doi.org/10.48550/ARXIV.QUANT-PH/0106080}{Quantum Info. Comput. {\bf 1}(3):\,70–78} (2001).

\bibitem{HDW08}
M.-H. Hsieh, I.~Devetak, and A.~Winter.
\newblock {\em ``Entanglement-Assisted Capacity of Quantum Multiple-Access Channels''}.
\newblock \href{http://dx.doi.org/10.1109/TIT.2008.924726}{IEEE Transactions on Information Theory {\bf 54}(7):\,3078--3090} (2008).

\bibitem{HLS18}
A.~Hu, J.~Li, and R.~Shapiro.
\newblock {\em ``Quantum Benchmarking on the [ [ 4 , 2 , 2 ] ] Code''}.
\newblock (2018).
\newblock Available online: \url{https://services.math.duke.edu/DOmath/DOmath2018/hu-li-shapiro.pdf}.

\bibitem{HRB08}
H.~Häffner, C.~Roos, and R.~Blatt.
\newblock {\em ``Quantum computing with trapped ions''}.
\newblock \href{http://dx.doi.org/10.1016/j.physrep.2008.09.003}{Physics Reports {\bf 469}(4):\,155--203} (2008).

\bibitem{IonQWebsiteAria}
{IonQ}.
\newblock {\em ``IonQ Aria - Specifications''}, (2023).
\newblock Available online: \url{https://ionq.com/quantum-systems/aria}.

\bibitem{IonQWebsiteHarm}
{IonQ}.
\newblock {\em ``IonQ Harmony - Specifications''}, (2023).
\newblock Available online: \url{https://ionq.com/quantum-systems/harmony}.

\bibitem{JLFCHZHZ22}
Y.~Jin, J.~Luo, L.~Fong, Y.~Chen, A.~B. Hayes, C.~Zhang, F.~Hua, and E.~Z. Zhang.
\newblock {\em ``A Structured Method for Compilation of QAOA Circuits in Quantum Computing''}, (2022).
\newblock \text{\href{http://arxiv.org/abs/2112.06143}{arXiv:\,2112.06143}}.

\bibitem{Johnson12}
S.~G. Johnson.
\newblock {\em ``Notes on the equivalence of norms''}.
\newblock Lecture notes, MIT course 18.335, (2012).
\newblock Available online: \url{https://math.mit.edu/~stevenj/18.335/norm-equivalence.pdf}.

\bibitem{JKP07}
N.~Johnston, D.~W. Kribs, and V.~I. Paulsen.
\newblock {\em ``Computing Stabilized Norms for Quantum Operations via the Theory of Completely Bounded Maps''}, (2007).
\newblock \text{\href{http://arxiv.org/abs/0711.3636}{arXiv:\,0711.3636}}.

\bibitem{King02}
C.~King.
\newblock {\em ``Additivity for unital qubit channels''}.
\newblock \href{http://dx.doi.org/10.1063/1.1500791}{Journal of Mathematical Physics {\bf 43}(10):\,4641} (2002).

\bibitem{Kitaev97}
A.~Y. Kitaev.
\newblock {\em ``Quantum computations: algorithms and error correction''}.
\newblock \href{http://dx.doi.org/10.1070/RM1997v052n06ABEH002155}{Russian Mathematical Surveys {\bf 52}(6):\,1191} (1997).

\bibitem{Kitaev03}
A.~Y. Kitaev.
\newblock {\em ``Fault-tolerant quantum computation by anyons''}.
\newblock \href{http://dx.doi.org/10.1016/s0003-4916(02)00018-0}{Annals of Physics {\bf 303}(1):\,2–30} (2003).

\bibitem{KL96}
E.~Knill and R.~Laflamme.
\newblock {\em ``Concatenated Quantum Codes''}.
\newblock \href{http://dx.doi.org/10.48550/ARXIV.QUANT-PH/9608012}{arXiv:quant-ph/9608012 } (1996).

\bibitem{KLZ98}
E.~Knill, R.~Laflamme, and W.~H. Zurek.
\newblock {\em ``Resilient quantum computation: error models and thresholds''}, (1998).
\newblock \text{\href{http://dx.doi.org/10.1098/rspa.1998.0166}{DOI:\,10.1098/rspa.1998.0166}}.

\bibitem{KW04}
D.~Kretschmann and R.~F. Werner.
\newblock {\em ``Tema con variazioni: quantum channel capacity''}.
\newblock \href{http://dx.doi.org/10.1088/1367-2630/6/1/026}{New Journal of Physics {\bf 6}:\,26--26} (2004).

\bibitem{LMPZ96}
R.~Laflamme, C.~Miquel, J.~P. Paz, and W.~H. Zurek.
\newblock {\em ``Perfect Quantum Error Correcting Code''}.
\newblock \href{http://dx.doi.org/10.1103/PhysRevLett.77.198}{Phys. Rev. Lett. {\bf 77}:\,198--201} (1996).

\bibitem{LSAA22}
L.~Lao, H.~van Someren, I.~Ashraf, and C.~G. Almudever.
\newblock {\em ``Timing and Resource-Aware Mapping of Quantum Circuits to Superconducting Processors''}.
\newblock \href{http://dx.doi.org/10.1109/TCAD.2021.3057583}{IEEE Transactions on Computer-Aided Design of Integrated Circuits and Systems {\bf 41}(2):\,359--371} (2022).

\bibitem{LKH22}
Z.~Li, I.~Kim, and P.~Hayden.
\newblock {\em ``Concatenation Schemes for Topological Fault-tolerant Quantum Error Correction''}, (2022).
\newblock \text{\href{http://arxiv.org/abs/2209.09390}{arXiv:\,2209.09390}}.

\bibitem{LR02}
E.~H. Lieb and M.~B. Ruskai.
\newblock {\em Proof of the strong subadditivity of quantum-mechanical entropy}, pages 63--66.
\newblock \href{http://dx.doi.org/10.1007/978-3-642-55925-9_6}{Springer Berlin Heidelberg} (2002).

\bibitem{LGLFDBM17}
N.~M. Linke, M.~Gutierrez, K.~A. Landsman, C.~Figgatt, S.~Debnath, K.~R. Brown, and C.~Monroe.
\newblock {\em ``Fault-tolerant quantum error detection''}.
\newblock \href{http://dx.doi.org/10.1126/sciadv.1701074}{Science Advances {\bf 3}(10)} (2017).

\bibitem{Lloyd97}
S.~Lloyd.
\newblock {\em ``Capacity of the noisy quantum channel''}.
\newblock \href{http://dx.doi.org/10.1103/physreva.55.1613}{Physical Review A {\bf 55}(3):\,1613--1622} (1997).

\bibitem{LP99}
H.-K. Lo and S.~Popescu.
\newblock {\em ``Classical Communication Cost of Entanglement Manipulation: Is Entanglement an Interconvertible Resource?''}.
\newblock \href{http://dx.doi.org/10.1103/physrevlett.83.1459}{Physical Review Letters {\bf 83}(7):\,1459--1462} (1999).

\bibitem{Maslov17}
D.~Maslov.
\newblock {\em ``Basic circuit compilation techniques for an ion-trap quantum machine''}.
\newblock \href{http://dx.doi.org/10.1088/1367-2630/aa5e47}{New Journal of Physics {\bf 19}(2):\,023035} (2017).

\bibitem{MGHHLPP14}
P.~Mazurek, A.~Grudka, M.~Horodecki, P.~Horodecki, J.~Łodyga, L.~Pankowski, and A.~Przysiężna.
\newblock {\em ``Long-distance quantum communication over noisy networks without long-time quantum memory''}.
\newblock \href{http://dx.doi.org/10.1103/physreva.90.062311}{Physical Review A {\bf 90}(6)} (2014).

\bibitem{MKT+00}
C.~Myatt, B.~King, Q.~Turchette, C.~Sackett, D.~Kielpinski, W.~Itano, C.~Monroe, and D.~Wineland.
\newblock {\em ``Decoherence of quantum superpositions through coupling to engineered reservoirs''}.
\newblock \href{http://dx.doi.org/10.1038/35002001}{Nature {\bf 403}:\,269--73} (2000).

\bibitem{NR07}
M.~Nathanson and M.~B. Ruskai.
\newblock {\em ``Pauli diagonal channels constant on axes''}.
\newblock \href{http://dx.doi.org/10.1088/1751-8113/40/28/S22}{Journal of Physics A: Mathematical and Theoretical {\bf 40}(28):\,8171} (2007).

\bibitem{Nicolaidis11}
M.~Nicolaidis.
\newblock {\em Soft Errors in Modern Electronic Systems}.
\newblock \href{http://dx.doi.org/https://doi.org/10.1007/978-1-4419-6993-4}{Springer, Boston, MA} (2011).

\bibitem{NC00}
M.~A. Nielsen and I.~L. Chuang.
\newblock {\em Quantum Computation and Quantum Information: 10th Anniversary Edition}.
\newblock \href{http://dx.doi.org/10.1017/CBO9780511976667}{Cambridge University Press} (2010).

\bibitem{OPH+16}
N.~Ofek, A.~Petrenko, R.~Heeres, P.~Reinhold, Z.~Leghtas, B.~Vlastakis, Y.~Liu, L.~Frunzio, S.~M. Girvin, L.~Jiang, M.~Mirrahimi, M.~H. Devoret, and R.~J. Schoelkopf.
\newblock {\em ``Demonstrating Quantum Error Correction that Extends the Lifetime of Quantum Information''}, (2016).
\newblock \text{\href{http://arxiv.org/abs/1602.04768}{arXiv:\,1602.04768}}.

\bibitem{PR13}
A.~Paetznick and B.~W. Reichardt.
\newblock {\em ``Fault-tolerant ancilla preparation and noise threshold lower bounds for the 23-qubit Golay code''}, (2013).
\newblock \text{\href{http://arxiv.org/abs/1106.2190}{arXiv:\,1106.2190}}.

\bibitem{PK22}
A.~Poshtvan and V.~Karimipour.
\newblock {\em ``Capacities of the covariant Pauli channel''}.
\newblock \href{http://dx.doi.org/10.1103/physreva.106.062408}{Physical Review A {\bf 106}(6)} (2022).

\bibitem{Preskill97}
J.~Preskill.
\newblock {\em ``Fault-tolerant quantum computation''}, (1997).
\newblock \text{\href{http://arxiv.org/abs/quant-ph/9712048}{arXiv:\,quant-ph/9712048}}.

\bibitem{Preskill18}
J.~Preskill.
\newblock {\em ``Quantum Computing in the {NISQ} era and beyond''}.
\newblock \href{http://dx.doi.org/10.22331/q-2018-08-06-79}{Quantum {\bf 2}:\,79} (2018).

\bibitem{Rouze19}
C.~Rouzé.
\newblock {\em ``Functional inequalities in quantum information theory''}.
\newblock \href{http://dx.doi.org/10.17863/CAM.38295}{Apollo - University of Cambridge Repository } (2019).

\bibitem{RSW01}
M.~B. Ruskai, S.~Szarek, and E.~Werner.
\newblock {\em ``An Analysis of Completely-Positive Trace-Preserving Maps on 2x2 Matrices''}, (2001).
\newblock \text{\href{http://arxiv.org/abs/quant-ph/0101003}{arXiv:\,quant-ph/0101003}}.

\bibitem{Sason15}
I.~Sason.
\newblock {\em ``On Reverse Pinsker Inequalities''}, (2015).
\newblock \text{\href{http://arxiv.org/abs/1503.07118}{arXiv:\,1503.07118}}.

\bibitem{Schlosshauer19}
M.~Schlosshauer.
\newblock {\em ``Quantum decoherence''}.
\newblock \href{http://dx.doi.org/10.1016/j.physrep.2019.10.001}{Physics Reports {\bf 831}:\,1--57} (2019).

\bibitem{Schumacher95}
B.~Schumacher.
\newblock {\em ``Quantum coding''}.
\newblock \href{http://dx.doi.org/10.1103/PhysRevA.51.2738}{Phys. Rev. A {\bf 51}:\,2738--2747} (1995).

\bibitem{SW97}
B.~{Schumacher} and M.~D. {Westmoreland}.
\newblock {\em ``{Sending classical information via noisy quantum channels}''}.
\newblock \href{http://dx.doi.org/10.1103/PhysRevA.56.131}{Physical Review A {\bf 56}(1):\,131--138} (1997).

\bibitem{SW01}
B.~Schumacher and M.~D. Westmoreland.
\newblock {\em ``Optimal signal ensembles''}.
\newblock \href{http://dx.doi.org/10.1103/PhysRevA.63.022308}{Phys. Rev. A {\bf 63}:\,022308} (2001).

\bibitem{Shannon48}
C.~E. Shannon.
\newblock {\em ``A mathematical theory of communication''}.
\newblock \href{http://dx.doi.org/10.1002/j.1538-7305.1948.tb01338.x}{The Bell System Technical Journal {\bf 27}(3):\,379--423} (1948).

\bibitem{Shi12}
M.~E. Shirokov.
\newblock {\em ``Conditions for Coincidence of the Classical Capacity and Entanglement-Assisted Capacity of a Quantum Channel''}.
\newblock \href{http://dx.doi.org/10.1134/S0032946012020019}{Probl. Inf. Transm. {\bf 48}(2):\,85–101} (2012).

\bibitem{Shirokov17}
M.~E. Shirokov.
\newblock {\em ``Tight uniform continuity bounds for the quantum conditional mutual information, for the {Holevo} quantity, and for capacities of quantum channels''}.
\newblock \href{http://dx.doi.org/10.1063/1.4987135}{Journal of Mathematical Physics {\bf 58}(10):\,102202} (2017).

\bibitem{SPS21}
A.~Shlosberg, A.~M. Polloreno, and G.~Smith.
\newblock {\em ``Towards Demonstrating Fault Tolerance in Small Circuits Using Bacon-Shor Codes''}, (2021).
\newblock \text{\href{http://arxiv.org/abs/2108.02079}{arXiv:\,2108.02079}}.

\bibitem{Shor95}
P.~W. Shor.
\newblock {\em ``Scheme for reducing decoherence in quantum computer memory''}.
\newblock \href{http://dx.doi.org/10.1103/PhysRevA.52.R2493}{Phys. Rev. A {\bf 52}:\,R2493--R2496} (1995).

\bibitem{Shor02_2}
P.~W. Shor.
\newblock {\em ``Additivity of the classical capacity of entanglement-breaking quantum channels''}.
\newblock \href{http://dx.doi.org/10.1063/1.1498000}{Journal of Mathematical Physics {\bf 43}(9):\,4334--4340} (2002).

\bibitem{Shor02}
P.~W. Shor.
\newblock {\em ``The quantum channel capacity and coherent information''}.
\newblock Lecture notes, MSRI Workshop on Quantum Computation, (2002).
\newblock Available online: \url{www.msri.org/programs/53}.

\bibitem{Siudzinska20}
K.~Siudzińska.
\newblock {\em ``Classical capacity of generalized Pauli channels''}.
\newblock \href{http://dx.doi.org/10.1088/1751-8121/abb276}{Journal of Physics A: Mathematical and Theoretical {\bf 53}(44):\,445301} (2020).

\bibitem{SM99}
A.~S\o{}rensen and K.~M\o{}lmer.
\newblock {\em ``Quantum Computation with Ions in Thermal Motion''}.
\newblock \href{http://dx.doi.org/10.1103/PhysRevLett.82.1971}{Phys. Rev. Lett. {\bf 82}:\,1971--1974} (1999).

\bibitem{Steane96}
A.~Steane.
\newblock {\em ``Multiple-particle interference and quantum error correction''}.
\newblock \href{http://dx.doi.org/10.1098/rspa.1996.0136}{Proceedings of the Royal Society of London. Series A: Mathematical, Physical and Engineering Sciences {\bf 452}(1954):\,2551–2577} (1996).

\bibitem{Steane06}
A.~M. Steane.
\newblock {\em ``A tutorial on quantum error correction''}.
\newblock \href{http://dx.doi.org/10.3254/978-1-61499-018-5-1}{Proceedings of the International School of Physics ``Enrico Fermi" - "Quantum Computers, Algorithms and Chaos" pages 1--31} (2006).

\bibitem{TCCCG17}
M.~Takita, A.~W. Cross, A.~C{\'{o}}rcoles, J.~M. Chow, and J.~M. Gambetta.
\newblock {\em ``Experimental Demonstration of Fault-Tolerant State Preparation with Superconducting Qubits''}.
\newblock \href{http://dx.doi.org/10.1103/physrevlett.119.180501}{Physical Review Letters {\bf 119}(18)} (2017).

\bibitem{Terhal23}
B.~Terhal.
\newblock {\em ``Fault tolerance''}.
\newblock QIP Tutorial, (2023).
\newblock Available online: \url{https://www.youtube.com/watch?v=Je7sVJGKMgU}.

\bibitem{ErrorCorrectionZoo}
{The Zoo Team}.
\newblock {\em ``The Error Correction Zoo''}, (2023).
\newblock Available online: \url{https://errorcorrectionzoo.org/}.

\bibitem{TT15}
M.~Tomamichel and V.~Y.~F. Tan.
\newblock {\em ``Second-Order Asymptotics for the Classical Capacity of Image-Additive Quantum Channels''}.
\newblock \href{http://dx.doi.org/10.1007/s00220-015-2382-0}{Communications in Mathematical Physics {\bf 338}(1):\,103--137} (2015).

\bibitem{TMK+00}
Q.~A. Turchette, C.~J. Myatt, B.~E. King, C.~A. Sackett, D.~Kielpinski, W.~M. Itano, C.~Monroe, and D.~J. Wineland.
\newblock {\em ``Decoherence and decay of motional quantum states of a trapped atom coupled to engineered reservoirs''}.
\newblock \href{http://dx.doi.org/10.1103/PhysRevA.62.053807}{Phys. Rev. A {\bf 62}:\,053807} (2000).

\bibitem{Uhlmann76}
A.~Uhlmann.
\newblock {\em ``The “transition probability” in the state space of a $*$-algebra''}.
\newblock \href{http://dx.doi.org/https://doi.org/10.1016/0034-4877(76)90060-4}{Reports on Mathematical Physics {\bf 9}(2):\,273--279} (1976).

\bibitem{VGW96}
L.~Vaidman, L.~Goldenberg, and S.~Wiesner.
\newblock {\em ``Error prevention scheme with four particles''}.
\newblock \href{http://dx.doi.org/10.1103/physreva.54.r1745}{Physical Review A {\bf 54}(3):\,R1745--R1748} (1996).

\bibitem{VDRF17}
D.~Venturelli, M.~Do, E.~Rieffel, and J.~Frank.
\newblock {\em ``Temporal Planning for Compilation of Quantum Approximate Optimization Circuits''}.
\newblock In {\em Proceedings of the Twenty-Sixth International Joint Conference on Artificial Intelligence, {IJCAI-17}}, pages 4440--4446, (2017), \text{\href{http://dx.doi.org/10.24963/ijcai.2017/620}{DOI:\,10.24963/ijcai.2017/620}}.

\bibitem{VAC+02}
D.~Vion, A.~Aassime, A.~Cottet, P.~Joyez, H.~Pothier, C.~Urbina, D.~Esteve, and M.~H. Devoret.
\newblock {\em ``Manipulating the Quantum State of an Electrical Circuit''}.
\newblock \href{http://dx.doi.org/10.1126/science.1069372}{Science {\bf 296}(5569):\,886--889} (2002).

\bibitem{Vuillot18}
C.~Vuillot.
\newblock {\em ``Is error detection helpful on IBM 5Q chips ?''}, (2018).
\newblock \text{\href{http://dx.doi.org/10.26421/qic18.11-12}{DOI:\,10.26421/qic18.11-12}}.

\bibitem{Watrous04}
J.~Watrous.
\newblock {\em ``Notes on super-operator norms induced by Schatten norms''}, (2004).
\newblock \text{\href{http://arxiv.org/abs/quant-ph/0411077}{arXiv:\,quant-ph/0411077}}.

\bibitem{Watrous}
J.~Watrous.
\newblock {\em ``The Theory of Quantum Information''}, (2018).
\newblock Available online: \url{https://cs.uwaterloo.ca/~watrous/TQI/TQI.pdf}.

\bibitem{Wilde13}
M.~M. Wilde.
\newblock {\em Quantum Information Theory}.
\newblock \href{http://dx.doi.org/10.1017/CBO9781139525343}{Cambridge University Press} (2013).

\bibitem{WNJRM17}
D.~Willsch, M.~Nocon, F.~Jin, H.~D. Raedt, and K.~Michielsen.
\newblock {\em ``Gate-error analysis in simulations of quantum computers with transmon qubits''}.
\newblock \href{http://dx.doi.org/10.1103/physreva.96.062302}{Physical Review A {\bf 96}(6)} (2017).

\bibitem{WWJdRM18}
D.~Willsch, M.~Willsch, F.~Jin, H.~D. Raedt, and K.~Michielsen.
\newblock {\em ``Testing quantum fault tolerance on small systems''}.
\newblock \href{http://dx.doi.org/10.1103/physreva.98.052348}{Physical Review A {\bf 98}(5)} (2018).

\bibitem{WWDHOE22}
A.~Winick, J.~J. Wallman, D.~Dahlen, I.~Hincks, E.~Ospadov, and J.~Emerson.
\newblock {\em ``Concepts and conditions for error suppression through randomized compiling''}, (2022).
\newblock \text{\href{http://arxiv.org/abs/2212.07500}{arXiv:\,2212.07500}}.

\bibitem{BenchmarkIonQ19}
K.~Wright, K.~M. Beck, Debnath, et~al.
\newblock {\em ``Benchmarking an 11-qubit quantum computer''}.
\newblock \href{http://dx.doi.org/10.1038/s41467-019-13534-2}{Nature Communications {\bf 10}(1)} (2019).

\bibitem{YK22}
H.~Yamasaki and M.~Koashi.
\newblock {\em ``Time-Efficient Constant-Space-Overhead Fault-Tolerant Quantum Computation''}, (2022).
\newblock \text{\href{http://arxiv.org/abs/2207.08826}{arXiv:\,2207.08826}}.

\bibitem{Zeh70}
H.~D. Zeh.
\newblock {\em ``{On the interpretation of measurement in quantum theory}''}.
\newblock \href{http://dx.doi.org/10.1007/BF00708656}{Found. Phys. {\bf 1}:\,69--76} (1970).

\bibitem{ZLZCLZK20}
S.~Zhang, Y.~Lu, K.~Zhang, W.~Chen, Y.~Li, J.-N. Zhang, and K.~Kim.
\newblock {\em ``Error-mitigated quantum gates exceeding physical fidelities in a trapped-ion system''}.
\newblock \href{http://dx.doi.org/10.1038/s41467-020-14376-z}{Nature Communications {\bf 11}(1)} (2020).

\bibitem{ZLB18}
Y.-C. Zheng, C.-Y. Lai, and T.~A. Brun.
\newblock {\em ``Efficient preparation of large-block-code ancilla states for fault-tolerant quantum computation''}.
\newblock \href{http://dx.doi.org/10.1103/physreva.97.032331}{Physical Review A {\bf 97}(3)} (2018).

\bibitem{ZZYS17}
Q.~Zhuang, E.~Y. Zhu, and P.~W. Shor.
\newblock {\em ``Additive Classical Capacity of Quantum Channels Assisted by Noisy Entanglement''}.
\newblock \href{http://dx.doi.org/10.1103/PhysRevLett.118.200503}{Phys. Rev. Lett. {\bf 118}:\,200503} (2017).

\end{thebibliography}

\newpage\chapter*{Acknowledgements}
\addcontentsline{toc}{chapter}{Acknowledgements}
\thispagestyle{empty}

The final part of this thesis concerns itself with how to most efficiently and effectively communicate how incredibly grateful I am for all the support, advice and opportunities that have been offered to me along the way.

First and foremost, thanks to my supervisor Matthias Christandl for giving me the opportunity to pursue a PhD in my favorite field, and for his enthusiasm and insights.

Many thanks to my scientific collaborator Alexander Müller-Hermes for many discussions, riddles and mentorship. Thanks also to Christoph Hirche, Paweł Mazurek and Marco Túlio Quintino for many interesting discussions.

I would also like to thank all current and former members of QMATH for many, many, many great coffee breaks. Thanks also to my fellow men and women without qualities.

Thank you to Marco Tomamichel for hosting me at CQT at NUS and a great collaboration. Thanks to the members of the Tomamichel group, who showed me many sides of Singapore and science, and how fun it is to write on the wall. 

Thank you to the commitee members Laura Mancinska, Omar Fawzi and Mario Berta for taking the time to read this thesis. Thanks also to the invaluable effort of the proofreaders Ashutosh Goswami, Freek Witteveen and Wolfgang Belzig for spotting many errors and giving many great suggestions! Mange tak til Sabiha og Emil for at oversætte mit resumé.

Thank you to Carsten Speckmann for understanding me better than any other human. Thank you to Mama, Papa, Josi, Oma Irene, Opa Herbert, Oma Gertrud and Opa Manfred for a whole lifetime of unconditional support. Thank you to my lovely friends all over the world.

\vspace{1cm}
This thesis was supported by the European Union’s Horizon 2020 research
and innovation programme under the Marie Skłodowska-Curie TALENT Doctoral fellowship (grant no. 801199). 


\newpage
\null

\thispagestyle{empty}
\end{document}